\documentclass[10.5pt,a4paper]{article}
\usepackage[utf8]{inputenc}
\usepackage{fullpage}
\usepackage{amsmath,amsthm,amsfonts,amssymb,graphicx,color}
\usepackage{mathrsfs}
\usepackage{quantikz}
\usepackage{cite}
\usepackage{faktor}
\usepackage{dsfont,soul}
\usepackage{cancel}
\usepackage{setspace,float}
\usepackage{xltabular}
\usepackage[all]{xy}
\usepackage{authblk}
\usepackage{multicol}
\usepackage{adjustbox}
\usepackage[labelfont=bf]{caption}
\usepackage{tocloft}

\usepackage[unicode=true,bookmarks=true,
bookmarksnumbered=false,bookmarksopen=false,
breaklinks=false,
pdfborder={0 0 1},
backref=false,colorlinks=true]{hyperref}

\usepackage[left=2cm,right=2cm,top=2.5cm,bottom=2.5cm]{geometry}

\hypersetup{pdftitle={A Correspondence Between Quantum Error Correcting Codes and Quantum Reference Frames}, pdfauthor={Sylvain Carrozza, Aidan Chatwin-Davies, Philipp A. Hohn, Fabio M. Mele}, citecolor=blue,linkcolor=blue,urlcolor=blue,citecolor=blue}

\newtheorem{thm}{Theorem}[section]
\newtheorem{prop}[thm]{Proposition}
\newtheorem{lem}[thm]{Lemma}
\newtheorem{cor}[thm]{Corollary}
\newtheorem{defn}[thm]{Definition}

\makeatletter
\def\th@definition{%
  \thm@notefont{}
  \normalfont 
}
\makeatother
\theoremstyle{definition}
\newtheorem{example}[thm]{Example}
\newcommand{\eox}{\hfill$\triangle$}

\newcommand{\cH}{\mathcal{H}}

\newcommand{\cF}{\mathcal{F}}
\newcommand{\cB}{\mathcal{B}}
\newcommand{\cV}{\mathcal{V}}
\newcommand{\cE}{\mathcal{E}}
\newcommand{\cP}{\mathcal{P}}
\newcommand{\cZ}{\mathcal{Z}}
\newcommand{\uU}{\underline{U}}

\newcommand\restr[2]{{
  \left.\kern-\nulldelimiterspace 
  #1 
  \vphantom{\big|} 
  \right|_{#2} 
  }}

\newcommand{\Sec}[1]{Sec.~\ref{#1}}

\newcommand{\App}[1]{App.~\ref{#1}}
\newcommand{\mbf}[1]{\mathbf{#1}}
\newcommand{\mrm}[1]{\mathrm{#1}}
\newcommand{\Hil}{\mathcal{H}}
\newcommand{\Oh}{\mathcal{O}}

\newcommand{\dd}{\mathrm{d}}

\newcommand{\ketbra}[2]{|#1\rangle \langle #2 |}

\DeclareMathOperator{\Tr}{Tr}
\DeclareMathOperator{\Span}{span}

\usepackage{xcolor}

\definecolor{purple}{rgb}{0.8,0,0.8}

\definecolor{darkcoral}{rgb}{0.8, 0.36, 0.27}

\title{\textbf{A Correspondence Between Quantum Error Correcting Codes and Quantum Reference Frames}}
\author[1]{Sylvain Carrozza\thanks{sylvain.carrozza@ube.fr}}
\author[2,3]{Aidan Chatwin-Davies\thanks{aidan.chatwindavies@uri.edu}}
\author[3]{Philipp A.\ H\"ohn\thanks{philipp.hoehn@oist.jp}}
\author[4,5]{Fabio M.\ Mele\thanks{fmele1@lsu.edu}}
\affil[1]{\normalsize{IMB UMR 5584, Universit\'{e} Bourgogne Europe, CNRS, F-21000 Dijon, France}}
\affil[2]{\normalsize{Department of Physics, University of Rhode Island, Kingston, RI 02881, United States of America}}
\affil[3]{\normalsize{Okinawa Institute of Science and Technology Graduate University, Onna, Okinawa 904-0495, Japan}}
\affil[4]{\normalsize{Department of Physics \& Astronomy, Western University, London, ON, N6A 3K7, Canada}}
\affil[5]{\normalsize{Department of Physics and Astronomy, Louisiana State University, Baton Rouge, LA 70803, USA}}
\date{\today}

\begin{document}

\maketitle

\begin{abstract}
In a gauge theory, a collection of kinematical degrees of freedom is used to redundantly describe a smaller
amount of gauge-invariant information. In a quantum error correcting code (QECC), a collection of computational
degrees of freedom that make up a device's physical layer is used to redundantly encode a smaller amount
of logical information. We elaborate this parallel in terms of quantum reference frames (QRFs), which are a
universal toolkit for dealing with symmetries in quantum systems and which define the gauge theory analog of encodings. The result is a precise dictionary between
QECCs and QRF setups within the perspective-neutral framework for gauge systems. Concepts from QECCs like error sets and correctability translate to
novel insights into the informational architecture of gauge theories. Conversely, the dictionary provides a
systematic procedure for constructing symmetry-based QECCs and characterizing
their error correcting properties. In this initial work, we scrutinize the dictionary between Pauli stabilizer codes
and their corresponding QRF setups. We show that there is a one-to-one correspondence between maximal correctable error sets and tensor factorizations splitting system from error-generated QRF degrees of freedom. Relative to this split, errors corrupt only redundant frame data, leading to a novel characterization of correctability. When passed through the dictionary,
standard Pauli errors behave as electric excitations that are dual, via Pontryagin duality, to magnetic excitations related to gauge-fixing. This gives rise to a new class of correctable errors and a systematic error duality. We illustrate our findings in surface codes, which themselves manifestly connect quantum error correction with gauge systems. Our
exploratory investigations pave the way for deeper foundational applications to gauge
theories and for eventual practical applications to quantum simulation.
\end{abstract}

\clearpage

\setlength{\cftbeforesecskip}{7pt}

\tableofcontents

\section{Introduction}

Gauge theories are ubiquitous in nature; they capture the physics of electrodynamics, nuclear forces, gravity, and more.
One of their defining features is that they describe physics redundantly in a way that is insensitive to certain local details.
This is similar to how quantum error correcting codes (QECCs) protect quantum information from local errors by redundantly encoding logical states into a larger physical space.
Quantum error correction (QEC) is thereby an essential ingredient for establishing fault-tolerant quantum computation in the omnipresence of quantum noise,  which corrupts stored information and which hinders gate implementation and state preparation if left unchecked.

Without redundantly encoding logical data into a quantum register, \emph{any} uncontrolled interactions of this register with its environment or itself can result in a logical error in the computation.~Without redundancy among kinematical degrees of freedom, a gauge symmetry cannot exist.~Which degrees of freedom are identified as redundant is in both contexts a choice and not specified by the theory.~In QEC this is the reason that a given code can, in principle, correct multiple different sets of errors, and in gauge theories this corresponds to the freedom to gauge fix in a continuum of different ways.~Which error set is practically interesting depends on the error model and physical implementation of the code, and which gauge is useful depends on the preference of the practitioner.~Moreover, in stabilizer codes, the logical information to be protected from errors is encoded so as to commute with the stabilizer group, just like how physical information in gauge theories commutes with gauge symmetries. 

There is thus a clear analogy between QEC and gauge theories and one may inquire whether this is merely a coincidence or whether there is, in fact, a deeper underlying structural relationship.~Can we in general understand gauge theories and the distribution of gauge-invariant quantum information in spacetime as a QECC and, conversely, can we meaningfully think of a QECC as a gauge theory?~If so, what are the foundational and practical insights to be gained from such a correspondence?~Our aim with this work is to start setting the stage for answering such questions.  

To be sure, QECCs have been usefully invoked in high-energy physics and particularly in quantum gravity before.
For example, quantum codes have been employed to explore how long it would take to decode information from Hawking radiation about objects that have fallen into a black hole \cite{Hayden:2007cs,Yoshida:2017non}.
Similarly, holography \cite{Maldacena:1997re,Witten:1998qj} has been shown to have a natural interpretation in terms of erasure codes \cite{Almheiri:2014lwa,Harlow:2016vwg,Faulkner:2020hzi}, and QEC also gives rise to a new angle on renormalization and splitting relevant from irrelevant information \cite{Furuya:2020tzv,Furuya:2021lgx}.
These applications have also had benefits for QEC, influencing code design and leading to new types of codes and tensor networks, such as, e.g., the HaPPY code \cite{Pastawski:2015qua}, hyperinvariant MERA \cite{Evenbly:2017hyg,Cao:2021wrb,Steinberg:2023wll,Steinberg:2024ack}, Quantum Lego \cite{Cao:2021ibt,Fan:2024ero}, and more.
By contrast, our aim with this work will not be specific applications of QEC in high-energy physics, but rather developing a fundamental correspondence between gauge theories and QEC.

Going beyond an analogy between QEC and gauge theories requires making sense of key ingredients of QEC---namely, encodings and correctable errors---on the gauge theory side, where such concepts are not usually considered.
This is where quantum reference frames (QRFs) come into the picture \cite{Aharonov:1967zza,Aharonov:1984,Bartlett:2006tzx,angeloPhysicsQuantumReference2011a,Giacomini:2017zju,delaHamette:2021oex,Hoehn:2023ehz,Hoehn:2019fsy,Hoehn:2020epv,Chataignier:2024eil,AliAhmad:2021adn,Hoehn:2021flk,delaHamette:2021piz,Vanrietvelde:2018pgb,Vanrietvelde:2018dit,Giacomini:2021gei,Castro-Ruiz:2019nnl,Suleymanov:2023wio,delaHamette:2020dyi,Krumm:2020fws,Kabel:2024lzr,Kabel:2023jve,Carette:2023wpz,loveridgeSymmetryReferenceFrames2018a,Giacomini:2018gxh,Castro-Ruiz:2021vnq}. 
They constitute a universal toolset that applies to any quantum system subject to a symmetry principle. 
More precisely, the correspondence established in this work relies on QRFs according to the \emph{perspective-neutral} (PN) framework, which is distinguished by being built using the structures of gauge theories, thus making it directly applicable to the latter \cite{delaHamette:2021oex,Hoehn:2023ehz,Hoehn:2019fsy,Hoehn:2020epv,Chataignier:2024eil,AliAhmad:2021adn,Hoehn:2021flk,delaHamette:2021piz,Vanrietvelde:2018pgb,Vanrietvelde:2018dit,Giacomini:2021gei,Castro-Ruiz:2019nnl,Suleymanov:2023wio}.
While the terminology of QRFs first arose in the foundations of quantum theory and quantum information \cite{Aharonov:1967zza,Aharonov:1984,angeloPhysicsQuantumReference2011a}, the concept is pervasive in high-energy physics and arises whenever constructing observables in gauge systems, either via a gauge-invariant dressing method, or, equivalently, a gauge-fixing procedure, e.g.\ see \cite{Carrozza:2021gju,Goeller:2022rsx,Carrozza:2022xut,Araujo-Regado:2024dpr,Araujo-Regado:2025ejs}.
In short, the field degrees of freedom that become frozen in a gauge-fixing procedure, or that are used to ``dress'' non-invariant ``bare'' quantities into invariant ones, constitute a QRF for the pertinent gauge group.
Now, crucially, according to the PN-framework, a choice of QRF in a gauge system is a choice of split between redundant (the QRF itself) and non-redundant information.
On the QEC side, this corresponds to a decoding and is deeply connected with sets of correctable errors.
This suggests a dictionary between QEC and gauge theories, using QRFs as a mediator, and indeed this is what we find.
The dictionary is summarized in Table~\ref{tab_dictionary} and is explained in detail in the main text.
In this article, we shall restrict our attention for simplicity and explicitness to  Pauli stabilizer codes \cite{Gottesman:1997zz}, while we will generalize the correspondence to general group-based codes and QRFs associated with general symmetry groups in upcoming work \cite{CCDHM2}.

Before proceeding, let us briefly contrast our work with previous discussions invoking QRFs in the context of QEC. Their primary aim was to deal with continuous unitary symmetries in the logical space and particularly the Eastin-Knill theorem \cite{Eastin:2009tem}, which says that a continuous group of transversal operators in finite-dimensional codes is not universal. This leaves two broad ways out for universal transversal gate sets: approximate QEC in finite dimensions or infinite-dimensional codes (which is natural in field theories, but not in practical implementations). QRFs have been invoked to address this challenge \cite{Hayden:2017jjm,Woods:2019fpy,Faist:2019ahr,Yang:2020zqm}, however, with several key differences to the present manuscript. First, the QRFs in \cite{Hayden:2017jjm,Woods:2019fpy,Faist:2019ahr,Yang:2020zqm} are external and appended to the code as ancilla systems, i.e.\ to the space of physical qubits implementing the code, whereas here we extract the QRFs from the physical space itself so that they are internal. Second, in \cite{Hayden:2017jjm,Woods:2019fpy,Faist:2019ahr,Yang:2020zqm} the QRFs are associated with a symmetry group acting on the logical space and they are treated with the distinct \emph{quantum information approach} to QRFs \cite{Bartlett:2006tzx} appropriate for physical symmetries. By contrast, our QRFs are associated with the stabilizer group and accordingly with the PN-approach appropriate for gauge symmetries. Third, the QRFs in \cite{Hayden:2017jjm,Woods:2019fpy,Faist:2019ahr,Yang:2020zqm} are invoked to set up a suitable formulation of \emph{approximate} QEC to sidestep the Eastin-Knill theorem, while we focus here on \emph{exact} codes. The scope and setting of our work are thus quite different, and we emphasize that the QECC/QRF correspondence to be established is by construction only feasible with the PN-framework.

What may be the use of this correspondence?~On the practical side, we anticipate that it may provide a novel avenue for improving \emph{quantum simulations} of gauge theories and quantum gravity, a topic that has recently garnered a lot of attention \cite{Banuls:2019bmf,Martinez:2016yna,Kokail:2018eiw,Yang:2020yer,Bauer:2022hpo,DiMeglio:2023nsa,Chakraborty:2020uhf,Honda:2021aum,Brown:2019hmk,Gharibyan:2020bab,Shapoval:2022xeo,Li:2017gvt,Li:2019holentropysim,Maceda:2024rrd,Darbha:2024cwk,Darbha:2024srr}.~Indeed, this may provide novel code designs for implementing error correction in these applications by aligning the code space with the space of gauge-invariant states.~Furthermore, QRFs offer a systematic toolkit for dealing with symmetries in quantum systems and this may well lead to novel QECC designs, based on more general groups and exploiting properties of QRFs.~We shall comment on these aspects in the conclusions.~Conversely, we anticipate that this structural dictionary may also lead to novel foundational insights into both QECCs, as well as gauge theories, e.g., in the latter case better understanding the encoding of gauge-invariant data in terms of quantum codes. Indeed, our exposition will already provide a few examples of new foundational insights gained via the dictionary.

The organization of this article is as follows.~In Sec.~\ref{sec_dictionary}, we give brief introductions to QEC and to the PN-approach to QRFs, while emphasizing their similarities.~We then provide a broad summary of the results of the paper, describing in particular the precise dictionary underlying the correspondence between QECCs and QRF setups for general compact stabilizer/gauge groups. 

In Sec.~\ref{sec_QRF2QEC}, we explain in more detail in which sense a QRF setup associated to a finite Abelian symmetry can be reinterpreted, via the PN-approach, as a code. In this section, for later use, we also extend the previous QRF-literature by permitting the QRFs and complementary systems to carry a projective (rather than proper) unitary representation of the pertinent group. 

From Sec.~\ref{sec_genstab} onward we focus on the converse direction, in the specific context of Pauli stabilizer codes: namely, we explain how any stabilizer code can be reinterpreted and analyzed in terms of the mathematical structures of the PN-approach to QRFs. Central to our analysis is a one-to-one correspondence between: 
\begin{itemize}
\item[(a)] certain tensor factorizations of the physical space of the code (kinematical Hilbert space in the language of gauge theories) into a QRF factor and complementary system factor; and 
\item[(b)] maximal correctable error sets (see Theorem~\ref{thm_errorQRF}) or, equivalently, complete sets of frame fields (see Theorem~\ref{thm:nonlocal_fact}). 
\end{itemize}
In particular, Theorem~\ref{thm_errorQRF} shows that \emph{the code restriction of every maximal set of correctable errors generates a QRF and the stabilizer group}.~Such QRFs are special and entirely absorb the code's redundancy in the sense that both the code-restricted errors and all stabilizer transformations act trivially on the subsystem complementary to it, which thus is plainly logical.~This split between fully redundant and non-redundant data gives rise to a novel interpretation of correctability of errors, making the relation with redundancy mentioned in the opening paragraphs unambiguously clear.~We show that \emph{a set of Pauli errors is correctable} if and only if such a QRF exists on which the errors act, or in other words, \emph{if and only if a tensor factorization can be found such that the stabilizer transformations and code-restricted errors act exclusively on a purely redundant tensor factor} (Corollary~\ref{cor:KL-via-QRFs}).~QRFs thereby afford a new perspective on the relation between correctability and redundancy, or, in other words, provide another way to make the general idea of a QECC as a collection of virtual ancilla qubits and virtual logical qubits precise \cite{PreskillNotes,Poulin:2005wry,Kao:2023ehh}. We will come back to this in Sec.~\ref{ssec_redundancy}.

Such factorizations are in general non-local relative to the tensor product structure canonically associated to the code's physical space. However, we show in Theorem~\ref{thm:existence_of_local_qrfs} that it always remains possible to resort to a local QRF in a stabilizer code (which may carry a projective representation of the stabilizer group). 

Sec.~\ref{sec_errorduality} reveals the existence of and formalizes a new notion of error duality. It relates ``electric'' to ``magnetic'' families of correctable errors, where standard Pauli errors are part of the ``electric'' family, by means of an operator-valued Fourier transform exploiting the mathematical notion of Pontryagin duality. All properties of such error sets, including syndromes and recovery operations, are dual under the exchange of ``electric'' and ``magnetic'' charge labels. This observation may also be viewed as a new foundational insight into the structure of QEC, coming from a QRF and gauge theory perspective.

Sec.~\ref{sec:surfacecodes} illustrates some of our main findings in the context of surface codes (which include Kitaev's toric code) with $\mathbb{Z}_2$ structure group. The complementary interpretations of such models, as both stabilizer codes and lattice gauge theories, makes them particularly relevant to our analysis. 

We conclude the paper with a technical discussion connecting our results to earlier works on the relation between correctability and redundancy of errors, as well as on operator algebra QEC (Sec.~\ref{sec:discussion}), and a general conclusion outlining future research directions (Sec.~\ref{sec:conclusion}). Finally, extensive technical discussions, including reviews of relevant mathematical tools and proofs, can be found in the Appendices.

\begin{xltabular}[t!]{\textwidth}{|X c|c X|}
\captionsetup{labelfont=bf}
\caption{The correspondence between QEC for $[[n,k]]$ stabilizer codes and the perspective-neutral framework for internal QRFs associated with the gauge group $G\simeq\mathbb{Z}_2^{\times(n-k)}$.~The detailed correspondence is explained in the main body text.~Every stabilizer code with stabilizer group $G$ gives rise to a multitude of \emph{ideal} QRFs (i.e., that possess perfectly distinguishable orientation states) and these are in one-to-one correspondence with --- and in fact, generated by --- maximal sets of correctable errors.
For notational simplicity, we opted here to only represent the Page-Wootters reduction $\mathcal{R}_R$ as an analog of decoding on the QRF side. There is a unitarily equivalent and dual trivialization map satisfying essentially the same properties, as we will explain in the main body. $\hat{G}$ denotes the Pontryagin dual of $G$.}
\label{tab_dictionary} \\

\multicolumn{4}{c}{\textbf{QECC/QRF Dictionary}}\\ \hline

\multicolumn{2}{|c|}{\textbf{Stabilizer QECCs}} & \multicolumn{2}{c|}{\textbf{Ideal QRFs}}\\ \hline
\endfirsthead

\multicolumn{4}{c}%
{\tablename\ \thetable{} -- continued from previous page} \\
\hline \multicolumn{2}{|c|}{\textbf{Stabilizer QECCs}} & \multicolumn{2}{c|}{\textbf{Ideal QRFs}}\\ \hline
&&&\\[-0.65em]
\endhead

\hline \multicolumn{4}{|r|}{{Continued on next page}} \\ \hline
\endfoot

\hline
\endlastfoot

&&&\\[-0.65em]

\parbox[c]{\hsize}{\small\setstretch{0.75}{\bf physical space}} & $\mathcal H_{\rm physical}$ & $\mathcal H_{\rm kin}$ & \parbox[c]{\hsize}{\small\setstretch{0.75}{\bf kinematical space}}\\[1em]

\parbox[c]{\hsize}{\small\setstretch{0.75}{\bf{stabilizer group}}} & $G$ & $G$ & \parbox[c]{\hsize}{\small\setstretch{0.75}{\bf gauge group}}\\[1em]

\parbox[c]{\hsize}{\small\setstretch{0.75}{\bf stabilizer operator}} & $\bigotimes_i U_i^g$ & $U^g$ & \parbox[c]{\hsize}{\small\setstretch{0.75} {\bf gauge transformation}}\\[1em]

\parbox[c]{\hsize}{\small\setstretch{0.75} {\bf code space}} & $\mathcal H_{\rm{code}}$ & $\mathcal H_{\rm pn}$ & \parbox[c]{\hsize}{\small\setstretch{0.75} {\bf perspective-neutral space}}\\[1em]

\parbox[c]{\hsize}{\small\setstretch{0.75} {\bf maximal set of\\correctable errors}} & $\mathcal{E}=\{E_i\Pi_{\rm code}\}$& $\mathcal{A}_{R}=\cB(\mathcal{H}_R)$&\parbox[c]{\hsize}{\small\setstretch{0.75}{\bf QRF algebra}}\\[1em]

\parbox[c]{\hsize}{\small\setstretch{0.75} {\bf choice of\\correctable error set}}&$\mathcal{E}$&$\Hil_R,\mathcal{A}_R$&\parbox[c]{\hsize}{\small\setstretch{0.75}{\bf choice of internal QRF}}\\
~ & ~ &$\text{s.t. }\Hil_{\rm kin}\simeq\Hil_R\otimes\Hil_S$& \parbox[c]{\hsize}{\small\setstretch{0.75}{\bf subsystem} (and\\complementary system $S$)} \\
&&\text{and }$U^g=U^g_R\otimes U^g_S$&\\[1em]

\parbox[c]{\hsize}{\small\setstretch{0.75} {\bf logical space}} & $\mathcal H_{\rm logical}$ & $\mathcal H_{|R}=\mathcal R_R[\Hil_{\rm pn}]\simeq\Hil_S$ &\parbox[c]{\hsize}{\small\setstretch{0.75}{\bf reduced system space}}\\[1em]

\parbox[c]{\hsize}{\small\setstretch{0.75}{\bf decoding}} & $\mathcal{D}(\bullet)$ & $\mathcal W_R(\bullet)=\mathcal R_R\bullet \mathcal R_R^\dag$ & \parbox[c]{\hsize}{\small\setstretch{0.75}{\bf{internal QRF\\perspective}} (reduction)}\\[1em]

\parbox[c]{\hsize}{\small\setstretch{0.75}{\bf encoding}} & $\mathcal{N}(\bullet)$ & $\mathcal W_R^\dag(\bullet)=\mathcal R_R^\dag\bullet \mathcal R_R$ & \parbox[c]{\hsize}{\small\setstretch{0.75}{\bf embedding}\\(inverse reduction)}\\[1em]

\parbox[c]{\hsize}{\small\setstretch{0.75}{\bf encoded logical operator} ($f_L\in \cB(\Hil_{\rm logical})$)} & $\mathcal N(f_L)$& $O_{f_S|R}=\mathcal R_R^{\dag}(f_S)\mathcal R_R$ &\parbox[c]{\hsize}{\small\setstretch{0.75}{\bf gauge-invariant\\relational observable}\\ ($f_S\in\cB(\Hil_{|R})$)}\\[1em]

\parbox[c]{\hsize}{\small\setstretch{0.75}{\bf change of encoding}} & $\mathcal D^{\prime}\circ\mathcal{N}=\mathcal N^{\prime\dag}\circ\mathcal N$ & $\mathcal V_{R\to R^{\prime}}=\mathcal W_{R^{\prime}}\circ \mathcal W_R^{\dag}$ & \parbox[c]{\hsize}{\small\setstretch{0.75}{\bf QRF transformation}}\\[1em]

\parbox[c]{\hsize}{\small\setstretch{0.75}{\bf correctable Pauli error}} &$E_{\chi}$ & $E_{\chi}$&\parbox[c]{\hsize}{\small\setstretch{0.75}{\bf``electric'' charge\\excitation}}\\[1em]

\parbox[c]{\hsize}{\small\setstretch{0.75}{\bf error space}\\($\chi\in\hat G$ irreducible characters of $G$, elements of Pontryagin dual $\hat G$)}&$\Hil_{\chi}=E_{\chi}(\Hil_{\rm code})$&$\Hil_{\chi}=E_{\chi}(\Hil_{\rm pn})$&\parbox[c]{\hsize}{\small\setstretch{0.75}{\bf space with ``electric''\\charge $\chi$}}\\[2em]

\parbox[c]{\hsize}{\small\setstretch{0.75}{\bf syndrome measurement}\\ $\{U^g\}_{g\in G_{\rm gen}}$ set of stabilizer generators;\\$\{P_{\chi}\}_{\chi\in\hat G}$ projectors onto charge space $\Hil_{\chi}$} &$\{U^g\}_{g\in G_{\rm gen}} $ or $\{P_{\chi}\}_{\chi\in\hat G}$&$\{U^g\}_{g\in G_{\rm gen}}$ or $\{P_{\chi}\}_{\chi\in\hat G}$&\parbox[c]{\hsize}{\small\setstretch{0.75}{\bf ``electric'' charge\\measurement}\\
$\{U^g\}_{g\in G_{\rm gen}}$ set of gauge generators; \\ $\{P_{\chi}\}_{\chi\in\hat G}$ projectors onto charge space $\Hil_{\chi}$}\\[3em]

\parbox[c]{\hsize}{\small\setstretch{0.75}{\bf error detection\\operation}}&$\sum_{\chi\in\hat G}\,P_{\chi}\bullet P_{\chi}$&$\frac{1}{|G|}\sum_{g\in G} U^g\bullet\,(U^g)^\dag$ & \parbox[c]{\hsize}{\small\setstretch{0.75}{\bf incoherent group\\average over $G$} ($G$-twirl)}\\[1.5em]

\parbox[c]{\hsize}{\small\setstretch{0.75}{\bf recovery} for (unitary)} error $E_{\chi}$&$E_{\chi}^\dag$ &$E_{\chi}^\dag$&\parbox[c]{\hsize}{\small\setstretch{0.75}{\bf ``electric'' charge\\de-excitation}}\\[2em]

\parbox[c]{\hsize}{\small\setstretch{0.75}{\bf recovery operation}\\for (unitary) errors $E_{\chi}$}&$\sum_{\chi\in\hat G}E_{\chi}^\dag P_{\chi}\bullet P_{\chi}E_{\chi}$ &$\sum_{\chi\in\hat G}E_{\chi}^\dag P_{\chi}\bullet P_{\chi}E_{\chi}$&\parbox[c]{\hsize}{\small\setstretch{0.75}{\bf recovery operation}\\for ``electric'' errors}\\
&&&\\

\parbox[c]{\hsize}{\small\setstretch{0.75} {unitary \bf dual $\hat G$-error}\\($g\in G$ labels the irreducible characters of $\hat G$)} &$\hat{E}_g$ & $\hat{E}_g$&\parbox[c]{\hsize}{\small\setstretch{0.75}{\bf{``magnetic'' charge\\excitation/ gauge fixing\\errors}}}\\
 &&&\\
 
\parbox[c]{\hsize}{\small\setstretch{0.75}{\bf dual error space}}&$\hat{\Hil}_g=\hat{E}_g(\Hil_{\rm code})$&$\hat{\Hil}_g=\hat{E}_g(\Hil_{\rm pn})$&\parbox[c]{\hsize}{\small\setstretch{0.75}{\bf gauge fixed space} (``magnetic'' charge $g$ space)}\\[1em]

\parbox[c]{\hsize}{\small\setstretch{0.75}{\bf $\hat{G}$-syndrome\\measurement}\\ $\{\hat{U}^{\chi}\}_{\chi\in\hat{G}_{\rm gen}}$ set of\\$\hat G$-generators;\\$\{\hat{P}_{g}\}_{g\in G}$ projectors onto\\dual charge space $\hat{\Hil}_{g}$} &$\{\hat{U}^{\chi}\}_{\chi\in\hat{G}_{\rm gen}} $ or $\{\hat{P}_g\}_{g\in G}$&$\{\hat{U}^{\chi}\}_{\chi\in\hat{G}_{\rm gen}} $ or $\{\hat{P}_g\}_{g\in G}$&\parbox[c]{\hsize}{\small\setstretch{0.75}{\bf ``magnetic'' charge\\measurement}}\\
&&&\\

\parbox[c]{\hsize}{\small\setstretch{0.75}{\bf $\hat G$-error detection\\operation}}&$\sum_{g\in G}\,\hat{P}_{g}\bullet \hat{P}_{g}$&$\frac{1}{|\hat G|}\sum_{\chi\in\hat G} \hat{U}^{\chi}\bullet\,(\hat{U}^{\chi})^\dag$ & \parbox[c]{\hsize}{\small\setstretch{0.75}{\bf incoherent group\\average over $\hat G$} ($\hat G$-twirl)}\\
&&&\\

\parbox[c]{\hsize}{\small\setstretch{0.75}{\bf recovery for\\dual error} $\hat{E}_g$}  &$\hat{E}_g^\dag$ &$\hat{E}_g^\dag$&\parbox[c]{\hsize}{\small\setstretch{0.75}{\bf ``magnetic'' charge\\de-excitation}}\\
&&&\\

\parbox[c]{\hsize}{\small\setstretch{0.75}{\bf recovery operation}\\for dual errors $\hat{E}_{g}$}&$\sum_{g\in G}\hat{E}_{g}^\dag \hat{P}_{g}\bullet \hat{P}_{g}\hat{E}_{g}$ &$\sum_{g\in G}\hat{E}_{g}^\dag \hat{P}_{g}\bullet \hat{P}_{g}\hat{E}_{g}$&\parbox[c]{\hsize}{\small\setstretch{0.75}{\bf recovery operation}\\for ``magnetic''/gauge-fixing errors}\\
[1em]
\end{xltabular}

\medskip

\noindent\textbf{Note added.} While completing this work, we were made aware of interesting concurrent efforts to understand relations between QECCs, the PN-framework and gauge theories by Lin-Qing Chen and Elias Rothlin as part of the MSc thesis of the latter \cite{CRthesis}. There is some overlap with identifications we also make in this work (e.g., code space = perspective-neutral space, encoded logical operators = relational observables) and the authors also consider the five-qubit code as an example like us in Example~\ref{sec:5qubit} to extract QRFs from it. Altogether, however, the error side of the dictionary, including the duality, which comprises the essence of our work and the fineprint of the dictionary, does not appear in \cite{CRthesis} and so the overall overlap is not extensive. We look forward to comparing our results in the future and thank Lin-Qing Chen and Elias Rothlin for sharing their results with us. We are also aware of concurrent efforts by Masazumi Honda \cite{Masazumi1,Masazumi2} to link lattice gauge theories with QECCs for applications in quantum simulations of the former (no link with QRFs is made). Also in this work, electric excitations as errors on the gauge theory side appear, in line with our observations. We thank Masazumi Honda for many discussions on this and initial collaboration on follow-up work.

\section{The QECC/QRF dictionary in a nutshell}\label{sec_dictionary}

Let us now provide some flavor of the correspondence between QECCs and the PN-framework for QRFs.~We will here be somewhat coarse without referring to a specific class of QECCs or QRFs to merely provide intuition.~Some facets of Table~\ref{tab_dictionary} will require further structure that we explain in detail in the succeeding sections.

Boiled down to one sentence, this correspondence rests on the fact that redundancy is the hallmark of both QEC and gauge theories and that QRFs are the tool to split redundancy from physical information in gauge systems.

\subsection{Quantum error correction in a nutshell}\label{ssec_QEC}

Figure~\ref{fig:QECC} depicts the basic structure of quantum error correction (see \cite{Gottesman:1997zz,NielesenChuang:book,gottesman2006quantum,Gottesman:2009zvw,PreskillNotes} for more complete introductions).
Constituting this structure are four different types of Hilbert spaces:
\begin{itemize}
\item[(a)] The \emph{physical space} $\mathcal H_{\rm physical}$, modeling \emph{all} the possible states of the physical system used to physically incarnate the code.
\item[(b)] The \emph{code space} $\mathcal H_{\rm code}$ which, in finite-dimensional systems, is a subspace of the physical space $\mathcal H_{\rm code}\subset\mathcal H_{\rm physical}$ and corresponds to the space of states of the physical system actually needed to physically incarnate the code.
\item[(c)] The \emph{logical space} $\mathcal H_{\rm logical}$, which corresponds to the abstract states and operations one wishes to implement physically.
\item[(d)] The \emph{error spaces} $\mathcal{H}_i\subset\Hil_{\rm physical}$, which are the target spaces of errors acting on the code space.
\end{itemize}
Technically, there is an obvious map between the first two spaces:
\begin{itemize}
    \item[(0)] The orthogonal code space projector $\Pi_{\rm code}:\Hil_{\rm physical}\rightarrow\Hil_{\rm code}$.
\end{itemize}
Quantum error correction as a process relates these four spaces through several operationally important maps:
\begin{itemize}
\item[(i)] The unitary \emph{encoding} $\mathcal{N}:\cB(\mathcal H_{\rm logical})\rightarrow\cB(\mathcal H_{\rm code})$ maps density matrices and operations from the logical space into its physical realization, the code space. In practice, one usually appends $m$ ancillary qudit systems, each in some ready state $\ket{1}$, to the logical state until $\dim(\Hil_{\rm logical}\otimes\left(\mathbb{C}^{d}\right)^{\otimes m})=\dim\Hil_{\rm physical}\equiv n$ to write
\begin{equation}\label{encoding}
    \mathcal N(\rho_L)=N^\dag\left(\rho_L\otimes\ket{1}\!\bra{1}\otimes\cdots\otimes\ket{1}\!\bra{1}\right)N
\end{equation}
for some suitable $n\times n$ unitary $N$.
Encoding is typically highly non-local, so as to spread the logical data across all physical subsystems.
Conversely, the \emph{decoding} is just the inverse unitary, $\mathcal D=\mathcal{N}^\dag:\cB(\Hil_{\rm code})\rightarrow\cB(\Hil_{\rm logical})$. 

\item[(ii)] Encoded \emph{logical operators} are operators that preserve the code space, $\bar{A} : \Hil_\mrm{code} \to \Hil_\mrm{code}$, and correspond to the operations one wishes to perform on the abstract logical state.~For a given $A \in \cB(\Hil_\mrm{logical})$, it follows that $\mathcal{D} \circ \mathcal{N}(A) = A$.~For $\bar{A}$ as an operator whose domain and range are both $\Hil_\mrm{code}$, $\mathcal{N} \circ \mathcal{D}(\bar{A}) = \bar{A}$ as well.~However, the domain of $\bar{A}$ is commonly extended to all of $\Hil_\mrm{physical}$, in which case $\mathcal{N} \circ \mathcal{D}(\bar{A})$ need only be equal to the extension of $\bar{A}$ up to multiplication by an operator that decodes to the identity, $I_\mrm{logical}$. 

\item[(iii)] In contrast, \emph{errors} (which we label by some index $i$) are operators that do not preserve the code space, and the image of $\Hil_\mrm{code}$ under a given error $E_i$ is the error space $\Hil_i$; to wit, $E_i:\mathcal{H}_{\rm code}\rightarrow\Hil_i\subset\Hil_{\rm physical}$.\footnote{$\Hil_i$ need not necessarily lie in the orthogonal complement $\Hil^\perp_{\rm code}$ of the code subspace in order for $E_i$ to be correctable. For example, in Sec.~\ref{ssec_dualerrors}  we will see that the ``gauge fixing errors'' $\hat{E}_g$ of Table~\ref{tab_dictionary}  do not map into $\Hil_{\rm code}^\perp$, yet are correctable.} We will often collect such errors in a set $\mathcal E=\{E_1,\ldots,E_m\}$ for some $m\in\mathbb N$, and we will always take the $E_i$ to be unitary (as maps from $\Hil_{\rm code}$ to $\Hil_i$), for otherwise they are not correctable. With some abuse of notation, we will sometimes also write $E_i$ even if the error-defining operator's domain is all of $\Hil_{\rm physical}$ (as with Pauli operators later), in which case we write $E_i\Pi_{\rm code}$ for the error acting only on the code.

\item[(iv)] An error set $\mathcal E$ gives rise to a continuum of \emph{error operations}, each of which models a specific noise source acting on the quantum registers and taking states (and operators) out of the code space: \begin{equation}
    \tilde{\mathcal{E}}:\cB(\Hil_{\rm code})\rightarrow\cB(\Hil_{\rm physical})\,,\qquad\qquad\qquad\tilde{\mathcal{E}}(\rho)=\sum_k \tilde{E}_k\,\rho\,\tilde{E}_k^\dag\,,
\end{equation}
The operation elements $\tilde E_k$ are linear combinations of the errors $E_i$ which need not be unitary and are such that $\sum_k\tilde{E}_k^\dag\,\tilde{E}_k\leq I$. Such quantum operations can also account for the case that \emph{no} error happened with a certain probability. Which such error operation (and error set $\mathcal E$) is relevant in practice depends on the errors that one wishes to model. 

\item[(v)] The purpose of QEC is to find ways to correct those errors and, in the perfect case, the original uncorrupted logical state is subsequently recovered after decoding.
For unitary errors, if one knew that a given $E_i$ had happened, recovery would be achieved by simply applying $E_i^\dag$.
To make this determination, one performs an error \emph{syndrome measurement}.
Generically, one can determine the syndrome by a projective measurement $\{P_i\}$, comprised by the  orthogonal projectors $P_i$ onto the error spaces $\Hil_i$.\footnote{In stabilizer codes, one can alternatively measure any generating set $\{S_i\}$ of the stabilizer group $\mathcal{G}$. These leave the error spaces invariant, given that they comprise the non-trivial irreducible representations of $G$; see Secs.~\ref{Sec:generalstabilizer} and~\ref{ssec_pauliduality} for further review.}
The measurement result then determines which recovery operation should be carried out.
In the end, provided the errors defining $\tilde{\mathcal{E}}$ are \emph{correctable}, one can combine syndrome measurement and recovery into a single \emph{error-correction operation} $\mathcal{O}:\cB(\Hil_{\rm physical})\rightarrow\cB(\Hil_{\rm code})$ that maps $\tilde{\mathcal{E}}(\rho_{\rm code})$ back into the code, so that
\begin{equation}
\mathcal{O}\left(\tilde{\mathcal{E}}(\rho_{\rm code})\right)=\rho_{\rm code}\,.
\end{equation}
Error correction fails when $\Oh$ results in the application of a logical operator to the original encoded state, or in other words, a \emph{logical error}.
Failure can be due to a faulty implementation of $\Oh$ or due to an error being \emph{uncorrectable}.
\end{itemize}
These maps are thus not independent; when the errors are correctable, they obey
\begin{equation}
\mathcal{O}\circ\tilde{\mathcal{E}}\circ\mathcal{N} = \mathcal{N}\,.\label{qeccrel}
\end{equation}

\begin{figure}
    \centering
    \includegraphics[width=0.5\linewidth]{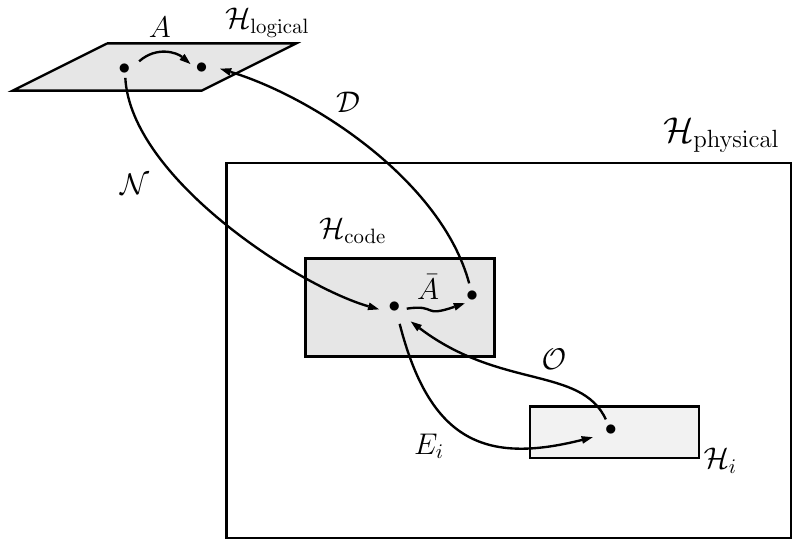}
    \caption{The Hilbert spaces in quantum error correction and the maps among them.}
    \label{fig:QECC}
\end{figure}

Now when is an error set correctable?
Intuitively, it is clear that errors and the corresponding error spaces must be distinguishable in some way, for otherwise one does not know which recovery to apply. 
It is a foundational result in QEC that an error set $\mathcal{E}=\{E_1,\ldots,E_m\}$ is correctable (this includes detection \emph{and} recovery) if and only if the \emph{Knill-Laflamme (KL) condition} \cite{Knill:1996ny}
\begin{equation}\label{KL}
    \Pi_{\rm code}\,E_i^\dag\,E_j\,\Pi_{\rm code}=C_{ij}\Pi_{\rm code}
\end{equation}
is satisfied for all $i,j=1,\ldots,m$, where $C_{ij}$ is a Hermitian matrix of complex coefficients.~Note that, if this condition holds for $\mathcal E$, then it also holds for its linear span, which means that all the operation elements $\tilde{E}_k$ above will satisfy it too.~In particular, since $C_{ij}$ is Hermitian, it can be diagonalized: $w^\dag C w=d$, where $d$ is diagonal and $w$ unitary.~Diagonalization results in reshuffled errors $E_i'=\sum_j w_{ij}E_j$, such that the error spaces $\Hil_i'=E_i'(\Hil_{\rm code})$ corresponding to errors $E'_i$ with non-zero entry in $d$ are mutually orthogonal.~States in different error spaces are thus perfectly distinguishable, and the error spaces are all of dimension $\dim\Hil_{\rm code}$ (recall that the $E_i$ are unitaries).~Thus, an error set $\mathcal E$ is correctable if and only if its linear span contains a basis of unitary errors mapping to non-overlapping orthogonal error subspaces of $\Hil_{\rm physical}$.~If one of the $\Hil_i'$ coincides with $\Hil_{\rm code}$ then the corresponding error must be trivial.~We call the rank of $C$ the rank of the error set $\mathcal E$.

\begin{example}[Three-qubit repetition code]
\label{sssec_3qubit}

To illustrate the basics of quantum error correction described above, let us begin by reviewing a simple example of a stabilizer code, namely, a repetition code on three qubits.~This 3-qubit code will be a workhorse for illustrating our results, and so here we will highlight those features that are most relevant to the coming analysis.~Pedagogical overviews are abundant in the literature, e.g., \cite{Gottesman:2009zvw,PreskillNotes}.

A 3-qubit repetition code encodes one logical qubit into three physical qubits.
$\Hil_\mrm{code}$ is thus a two-dimensional subspace of $\Hil_\mrm{physical}$, which we can choose to span with the orthogonal codewords
\begin{equation} \label{eq:3qubitrep}
    \ket{\bar{0}} \equiv \ket{000}\;,\quad\ket{\bar{1}} \equiv \ket{111}.
\end{equation}
As a matter of notation, we will often put a bar over a state in $\Hil_\mrm{code} \subset \Hil_\mrm{physical}$ to denote that it is the encoded image of the corresponding logical state in $\Hil_\mrm{logical}$.
The relationships among the various Hilbert spaces introduced in the previous section are then
\begin{equation}
    \mathbb{C}^2 \simeq \Hil_{\text{logical}} \simeq \Hil_{\text{code}} \subset \Hil_{\text{physical}} \simeq (\mathbb{C}^2)^{\otimes 3},
\end{equation}
and the encoding map for the choice of basis states \eqref{eq:3qubitrep} is thus
\begin{equation}
    N ~ : ~ \ket{i}_\mrm{logical} \mapsto \ket{\bar{i}}, \quad i = 0, 1.
\end{equation}

The 3-qubit code is an example of \emph{a stabilizer code}: another way to characterize $\Hil_\mrm{code}$ is as the joint $+1$ eigenspace of the \emph{stabilizer group}
\begin{equation} \label{eq:3qubitstab}
    \mathcal{G} = \{I, Z_1 Z_2, Z_2 Z_3, Z_1 Z_3\},
\end{equation}
which is generated by any pair of the non-identity stabilizers.
Stabilizer codes have the property that
\begin{equation} \label{eq:Pi_code_stabilizers}
    \Pi_{\rm code}=\frac{1}{|\mathcal{G}|}\prod_{i=1}^{\ln|\mathcal{G}|}\,\left(I+S_i\right)\,,
\end{equation}
where the $\{S_i\}$ are any set of generators of $\mathcal{G}$ and $|\mathcal{G}|$ is its cardinality (or order).
~It is a straightforward exercise to check that this expression indeed gives $\Pi_\mrm{code} = \ketbra{\bar{0}}{\bar{0}} + \ketbra{\bar{1}}{\bar{1}}$.

The 3-qubit code corrects a single erroneous $X_i$, whose location can be determined by measuring a pair of stabilizer generators.
For example, suppose the error $X_1$ occurs on an encoded state $\ket{\bar{\psi}} \in \Hil_\mrm{code}$ and we subsequently measure the generators $Z_1 Z_2$ and $Z_2 Z_3$.
Since
\begin{equation}
\begin{aligned}
    Z_1 Z_2 (X_1 \ket{\bar{\psi}}) &= - X_1 (Z_1 Z_2 \ket{\bar{\psi}}) = - (X_1 \ket{\bar{\psi}}) \\
    Z_2 Z_3 (X_1 \ket{\bar{\psi}}) &=  X_1 (Z_2 Z_3 \ket{\bar{\psi}}) = (X_1 \ket{\bar{\psi}}),
\end{aligned}
\end{equation}
we will measure the eigenvalue $-1$ for $Z_1 Z_2$ and $+1$ for $Z_2 Z_3$.
More generally, each single-qubit $X_i$ will present a unique syndrome (i.e., set of measurement results) for a given pair of stabilizers when measured, and so the location of a single $X_i$ error can be diagnosed and then corrected by deliberately applying $X_i$ again.

The logical operators $X_\mrm{logical}$ and $Z_\mrm{logical}$ have many representations on $\Hil_\mrm{code}$.
For example, $\bar{X} = X_1 X_2 X_3$ and $\bar{Z} = Z_i$ where $i \in \{1, 2, 3\}$ will do the trick.
Note that $\bar{X}$ and $\bar{Z}$ commute with the stabilizers, and so they are ``undetectable'' via syndrome measurement --- as they must be, since they preserve the code subspace!

In the language of the previous section, $\mathcal{E} = \{I, X_1, X_2, X_3\}$ is a set of correctable errors.~Indeed, the Knill-Laflamme condition \eqref{KL} is satisfied with $C_{ij} = \delta_{ij}$.~The error spaces are then $\Hil_i = X_i \Hil_\mrm{code}$ for $i = 1, 2, 3$, and we can build an error channel, for example, by assigning probabilities $p_k$, $k = 0, \dots, 3$, to each member $E_k \in \mathcal{E}$ (with the convention that $E_0 = I$).~Explicitly, we have that
\begin{equation}
    \tilde{\mathcal{E}}(\rho) = \sum_{k=0}^3 p_k E_k \rho E_k^\dagger
\end{equation}
(i.e., $\tilde E_k = \sqrt{p_k} E_k$).~Denoting the projector onto $\Hil_k$ by $P_k$ (with $P_0 \equiv \Pi_\mrm{code}$), we can formally write the recovery operation as
\begin{equation}
    \Oh(\rho) = \sum_{k=0}^3 E_k P_k \rho P_k E_k^\dagger.
\end{equation}
However, operationally, we can correct for errors by following the syndrome measurement and correction protocol described above.

The choice of correctable error set is not unique.
For example, $\{I,Y_1,Y_2,Y_3\}$ is also a set of correctable errors, and so is $\{I, X_1, X_2, X_1 X_2\}$, albeit not one that we usually choose, given that we are traditionally concerned with error sets that consist of all Pauli strings up to a given weight.
$\{X_1, Z_2\}$ is also technically a correctable error set (because $I$ is not contained in the set), but, intuitively, it is somehow not as big as it could be.
It turns out that $\{X_1, Z_2\}$  is not a \emph{maximal} set of correctable errors, which is a notion that we will define carefully in Sec.~\ref{Sec:error/QRFcorrespondence}.
At a formal level, all these error sets satisfy the Knill-Laflamme condition.
At an operational level, it comes down to how one chooses to interpret the syndrome measurement.
For example, the operators $X_3$ and $X_1 X_2$ both present the same syndrome when we measure a pair of stabilizer generators; what we consequently do to recover determines the error set that we correct for.
\eox
\end{example}

\clearpage

\subsection{The perspective-neutral framework for QRFs in a nutshell} \label{ssec_PN}

Let us now provide a similar nutshell review of the PN-framework for QRFs, followed by an illustrative example.

\subsubsection{Structure of gauge theories}
To understand how QRFs add the necessary structure for writing gauge theories in a form analogous to QECCs, we begin with the Hamiltonian formulation of gauge theories \cite{Henneaux:1992ig}. Underlying it are three different types of Hilbert spaces (the primed labels on the left link with those of QEC, e.g.\ (a') is the analog of (a), etc.):
\begin{itemize}
\item[(a')] The \emph{kinematical Hilbert space} $\mathcal H_{\rm kin}$, carrying \emph{all} (gauge-variant and -invariant) degrees of freedom in the problem and a fixed unitary representation  $\mathcal{G}=U(G)$ of the gauge group $G$. For what follows, it will suffice to assume $G$ to be compact.\footnote{This does not work for continuum gauge field theories, where $G$ is a suitable set of $G'$-valued spacetime functions and thus infinite-dimensional. Here, $G'$ denotes the (typically compact) structure group, e.g.\ $\rm{U}(1)$ in QED. The assumption is consistent with gauge theories on finite lattices, however, where $G=(G')^{\times (\# \text{vertices})}$, as well as finite-dimensional gauge systems in mechanics. } 
\item[(b')] The \emph{gauge-invariant Hilbert space} $\mathcal H_{\rm pn}$ of $\mathcal{G}$-invariant 
 states, $\ket{\psi_{\rm pn}}=U^g\ket{\psi_{\rm pn}}$ for all $g\in{G}$.\footnote{In the terminology of constraint quantization, this Hilbert space is often denoted by $\Hil_{\rm phys}$.~To avoid confusion with the terminology of QEC, in this work we shall refer to this Hilbert space and its states as (internal QRF) \emph{perspective-neutral}, given that this is their role in the PN-framework (see below), and hence utilize the subscript notation ``pn'' rather than ``phys''.~As discussed below, in our dictionary $\mathcal H_{\rm pn}$ will be identified with $\Hil_{\rm code}$, while $\mathcal H_{\rm physical}$ of QEC will be identified with $\Hil_{\rm kin}$.}  This space encodes the physically predictive content of the theory and is a subspace of $\Hil_{\rm kin}$ when $G$ is compact.

 \item[(d')] The non-trivially \emph{charged sectors} $\Hil_i$ carrying non-trivial irreducible representations of ${G}$ in the decomposition $\Hil_{\rm kin}=\Hil_{\rm pn}\oplus(\bigoplus_i\Hil_i)$. In Abelian gauge theories, these correspond to electric charges. If no further degrees of freedom are added to build new trivial representations, then the $\Hil_i$ break gauge-invariance. 
 \end{itemize}

What is the gauge theory counterpart of the logical space?~The redundancy-based analogy of gauge theories and QECCs suggests that it should constitute a redundancy-free description of $\Hil_{\rm pn}$.~This is where QRFs according to the PN-framework enter the scene to make this intuition precise.~The basic structure of the PN-framework is depicted in Fig.~\ref{fig:QRF} and is analogous to that of QEC shown in Fig.~\ref{fig:QECC}.

\subsubsection{The perspective-neutral framework for QRFs}\label{sssec_PN}

Within the PN-framework \cite{delaHamette:2021oex,Hoehn:2023ehz,Hoehn:2019fsy,Hoehn:2020epv,Chataignier:2024eil,AliAhmad:2021adn,Hoehn:2021flk,delaHamette:2021piz,Vanrietvelde:2018pgb,Vanrietvelde:2018dit,Giacomini:2021gei,Castro-Ruiz:2019nnl,Suleymanov:2023wio} (see \cite[Sec.~II]{Hoehn:2023ehz} for an introduction), which is built using the constraint structures of gauge theories, 
\begin{quote}
    \emph{a choice of QRF is tantamount to a choice of split of the kinematical degrees of freedom into redundant (the QRF itself) and non-redundant ones (the QRF's kinematical complement)}.
\end{quote}
In general, there is no unique such split and so any gauge theory will admit many distinct internal QRFs. A QRF has to transform non-trivially under $\mathcal{G}$, for otherwise it cannot be redundant; it will thus always be associated with some non-trivial representation of $\mathcal{G}$. 

More precisely, a QRF $R$, where $R$ stands for both ``reference'' and ``redundant'', associated with  ${G}$ is described by some Hilbert space $\mathcal{H}_R$ that carries a non-trivial (possibly projective and not necessarily irreducible) unitary representation $U_R$ of ${G}$. Our aim is to use it as a reference to describe the remaining degrees of freedom, for compactness called $S$ for ``system'',\footnote{For simplicity, we leave the  $R$-dependence of $S$ implicit.} in a gauge-invariant manner. Hence, a choice of QRF subsystem is nothing but a choice of kinematical tensor product structure (TPS) $\Hil_{\rm kin}\simeq\Hil_R\otimes\Hil_S$. For simplicity, we restrict our attention to partitions such that the representation $U(G)$  is of tensor product form, $U^g=U_R^g\otimes U_S^g$, $g\in G$. While this is not necessary, it will suffice for all our considerations below. $\Hil_S$ may itself have a further refined TPS (as will typically be the case in examples below). 

The QRF gives rise to an additional type of Hilbert space --- the counterpart of the logical space (c) in QEC:
\begin{itemize}
\item[(c')] The \emph{internal QRF perspective Hilbert space} $\Hil_{|R}$, corresponding to the  \emph{redundancy-free} description of $S$ relative to $R$.  As such, it is sometimes also called a \emph{reduced Hilbert space}.
\end{itemize}
QRFs equip the gauge-invariant Hilbert space with an additional interpretation and lead to a further type of Hilbert space on the same level as, but dual  to the charged sectors (in a sense to be described later):
\begin{itemize}
\item[(b')] $\Hil_{\rm pn}$ is a \emph{perspective-neutral space} that encodes and links all possible internal QRF perspectives (cf.~(vi')).
\item[(d'')]  The \emph{gauge-fixed spaces} $\hat{\Hil}_g\subset\Hil_{\rm kin}$ obtained by fixing $\Hil_{\rm pn}$ to a specific \emph{orientation} of $R$ (labeled by $g\in{G}$ and explained in detail later). These are isomorphic to $\Hil_{\rm pn}$ and sometimes also referred to as \emph{QRF-aligned Hilbert spaces} \cite{Krumm:2020fws,Hoehn:2021flk}.\footnote{Such aligned descriptions also appear in the distinct perspectival QRF approach in \cite{delaHamette:2020dyi}. This approach is only equivalent to the PN-framework for ideal QRFs (sharp orientation states) \cite{delaHamette:2021oex}. The spaces of \cite{delaHamette:2020dyi} can be viewed as gauge-fixed only in that case.} 
\end{itemize}

\begin{figure}[t!]
    \centering
    \includegraphics[width=0.5\linewidth]{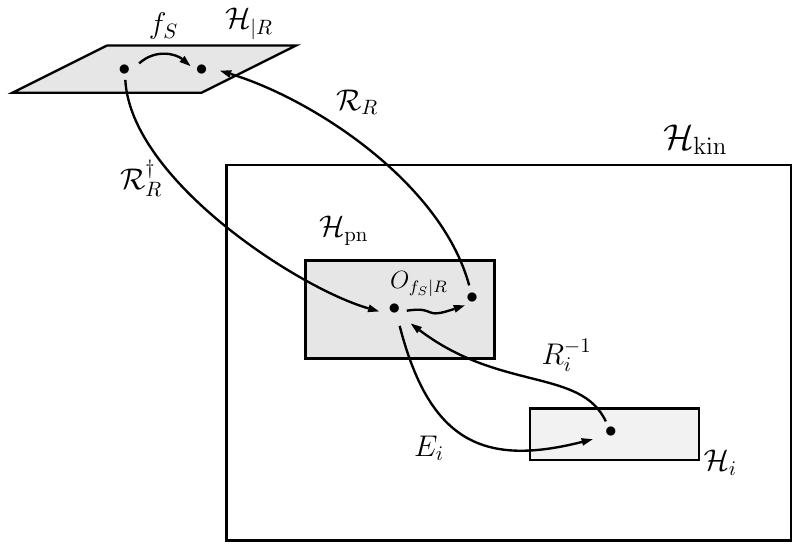}
    \caption{The Hilbert spaces in the perspective-neutral framework for QRFs and the maps among them.} 
    \label{fig:QRF}
\end{figure}

The PN-framework relates these Hilbert spaces through the following maps:
\begin{itemize}
    \item[(0')] The orthogonal projector $\Pi_{\rm pn}:\Hil_{\rm kin}\rightarrow\Hil_{\rm pn}$ onto the perspective-neutral space (zero charge sector), obtained by \emph{coherent group averaging} of states and observables over the gauge orbits:
    \begin{equation}
        \Pi_{\rm pn}=\frac{1}{\rm{Vol}(G)}\int_G \dd{g}\,U^g\,,
    \end{equation}
    where $\rm{Vol}(G)$ and $\dd{g}$ are the volume and Haar measure of $G$. 
   \item[(i')] A core aspect of the PN-framework is to spell out how one can ``jump into the perspective of a QRF $R$''. Two types of unitary \emph{reduction maps} achieve this \cite{Hoehn:2019fsy,Hoehn:2021flk,delaHamette:2021oex}: (1) $\mathcal{R}_R:\Hil_{\rm pn}\to\Hil_{|R}$ by gauge fixing the QRF $R$ into a specific orientation (Page-Wootters reduction); (2) $T_R:\Hil_{\rm pn}\to\ket{1}_R\otimes\Hil_{|R}$ by disentangling the  QRF $R$, where $\ket{1}_R$ is a fixed QRF state (trivialization). (Note that this fixed fiducial QRF state should not be confused with the Pauli $Z$ eigenstate.)  Both maps thus remove redundancy in the description of $\Hil_{\rm pn}$ in terms of \emph{kinematical} data and they will turn out to be dual to one another. 
    Thus, density operators $\rho_{|R}$ describing $S$ relative to $R$ relate to their perspective-neutral counterpart as
    \begin{equation}
    \rho_{\rm pn}=\mathcal{R}_R^\dag\,\rho_{|R}\,\mathcal{R}_R=T_R^\dag\left(
    \ket{1}\!\bra{1}_R\otimes\rho_{|R}\right)T_R\,.
    \end{equation}
    Comparing with the encoding in Eq.~\eqref{encoding}, we can view $\ket{1}_R$ as an ancilla ready state.
   
   \item[(ii')] All gauge-invariant operators $O:\Hil_{\rm pn}\rightarrow\Hil_{\rm pn}$ can be written as \emph{relational observables}, describing an $S$ property $f_S\in\mathcal{B}(\Hil_{|R})$ relative to $R$ by embedding $\mathcal{B}(\Hil_{|R})\rightarrow\mathcal{B}(\Hil_{\rm pn})$ \cite{Hoehn:2019fsy,Hoehn:2021flk,delaHamette:2021oex}
   \begin{equation}
O_{f_S|R}=\mathcal{R}_R^\dag\,f_S\,\mathcal{R}_R=T_R^\dag\left(\ket{1}\!\bra{1}_R\otimes f_S\right)T_R\,.
   \end{equation}
   \item[(iii')] The maps $E_i:\Hil_{\rm pn}\rightarrow\Hil_{i}$ to the non-trivial irreps of ${G}$ have the interpretation of \emph{electric charge excitations} in Abelian gauge theories. In that case, all irreps are isomorphic and the $E_i$ \emph{can} be unitary (unique only up to unitaries in the domain and target space).
   \item[(vi')] There are many QRFs, so we can change the internal perspective with a \emph{QRF transformation} that takes the form of \emph{quantum coordinate transformations}.
For two QRFs $R,R'$ and (i')(1) it reads 
\begin{equation}\label{QRFchange}
  {V}_{R\to R'}:\Hil_{|R}\rightarrow\Hil_{|R'}\,,\qquad\qquad  {V}_{R\to R'}=\mathcal{R}_{R'}\circ\mathcal{R}_R^\dag\,,
\end{equation}
extending to operators in an obvious manner, and taking a similar form for (i')(2). $V_{R\to R'}$ passes through the QRF-perspective-neutral space $\Hil_{\rm pn}$ that thus links the internal perspectives akin to how a manifold links different coordinate frame descriptions. This provides the PN-framework with its name. 
$R,R'$ are two distinct  QRF subsystems, so $\Hil_{\rm kin}\simeq\Hil_R\otimes\Hil_S$ and $\Hil_{\rm kin}\simeq\Hil_{R'}\otimes\Hil_{S'}$.  These partitions are related by a change of TPS of $\Hil_{\rm kin}$, unless $R'\subset S$ (the previous standard case in much of the QRF literature). We have suppressed it  for expositional simplicity in Eq.~\eqref{QRFchange} (see App.~\ref{app_QRFchange} and \cite{TBH} for details).
Expositions of QRF transformations for finite groups can be found in \cite{Hoehn:2021flk,Hoehn:2023ehz,delaHamette:2021oex,Krumm:2020fws,delaHamette:2020dyi}. 
\item[(iii'')] The maps $\hat{E}_g:\Hil_{\rm pn}\rightarrow \hat{\Hil}_g$ gauge fix $R$'s orientation (and are equivalent to (i')(1) above).

\end{itemize}
We have deliberately left out the QRF analogs of error operations (iv') and recovery (v'), as these are typically not considered in this context. We will discuss them in the coming sections, while refining the dictionary, and identify gauge theoretic structure that parallels Eq.~\eqref{qeccrel}. Anticipating a little, we will show that: (iv') a choice of correctable error set $\{E_i\}$ can be reinterpreted as a choice of \emph{QRF splitting} of the Hilbert space (see Sec.~\ref{Sec:error/QRFcorrespondence}); (v') the recovery operation $\mathcal{O}$ can be understood, in gauge-theoretic terms, as a \emph{dressing} operation (see Sec.~\ref{sec:dressing_field}). As a result, all the elements of Fig.~\ref{fig:QECC}, illustrating the general structure of a QECC, have a  counterpart in the QRF reinterpretation of Fig.~\ref{fig:QRF}.

\begin{example}[QRFs in the three-qubit repetition code]
\label{Sec:3qubitcodeasQRF}

To illustrate the PN-framework and to preview the results in the next sections, let us reinterpret the 3-qubit bit flip code in terms of internal QRFs that transform under the stabilizer group $\mathcal{G}=U(G)$.
The basic idea is to identify (a) $\Hil_{\rm physical}$ from QEC with (a') $\Hil_{\rm kin}$  from gauge theories, and to seek a factorization into redundant and non-redundant information, where the latter can be identified with the logical information in the code.
In concrete terms, we therefore seek a factorization $\Hil_{\rm kin} = \Hil_R \otimes \Hil_S$ into frame ($R$) and system ($S$) with $\Hil_S \simeq \Hil_\mrm{code}$.
Furthermore, we seek a representation of $G$ on $\Hil_R$ and an associated basis of configuration states that turn out to define the reductions in (i').

Here, let us look for a factorization $\Hil_{\rm kin} = \Hil_R \otimes \Hil_S$ that is local with respect to the same 3-qubit tensor product structure as $\Hil_{\rm kin}$ (later we will generalize this).~To perform this split, we first note that, as an Abelian group with four elements, each of which squares to the identity, the stabilizer group~\eqref{eq:3qubitstab} is isomorphic to $G=\mathbb{Z}_2\times\mathbb{Z}_2$.~If we require that the frame supports a regular representation of this group, then we must have that $\Hil_R=\ell^2(\mathbb{Z}_2)\otimes\ell^2(\mathbb{Z}_2)\simeq\mathbb{C}^2\otimes\mathbb{C}^2$ is isomorphic to a pair of qubits.~Let us label the group elements as $\{++,+-,-+,--\}$.~We can then define the group configuration basis $\ket{g}_R$, $g\in\mathbb{Z}_2\times\mathbb{Z}_2$, identified with the QRF's possible \emph{orientations}, for example, via pairs of $X$-basis states\footnote{This choice is of course not unique.~For example, the $Y$-eigenbasis $Y\ket{\mathsf{y}} = \ket{\mathsf{y}}$, $Y\ket{\lambda} = -\ket{\lambda}$ with the identification $+ \to \mathsf{y}$, $- \to \lambda$ would work equally well.}
\begin{equation} \label{eq:3qubitOrientations}
    \ket{g}_R\in\big\{\ket{++},\ket{+-},\ket{-+},\ket{--}\big\}\,.
\end{equation}
Clearly, these are orthonormal and it is straightforward to see that a unitary representation of $G$ on $\mathbb{C}^2\otimes\mathbb{C}^2$ is given by
\begin{eqnarray}\label{eq:3qbfUR}
    U_R^{++} &=& I\otimes I\,,\quad\quad U_R^{+-}=I\otimes Z\,,\nonumber\\
    U_R^{-+}&=&Z\otimes I\,,\quad\quad U_R^{--}=Z\otimes Z\,,
\end{eqnarray}
such that we have the covariance property $U_R^g\ket{g'}_R = \ket{gg'}_R$.~Such QRFs with a perfectly distinguishable covariant orientation basis  are called \emph{ideal}.\footnote{In the stabilizer code context of this article, it turns out to suffice to almost exclusively work with ideal frames.~We will, however, encounter non-ideal frames in Sec.~\ref{sec:surfacecodes} and touch on them again in Sec.~\ref{sec:conclusion}.}

Since the QRF has to come with a four-dimensional Hilbert space, $\Hil_S$ must be two-dimensional and hence isomorphic to the logical qubit.~We now seek a unitary representation of $G$ on $\Hil_S\simeq\mathbb{C}^2$.~For dimensional reasons, it cannot be faithful.
However,
\begin{eqnarray}\label{eq:3qbfUS}
    U_S^{++}&=&I_S\,,\quad\quad U_S^{+-}=Z_S\,,\nonumber\\
    U_S^{-+}&=&Z_S\,,\quad\quad U_S^{--}=I_S\label{Srep}
\end{eqnarray}
furnishes a unitary representation of $\mathbb{Z}_2\times\mathbb{Z}_2$ on $\mathbb{C}^2$.
Indeed, $U_S^{g} U_S^{g'} = U_S^{g'}U_S^{g} = U_S^{gg'}$ for all $g, g' \in G$.

In order to now factorize the physical space as $\Hil_{\rm physical}=\mathbb{C}^2\otimes\mathbb{C}^2\otimes\mathbb{C}^2=\Hil_R\otimes\Hil_S$, let us simply choose any two physical qubits to make up the internal QRF and the remaining qubit as the system.
Without loss of generality, suppose we choose $R=12$ and $S=3$.
Then we see that 
\begin{equation} \label{eq:3qubit-Urep}
\begin{array}{lll}
U^{++} = U_R^{++}\otimes U_S^{++} = I_1\otimes I_2\otimes I_3\, , & ~ & U^{+-} = U_R^{+-}\otimes U_S^{+-}=I_1\otimes Z_2\otimes Z_3\, , \\[2mm]
    U^{-+} = U_R^{-+}\otimes U_S^{-+} = Z_1\otimes I_2\otimes Z_3\, , & ~ & U^{--} = U_R^{--}\otimes U_S^{--}=Z_1\otimes Z_2\otimes I_3 \, ,
\end{array}
\end{equation}
which coincides exactly with the stabilizers in~\eqref{eq:3qubitstab}.
Note that this coincidence holds regardless of which pair of physical qubits we choose as the internal QRF.
The stabilizer group can thus be interpreted as the gauge group with which $R$ is associated (in Sec.~\ref{sec_QRF2QEC} we will see that this group can more generally be interpreted as a group of \emph{external} frame transformations \cite[Sec.~II]{Hoehn:2023ehz}).
This implies that the coherent average over the stabilizer group --- which, in error correction terms, is the projector onto the code subspace --- is the standard group averaging, defining the projector onto the perspective-neutral Hilbert space used in the internal QRF setup:
\begin{eqnarray}\label{eq:3qubit:PicodePipn}
    \Pi_{\rm code}=\ket{\bar0}\!\bra{\bar0}+\ket{\bar1}\!\bra{\bar1}=\frac{1}{4}\left(I+Z_1Z_2+Z_2Z_3+Z_1Z_3\right) = \frac{1}{|G|}\sum_{g\in G} U_R^g\otimes U_S^g=\Pi_{\rm pn}\,.
\end{eqnarray}

With a choice of QRF in place, we can write down the mappings that were introduced above which remove the redundancy in the perspective-neutral description.
For finite groups, the Page-Wootters (PW) reduction map in (i')(1) reads $\mathcal{R}_R^g=\sqrt{|G|}(\bra{g}_R\otimes I_S)\Pi_{\rm pn}$ \cite{Hoehn:2021flk,delaHamette:2021oex}, so continuing with the choice $R=12$, and, e.g., choosing $g=e=++$, we have
\begin{equation}
    \mathcal{R}^{++}_{12}=2(\bra{++}_{12}\otimes I_3)\Pi_{\rm pn}=\bra{00}_{12}\otimes\frac{1}{2}(I_3+Z_3)+\bra{11}_{12}\otimes\frac{1}{2}(I_3-Z_3),
\end{equation}
and so
\begin{equation}\label{eq:logicalbasisPWencoding}
    \mathcal{R}^{++}_{12}\ket{\bar 0}=\ket{0}_3\,,\quad\quad \mathcal{R}^{++}_{12}\ket{\bar1}=\ket{1}_3\,.
\end{equation}
In other words, relative to the QRF $R=12$ in the ``identity orientation,'' the logical qubits are encoded in the up and down states of the $Z$-basis on the system qubit $S=3$.
This encoding would change if we made a different choice of orientation $g$, but can maximally involve an application of $Z_3$.
Explicitly, for $\ket{\bar{\psi}} \in \Hil_\mrm{code}$,
\begin{equation}
    \mathcal{R}_{12}^g \ket{\bar{\psi}} = \left\{
    \begin{array}{ll}
        \ket{\psi}_3 & g = ++,-- \\
        Z_3 \ket{\psi}_3 & g = +-, -+ .
    \end{array}
    \right.
\end{equation}
In particular, this defines a \emph{covariant} encoding map $\mathcal{N}_g:\mathcal{B}(\Hil_{3})\rightarrow\mathcal{B}(\Hil_{\rm pn})$, yielding the relational observables encoding $f_3\in\mathcal{B}(\Hil_3)$ relative to $R=12$ by 
\begin{equation}\label{3qubit:covencoding}
\mathcal{N}_g(f_3):=O_{f_3|12}=\left(\mathcal{R}_{12}^g\right)^\dag f_3\,\mathcal{R}_{12}^g =4\,\Pi_{\rm pn}(\ket{g}\!\bra{g}_{12}\otimes f_3)\Pi_{\rm pn}\,.
\end{equation}
Covariance follows because for all $g,g'\in\{++,+-,-+,--\}$,
\begin{equation}
  \mathcal{N}_g\left(U_3^{g'}\,f_3\,U_3^{g'} \right)=\left(I_{12}\otimes U_3^{g'}\right)\mathcal{N}_g(f_3)\left(I_{12}\otimes U_3^{g'}\right) \,.
\end{equation}
The disentangler in (i')(2), on the other hand, reads for finite groups $T_R=\sum_{g\in G}\ket{g}\!\bra{g}_R\otimes \left(U_S^g\right)^\dag$ \cite{Hoehn:2021flk,delaHamette:2021oex}, and so for the 3-qubit code is given by
\begin{equation}\label{3qubit:triv}
    T_{12} = I_{12} \otimes \tfrac{1}{2}(I_3 + Z_3) + X_1 X_2 \otimes \tfrac{1}{2}(I_3 - Z_3)\,.
\end{equation}
For any $\ket{\bar\psi} \in \Hil_\mrm{code}$, we have that $T_{12}\ket{\bar \psi} = \ket{00}_{12} \otimes \ket{\psi}_3$.

Page-Wootters reductions and trivialization act in a similar way on encoded states by pushing the logical information onto the system $S$.~However, their action on the QRF is \emph{dual} (which is best seen via the next paragraph): identifying $\ket{00}_{12}\equiv\ket{1}_R$ as the ``ready state'' of the QRF, we have that $\braket{g}{1}_R=\pm\frac{1}{2}$ for all $g\in\mathbb{Z}_2\times\mathbb{Z}_2$;~the $X$- and $Z$-basis of $R$ are mutually unbiased.~We will later see that the two are related by a group Fourier transform and more generally by what is known as Pontryagin duality (cf.~App.~\ref{app:Pontryagin}) that will play a central role in our work.~Comparing with (i) from the QEC side, we can view the two maps as dual decodings.~There is a further difference between the two: while $\mathcal{R}_R$ is only unitary on $\Hil_{\rm pn}$, $T_R$ is unitary on \emph{all} of $\Hil_\mrm{kin}$, and so it retains information about any errors that may have occurred.~This will prove useful for error correction (cf.\ Eq.~\eqref{eq:3qubitTerror}).

Finally, the ``electric charge excitations'' are  $\mathcal{E}=\{I,X_1,X_2,X_3\}$ as already described in Sec.~\ref{sssec_3qubit}.~On the other hand, the gauge fixing maps $\hat{E}_g$ turn out to be equivalent to fixing the QRF's orientation to $g\in \{++,+-,-+,--\}$ via
\begin{equation}
    \hat{P}_g:=\ket{g}\!\bra{g}_{12}\otimes I_3\,,
\end{equation}
and thereby map $\hat{E}_g:\Hil_{\rm pn}\rightarrow\hat{\Hil}_g=\ket{g}_{12}\otimes(\mathbb{C}^2)_{3}$.~Using the covariance of the orientation states and invariance of $\Hil_{\rm pn}$, it is easy to see that $U^g(\hat{\Hil}_{g'})=\hat{\Hil}_{gg'}$.~Clearly, these maps are equivalent to PW-reduction; however, their image now resides in $\Hil_{\rm kin}$.~As such, they can also be seen as decodings with $\ket{g}_R$ now the QRF ``ready state'' (dual to the ``ready state'' $\ket{1}_R$ from the trivialization above).~In Sec.~\ref{ssec_dualerrors} we will see how $\hat{E}_g$ can be replaced by a unitary on $\Hil_{\rm kin}$, giving rise to an error set that is \emph{dual} to $\mathcal{E}$.
\eox
\end{example}

\subsection{Interlude: how to think of errors in gauge theory and QRF setups}\label{ssec_interlude}

Perhaps the most evident identification is that of the code space of QEC with the gauge-invariant/perspective-neutral Hilbert space of gauge theories and QRFs.~Before proceeding with a brief summary of the QECC/QRF dictionary, let us pause to ask a key question concerning the interpretation of this correspondence.~Namely, what is the meaning of errors in a gauge theory, where operations taking us out of the space of gauge-invariant states are not usually considered physical?~In fact, there are at least five circumstances under which it \emph{does} make sense to consider “errors” in a gauge theory or a QRF setup:
\begin{itemize}
    \item They are literally errors in a quantum \emph{simulation} of a gauge theory.
    \item The gauge symmetry is \emph{emergent} (as in the toric code) and errors are simply excitations beyond the symmetric sector (e.g.\ electric or magnetic excitations in Abelian theories), thereby breaking a non-fundamental gauge invariance. This case will be discussed in the surface codes of Sec.~\ref{sec:surfacecodes}. 
    \item The errors do not break gauge invariance, but excite a previously “frozen” field species, thereby changing the shape of the gauge constraints. For example, the errors could be excitations beyond the matter vacuum in quantum electrodynamics, thus changing the Gauss law from its vacuum form to the version containing a non-trivial charge density. This would amount to an electric charge excitation and the ensuing non-invariance of the gauge field sector can be compensated by matter degrees of freedom such as electrons. The ``error'' would thereby map from one physical subsector of the theory to another. Such an observation has also been made in \cite{Masazumi1,Masazumi2}. 
    \item As we will explain in Sec.~\ref{sec_QRF2QEC}, in operational QRF scenarios, the group $G$ need not literally be a gauge group, but can be considered a group of external frame transformations. A standard scenario involves two agents who do not share an external frame and thus aim to restrict their quantum communication to external-frame-independent information \cite{Bartlett:2006tzx}. In this case, errors such as the $\{E_i\}$ discussed in (iii, iii') can be viewed as reintroducing external frame information into the communicated quantum information, thereby corrupting it.
    \item Gauge fixings provide another way of mapping states out of the space of gauge-invariant states without irrevocably losing physical information. As we will explain in Sec.~\ref{sec_errorduality}, these will indeed take the form of correctable errors obeying the KL condition and give rise to a novel error class that is dual to standard ones in QECCs, akin to electromagnetic duality (with the essential difference being that the dual charges are complementary). 
\end{itemize}
There may well be other circumstances under which gauge theory/QRF analogs of errors have a physical meaning.~Regardless, the points here already motivate us to consider errors, syndrome measurements and recovery operations on both sides of the QECC/QRF correspondence.

\subsection{The dictionary in a nutshell}

The preceding nutshell reviews of QEC and the PN-framework for QRFs were meant to provide a flavor of the QECC/QRF correspondence summarized  in Table~\ref{tab_dictionary}. They do not yet suffice to understand the full scope of the dictionary presented in that table, for which the rest of this paper will be required. In particular, it will be necessary to consider a specific QECC and corresponding QRF class to appreciate its detailed facets. However, the reviews already support the general identification of the structures (a--d) and (0, i--iii) of QEC with (a'--d') and (0', i'--iii'), respectively, from the PN-framework, providing part of the content of Table~\ref{tab_dictionary}. 

We did not yet discuss the QRF/gauge theory counterpart of error operations (iv), as well as syndrome measurements  and recovery operations (v), as these are not usually considered in that context. However, the discussion in Sec.~\ref{ssec_interlude} implies that one can equally well consider those in a QRF/gauge theory setting, where they describe charge excitation operations (iv') and charge measurement and de-excitation operations (v'). Similarly, they may describe operations reinserting external frame information (iv'), as well as detection of that information and recovery of the original encoding (v').

Conversely, we did not yet discuss the QEC counterpart of QRF changes (vi').
This corresponds to a feature that is not usually emphasized in the literature on QEC; namely, that for any code space there exist many decodings (and thus also encodings), each of which splits off a distinct set of (possibly virtual \cite{Zanardi:2001zz,Zanardi:2004zz}) qudit subsystems in some ready state. This corresponds to different ways of splitting off redundancy. Let $\mathcal{N},\mathcal{N}'$ be two encodings. A \emph{change of decoding} (vi) is given by
\begin{equation}\label{enchange}
\mathcal{D}'\circ\mathcal{N}=\mathcal{N}'^\dag\circ\mathcal{N}\,.
\end{equation}
In order for $\mathcal{N},\mathcal{N}'$ to be genuinely distinct encodings, $\mathcal{D}'\circ\mathcal{N}$ must be a non-local unitary conjugation, for otherwise the encodings will differ only by a local change of basis for the physical subsystems of the register. Similarly, QRF transformations are non-local unitaries \cite{Hoehn:2023ehz,AliAhmad:2021adn,delaHamette:2021oex,Hoehn:2019fsy,Giacomini:2017zju}.

In the remainder of this article, we shall explain in detail also the rest of the dictionary in Table~\ref{tab_dictionary}, which necessitates more structure. The precise form of those further ingredients in the dictionary relies on a restriction to stabilizer QECCs and QRFs associated with finite Abelian groups, to which we now turn, though we anticipate a generalization of those results in an appropriate form also for more general groups and codes. 

This will lead to two core results:~in Sec.~\ref{sec_genstab} we will prove that every correctable set of standard Pauli errors that is maximal in a certain sense generates an ideal QRF with special properties and that, conversely, such a QRF is associated with a unique equivalence class of maximal correctable error sets.~This will explain the translation between the choice of correctable error set and choice of QRF $R$ in Table~\ref{tab_dictionary} and lead to a reinterpretation of the KL condition with a QRF perspective. 

Furthermore, in Sec.~\ref{sec_errorduality}, we will describe a novel duality between standard Pauli errors and ``gauge fixing errors'' as briefly alluded to above. The duality will manifest in every property of the errors, syndrome measurements, and recovery operations. This duality is based on so-called Pontryagin duality (cf.~App.~\ref{app:Pontryagin}) and is akin to electromagnetic duality, except that in our case the ``electric'' and ``magnetic'' errors/charge excitations do not commute. It invokes representations of the Pontryagin dual group $\hat{G}$ of the stabilizer group $G$ that turn out to be dual in a specific sense to the stabilizer representation. This error duality can be viewed as a discrete form of the Fourier transform duality between position and momentum variables in quantum theory; indeed, position and momentum translations are each other's Pontryagin dual (cf.~the discussion around Eq.~\eqref{eq:Weyl}).

As a point of semantics, we will henceforth use the dictionary to allow us to use QEC and QRF notation for mathematical objects interchangeably.

\section{From QRFs to QECCs: finite Abelian groups}\label{sec_QRF2QEC}

In this section, we will start with QRFs according to the PN-framework and show how this formally admits the structure of a QECC. This will provide the nutshell synopsis in Sec.~\ref{sssec_PN} with more context and detail. In the next section, we will do the converse, starting with stabilizer QECCs and showing how each constitutes an instance of the PN-framework with many possible choices of QRFs inside. 

It is worthwhile to emphasize that the PN-framework is of broader applicability than gauge theories and gravity.~Specifically, the framework is agnostic as to whether the symmetry group $G$ is an actual gauge symmetry or corresponds to some operational restriction.~For example, the framework encompasses also the quantum simulation of gauge theories, alluded to in the introduction, where the actual physical system in the laboratory is not a gauge system, but one restricts its set of physically meaningful states to a subset that mimics the gauge-invariant states of, say, some lattice gauge theory.~This highlights that the QECC/QRF correspondence spelled out in this article is not restricted to pure QEC and gauge theory settings, but has a larger applicability in other operationally relevant situations.

\subsection{External vs.\ internal frames}\label{sec:extvsintframe}

The starting point of the PN-framework is the distinction between an \emph{external} frame (not part of the quantum system) and \emph{internal} QRFs --- which are subsystems of the total quantum system subjected to some symmetry principle --- and the states and observables that they can each distinguish.\footnote{See \cite[Sec.~II]{Hoehn:2023ehz} for an introduction to this philosophy and \cite{delaHamette:2021oex} for further details.} Invoking the language from Sec.~\ref{sssec_PN}, the total composite quantum system would be $RS$, where $R$ is one internal QRF and $S$ its complement (which may encompass many more internal QRFs). The joint description of $RS$ is the one relative to the external frame, as is this subsystem decomposition. Thus, we have a reinterpretation:
\begin{itemize}
    \item[(a')] The \emph{kinematical Hilbert space} $\Hil_{\rm kin}\simeq\Hil_R\otimes\Hil_S$ is the \emph{space of externally distinguishable states}, i.e.\ states that are distinguishable relative to the external frame. Similarly, the subsystem partition between $R$ and $S$ (and corresponding TPS) is defined relative to the external frame.
\end{itemize}
The choice of internal $R$ frame thereby becomes tantamount to a choice of kinematical TPS.

The aim of the program is to develop a purely internal description, i.e.\ to describe the composite quantum system from within, relative to an internal QRF subsystem such as $R$. The group ${G}$ capturing the symmetries is identified with external frame transformations. Since we intend to formulate a purely internal description, we aim to wipe out any external frame information, thereby constructing a standalone internal quantum theory from the original external description. Thus, we shall treat ${G}$ as gauge and refer to the external frame transformations simply as gauge transformations; gauge-invariance thereby becomes another term for external frame-independence. For the framework and its applicability, it does not matter whether the external frame and its transformations are physically meaningful or entirely fictitious. They would be meaningful in certain laboratory situations including quantum simulations of gauge theories and quantum communication between two agents not sharing an external frame, where $G$ encodes transformations between different external laboratory frames. By contrast, they may be fictitious as in gauge theories and gravity, where ${G}$ may be the actual gauge group of the theory (e.g., ${\rm{SU}}(N)$-valued spacetime functions in Yang-Mills theory, or compactly supported diffeomorphisms in gravity).

To proceed with this aim, we need to introduce more structure. As we mentioned before, in this framework, 
\begin{quote}
    \emph{a choice of QRF is tantamount to a choice of split of the externally distinguishable degrees of freedom into redundant (the QRF itself) and non-redundant ones (the QRF's complement)}.
\end{quote}
To this end, $\Hil_R$ needs to be ``big enough'' to absorb all $G$-induced gauge redundancy and serve as a reference for describing its complement $S$ in a gauge-invariant manner. In particular, $\Hil_R$ must carry a non-trivial unitary representation $U_R$ of $G$; this representation may be projective and need not be irreducible. This means that the representation obeys\footnote{In much of the previous literature on QRFs, $U_R$ is, for simplicity, assumed to be a non-projective unitary representation, so $c(g,g')=1$ for all $g,g'\in G$ (periodic clocks constitute an exception \cite{Chataignier:2024eil}). In the next section, we will see that it will be necessary to slightly extend to projective representations in this work.}
\begin{equation}
    U^g_R\,U^{h}_R=c(g,h)\,U^{gh}_R\,,\quad\forall\,g,h\in G\,,
\end{equation}
where $|c(g,h)|=1$ is a phase term that differs from $1$ whenever $U_R$ is projective. As $U_R^e=I_R$, we have the condition that $c(g,e)=c(e,g)=1$, for all $g\in G$. In App.~\ref{app_projrep}, we exhibit further constraints on the phases $c$, which constitute a cocycle, and we show that we have the freedom to choose it such that:
\begin{equation}\label{eq:cocycle_gauge_restriction}
    \forall g \in G\,, \qquad c(g, g^{-1})=1\,,
\end{equation}
implying
\begin{equation}\label{eq:inverse_of_UR}
    (U_R^g)^\dagger = (U_R^g)^{-1} = U_R^{(g^{-1})}\,. 
\end{equation}
Furthermore, such a cocycle also obeys the useful identities 
\begin{equation}\label{eq:conjugate_of_cocycle}
   \forall\,g,h\in G\,,\qquad c^{*}(g,h) = c(h^{-1} , g^{-1} ) \,.
\end{equation}
We will assume these choices have been made from now on.\footnote{In Sec.~\ref{Sec:genstabilizerQRF}, the $U_R^g$ will be strings of Pauli operators, and so by construction they already satisfy $c(g,g^{-1}) = 1$, since strings of Pauli operators square to the identity.}

As we already mentioned, we will for simplicity restrict to unitary tensor product representations on $\Hil_{\rm kin}$, such that $U(G)=\mathcal{G}$ takes the form 
\begin{equation}
   U^g=U^g_R\otimes U^g_S,\qquad\qquad\forall\,g\in G\,,
\end{equation}
where $U^g$ is unitary (and not projective) and $U_S$ defines a unitary representation on $\Hil_S$. Given that $U^g$ is a non-projective representation, this means that $U_S$ must be projective whenever $U_R$ is:
\begin{equation} \label{eq:USgUSh}
    U^g_S\,U^{h}_S=c^*(g,h)\,U_S^{gh}\,,\quad\forall\,g,h\in G\,.
\end{equation} 
Otherwise, we do not impose further restrictions on $U_S$ (in particular, it may be trivial, as we will see in interesting cases in Sec.~\ref{sec_genstab}). 
This will suffice for the purposes of this article (and going beyond it is somewhat non-trivial).

As an intermediate step toward an internal description, we first build the state space that is independent of the external frame. It is thus the gauge-invariant Hilbert space:
\begin{itemize}
    \item[(b')] The perspective-neutral Hilbert space $\Hil_{\rm pn}$ of $\mathcal{G}$-invariant states is the \emph{space of internally distinguishable states}, i.e.\ states that are distinguishable relative to an internal QRF. By being $\mathcal{G}$-invariant, this space is external-frame-independent. It is further neutral with respect to the choice of internal frame.
\end{itemize}
The construction of this space is simple in our case of interest.

As we are interested in the connection between Pauli stabilizer codes and QRFs, we shall henceforth restrict our exposition to $G$ being a finite Abelian group, as stabilizer groups for $[[n,k]]$ Pauli stabilizer codes are always of the form $G=\mathbb{Z}_2^{\times(n-k)}$. The following discussion of the PN-framework for finite Abelian groups is based on \cite{Hoehn:2021flk,Hoehn:2023ehz}.\footnote{See also \cite{delaHamette:2020dyi,Krumm:2020fws} for discussions of QRFs for finite groups in a different approach.} A description of the PN-framework for general groups can be found in \cite{delaHamette:2021oex} (see also \cite{Hoehn:2019fsy,Hoehn:2020epv,AliAhmad:2021adn,Vanrietvelde:2018dit,Vanrietvelde:2018pgb}).

The perspective-neutral space can be easily constructed by coherently averaging over the gauge orbits, i.e.\
\begin{equation}
    \Pi_{\rm pn}:\Hil_{\rm kin}\rightarrow\Hil_{\rm pn}\,\qquad\qquad \Pi_{\rm pn}:=\frac{1}{|G|}\sum_g\,U^g_R\otimes U^g_S\,.
\end{equation}
Here, $|G|$ denotes the cardinality of the group. $\Pi_{\rm pn}$ constitutes an orthogonal projector onto the trivial representation of $G$ in $\Hil_{\rm kin}$, cf.~App.~\ref{app:stabilizer_reps}. We can interpret $\Hil_{\rm pn}$ as the code space on which the internal relational physics is encoded. The relational data thereby assumes the role of the logical data. 

\subsection{The QRF's orientations}\label{ssec_orientations}

To make this more precise, we need to build internal relational descriptions that define the encoding and decoding of the relational data. To this end, we need distinguished variables of $R$ that we can relate $S$'s properties to. For example, if $R$ was a clock, we would need to define what its readings are. We will call the QRF's configurations as measured by these reference variables its \emph{orientations}. These need to cover a large enough parameter space to deparametrize the gauge orbits. To define them, we construct a coherent state system 
\begin{equation}\label{cohstate}
\{U_R,\ket{g}_R\}\,,\qquad\text{such that}\qquad U_R^g\ket{h}_R=c(g,h)\ket{gh}_R\,, \forall\,g,h\in{G}\,,
\end{equation}
comprising the orientation states of the frame \cite{delaHamette:2021oex}. Frame orientations are thus group-valued and not gauge-invariant; indeed, tetrads in special relativity are valued in the Lorentz group, the position variable of a particle takes value in the translation group, etc. There are many distinct possible orientation state systems. A natural way to construct one is to select a seed state $\ket{e}_R$ representing the neutral element $e$ in the Hilbert space of $R$, and then define $\ket{g}_R := U_R(g)\ket{e}_R$ for any $g \in G$ (recall that $c(g,e)=1$, so this construction is indeed consistent with \eqref{cohstate}). Different choices of seed state $\ket{e}_R$ will lead to different orientation state systems.\footnote{Beyond the choice of seed state, equivalence among cocycles gives further freedom in picking a set of orientation states. Namely, we can always redefine the cocycle via $c(g,h)\to c(g,h) f(gh)f(g)^{-1}f(h)^{-1}$, for any function $f:G \to U(1)$. In practice, in Sec.~\ref{sec_genstab}, projective instances of a representation $\{U_R^g\}$ will always be ``given'', in the sense that they are constructed from a QECC's list of stabilizers, and so we will not exploit this freedom.} In this work, we will be mostly interested in \emph{complete QRFs} (see \cite{delaHamette:2021oex} for other types), which are modeled by systems $\{U_R,\ket{g}_R\}$ on which $U_R$ defines a regular action\footnote{Note that the representation $U_R$ need not be regular in order for its action on the coherent state system to be regular. For example, the action can still be regular when the coherent states comprise an overcomplete basis \cite{delaHamette:2021oex}.} and which form a basis of $\Hil_R$. $R$'s orientations are then in one-to-one correspondence with group elements and can be used to parametrize the $\mathcal{G}$ orbits in $\Hil_{\rm kin}$. 

A QRF is called \emph{ideal} when the chosen orientation states are perfectly distinguishable,
\begin{equation} 
\braket{g}{h}_R=\delta_{g,h}\,\qquad\forall\,g,h\in G\,, \end{equation}
with $\delta_{g,h}$ the Kronecker delta  on ${G}$. Otherwise, we call it non-ideal. Idealness is a property of the representation \emph{and} chosen seed state, so that $R$ may have sharp, as well as fuzzy orientation state systems.\footnote{For instance, if $\ket{g}_R$ is an ideal QRF orientation system, then smearing it with a group Gaussian can produce a fuzzy version of it and still give rise to a covariant system, $U^g_R\,\ket{g'}_R=\ket{gg'}_R$.} 

A QRF orientation observable is given by the covariant positive operator-valued measure (POVM) defined by the orientation states, which equips orientations with a probability distribution \cite{delaHamette:2021oex,Hoehn:2019fsy,Chataignier:2024eil}.\footnote{One can consider more general covariant POVMs to model QRF orientations \cite{Carette:2023wpz} but this will not be relevant here.} In the finite group context of stabilizer codes below, these are simply given by $\{\ket{g}\!\bra{g}_R\}_{g\in G}$ and constitute a resolution of the identity
\begin{equation}
    \sum_{g\in G}\ket{g}\!\bra{g}_R=I_R\,.
\end{equation}

\subsection{Internal QRF perspectives as decodings}

We may now exploit the QRF's orientation states to ``jump into $R$'s internal perspective''. There exist two unitarily equivalent ways of doing so; both remove/isolate redundant information (that of the QRF) and give rise to a \emph{relational Schr\"odinger and Heisenberg picture}, respectively. We may think of them as two distinct decodings or (equivalently) encodings of internal descriptions.
    \begin{enumerate}
        \item The first is known as \emph{Page-Wootters reduction} and conditions on $R$'s orientation, which is nothing but a gauge fixing.\footnote{It is a generalization of the Page-Wootters construction for relational quantum dynamics \cite{Page:1983uc} to general groups \cite{delaHamette:2021oex}.} As such, it is unitary (on $\Hil_{\rm pn}$) and for finite groups it reads \cite{Hoehn:2021flk,Hoehn:2023ehz}:\footnote{These references discuss proper unitary representations of $G$ on $\mathcal{H}_R$. However, it can be easily checked that $\mathcal{R}_R^g$ remains unitary also for projective unitary representations, in particular that $\mathcal{R}^g_R\left(\mathcal{R}^g_R\right)^\dag=\Pi_{|R}$ and $\left(\mathcal{R}_R^g\right)^\dag\mathcal{R}_R^g=\Pi_{\rm phys}$.}
\begin{equation}\label{PWred}
        \mathcal{R}_R^g:\Hil_{\rm pn}\rightarrow\Hil_{|R}\,,\qquad \mathcal{R}_R^g:=\sqrt{|{G}|}\left(\bra{g}_R\otimes I_S\right)\Pi_{\rm pn} \,,
        \end{equation}
where the internal $R$-perspective Hilbert space is 
\begin{equation}
        \Hil_{|R}=\Pi_{|R}\left(\Hil_S\right)\,,\qquad\qquad \Pi_{|R}:=\mathcal{R}^g_R\,\Pi_{\rm pn}\left(\mathcal{R}^g_R\right)^\dag=\sum_{g'\in G} c(g',g)c^*(g,g')\,\braket{e}{g'}_R\, U_S^{g'}\,.
    \end{equation}
    Here, $\Pi_{|R}$ is an orthogonal projector onto the subspace of $S$ consistent with the gauge-invariant states $\ket{\psi}_{\rm pn}$ \cite{delaHamette:2021oex}; for ideal QRFs one has $\Pi_{|R}=I_S$.
Its image is interpreted as describing the states of $S$ relative to $R$. We note that the cocycle term  in $\Pi_{|R}$ obeys 
\begin{equation}\label{eq:cocycle1}c(g',g)c^*(g,g')=1
\end{equation}
if and only if the projective representation of $G$ is Abelian, i.e.\ $U^g_RU^{g'}_R=U_R^{g'}U_R^g$, in which case the form of the projector simplifies (becoming trivial for ideal QRFs).

We denote by $\ket{\psi(g)}_{|R}:=\mathcal{R}_R^g\ket{\psi}_{\rm pn}$ the reduced state associated to $\ket{\psi}_{\rm pn}$ when the frame is in orientation $g\in G$.\footnote{We draw the reader's attention to the slight abuse of notation in equation \eqref{PWred2}: $\ket{\psi(g)}_{|R}$ simply denotes a state that depends on $\ket{\psi}_{\rm pn}$ and $g$; there is no function $\psi$ on the group $G$ that one can evaluate to obtain a state label $\psi(g)$, despite what the notation may suggest. Such notation is common in the literature on the Page-Wooters formalism, so we also adopt it here.} These reduced states obey the covariance 
\begin{equation}
\ket{\psi(g)}_{|R}=U_S^g\,\ket{\psi(e)}_{|R}\,,\label{PWred2}
\end{equation}
which follows from the invariance of $\ket{\psi}_{\rm pn}$ and the covariance of the QRF orientation states in Eq.~\eqref{cohstate}; this is the origin of the name \emph{relational Schr\"odinger picture} (e.g.\ when $G$ is the time translation group  this yields the unitary time evolution of system states \cite{Page:1983uc,Hoehn:2019fsy,Hoehn:2020epv,Chataignier:2024eil}). 

Clearly, $(\mathcal{R}_R^g)^\dag(\bullet)\mathcal{R}_R^g:\cB(\Hil_{|R})\rightarrow\cB(\Hil_{\rm pn})$ and this map embeds $S$-observables as QRF-relational observables into the algebra on the perspective-neutral space \cite{Hoehn:2019fsy,Hoehn:2021flk,delaHamette:2021oex}:
\begin{equation}\label{relobs}
    O_{f_S|R}=\left(\mathcal{R}_R^g\right)^\dag\,f_S\,\mathcal{R}_R^g=|G|\,\Pi_{\rm pn}\left(\ket{g}\!\bra{g}_R\otimes f_S\right)\,\Pi_{\rm pn}=:\mathcal{N}_g(f_S)\,.
\end{equation}
This measures the system observable $S$, conditional on the QRF $R$ being in orientation $g$. In fact, whenever Eq.~\eqref{eq:cocycle1} holds, this relational observable description constitutes a \emph{covariant} encoding of the internal data into the perspective-neutral space as 
\begin{equation}
  \mathcal{N}_g\left(U_S^{g'}\,f_S\,\left(U_S^{g'}\right)^\dag \right)=\left(I_{R}\otimes U_S^{g'}\right)\mathcal{N}_g(f_S)\left(I_{R}\otimes U_S^{g'}\right)^\dag \,.
\end{equation}
The relational observables thus faithfully encode the internal description of $S$ relative to $R$ on the ``code space'' $\Hil_{\rm pn}$. Thus, they assume the role of encoded logical operators. This is specifically true for proper projective representations. When Eq.~\eqref{eq:cocycle1} does not hold, $I_R\otimes U_S^g$ does not commute with the gauge transformations $U_R^{g'}\otimes U_S^{g'}$ and therefore does not correspond to a physical operation. In that case, the notion of covariance has to be adapted. As this will not be relevant in the remainder, we shall not pursue this here.

   \item The second is based on the \emph{trivialization} or \emph{QRF-disentangler}, which reads
    \begin{eqnarray}
      && {T}_R:\Hil_{\rm pn}\rightarrow \ket{1}_R\otimes \Hil_{|R}\,,\qquad \qquad{T}_R=\sum_{g\in G}\ket{g}\!\bra{g}_R\otimes (U^g_S)^\dag\label{triv}
      \end{eqnarray}
      and acts as 
      \begin{eqnarray}
          \ket{1}_R\otimes\ket{\psi(e)}_{|R}={T}_R\ket{\psi}_{\rm pn}\,,\qquad\qquad\ket{1}_R:=\frac{1}{\sqrt{|{G}|}}\sum_{g\in G}\ket{g}_R\, \label{eq:R_ready}
    \end{eqnarray}
   on perspective-neutral states.~The latter equation explains why it is called a disentangler: $\ket{\psi}_{\rm pn}$ is entangled with respect to $\Hil_{\rm kin}\simeq\Hil_R\otimes\Hil_S$ when $U_S$ is non-trivial, so ${T}_R$ disentangles the selected redundant QRF subsystem $R$ from its complement $S$.

   In the remainder of our discussion of the disentangler, let us now assume that the QRF is \emph{ideal}, in which case a few additional properties hold. First, in this case, $T_R$ is unitary on \emph{all} of $\mathcal{H}_{\rm kin}$. Second, $T_R$ trivializes gauge transformations to exclusively act on the selected redundant subsystem $R$ via a new representation, $V_R^g$:
\begin{equation}\label{eq:trivialization_projective_sharp}
 T_R (U_R^g \otimes U_S^g ) T_R^\dagger \underset{\text{ideal QRF}}{=} V_R^g \otimes I_S \,, \qquad \mathrm{with} \quad V_R^g := U_R^g \sum_{h \in G}  c^*(g,h) \ket{h}\!\bra{h}_R=\sum_{g\in G}\ket{gh}\!\bra{h}_R \,.
\end{equation}
We can easily check that $V_R^g$ is unitary for any $g\in G$. Furthermore, for any $g,h \in G$, we find that
\begin{align}\label{eq:rep_VR}
    V_R^g \ket{h}_R &= U_R^g \sum_{k \in G} c^*(g,k) \ket{k}\!\bra{k}_R \ket{h}_R \underset{\text{ideal QRF}}{=} U_R^g \, c^*(g,h) \ket{h}_R = c(g,h)c^*(g,h) \ket{gh}_R \nonumber \\
    &= \ket{gh}_R\,.
\end{align}
Hence $V_R: g \mapsto V_R^g$ defines a regular representation of $G$; in particular, it is a non-projective unitary representation, even when $U_R$ is projective. This is not surprising, given that $T_R$ is unitary and that $U_R\otimes U_S$ furnishes a proper unitary representation of $G$.
    
     The trivialization is often followed by a Page-Wootters reduction to remove the redundant ``ready state'' $\ket{1}_R$, resulting in a \emph{relational Heisenberg picture} on $\Hil_{|R}$, in which observables, not states, transform unitarily under $\mathcal{G}$ \cite{Hoehn:2021flk,delaHamette:2021oex,Hoehn:2019fsy,Chataignier:2024eil}; this will however not be relevant here.
     
     We can interpret $\ket{1}_R$ as the QRF ``ancilla ready state'' absorbing the redundancy, and so similarly to the encoding Eq.~\eqref{encoding}, we have (for ideal QRFs)
\begin{equation}\label{eq:observables_triv}
        O_{f_S|R}=T_R^\dag\left(\ket{1}\!\bra{1}_R\otimes f_S\right)T_R=|G|\,\Pi_{\rm pn}\left(\ket{e}\!\bra{e}_R\otimes f_S\right)\,\Pi_{\rm pn}\,.
    \end{equation}
   When $U_S$ is non-trivial, this is a non-local unitary, spreading the information about $f_S$ across $R$ and $S$. We may thus legitimately view the QRF-disentangler as an en- and de-coding of internal relational data, in line with the interpretation that the latter assumes the role of logical data in QEC. As with the Page-Wootters maps, it constitutes a covariant encoding.

\end{enumerate}

Both reduction maps remove the redundant information of $R$, yielding the redundancy-free description of $S$ on the Hilbert space $\Hil_{|R}$, which assumes the role of the logical Hilbert space. The logical information is the description of $S$ relative to $R$. 

\subsection{Changing internal QRF perspective}

Removing redundancy from the description of $\Hil_{\rm pn}$ in terms of kinematical variables is not unique. Indeed, the redundancy has now moved into there existing \emph{many} redundancy-free descriptions and reduced Hilbert spaces. 
These correspond to different internal QRF choices and their perspectives on distinct complementary systems $S$.
We can change the internal perspective with a \emph{QRF transformation} and these take the form of \emph{quantum coordinate transformations}. Suppose $R$ and $R'$ are two distinct choices of QRF subsystem, so $\Hil_{\rm kin}\simeq\Hil_R\otimes\Hil_S$ and $\Hil_{\rm kin}\simeq\Hil_{R'}\otimes\Hil_{S'}$. If $R'$ is a subsystem of $S$ (previous standard case in much of the QRF literature), then the two TPSs are equivalent. Otherwise, the change of internal QRF subsystem also implies a change of TPS of $\Hil_{\rm kin}$ relative to the external frame.

    We described two ways of ``jumping into a QRF perspective'' and so will also summarize the corresponding changes of QRF. For expositional simplicity leaving the change of TPS on $\Hil_{\rm kin}$ implicit here (see App.~\ref{app_QRFchange} and \cite{TBH} for an explicit discussion of including the kinematical TPS change in the QRF transformation), we have:
    \begin{enumerate}
        \item In the relational Schr\"odinger picture the QRF transformations are given by
        \begin{equation}\label{eq:QRFtransf}
            V_{R\to R'}=\mathcal{R}_{R'}^{g'}\circ \left(\mathcal{R}_R^g\right)^\dag\,,
        \end{equation}
        where the dependence on the frame orientations $g,g'$ on the left-hand side has been omitted to ease the notation. The QRF transformation \eqref{eq:QRFtransf} takes an explicit form in the usual standard case $S=R'\otimes\tilde{S}$, e.g.\ see \cite{delaHamette:2021oex,Hoehn:2021flk,Hoehn:2019fsy}. It is a non-local unitary when $U_S$ is non-trivial and a change of TPS on $\Hil_{\rm pn}$ \cite{Hoehn:2023ehz}.
        \item The change of trivialization/disentangler similarly reads
        $
            T_{R'}\circ  T_R^\dag $, which again is a non-local unitary (unless $U_S$ is trivial).
    \end{enumerate}
Both of these QRF change maps can be viewed as a change of decoding, cf.~Eq.~\eqref{enchange}, highlighting that there are many ways to describe the physics that is distinguishable relative to internal frames. 

Note that both forms of the QRF transformations pass through the perspective-neutral space $\Hil_{\rm pn}$, just like coordinate changes pass through the coordinate-independent manifold. The perspective-neutral space thereby encodes and links all the internal QRF perspectives, thus being neutral with respect to an internal QRF perspective. This is the origin of the framework's name.

\subsection{Reversible maps out of the perspective-neutral space}

What, now, are interesting maps that take us out of the perspective-neutral space and into the ambient $\Hil_{\rm kin}$ that admit the interpretation of correctable ``errors''? We will be somewhat brief on this point, both because the general idea has already been discussed in Sec.~\ref{sec_dictionary} and because we will discuss this much more extensively in line with the error duality in Sec.~\ref{sec_errorduality}. This will also include syndrome measurements and recovery operations, the discussion of which we thus defer to that section. 

One natural class of maps are what we called \emph{electric charge excitations} before, $E_i:\Hil_{\rm pn}\rightarrow \Hil_i$, which take us into the charged spaces $\Hil_i$, corresponding to the non-trivial irreps of ${G}$. In Abelian theories, these are all isomorphic and so the maps are unitary. Indeed, in Abelian gauge theories, these spaces describe sectors of the fields now carrying a non-trivial ``electric'' charge. These break gauge invariance when no additional degrees of freedom are added to compensate for the charge (e.g.\ electrons in QED).\footnote{Charge sectors as resources for quantum encryption and QEC have been explored in Refs.~\cite{Kitaev:2003zj} and \cite{Bao:2023lrb}, respectively.} As unitaries, the $E_i$ will satisfy the KL condition \eqref{KL}, provided  that no two distinct maps to the sector $i$ are included in the ``error set''. For example, $E_i$ and $U_i\,E_i\,U_{\rm pn}$ with non-trivial unitaries $U_i,U_{\rm pn}$ on target and source spaces, respectively, would yield two maps that are indistinguishable with an electric charge measurement, yet they have different inverses and so recovery operations.

On the other hand, in operational quantum communication scenarios between two agents without shared external frame, the $\Hil_i$ correspond to non-trivial superselection sectors for the external frame transformations \cite{Bartlett:2006tzx}. As argued in \cite[Sec.~II B 3]{Hoehn:2023ehz}, transitioning to non-trivial charge sectors means reintroducing external frame information into the description. When $E_i$ is unitary, however, this means that the external-frame-independent data is not irretrievably lost. In fact, explicitly including non-trivial charge sectors in the description connects the PN-framework with the older quantum information-theoretic approach to QRFs \cite{Bartlett:2006tzx} of which it is essentially the trivial charge sector. Coincidentally, the quantum information-theoretic approach has its origins in charge superselection sector considerations in gauge theories \cite{Aharonov:1967zza,Aharonov:1967zz}.

 Another natural set of reversible maps from the invariant space $\Hil_{\rm pn}$ into $\Hil_{\rm kin}$ is given by gauge fixings. For simplicity, let us assume that $R$ is an ideal QRF, as this will suffice for most of our considerations in sections to come. The maps $\hat{E}_g:\Hil_{\rm pn}\rightarrow\hat{\Hil}_g\subset\Hil_{\rm kin}$ into the QRF orientation  gauge-fixed spaces are given by the projectors
\begin{equation}
    \hat{P}_g=\ket{g}\!\bra{g}_R\otimes I_S\,,\qquad g\in{G}\,.
\end{equation}
By the covariance Eq.~\eqref{cohstate}, $\mathcal{G}$ acts transitively on the projectors, $U^g\,\hat{P}_{g'}\,(U^{g})^\dag=\hat{P}_{gg'}$. It is clear that this gauge fixing is equivalent to the Page-Wootters reduction in Eq.~\eqref{PWred} upon the replacement $\bra{g}_R\mapsto \ket{g}\!\bra{g}_R$ in Eq.~\eqref{PWred} and thus unitary (up to a rescaling by $1/\sqrt{|{G}|})$.
Hence, 
\begin{equation}
    \hat{\Hil}_g=\ket{g}_R\otimes\Hil_{|R}\,,
\end{equation}
    and the gauge fixing yields QRF-aligned states \cite{Hoehn:2021flk}
\begin{equation}\label{alignedstate}
    \hat{P}_g\,\ket{\psi}_{\rm pn}=\frac{1}{\sqrt{| G|}}\ket{g}_R\otimes\ket{\psi(g)}_{|R}
    \end{equation}
    with
    \begin{equation}
 U^{g'}\,\hat{P}_g\,\ket{\psi}_{\rm pn}=\hat{P}_{g'g}\ket{\psi}_{\rm pn}=\frac{1}{\sqrt{| G|}}\ket{g'g}_R\otimes\ket{\psi(g'g)}_{|R}\,.
\end{equation}
Clearly, as a gauge fixing does not delete gauge-invariant information, the gauge-fixed target spaces are mutually orthogonal. It is thus not surprising that also these gauge fixings satisfy the KL condition \eqref{KL} and we may thus view these as correctable errors. Again, this will be discussed in detail in Sec.~\ref{sec_errorduality}, where we will explain how these ``gauge fixing errors'' are dual to the electric ones above and we may view them as \emph{magnetic charge excitations}. There we will also explain how these gauge fixing projectors can be replaced by unitaries $\hat{E}_g$ on $\Hil_{\rm kin}$, bringing these errors on an equal footing with the electric charge excitations $E_i$.

Lastly, we note that we may view these gauge fixing errors alternatively as yet another covariant decoding, equivalent to the Page-Wootters one, where $\ket{g}_R$ in Eq.~\eqref{alignedstate} now assumes the role of the ancilla QRF ``ready state''. This decoding and encoding will ultimately be best described in terms of the unitary $\hat{E}_g$ to be described in Sec.~\ref{sec_errorduality}. For now we only note that it will be dual to the one defined by the trivialization map. Indeed, the two ready states are related by a Fourier transform over the group $G$
\begin{equation}
    \ket{1}_R=\frac{1}{\sqrt{|G|}}\sum_g\ket{g}_R\,,
\end{equation}
so that their overlap is maximally unbiased
\begin{equation}
    \braket{g}{1}_R=\frac{1}{\sqrt{|G|}}\,\qquad\qquad\forall\,g\in G\,.
\end{equation}
This will be made more precise with Pontryagin duality in Sec.~\ref{sec_errorduality}.

In summary, the PN-framework naturally has the structure of a QECC with a multitude of different encodings between which it translates and there are several natural analogs of correctable error sets.

\section{From QECCs to QRFs: General Pauli stabilizer codes}\label{sec_genstab}

In this section, we show that a generic $[[n,k]]$ Pauli stabilizer QECC\footnote{We omit the code distance $d$ in the bracket $[[n,k]]$, since we will not discuss the role it plays for QRFs here.} gives rise to QRFs according to the PN-framework with symmetry group $G = \mathbb{Z}_2^{\times (n-k)}$.~Our starting point is to identify the collection of $n$ physical qubits as the kinematical Hilbert space, $\Hil_\mrm{kin} \equiv \Hil_\mrm{physical}$.~The stabilizers of the QECC thus furnish a unitary representation of $G$ on all of $\Hil_\mrm{physical}$.~The task is then to 
\begin{enumerate}
    \item identify a tensor factorization of $\Hil_\mrm{physical}$ into system, $S$, and frame, $R$, such that the frame carries a nontrivial (possibly projective) unitary representation $\{U_R^g\}$ of $G$ (cf.~Sec.~\ref{sec:extvsintframe}), and then
    \item construct a basis for $\Hil_R$ consisting of coherent, frame orientation states $\ket{g}_R$ that transform under the unitary representation: $U_R^g \ket{g'}_R = c(g,g') \ket{gg'}_R$, where the phase $|c(g,g')| = 1$ may not itself be equal to 1 if the representation is projective (cf.~Sec.~\ref{ssec_orientations}).
\end{enumerate}

In what follows, we present two distinct methods of constructing a QRF.
The first method identifies an ideal QRF that consists of $n-k$ qubits, which are a subset of the $n$ physical qubits.
The factorization of $\Hil_\mrm{physical} = \Hil_R \otimes \Hil_S$ into frame and system is therefore a \emph{local} one with respect to the inherent tensor product structure of the $n$ physical qubits.
However, the price that we pay for this local factorization is that $U^g = U_R^g \otimes U_S^g$ acts nontrivially on $S$ as well. This means we cannot directly associate the code's logical degrees of freedom with this kinematical $S$, but only with the (non-gauge-redundant) relational description of $S$ relative to $R$. It will also have the consequence that the frame $R$ itself will not be associated with a unique (equivalence class) of correctable error sets.
The second method instead uncovers a constitutive relationship between error sets and frames. \emph{Any} given maximal correctable error set $\mathcal{E}$ will be shown to generate a frame $R'$ with special properties. It yields a \emph{nonlocal} factorization $\Hil_\mrm{physical} \simeq \Hil_{R'} \otimes \Hil_{S'}$, but still with $\dim \Hil_{R'} = 2^{n-k}$ and $\dim \Hil_{S'} = 2^k$.
There the frame-system tensor product structure (TPS) is inequivalent to the original TPS, with $R'$ and $S'$ spread across (possibly all) the original physical qubits. The action of $\mathcal{E}$ on $\Hil_\mrm{code}$, as well as the full representation $U(G)$, localize to the frame $R'$ while encoded logical manipulations localize entirely to the kinematical $S'$, thus making sharp within QRF language the QECC picture of a stabilizer code as defining $k$ virtual qubits \cite{PreskillNotes,Poulin:2005wry,Kao:2023ehh}.

We begin by reviewing the basics of stabilizer codes in Sec~\ref{Sec:generalstabilizer} before constructing local frames in Sec.~\ref{Sec:genstabilizerQRF}.
We then construct nonlocal frames in Secs.~\ref{ssec:QRFsandPontryagin}-\ref{Sec:error/QRFcorrespondence} and end by using QRF tools to construct explicit decoding and recovery operations in the remaining subsections.

\subsection{Stabilizer basics}\label{Sec:generalstabilizer}

To establish notation, we begin by reviewing the essential definitions and ingredients of Pauli stabilizer QECCs.
Once again, we will only review the most relevant points here; see, e.g., \cite{Gottesman:1997zz,NielesenChuang:book,Gottesman:2009zvw,PreskillNotes} for a more detailed exposition. Along the way, we will provide a representation-theoretic interpretation of the mathematical structure underlying stabilizer codes, which will prove useful in the rest of the paper, see Appendix~\ref{app:stabilizer_reps} for supplementary details.

The single-qubit Pauli group $\mathcal P_1$ is the matrix group consisting of all Pauli matrices with multiplicative factors $\pm 1$ and $\pm i$; that is,
\begin{equation}
    \mathcal P_1=\left\{\pm I,\pm iI,\pm X,\pm iX,\pm Y,\pm iY,\pm Z,\pm iZ\right\}\,.
\end{equation}
The $n$-qubit Pauli group $\mathcal P_n$ consists of all the $n$-fold tensor products of the above matrices.
A generic $[[n,k]]$ stabilizer code is a $2^k$-dimensional subspace $\Hil_\mrm{code}$ of a $2^n$-dimensional $n$-qubit Hilbert space $\Hil_{\rm physical}\simeq(\mathbb C^2)^{\otimes n}$ that is stabilized by an Abelian subgroup $\mathcal G \subset \mathcal{P}_n$ with $n-k$ independent and commuting generators and such that $-I\not\in\mathcal G$:
\begin{equation}
\mathcal G=\langle S_1,\dots, S_{n-k}\rangle\,.
\end{equation}
As strings of Pauli operators, the generators of $\mathcal G$ are unitary, Hermitian, and square to the identity. Equivalently, $\mathcal{G}$ can be defined as the image of a unitary representation $U: G \to {\rm Aut}(\cH_{\rm kin})$ of the abstract gauge group $G= \mathbb{Z}_2^{\times (n-k)}$, obeying two technical conditions: a) $U$ must be \emph{faithful}; b) ${\rm Im}(U)=\mathcal{G}$ should not contain $-I$.\footnote{Those two conditions are themselves equivalent to the fact that the unique $*$-algebra representation $\uU: \ell^2 (G) \to \cB(\cH_{\rm kin})$ induced by $U$ is itself \emph{faithful}; see Lemma \ref{lem:faithful_alg_rep} from Appendix \ref{app:stabilizer_reps}.} Those conditions ensure that 
$\Hil_\mrm{code}$, which is nothing but the invariant subspace of the representation $U$, is isomorphic to a collection of $k$ logical qubits, and so we have that (see Lemma \ref{lemma:dim_Hj} for a representation-theoretic proof of this fact)
\begin{equation}
    (\mathbb{C}^2)^{\otimes k} \simeq \Hil_\mrm{logical} \simeq \Hil_\mrm{code} \subset \Hil_\mrm{physical} \simeq (\mathbb{C}^2)^{\otimes n}.
\end{equation}
An explicit mapping between $\Hil_\mrm{logical}$ and $\Hil_\mrm{code}$ defines the computational encoding of the $k$ logical qubits, and it is customary to characterize it as follows.
Recall that the normalizer $\mathcal N_{\mathcal P_n}(\mathcal G)$ of $\mathcal G$ in $\mathcal P_n$ is the set of elements $N \in \mathcal P_n$ such that $N\mathcal G N^{\dagger}\in\mathcal G$.
Owing to the structure of $\mathcal G$ and $\mathcal P_n$, $\mathcal N_{\mathcal P_n}(\mathcal G)$ coincides with the centralizer of $\mathcal G$ in $\mathcal P_n$; that is, the set of elements of $\mathcal P_n$ that commute with all of $\mathcal G$.\footnote{For any $N\in\mathcal P_n$, $U\in\mathcal G$, one has $N^{\dagger} U N=\pm N^{\dagger} N U=\pm U$ so that, since $-I\not\in\mathcal G$, $ N\in\mathcal N_{\mathcal P_n}(\mathcal G)$ if and only if $N$ belongs to the centralizer of $\mathcal G$ in $\mathcal P_n$.}
Clearly $\mathcal G\subseteq\mathcal N_{\mathcal P_n}(\mathcal G)$, and so the elements of $\mathcal N_{\mathcal P_n}(\mathcal G)/\mathcal G$ act nontrivially on $\Hil_{\rm code}$ and thus correspond to encoded logical unitaries.
Choosing a basis for $\Hil_{\rm code}$ that consists of the eigenvectors of $n$  commuting elements of $\mathcal N_{\mathcal P_n}(\mathcal G)$ then defines an isomorphism between $\mathcal N_{\mathcal P_n}(\mathcal G)/\mathcal G$ and $\mathcal P_k$, and by linearity, $\cB(\Hil_\mrm{code})$ and $\cB(\Hil_\mrm{logical})$ as well.

Under this isomorphism, the logical Pauli operators $X_i, Z_i \in \mathcal{P}_k$ that act on the $i$\textsuperscript{th} logical qubit map to encoded representations $\bar{X}_i$, $\bar{Z}_i$ that satisfy
\begin{equation}
    [\bar{X}_i,\bar{X}_j]=0=[\bar{Z}_i,\bar{Z}_j]\qquad,\qquad
    [\bar{Z}_i,\bar{X}_j]=0\quad i\neq j\qquad,\qquad
    \{\bar{Z}_i,\bar{X}_i\}=0\;.
\end{equation}
The encoded logical computational basis is given by the simultaneous eigenstates of the commuting set $\{S_1,\dots,\allowbreak S_{n-k},\bar{Z}_1,\dots,\bar{Z}_k\} \subset \mathcal N_{\mathcal P_n}(\mathcal G)$.~We shall refer to the mapping of the logical computational basis from $\Hil_{\rm logical}\simeq(\mathbb C^2)^{\otimes k}$ into $\Hil_\mrm{code}$ as the \emph{computational encoding}, denoted $C_{\rm comp}$, to distinguish it from QRF-induced encodings.~Accordingly, the corresponding \emph{computational decoding} map is $\mathcal D_{\rm comp}(\cdot)=C_{\rm comp}^{-1}(\cdot)C_{\rm comp}$, and it satisfies $\mathcal D_{\rm comp}(\bar{X}_i)=X_i$, $\mathcal D_{\rm comp}(\bar{Z}_i)=Z_i$.~The encoded computational basis states are then $C_{\rm comp}\ket{z_1,\dots,z_k}_{\rm logical}$, $z_{\beta}\in\{0,1\}$, $\beta=1,\dots,k$.

Recall that errors are unitary operators that map states in $\Hil_\mrm{code}$ into states that have support in the complement of $\Hil_\mrm{code}$.
In terms of the normalizer, errors are spanned by the elements of $\mathcal{P}_n \setminus \mathcal N_{\mathcal P_n}(\mathcal G)$, and any operator in the linear span of an error set $\mathcal{E} = \{E_1, E_2, \dots, E_m\}$ that satisfies the Knill-Laflamme condition \eqref{KL} is in principle correctable.
When acting on a codeword, an error $E_i$ moves the codeword from the $+1$ eigenspace to the $-1$ eigenspace of those stabilizers with which it anticommutes.
Therefore, measuring the eigenvalues of a set of generators of $\mathcal G$ reveals information about a possible error that may have occurred and constitutes a complete syndrome measurement for stabilizer QECCs.
The set of such eigenvalues is an $(n-k)$-dimensional binary vector known as the \emph{syndrome}.

Representation theory allows to decompose both $\cH_{\rm kin}$ (resp.\ $\mathcal{B}(\Hil_{\rm kin})$) into canonically defined orthogonal subspaces $\{\cH_\chi\}$ (resp.\ $\{\cB_\chi\}$) called \emph{isotypes}. They are labeled by the inequivalent irreducible representations, or characters $\chi$ of the (Abelian) gauge group $G$, which can also be seen as elements of the dual gauge group $\hat{G}$ (see Appendix \ref{app:Pontryagin}). Any vector in a subspace labeled by $\chi$ generates a one-dimensional irreducible representation isomorphic to $\chi$. The code subspace corresponds to the trivial isotype $\cH_1 = \cH_{\rm code}$ (see the detailed discussion around \eqref{eq:H_trivial_code_pn}), while logical operators can be identified with $\cB_1$ modulo equivalence on $\cH_{\rm code}$ (see the detailed discussion around \eqref{eq:logical_operations}). Furthermore, the Pauli group itself decomposes as the disjoint union $\cP_n = \sqcup_\chi C_\chi$, where $C_\chi := \cB_\chi \cap \cP_n$ and $\Span (C_\chi) = \cB_\chi$ for any $\chi \in \hat{G}$ (see Proposition \ref{prop:decomp_Pauli}). Determining the syndrome associated to an error $E \in \cP_n$ turns out to be equivalent to determining the unique character $\chi \in \hat{G}$ such that $E \in C_\chi$. Such an error then isometrically maps the code subspace $\cH_{\rm code}$ to $\cH_\chi$, which can be interpreted as an \emph{error subspace}. See Appendix \ref{app:stabilizer_reps} for further detail.

As reviewed in \Sec{ssec_QEC} and illustrated for the 3-qubit code in Example~\ref{sssec_3qubit}, interpreting the syndrome is tantamount to making a choice of the error set $\mathcal{E}$ that one corrects for.~Having made this choice, since each $E_i \in \mathcal P_n$, a known error can be corrected by applying $E_i$ (or one equivalent to it by multiplication by a stabilizer) a second time.~Even if the original error that occurred is a nontrivial linear combination of errors in $\mathcal P_n$, the process of syndrome measurement projects the erroneous encoded state onto a state on which a single $E_i$ has acted.

\subsection{Local QRFs}\label{Sec:genstabilizerQRF}

Identifying $\Hil_\mrm{kin} \equiv \Hil_\mrm{physical} \simeq (\mathbb{C}^2)^{\otimes n}$, we already noted earlier that the stabilizer group $\mathcal{G}$ furnishes a \emph{faithful} unitary representation of $G = \mathbb{Z}_2^{\times(n-k)}$.~This follows from the fact that $G$ is Abelian and has $2^{n-k}$ elements, each of which squares to the identity.~As such, let us index the elements of $\mathcal{G}$ with the elements of $G$, that is $\mathcal{G} = \{U^g ~ | ~ g \in G\}$, in such a way that $U^{g}U^{h}=U^{gh}$ for any $g,h \in G$.~Generalizing what we have seen in the 3-qubit code example discussed in Sec.~\ref{Sec:3qubitcodeasQRF} (cf.~Eq.~\eqref{eq:3qubit:PicodePipn}, Example~\ref{Sec:3qubitcodeasQRF}), we then naturally have that the projector onto the code subspace, $\Pi_\mrm{code}$, which can be written in terms of any generating set of stabilizers as in Eq.~\eqref{eq:Pi_code_stabilizers}, coincides with the group-averaging projector, $\Pi_\mrm{pn}$, that projects onto the perspective-neutral Hilbert space:
\begin{lem}[\textbf{Equality of the code and perspective-neutral projectors}]\label{lem:PicodeasGA}
The orthogonal projector $\Pi_{\rm code}$ onto $\mathcal \Hil_{\rm code}\subset\Hil_{\rm physical}$ can be written as a coherent group averaging projector over the stabilizer group.
That is,
\begin{equation}\label{eq:stabilizerGA}
    \Pi_{\rm code}=\frac{1}{|G|}\sum_{g\in G}U^g \equiv \Pi_\mrm{pn}\;,
\end{equation}
with  $U^g$, $g\in G=\mathbb Z_2^{\times(n-k)}$, a unitary representation of the stabilizer group on $\Hil_{\rm physical}$.
\end{lem}
\begin{proof}
The proof can be found in App.~\ref{App:proofs}.
\end{proof}

Next, we will exhibit a factorization $\Hil_\mrm{physical} = \Hil_R \otimes \Hil_S$ into frame and system, with $\Hil_R$ carrying a nontrivial and faithful (possibly projective) unitary representation of $G$, and then a basis of coherent states for $\Hil_R$.~Per the correspondence, we must have that $\Hil_S \simeq \Hil_\mrm{code} \simeq \Hil_\mrm{logical}$, and so by counting dimensions, we will have that $\dim \Hil_R = 2^{n-k}$.~$R$ can thus be an \emph{ideal} frame, i.e., having perfectly distinguishable (mutually orthogonal) frame orientation states, and we will demand that this be the case.

Interestingly, it turns out that it is always possible to construct a frame such that $\Hil_R$ consist of $n-k$ of the original $n$ physical qubits.
The tensor factorization into frame and system is thus (up to permutations of tensor factors) local with respect to the original TPS on $\Hil_\mrm{physical}$.

The result follows from a chain of lemmas.
The first lemma establishes that it is always possible to choose $n-k$ qubits such that the set of Pauli strings formed by truncating the stabilizers to these qubits still forms a faithful, albeit possibly projective, unitary representation of $G$.

\begin{lem}[\textbf{Existence of a faithful representation of $G$ on a subset of physical qubits}]\label{lemma:URexistence}
    Let $0 \leq k \leq n$ be fixed.
    Let $\mathcal{G} = \{U^g ~ | ~ g \in G = \mathbb{Z}_2^{\times(n-k)}\}$, where $U^e = I^{\otimes n}$, be a Pauli stabilizer group on $\Hil \simeq (\mathbb{C}^2)^{\otimes n}$, i.e., $-I \notin \mathcal{G}$ and $\forall~g,h\in G$: $U^g \in \mathcal{P}_n$, $[U^g,U^h] = 0$, $U^g U^h = U^{gh} \in \mathcal{G}$, and $U^g = U^h$ if and only if $g = h$.
    Then, there exists a choice of $A \subseteq \{1, \dots, n\}$ consisting of $k$ distinct elements such that $\pi_A(U^g) = I^{\otimes n}$  if and only if $g=e$, where $\pi_A = \prod_{a \in A} \pi_a$ and with
    $$
    \begin{aligned}
        \pi_a ~ : ~ &\mathcal{P}_n \to \mathcal{P}_n \\
        & i^\lambda \, \Oh_1 \cdots \Oh_{a-1} \Oh_a \Oh_{a+1} \cdots \Oh_n ~ \mapsto ~ \Oh_1 \cdots \Oh_{a-1} I \Oh_{a+1} \cdots \Oh_n
    \end{aligned}
    $$
    denoting the operator that sets the $a^\text{th}$ letter of an $n$-qubit Pauli string to the identity and removes any prefactor that is not 1.
\end{lem}

\begin{proof}
The proof can be found in App.~\ref{App:proofs}.
\end{proof}

Consequently, there exists a choice of $R$ consisting of $n-k$ physical qubits and $S$ consisting of the remaining $k$ physical qubits such that the operators 
\begin{equation} \label{eq:URg_generalstab}
    U_R^g := \frac{1}{2^k} \Tr_S[\pi_S (U^g)]
\end{equation}
faithfully represent $G$.~That $\mathcal{G}_R = \{U_R^g\}$ is a representation of $G$ is guaranteed by the fact that $\{U^g\}$ is a representation of $G$ consisting of local unitaries. Faithfulness follows because $U_R^g = I^{\otimes(n-k)}$ if and only if $\pi_S (U^g) = I^{\otimes n}$, which Lemma~\ref{lemma:URexistence} guarantees only happens when $U^g = U^e \equiv I^{\otimes n}$. Defining $\pi_a$ to remove any nontrivial phase from the $U_R^g$ themselves is a matter of convenience, which ensures that $\mathcal{G}_R$ is a list of Pauli strings with no overall nontrivial phases and which simplifies the proof of Lemma~\ref{lemma:URexistence}.
(A stabilizer can at most carry a phase of $-1$; otherwise, a putative stabilizer $\tilde U$ with an overall prefactor of $\pm i$ would obey $\tilde U^2 = -I$, which is not allowed in $\mathcal{G}$.)

The representation $\mathcal{G}_R$ may be projective, however, with
\begin{equation}\label{eq:phasedefinition}
    U_R^g U_R^{g'} = c(g,g') U_R^{gg'},
\end{equation}
for a possible phase $c(g,g')$ (see also the discussion in Sec.~\ref{sec:extvsintframe} and App.~\ref{app_projrep}).
We can view Eq.~\eqref{eq:phasedefinition} as defining $c(g,g')$, since $U_R^{gg'} = U_R^{g'g}$ is also guaranteed to be an element of $\mathcal{G}_R$.
Consequently, $c(g_1, g_2)$ satisfies  
\begin{equation}\label{eq:phasecondition}
\begin{array}{ll}
    c(g_1, g_2) \in \{+1, -1\}\,, & c(g_1, g_2) = c(g_2,g_1) \quad \text{if} \;\, [U_R^{g_1},U_R^{g_2}] = 0 \\[2mm]
    c(g_1, g_2) \in \{+i, -i\}\,, & c(g_1, g_2) = -c(g_2,g_1) \quad \text{if} \;\, \{U_R^{g_1},U_R^{g_2}\} = 0,
\end{array}
\end{equation}
which is a consequence of $\mathcal G$ being Abelian.~The projectiveness of the representation does not cause any obstructions to constructing an ideal frame on $R$, however, and if we lift the unitary representation of $\mathcal{G}$ in \eqref{eq:URg_generalstab} to the space of linear operators on $\Hil_R$, then the representation is no longer projective.~Explicitly, letting
\begin{equation}
    \begin{aligned}
        \mathcal{U}^g_R ~ : ~ &\mathcal{B}(\Hil_R) ~ \to ~ \mathcal{B}(\Hil_R) \\
        & \Oh ~ \mapsto ~ U^g_R \Oh (U^g_R)^\dagger,
    \end{aligned}
\end{equation}
we then have $\mathcal{U}^g_R \circ \mathcal{U}^{g'}_R = \mathcal{U}^{gg'}_R$.

We can define the stabilizer fragments on $S$, denoted $U^g_S$, by writing
\begin{equation}
    U^g = U_R^g \otimes U_S^g
\end{equation}
for every $g \in G$, i.e., the $U^g_S$ carry any phases of $-1$ in the stabilizers.~The set $\mathcal{G}_S = \{U_S^g\}$ similarly forms a  projective unitary representation of $G$, but in general it need not be faithful.~Likewise, the representation obeys (cf. Eq.~\eqref{eq:USgUSh})
\begin{equation} \label{eq:phasedefinitionS}
    U_S^g U_S^{g'} = c^*(g,g') U_S^{gg'}
\end{equation}
as a consequence of $\mathcal{G}$ being an ordinary representation.

The next lemma establishes that there exists a \emph{seed state} in $\Hil_R$ that can be used to construct a basis of ideal frame orientation states for $\Hil_R$.
This is a state $\ket{e}_R$ such that the $2^{n-k}$ states
\begin{equation}
    \ket{g}_R = U_R^g \ket{e}_R
\end{equation}
are distinct and orthonormal.

\begin{lem}[\textbf{Existence of a seed state for local QRFs}]\label{lemma:seedstate}
    Let $\{P^0 = I^{\otimes m}, P^1, \dots, P^{2^m-1}\} \subset \mathcal{P}_m$, $m\geq1$, be a set of distinct Pauli strings with no non-unit phases that is closed under multiplication up to constants,~i.e., for all $i, j$, $P^i \neq P^j$ if $i \neq j$ and $P^i P^j \propto P^k$ for some $k$.
    Then, there exists a state $\ket{\phi} \in (\mathbb{C}^2)^{\otimes m}$ such that $\bra{\phi} P^i \ket{\phi} = 0$ $\forall~i\neq 0$.
\end{lem}

\begin{proof}
The proof can be found in App.~\ref{App:proofs}.
\end{proof}

We can now prove the central claim of this subsection:

\begin{thm}[\textbf{Local QRFs for stabilizer codes}]\label{thm:existence_of_local_qrfs}
    Let $\mathcal{G} = \{U^g ~ | ~ g \in G = \mathbb{Z}_2^{\times(n-k)}\}$ be the stabilizer group of a $[[n,k]]$ Pauli QECC, indexed such that $U^e = I^{\otimes n}$ and $U^{g}U^{h} = U^{gh}$.
    There exists a subset $R$ consisting of $n-k$ of the $n$ physical qubits such that
    \begin{enumerate}
        \item $\mathcal{G}_R = \{U_R^g ~ | ~ g \in G\}$ with $U_R^g$ defined in Eq.~\eqref{eq:URg_generalstab} is a faithful, unitary, and possibly projective representation of $G$ satisfying $U_R^g U_R^{g'} = c(g,g') U_R^{gg'}$, with $c(g,g') \in \{\pm 1, \pm i\}$ as defined in Eq.~\eqref{eq:phasedefinition}, and
        \item there further exists a seed state $\ket{e}_R \in \Hil_R$ such that the set $\{\ket{g}_R = U_R^g \ket{e}_R ~ | ~ g \in G\}$ is an orthonormal basis for $\Hil_R$ that transforms as $U_R^g \ket{g'}_R = c(g,g')\ket{gg'}_R$.
    \end{enumerate}
\end{thm}

\begin{proof}
The existence of $R$ that yields $\mathcal{G}_R$ is assured by Lemma~\ref{lemma:URexistence}.
By construction, $\mathcal{G}_R$ satisfies the assumptions of Lemma~\ref{lemma:seedstate}, and so there exists a state $\ket{\phi}_R=\ket{e}_R \in \Hil_R$ such that $\bra{e} U_R^g \ket{e}_R = 0$ for all $g \in G$, $g\neq e$.
But then, for any $g, g' \in G$, we have that
\begin{equation}
    \braket{g}{g'}_R = \bra{e} U_R^g U_R^{g'} \ket{e}_R \propto \bra{e} U_R^{gg'} \ket{e}_R \, .
\end{equation}
(Note that we used the fact that $(U_R^g)^\dagger = U_R^g$.) This vanishes for $g\neq g'$ and equals $1$ for $g=g'$.
Therefore, the set $\{ \ket{g}_R \, | \, g \in G \}$ indeed forms an orthonormal basis for $\Hil_R$ that transforms as $U_R(g) \ket{g'}_R = c(g,g') \ket{gg'}_R$.
\end{proof}

The choice of the $n-k$ physical qubits for the frame $R$ need not be unique.~For example, in a cyclic code, any choice of $n-k$ qubits could be used to construct the frame.~Moreover, even for a given choice of frame qubits, the seed states need neither be unique nor separable, as we shall demonstrate in the examples below.

\begin{example}[Non uniqueness of local QRFs and their orientation states in three-qubit codes]

Consider the 3-qubit bit flip code and the QRF with $G = \mathbb{Z}_2 \times \mathbb{Z}_2$ that we extracted from it in Sec.~\ref{Sec:3qubitcodeasQRF}.
Eqs.~\eqref{eq:3qubitOrientations}-\eqref{eq:3qubit-Urep} are a realization of the general structure developed above, with $n = 3$ and $k = 1$.
While we chose $R = 12$ for convenience, indeed any choice of the two physical qubits for $R$ would do since the code is invariant under permutations of the physical qubits.
We chose $\ket{++}_R$ (a $+1$ eigenstate of $XX$) to be the seed state, but we could equally have used, e.g., any eigenstate of an operator $\Phi=\bigotimes_{\alpha=1}^{n-k}\Phi_{\alpha}$ with $\Phi_{\alpha}\in\{X,Y,Z\}$ such that, for each Pauli string $U_R^g \neq I^{\otimes{n-k}}$, at least one letter $\Phi_\alpha$ anticommutes with the corresponding letter of $U_R^g$.
In the 3-qubit code, the non-identity two-letter strings are $\mathcal{G}_R \setminus \{II\} = \{IZ, ZI, ZZ\}$.
At least one letter of $XX$ anticommutes with a letter of each of these strings, but equally good choices would have been $\Phi \in \{YX, XY, YY\}$.
\eox
\end{example}
\begin{example}[Non-separable QRF seed states]
The seed states in the previous example are all separable but this is not necessary---and not always possible---in general.~As an illustration, consider the following 3-qubit operators, which we can view as stabilizer fragments $U_R^g$ for some code with $n-k = 3$ (and thus $G = \mathbb{Z}_2 \times \mathbb{Z}_2 \times \mathbb{Z}_2$):
\begin{center}
\begin{tabular}{c c c c c c c c c}
    $U_R^g$ & $III$ & $IIX$ & $IIY$ & $XXI$ & $IIZ$ & $XXX$ & $XXY$ & $XXZ$\\
    $g$ & $+++$ & $-++$ & $+-+$ & $++-$ & $--+$ & $-+-$ & $+--$ & $---$
\end{tabular}
\end{center}
These form a faithful, projective unitary representation of $\mathbb{Z}_2 \times \mathbb{Z}_2 \times \mathbb{Z}_2$, with a consistent assignment of group elements in the second line.
In particular, note that there exists no operator $\Phi = \Phi_1 \Phi_2 \Phi_3$ such that there is at least one $\Phi_{\alpha}$ anticommuting with one of the Pauli operators in each of the above 3-qubit Pauli strings, and so the $\Phi$-constrution is not possible here.
Nevertheless, as can be checked, the entangled state $\frac{1}{\sqrt{2}}(\ket{011}+\ket{110})$ is a valid seed state that yields a set of $2^3$ orthonormal frame orientation states by acting on it with the above Pauli strings.
\eox
\end{example}

\begin{example}[Five-qubit code]
\label{sec:5qubit}

As an example of a QECC that has less trivial error correcting properties, let us consider the perfect 5-qubit code.
This is a $[[5,1]]$ stabilizer code that corrects a single erroneous Pauli operator at an unknown location.
As such, now we have
\begin{equation}
    \mathbb{C}^2 \simeq \Hil_\mrm{logical} \simeq \Hil_\mrm{code} \subset \Hil_\mrm{physical} \simeq (\mathbb{C}^2)^{\otimes 5}.
\end{equation}
A full list of the stabilizers is
\begin{equation} \label{eq:5qubitstabilizers}
    \mathcal{G} = \left\{
    \begin{array}{cccc}
        IIIII & ZYYZI & XYIYX & YXXYI \\
        ZXIXZ & IYXXY & XXYIY & XIXZZ \\
        IZYYZ & ZIZYY & YZIZY & YYZIZ \\
        IXZZX & XZZXI & ZZXIX & YIYXX \\
    \end{array} \right\},
\end{equation}
which can be generated by any 4 independent elements.
Following convention, we take the states
\begin{equation}\label{eq:5qubitencodedlogical01}
\begin{aligned}
    \ket{\bar 0}&=\frac{1}{4}\bigl(\ket{00000}+\ket{10010}+\ket{01001}+\ket{10100}+\ket{01010}-\ket{11011}-\ket{00110}-\ket{11000}\\
    &\qquad\quad-\ket{11101}-\ket{00011}-\ket{11110}-\ket{01111}-\ket{10001}-\ket{01100}-\ket{10111}+\ket{00101}\bigr),\\[0.5em]
    \ket{\bar 1}&=\frac{1}{4}\bigl(\ket{11111}+\ket{01101}+\ket{10110}+\ket{01011}+\ket{10101}-\ket{00100}-\ket{11001}-\ket{00111}\\
    &\qquad\quad-\ket{00010}-\ket{11100}-\ket{00001}-\ket{10000}-\ket{01110}-\ket{10011}-\ket{01000}+\ket{11010}\bigr),
\end{aligned}
\end{equation}
to be the encoded logical $0$ and $1$ states that span $\Hil_\mrm{code}$.
Since here $\mathcal{G}$ is Abelian with sixteen elements, each of which squares to the identity, the stabilizer group \eqref{eq:5qubitstabilizers} is isomorphic to $G = (\mathbb{Z}_2)^{\times 4}$.

We proceed to factorize $\Hil_\mrm{physical} = \Hil_R \otimes \Hil_S$ by dividing the physical qubits into a set of 4 frame qubits and a single system qubit.
Since the 5-qubit code is cyclic, any collection of 4 qubits can be used to define $R$, with the last qubit being $S$.
Without loss of generality, let us choose $R = 1234$ and $S = 5$.
The pattern is then identical to Eqs.~\eqref{eq:3qbfUR} and \eqref{eq:3qbfUS} in the 3-qubit code.
The first four letters of each Pauli string in \eqref{eq:5qubitstabilizers} furnish a faithful representation of $G$ on $R$, with the last letter supplying the corresponding unfaithful representation on $S$, so that each stabilizer element factorizes as $U^g = U^g_R \otimes U^g_S$.
Denoting the elements of $G$ by strings of $+$'s and $-$'s, we have that:
\begin{equation}\label{eq:5qbRrep}
    \begin{array}{lllllll}
        U^{++++}_R = IIII && U^{----}_R = ZYYZ && U^{+-+-}_R = XYIY && U^{-++-}_R = YXXY\\
        U^{-+++}_R = ZXIX && U^{+-++}_R = IYXX && U^{++-+}_R = XXYI && U^{+++-}_R = XIXZ \\
        U^{+---}_R = IZYY && U^{-+--}_R = ZIZY && U^{--+-}_R = YZIZ && U^{---+}_R = YYZI \\
        U^{++--}_R = IXZZ && U^{+--+}_R = XZZX && U^{--++}_R = ZZXI &&  U^{-+-+}_R = YIYX \\
    \end{array}
\end{equation}

\begin{equation}
    \begin{array}{lllllll}
        U^{++++}_S = I && U^{----}_S = I && U^{+-+-}_S = X && U^{-++-}_S = I\\
        U^{-+++}_S = Z && U^{+-++}_S = Y && U^{++-+}_S = Y && U^{+++-}_S = Z \\
        U^{+---}_S = Z && U^{-+--}_S = Y && U^{--+-}_S = Y && U^{---+}_S = Z \\
        U^{++--}_S = X && U^{+--+}_S = I && U^{--++}_S = X && U^{-+-+}_S = X \\
    \end{array}
\end{equation}
The rule for this assignment is to match $I, X$ to $+$ and $Z, Y$ to $-$, which ensures that the $U_R^g$ and $U_S^g$ indeed represent $G$.
In contrast to the 3-qubit code example, this representation of $G$ on $R$ is projective (and so is the representation on $S$).
The fragments $U_R^g$ and $U_S^g$ therefore obey the conditions \eqref{eq:phasedefinition}, \eqref{eq:phasecondition}, and \eqref{eq:phasedefinitionS}.

The last necessary ingredient is to choose a seed state for $R$.
As before, we may conveniently choose the seed state to be the all-up $X$ eigenstate,
\begin{equation} \label{eq:5qubitseed}
    \ket{e}_R = \ket{++++}_R,
\end{equation}
but this choice is highly non-unique.
\eox
\end{example}

\subsection{QRFs and dual representations of the stabilizer group's Pontryagin dual} \label{ssec:QRFsandPontryagin}

In stabilizer codes, we start by assuming that we are given a unitary representation $U:G\to\rm{Aut}(\Hil_{\rm kin})$, $g\mapsto U^g$ of the stabilizer group. Its irreducible representations are labeled by the characters $\chi\in\hat{G}$, which are the elements of the Pontryagin dual group $\hat{G}$ (see Appendix \ref{app:Pontryagin}). This Pontryagin dual will turn out to play a central role in what follows. Let us now thus see how each ideal QRF defines unitary representations of $\hat{G}$ also. In fact, the following establishes a one-to-one correspondence between ideal QRFs and representations of $\hat G$ that are \emph{dual} to $U$, according to Definition~\ref{def:dual_reps_faithful}. 

Recall from the discussion in Sec.~\ref{Sec:genstabilizerQRF} that for each stabilizer code with \emph{faithful} representation of the stabilizer group such that $-I \notin \mathcal{G}$, there always exist partitions of the physical space into an \emph{ideal} QRF (subject to a possibly projective unitary representation of $G$ on $\mathcal{H}_R$)  and a complementary system. 
In general, there will exist many such partitions and in the following two subsections, we shall see more of them.

\begin{prop}[\textbf{QRF orientation basis/Pontryagin dual representation correspondence}]\label{prop_QRFPcorresp}
   Let $R$ be an ideal QRF with a faithful (possibly projective) unitary representation $U_R(G)=\mathcal{G}_R$ of the stabilizer group $G=\mathbb{Z}_2^{\times(n-k)}$ on $\Hil_R\simeq(\mathbb{C}^{2})^{\otimes(n-k)}$ such that $-I_R \notin \mathcal{G}_R$.~For each covariant, orthonormal QRF orientation  basis $\{\ket{g}_R\}_{g\in G}$ there is a representation $\hat{U}_R: \hat{G}\to {\rm Aut}(\cH_{R})$ of the Pontryagin dual on $\Hil_R$  defined by
    \begin{equation}
        \forall \chi \in \hat{G}\,, \qquad \hat{U}^\chi_R = \sum_{g \in G} \chi(g) \ket{g}\!\bra{g}_R\,
    \end{equation}
that is \emph{dual} to $\mathcal{G}_R$ in the sense that
\begin{equation}
    \forall  g \in G \, , \chi \in \hat{G}\,, \qquad U_R^g \,\hat{U}_R^\chi = \chi(g) \,\hat{U}^\chi_R \,U_R^g \, .
\end{equation}
Conversely, each representation $\hat{U}_R(\hat{G}) = \hat{\mathcal{G}}_R$ of $\hat{G}$ on $\Hil_R$ that is dual to $\mathcal{G}_R$ defines a covariant, orthonormal QRF orientation basis for $\Hil_R$.

Furthermore, the elements of $\hat{\mathcal{G}}_R$  have the properties of Pauli operators, namely: they square to the identity and
    \begin{equation}
        \forall \chi, \eta \in \hat{G}\,, \qquad \Tr (\hat{U}^\chi_R \hat{U}^\eta_R) = 2^{n-k} \delta_{\chi , \eta }\,.
    \end{equation}  
\end{prop}

\noindent{\bf Remark.} Though we will not prove it here, the last property implies that we can find a unitary representation $\hat{\cP}_{n-k}$ of the Pauli group, of which $\hat{\mathcal{G}}_R$ is a subgroup; furthermore, $\hat{\mathcal{G}}_R$ is maximal among subgroups of $\hat{\cP}_{n-k}$ that do not contain $-I$. We denoted the Pauli group $\hat{\cP}_{n-k}$ with a hat because it will not in general coincide with the Pauli group $\mathcal{P}_{n-k}$ containing the unitary representation $\mathcal{G}_R$ of the stabilizer group, but rather constitute a unitarily rotated copy of it.~This will become relevant in the next subsection when we apply this result to stabilizer codes and $\hat{\mathcal{G}}_R$ will comprise a set of correctable errors, which thereby need not necessarily be contained in the set of standard Pauli errors. 

\begin{proof}
    Let $\{\ket{g}_R\}_{g\in G}$ be an orthonormal group basis in $\mathcal{H}_R\simeq\mathbb{C}^{2\otimes(n-k)}$, such that $U_R^g\ket{g'}_R=c(g,g')\ket{gg'}_R$. We can then define
\begin{equation}\label{eq:projdualisotype}
    \hat{P}_g:=\ket{g}\!\bra{g}_R
\end{equation}
   and thereby (this is, the group Fourier transform Eq.~\eqref{dualfourier} of Eq.~\eqref{eq:projector_isotype-dual})
\begin{equation}\label{Uhat}
    \hat{U}_R^\chi:=\sum_g \chi(g)\hat{P}_{g}\,,
\end{equation}
where the $\chi(g)\in\{\pm1\}$ are the stabilizer characters.~Clearly, $(\hat{U}_R^\chi)^2=I_R$, $\hat{U}_R^1=I_R$, $\hat{U}_R^\chi\hat{U}_R^\eta=\hat{U}_R^\eta\hat{U}_R^\chi$ and $(\hat{U}_R^\chi)^\dag=\hat{U}_R^\chi$.~Moreover, for any $\chi,\eta\in\hat{G}$,
\begin{equation}\label{TrUU}
\hat{U}_R^\chi\hat{U}_R^\eta=\sum_{g,g'}\chi(g)\eta(g')\delta_{g,g'}\hat{P}_g=\sum_{g}\chi\eta(g)\hat{P}_g=\hat{U}_R^{\chi\eta}\,,
\end{equation}
where we used that characters are multiplicative, cf.~App.~\ref{app:Pontryagin}. Thus, the $\hat{U}_R^\chi$ form a unitary representation of an Abelian group with cardinality ${|\hat{{G}}|=|{G}|=2^{n-k}}$; this is  the Pontryagin dual $\hat{{G}}$ of the stabilizer group ${G}$, see Apps.~\ref{app:Pontryagin} and~\ref{app:reps}. 

Furthermore, using $g=g^{-1}$, we have
\begin{align}
    U_R^g \,\hat{U}^\chi_R &= U_R^g \sum_{h \in G} \chi(h) \ket{h}\!\bra{h}_R = \sum_{h \in G} \chi(h) c(g,h) \ket{gh}\!\bra{h}_R = \sum_{h \in G} \chi(g^{-1}h) c(g,g^{-1}h) \ket{h}\!\bra{g^{-1}h}_R \\
    &= \chi(g) \sum_{h \in G} \chi(h) c(g,g^{-1}h) \ket{h}\!\bra{h}_R (U_R^{g^{-1}} c^*(g^{-1},h))^\dagger \\
    &= \chi(g) \sum_{h \in G} \chi(h) \underbrace{c(g,g^{-1}h) c(g^{-1},h)}_{=1} \ket{h}\!\bra{h}_R U_R^{g} = \chi(g)\, \hat{U}^\chi_R\, U_R^{g}\,.
\end{align}
We have used  Eq.~\eqref{eq:inverse_of_UR}, as well as the cocycle condition Eq.~\eqref{eq:cocycle_cond} and Eq.~\eqref{eq:cocycle_gauge_restriction}, which allow to write
\begin{equation}
    c(g, g^{-1} h)c(g^{-1}, h) = c(g, g^{-1}) c(e,h) = 1\,.
\end{equation}
Hence, $\hat{U}_R$ is \emph{dual}, according to Definition~\ref{def:dual_reps_faithful}, to $U_R$.

    Conversely, suppose $\hat{U}_R(\hat{G})=\hat{\mathcal{G}}_R$ is a representation of $\hat{G}$ on $\Hil_R$ that is dual to $\mathcal{G}_R$. Given that $\Hil_R\simeq\mathbb{C}^{2\otimes(n-k)}$ and $U_R$ is faithful, we have by Eq.~\eqref{equaldim} that each $\hat{P}_g$ in Eq.~\eqref{eq:projector_isotype-dual} is a one-dimensional projector on $\Hil_R$. Hence, Eq.~\eqref{Pgcov} implies that $\{\hat{P}_g\}_{g\in G}$ defines a covariant and orthnormal basis $\{\ket{g}_R\}_{g\in G}$ on $\Hil_R$.

    Finally, let us verify the last claim. Any element of a dual representation $\hat{\mathcal{G}}_R$ on $\Hil_R$ can be written in the form \eqref{Uhat}. Thus, $\Tr_R(\hat{U}_R^{\chi\neq1})=0$ because $\sum_g\chi(g)=0$ for $\chi\neq1$ by Eq.~\eqref{eq:orthogonality_rel_charac} (with $\eta =1$).  Furthermore, Eq.~\eqref{TrUU} tells us that also $\Tr_R(\hat{U}_R^\chi\hat{U}_R^\eta)=\delta_{\chi,\eta}2^{n-k}$, and so, together with the other properties established above, we conclude that the $\{\hat{U}_R^\chi\}_{\chi\in\hat{G}}$ have the properties of Pauli operators. 
\end{proof}
\noindent{\bf Remark.} The family of projectors $\{\ket{g}\!\bra{g}_R\}_{g\in G}$ and the dual representation $\{\hat{U}^\chi_R \}_{\chi\in \hat{G}}$ are related by an operator-valued Fourier transform; see Appendix \ref{app:reps}.

\medskip

The observations made in the previous proposition will become relevant in the next subsection, as we start addressing the role of errors, and their correspondence to QRFs. They will play an even more central role in Sec.~\ref{sec_errorduality}, where we will introduce a notion of error duality based on the same kind of Fourier analysis.

\subsection{Error-set/QRF correspondence}\label{Sec:error/QRFcorrespondence}

In this subsection, we will show that \emph{any} maximal set of correctable Pauli errors generates a QRF with especially nice properties that elucidate the connection between redundancy and correctability. In fact, this will turn out to be a two-way correspondence between error sets and QRFs from a special class. These frames will typically not be local in the original TPS underlying the physical space.

Let us first define maximality of a set of correctable errors. 
\begin{defn}[\textbf{Maximal sets of correctable errors}]\label{def:max_error_set}
    Let $\mathcal{E}=\{E_1,\ldots,E_m\}$ be a set of errors that is correctable according to the KL condition \eqref{KL}. We shall call such a set \emph{maximal} if the coefficient matrix $C_{ij}$ has the maximal rank that is attainable among all possible correctable error sets for the given code, i.e.\ if
    \begin{equation}
        \rm{rank}(C)=\Bigg\lfloor\frac{\dim(\cH_{\rm kin})}{\dim(\cH_{\rm code})}\Bigg\rfloor\,.
    \end{equation}
(Here, $\lfloor\cdot\rfloor$ denotes the floor function.)
\end{defn}

What the maximal rank is depends on the code. For an $[[n,k]]$ stabilizer code it is $2^{n-k}$. Note that maximal sets of correctable errors can be degenerate in the sense that the matrix $C_{ij}$ is not invertible. This happens, for example, when $\mathcal{E}$ contains distinct errors whose action on the code space $\mathcal{H}_{\rm code}$ coincides. This motivates us to define equivalence classes of error sets.
\begin{defn}[\textbf{Equivalence classes of correctable error sets}]\label{def_equiv}
   Two sets of correctable errors $\mathcal{E}$ and $\mathcal{E}'$ (of not necessarily the same cardinality) are called equivalent if
   \begin{equation}
       \Span_\mathbb{C}\{E_i\,\Pi_{\rm pn}\,|\,E_i\in\mathcal{E}\}=\Span_\mathbb{C}\{E_j'\,\Pi_{\rm pn}\,|\,E_j'\in\mathcal{E}'\}\,,
   \end{equation}
   that is, if the linear span of their code space restriction coincides.~The equivalence class $[\mathcal{E}]$ is then comprised of all possible errors sets that are equivalent to $\mathcal{E}$ in this sense.
\end{defn}
 
Next, we may ask ourselves whether Pauli error sets that are maximal (in the sense of Definition \ref{def:max_error_set}) exist or not. This is a non-trivial question \emph{a priori} since $\mathcal{P}_n$ only provides a discrete set of errors. However, representation theory allows one to answer in the affirmative. Indeed, we recall in Appendix \ref{app:stabilizer_reps} that both the kinematical Hilbert space $\cH_{\rm kin}$ and the algebra of operators $\cB(\cH_{\rm kin})$ can be decomposed into isotypes associated to the representation $U$:
\begin{equation}
    \cH_{\rm kin} = \bigoplus_{\chi \in \hat{G}} \cH_\chi\qquad \mathrm{and} \qquad
    \cB(\cH_{\rm kin})= \bigoplus_{\chi \in \hat{G}} \cB_\chi\,.
\end{equation}
In addition, those decompositions are \emph{orthogonal}. Furthermore, we show in Proposition \ref{prop:decomp_Pauli} that: a) the set of Pauli operators itself decomposes as a disjoint union
\begin{equation}
    \cP_n = \bigsqcup_{\chi \in \hat{G}} C_\chi\,, \quad \mathrm{with} \quad C_\chi := \cP_n \cap \cB_\chi \quad \mathrm{and} \quad \Span(C_\chi)= \cB_\chi \quad \forall \chi \in \hat{G}\,;
\end{equation}
and b) for any $\chi \in \hat{G}$, an error $E_\chi \in C_\chi$ isometrically maps the code subspace $\cH_{\rm pn}$ to the orthogonal error subspace $\cH_\chi$. It follows that we have the following characterization of correctable error sets drawn from $\cP_n$. 

\begin{prop}[\textbf{Characterization of correctable Pauli error sets}]\label{prop:correctable_CharTh}
    Let $\cE\subset \cP_n$ be a Pauli error set containing the identity. $\mathcal{E}$ is correctable if and only if: for any $\chi \in \hat{G}$ and any $E , \tilde{E} \in \mathcal{E}\cap C_\chi$, $E \Pi_{\rm pn}$ and $\tilde{E} \Pi_{\rm pn}$ are linearly dependent (meaning that $E$ and $\tilde{E}$ coincide on the code subspace, up to a phase). 
\end{prop}

\noindent{\bf Remark.} In particular, since $\mathcal{E}$ contains the identity, a logical operation $E \in \mathcal{E}\cap C_1$ is correctable if and only if $E \Pi_{\rm pn} = \lambda \, \Pi_{\rm pn}$ for some $\lambda \in \{\pm 1 , \pm i\}$.

\medskip

The previous proposition shows that to obtain a maximal Pauli error set $\cE\subset \cP_n$ we need to first select, for every $\chi \in \hat{C}$, \emph{one} Pauli operator $E_\chi$ in the sector $C_\chi$ (which is possible since $C_\chi \neq \emptyset$). We then have $2^{n-k}$ inequivalent elements in $\cE$, which makes it maximal. Any other Pauli error we may want to add to $\cE$ after that, while preserving its correctable character, will have to be proportional to one of the operators $E_\chi$ when restricted to the code subspace. Hence, any error set obtained in this way will be equivalent to $\cE$. This leads to the following corollary.
\begin{cor}[\textbf{Characterization of maximal Pauli error sets}]\label{cor:max_error_sets}
    Maximal correctable Pauli error sets exist. Furthermore, any such set is equivalent (in the sense of Definition \ref{def:max_error_set}) to a set of the form $\cE = \{E_\chi \}_{\chi \in \hat{G}}$, where: a) $E_1 = I$; b) for any $\chi \in \hat{G}\setminus \{1\}$, $E_\chi \in C_\chi$.
\end{cor}

We will deal with equivalent maximal error sets below. Generic maximal sets are not equivalent. For example, let $\{E_\chi\}_{\chi\in\hat{G}}\subset\mathcal{P}_{n}$  be a maximal set of Pauli errors. Then $\{V_\chi E_\chi\}_{\chi\in\hat{G}}$, where $V_\chi$ is a  unitary on the isotype $\Hil_\chi$ (i.e.\ the corresponding error space) with the condition that $V_1=I$, is inequivalent if at least one $V_{\chi\neq1}\neq I_\chi$. 

We are now ready to prove, for stabilizer codes, that any (equivalence class of a) maximal set of correctable Pauli errors according to the KL condition defines a \emph{unique ideal} QRF $R$, distinguished by the fact that both gauge transformations and errors act non-trivially only on $R$ and trivially on its complementary system $S$. A converse statement also holds. Every such error set thus gives rise to a clean \emph{kinematical} partitioning between purely redundant and logical/gauge-invariant information.~As there are many inequivalent such error sets, there exist many such splits between redundant (a QRF $R$) and non-redundant (a system $S$) subsystems for each stabilizer code/perspective-neutral setup. 
\begin{thm}[\textbf{Error-set/QRF correspondence}]\label{thm_errorQRF}
Consider a $[[n,k]]$ stabilizer code and let $\mathcal E=\{E_0=I,\ldots,E_m\}\subset\mathcal P_n$ be any maximal set of correctable errors. The unital $*$-algebra $\mathcal A_R$ generated by the restriction $E_i\Pi_{\rm pn}$, $i=0,\ldots,m$, of the errors to $\mathcal H_{\rm code}$ 
\begin{itemize} 
\item[(i)] is independent of the representative $\mathcal{E}\in[\mathcal{E}]$,
\item[(ii)] contains the stabilizer group $\mathcal G$, and \item[(iii)] defines a gauge $(n-k)$-qubit Hilbert space $\mathcal H_R\simeq(\mathbb{C}^2)^{\otimes (n-k)}$ such that $\mathcal A_R=\cB(\mathcal{H}_R)$.  \end{itemize}
$R$ thus constitutes a QRF with special properties:~its commutant $\mathcal{A}_S:=\mathcal{A}_R'$ defines a gauge-invariant virtual $k$-qubit system $S$ and together they define a tensor product structure $\mathcal H_{\rm kin}\simeq\mathcal{H}_R\otimes\mathcal{H}_S$ such that $\mathcal{G}$ acts trivially on $\mathcal{H}_S$, i.e.\ $U^g\mapsto U^g_R\otimes I_S$ for all $g\in{G}$.~Furthermore, $R$ admits orthogonal orientation states and is thus an ideal QRF.

Conversely, let $\Hil_{\rm kin}\simeq\Hil_R\otimes\Hil_S$ with $\Hil_R\simeq(\mathbb{C}^2)^{\otimes (n-k)}$ and suppose $\mathcal{G}=U(G)$ is faithful on $\Hil_R$ and acts trivially on $\Hil_S$. Every such QRF $R$ is associated with a unique equivalence class $[\hat{\mathcal{E}}]$ of maximal sets of correctable errors distinguished by the property that their code space restriction acts trivially on $S$, i.e.\ $\hat{E}_\chi\Pi_{\rm pn}=E^\chi_R\otimes I_S$ for all $\hat{E}_\chi\in\hat{\mathcal{E}}$. This is the equivalence class $[\hat{\mathcal{E}}]=[\hat{\mathcal{G}}]$ defined by containing all the dual representations on $\Hil_R$ of the Pontryagin dual $\hat{{G}}$  of the stabilizer group ${G}$.
\end{thm}

\noindent{\bf Remark.}
The dual representations $\hat{\mathcal{G}}=\hat{U}(\hat{G})$ of the Pontryagin dual  thus each comprise maximal sets of correctable errors.~It is in general not obvious whether $[\hat{\mathcal{E}}]$ always contains an error set $\mathcal{E}$ comprised exclusively of Pauli errors from the original set $\mathcal{P}_{n}$.~However, by Proposition \ref{prop_QRFPcorresp}, the dual representations $\hat{\mathcal{G}}$ are each guaranteed to reside in a possibly unitarily rotated Pauli group copy $\hat{\mathcal{P}}_n$.~In the examples below it will always be possible to choose a dual representation  $\hat{\mathcal{G}}$ that is contained in $\mathcal{P}_n$. \\

As only the code space restrictions of the errors matter, we thus see that the correspondence is between certain equivalence classes of maximal sets of correctable errors and QRFs such that these restricted errors  generate the stabilizer group \emph{and} the QRF.
In particular, the stabilizer gauge transformations are  generated by (the code space restriction of) the errors in \emph{any} maximal correctable set of Pauli operators.

\begin{proof}
Since the errors are assumed to be Pauli group elements, we can without loss of generality take $E_i=E_i^\dag$. Furthermore, the code is degenerate with respect to $\mathcal E$ only if there exist $E_i,E_j\in\mathcal E$ such that $E_iE_j\in\mathcal{G}$, in which case their action on $\mathcal{H}_{\rm pn}$ is identical and we place them in an equivalence class, $E_i\sim E_j$. It suffices to consider a single (arbitrary) representative from each class and so, without loss of generality, we can henceforth assume that the code is nondegenerate with respect to $\mathcal{E}$. The KL condition \eqref{KL} then reads \cite{Gottesman:1997zz}
\begin{equation}\label{KLstab}
    \Pi_{\rm pn} E_iE_j\Pi_{\rm pn}=\delta_{ij}\Pi_{\rm pn}\,,
\end{equation}
where the index runs over $i=0,\ldots,2^{n-k}-1$ because we assumed $\mathcal E$ to be maximal.

Consider now the $*$-algebra $\mathcal{A}_R$ generated by $E_i\Pi_{\rm pn}$. Using the KL condition \eqref{KLstab}, it may be checked that
\begin{equation}
    \mathcal{A}_R=\Span_\mathbb{C}\{E_i\Pi_{\rm pn}E_j, \forall\,i,j=0,\ldots,2^{n-k}-1\}\,,
\end{equation}
which also clarifies that $\mathcal{A}_R$ is independent of the representative of $[\mathcal{E}]$. This is a $4^{n-k}$-dimensional complex vector space. To see this, consider an arbitrary linear combination of the basis elements $\sum_{i,j}\lambda_{ij}E_i\Pi_{\rm pn} E_j$ and act with $\Pi_{\rm pn}\,E_k$ from the left and with $E_l\,\Pi_{\rm pn}$ from the right on it. Thanks to Eq.~\eqref{KLstab}, this yields $\lambda_{kl}\,\Pi_{\rm pn}$. Since this holds for all $k,l=0,\ldots,2^{n-k}-1$ and the $\lambda_{ij}$ are arbitrary, we conclude that the $4^{n-k}$ basis elements are indeed linearly independent.

The KL condition \eqref{KLstab} further implies
\begin{equation}\label{eprod}
   ( E_i\Pi_{\rm pn}E_j)(E_k\Pi_{\rm pn}E_l)=\delta_{jk}E_i\Pi_{\rm pn}E_l
\end{equation}
and so the $2^{n-k}$ diagonal elements $E_i\Pi_{\rm pn}E_i$ are mutually orthogonal projectors of rank $2^k$ since 
\begin{equation}
   \Tr(E_i\,\Pi_{\rm pn}\,E_i)=\Tr\Pi_{\rm pn}=2^k
\end{equation} 
by the cyclicity of the trace. Moreover, the basis elements are Hilbert-Schmidt-orthogonal. This follows from Eq.~\eqref{eprod} and  
\begin{equation}
    \Tr (E_i\,\Pi_{\rm pn}\,E_j)=\Tr (E_i\,\Pi_{\rm pn} \Pi_{\rm pn}\,E_j)= \Tr(\Pi_{\rm pn}\,E_jE_i \Pi_{\rm pn})= \delta_{ij} \Tr(\Pi_{\rm pn})\,.
\end{equation}
Finally, this algebras is unital with
\begin{equation}
    \sum_{i=0}^{2^{n-k}-1} E_i\Pi_{\rm pn}E_i=I\,.
\end{equation}
Thus, $\mathcal{A}_R$ is isomorphic to the algebra of complex $2^{n-k}\times2^{n-k}$ matrices with Hilbert-Schmidt inner product and, taking the rank into account, we can write $\mathcal{A}_R\simeq\cB((\mathbb{C}^2)^{\otimes(n-k)})\otimes I_{2^k\times 2^k}$. 

Next, we show that $\mathcal{G}\subset\mathcal{A}_R$. To this end, let us rely on a more canonical labeling of the errors $\cE = \{ E_i \}_{0\leq i\leq 2^{n-k}-1}$, which makes their commutation relations with the gauge group elements $\{ U^g \}_{g \in G}$ manifest. As established in Corollary \ref{cor:max_error_sets} based on results reviewed in Appendix \ref{app:stabilizer_reps} (see also the Appendices \ref{app:Pontryagin} and \ref{app:reps} for related material), we can, in fact, label the elements of any maximal Pauli error set by the characters of $G$: $\cE = \{ E_\chi \}_{\chi \in \hat{G}}$, with the condition that $E_1 = I$ (where $1$ denotes the trivial character). For any $\chi \in \hat{G}$ and $g \in G$, we then have: $E_\chi U^g = \chi(g) U^g E_\chi$, with $\chi(g) \in \{\pm 1\}$. Consider then the projector associated to $\chi\neq 1$, 
\begin{equation}\label{Pj}
    E_\chi \Pi_{\rm pn}E_\chi =\frac{1}{2^{n-k}}\sum_{g\in G}\chi(g)U^g\,,
\end{equation} 
which coincides with the projector $P_\chi$ onto the isotype of $\cH_{\rm kin}$ labeled by $\chi$ (see Eqs.~\eqref{eq:projectors_on_isotypes} and Proposition \eqref{prop:decomp_Pauli}), and is clearly gauge-invariant. For any $g \in G$, we have 
\begin{equation}\label{cig}
    \chi(g)I=(E_\chi U^g)^2=(U^g E_\chi)^2\,,
\end{equation}
and so $1(g)=\chi(e)=+1$ (as one should expect for characters). Furthermore, for any $\chi \in \hat{G}$ and $g \in G$,
\begin{equation}\label{eq:ev_U}
    U^gE_\chi\Pi_{\rm pn}=\chi(g)E_\chi \Pi_{\rm pn}
\end{equation}
so the $\{\chi(g)\}$, for fixed $\chi$ (and any generating subset of $\{g\in G\}$), define the syndrome of error $E_\chi$. Now $U^g$ is a Pauli operator and thus has equally many $+1$ and $-1$ eigenvalues for $g\neq e$. Given that all the eigenvalues of $U^g$ are precisely the $\{ \chi(g) \}_{\chi\in \hat{G}}$, each appearing with multiplicity $2^k$ (by \eqref{eq:ev_U}, and because we have $2^{n-k}$ errors $E_\chi$ and $\dim\mathcal H_{\rm code}=2^k$), this implies that $\displaystyle{\sum_{\chi \in \hat{G}} \chi(g)=0}$ for $g\neq e$ (see also Eq.~\eqref{eq:orthogonality_rel_dual} evaluated with $h=e$). Similarly, we have $\displaystyle{\sum_{\chi\in \hat{G}} \chi(e)=2^{n-k}}$. 

Using Eq.~\eqref{cig}, we find the orthogonality relation (thus recovering \eqref{eq:orthogonality_rel_dual}, which can be interpreted as an orthogonality relation for dual characters)
\begin{eqnarray}
    \sum_{\chi \in \hat{G}}\chi(g)\chi(g')&=&\sum_{\chi \in \hat{G}} (U^gE_\chi)^2(E_\chi U^{g'})^2=U^g\Biggl(\sum_{\chi \in \hat{G}} E_\chi U^{gg'}E_\chi\Biggr)U^{g'}\nonumber\\
    &=&\sum_{\chi \in \hat{G}} \chi(gg')=2^{n-k}\delta_{g,g'}\,,
\end{eqnarray}
from which we infer that\footnote{This formula can be interpreted as an operator-valued Fourier transform, see Appendix \ref{app:reps}.}
\begin{equation}
    \sum_{\chi \in \hat{G}} \chi(g)E_\chi \Pi_{\rm pn}E_\chi=\frac{1}{2^{n-k}}\sum_{\chi \in \hat{G}, g' \in G}\chi(g)\chi(g')U^{g'}=U^g\,.
\end{equation}
Hence, $U^g\in\mathcal{A}_R$, for all $g\in{G}$. 
Now $\mathcal{A}_R$ was shown to be isomorphic to the algebra of (bounded) operators on $(\mathbb{C}^2)^{\otimes(n-k)}$ and is thus a (Type $\rm{I}$) factor, i.e.\ has trivial center (only multiples of the identity commute with all elements in $\mathcal{A}_R$). Its commutant $\mathcal{A}_S:=\mathcal{A}_R'$  on $\mathcal{H}_{\rm kin}$ is thus a gauge-invariant $k$-qubit algebra, a factor, and $\mathcal{A}_R\cap\mathcal{A}_S=\mathbb{C}I$. Together with the property that the join of $\mathcal{A}_R$ and $\mathcal{A}_S$ generates all of $\cB(\mathcal{H}_{\rm kin})$, this implies \cite{Zanardi:2001zz,Zanardi:2004zz} that they define a tensor product structure $\mathcal{H}_{\rm kin}\simeq\mathcal{H}_R\otimes\mathcal{H}_S$  with the claimed properties. 

Now the stabilizer group ${G}=\mathbb{Z}_2^{\times(n-k)}$ is faithfully represented on $\mathcal{H}_{\rm kin}$ (and we have $-I \notin \mathcal{G}$). Then, since $U^g\in\mathcal{A}_R$ for any $g\in{G}$, it is also faithfully represented on $\mathcal{H}_R$ and $-I_R$ is not part of this representation. According to the analysis of Appendix \ref{app:stabilizer_reps}, $\cH_R$ thus decomposes as $2^{n-k}$ non-trivial isotypes of equal size, and since $\dim(\cH_R)=2^{n-k}$, they must all be one-dimensional, and in particular irreducible. Hence, $\cH_R$ contains exactly one copy of each irreducible representation of $G$; it is therefore isomorphic to the regular representation \cite{serre1977linear, simon1996representations}. Consequently, we have that $\mathcal{H}_R\simeq\ell^2({G})$ carries the regular representation of $G$ and thereby contains an orthonormal group orientation basis $\{\ket{g}_R\}_{g\in{G}}$, with $U_R^g\ket{h}_R=\ket{gh}_R$ for any $g,h \in G$. (See also Theorem \ref{thm:nonlocal_fact}, where such a group basis will be explicitly constructed.) In conclusion, $R$ is an ideal QRF.

Conversely, let ${G}=\mathbb{Z}^{\times(n-k)}_2$ be faithfully represented on $\Hil_{\rm kin}$ and suppose $R$ is a QRF such that $U^g=U^g_R\otimes I_S$ for all $g\in{G}$. From the preceding paragraph it follows that there exists an orthonormal group basis $\ket{g}_R\in\mathcal{H}_R\simeq(\mathbb{C}^{2})^{\otimes(n-k)}$, such that $U_R^g\ket{g'}_R=\ket{gg'}_R$. We can then use Eq.~\eqref{Uhat}
to define an error set $\hat{\mathcal{E}}=\{\hat{E}_\chi\}_{\chi\in\hat{G}}$ via a dual representation of the Pontryagin dual $\hat{G}$ on $\Hil_R$ from Proposition \ref{prop_QRFPcorresp}:
\begin{equation}\label{Uhat2}
    \hat{E}_\chi:=\hat{U}^\chi_R\otimes I_S\,.
\end{equation}
Then $\hat{E}_\chi\Pi_{\rm pn}=\ket{\chi}\!\bra{1}_R\otimes I_S$, where $\ket{\chi}_R=\frac{1}{\sqrt{|\mathcal{G}|}}\sum_g \chi(g)\ket{g}_R$ and so $\ket{1}_R=\frac{1}{\sqrt{|\mathcal{G}|}}\sum_g\ket{g}_R$. Since $\braket{\chi}{\eta}=\delta_{\chi,\eta}$, we conclude that the code space restrictions of the $\hat{E}_\chi\,\Pi_{\rm pn}$ remain linearly independent, and since further $\braket{\chi}{1}_R=\frac{1}{|\mathcal{G}|}\sum_g \chi(g)=0$, for $\chi\neq1$, we also have that each such error $\hat{E}_{\chi\neq1}$ maps into the orthogonal complement of $\mathcal{H}_{\rm code}=\ket{1}_R\otimes\mathcal{H}_S$. Hence, $\hat{\mathcal{E}}$ obeys the KL condition Eq.~\eqref{KLstab} and is maximal. Furthermore, since $\dim\Hil_{\rm code}^\perp=(2^{n-k}-1)2^k$ and the $2^k$-dimensional $\Hil_S$ is gauge-invariant, we have that 
$
\Span_{\mathbb{C}}\{\hat{E}_\chi\,\Pi_{\rm pn}\,|\,\hat{E}_\chi\in\hat{\mathcal{E}}\}$
contains the code space restriction of all possible correctable errors whose restriction  acts trivially on $\mathcal{H}_S$.

When the group basis $\ket{g}_R$ is non-unique, each such basis yields a representation of the Pontryagin dual $\hat{{G}}$ of ${G}$ with the same properties as just discussed. In particular, since all of them  act trivially on $\Hil_S$ they are contained in $[\hat{\mathcal{E}}]$. Thanks to Proposition \ref{prop_QRFPcorresp} this accounts for all possible  dual representations of $\hat{{G}}$ on $\Hil_R$.
\end{proof}

This result implies a reinterpretation of the KL-condition for stabilizer codes in terms of QRFs, showing explicitly that a necessary and sufficient condition for an error set comprised of Pauli operators to be correctable is that it acts exclusively on a choice of redundant (QRF) degrees of freedom, leaving the logical data untouched.
\begin{cor}[\textbf{QRF-reinterpretation of the Knill-Laflamme condition}]
\label{cor:KL-via-QRFs}
    Consider again a $[[n,k]]$ stabilizer code.~A (not necessarily maximal) set $\mathcal{E}=\{E_0,\ldots,E_m\}\subset\mathcal{P}_n$ of Pauli errors is correctable if and only if there exists a subsystem partition $\mathcal{H}_{R}\otimes\mathcal{H}_S$ of $\mathcal{H}_{\rm kin}$ with $\Hil_R\simeq(\mathbb{C}^{2})^{\otimes(n-k)}$ such that both $\mathcal{G}$ and the code space restricted errors $E_i\,\Pi_{\rm pn}$ act trivially on $S$, i.e.\ $g\mapsto U^g_R\otimes I_S$ for all $g\in\mathbb{R}$ and $E_i\Pi_{\rm pn}=\tilde{E}_i\otimes I_S$ for $i=0,\ldots,m$.~Thus, $R$ is entirely redundant and comprises a QRF $R$. 
\end{cor}

\begin{proof}
    The ``only if'' direction is provided by Theorem~\ref{thm_errorQRF} because every correctable set of Pauli operators is part of some maximal one. Conversely, suppose a tensor product structure with the stated properties can be found, so that we can ignore the system factor $\mathcal{H}_S$ for our purposes and consider $\Pi_{\rm pn}^{R}(\mathcal{H}_R)$ as a code subspace in its own right, where now $\Pi_{\rm pn}^{R}=\frac{1}{|\mathcal{G}|}\sum_{g\in\mathcal{G}}\,U_R^g$. Since $\mathcal{G}=U(G)$ is faithful and devoid of $-I$, also the representation $U_R$ on $\mathcal{H}_R\simeq(\mathbb{C}^{2})^{\otimes(n-k)}$ is faithful and does not have $-I_R$ in its image, so it decomposes into $2^{n-k}$ isotypes, each of dimension $1$ (see Appendix \ref{app:stabilizer_reps}). Hence, there are no non-trivial logical operators. This implies that  \emph{any} error subset is correctable and in particular $\mathcal{E}$.\footnote{Here is an explicit error correction protocol: a) perform a projective measurement to determine in which one-dimensional isotype $(\cH_{R})_\chi$ the state lies; b) apply a unitary that maps $(\cH_{R})_\chi$ to $(\cH_{R})_1$ (any one of them will do because the two subspaces are one-dimensional).} 
    \end{proof}

\subsection{Error-correction in terms of dressing fields, and associated frame-dependent tensor product}\label{sec:dressing_field}

We have seen in Theorem \ref{thm_errorQRF} that a maximal Pauli error set defines a unique QRF splitting of the kinematical Hilbert space. The purpose of this section is twofold. We will first draw some analogies with gauge theory by explaining how such a splitting can be equivalently described in terms of a family of \emph{frame fields} $\{ R_\chi \}$ living on the constrained surface $\cH_{\rm pn}$, and taking values in the different isotypes labeled by $\chi$. A quantum error correction protocol can then be reinterpreted, in this language, as the act of \emph{dressing} gauge-variant operators (the errors) to form gauge-invariant composite operators, that hence preserve the code subspace. There is therefore a direct correspondence between (equivalence classes of) frame fields and (equivalence classes of) correctable error sets. Secondly, we will show that a complete set of frame fields $\{ R_\chi \}$ defines a unique tensor product on $\cH_{\rm kin}$, which provides a complementary way to understand the QRF splitting already derived in Theorem \ref{thm_errorQRF}.

\medskip

To begin with, let us introduce the relevant notion of \emph{frame field} for our purposes.
\begin{defn}\label{def:complete_frame} Let $\chi \in \hat{G}\setminus \{ 1\}$. A \emph{frame field} with \emph{charge number} $\chi$ on $\cH_{\rm pn}$ is defined to be an isometry 
\begin{equation}
    R_\chi : \cH_{\rm pn} \to \cH_\chi\,.
\end{equation} 
It transforms under the gauge group as
    \begin{equation}
        R_\chi \to U^g \triangleright R_\chi 
        :=U^g R_\chi U^{g^{-1}} =U^g R_\chi = \chi(g)\, R_\chi\,.
    \end{equation}
A \emph{complete set of frame fields} for the full Hilbert space $\cH_{\rm kin}$ is a collection $R= \{  R_\chi \,\vert\,  \chi \in \hat{G}\setminus \{1\}\}$ of frame fields, one for each non-trivial charge number.
\end{defn}
We will comment on the connection of complete frame fields with our previous notion of QRF shortly.

Let $R$ be a complete set of frame fields, $\mathcal{E} \subset \mathcal{P}_n$ a maximal set of correctable errors containing the identity, and $\chi \in \hat{G}\setminus \{1\}$. An elementary error $E_\chi \in C_\chi \cap \mathcal{E}$ brings the code subspace into the non-trivially charged sector $\cH_\chi$ (see Proposition \ref{prop:decomp_Pauli}). The restriction $E_\chi \Pi_{\rm pn}$ thus transforms in the same way as $R_\chi$ under the gauge group, namely:
\begin{equation}
    E_\chi \Pi_{\rm pn} \rightarrow U^g\triangleright ( E_\chi \Pi_{\rm pn} ) = \chi(g) \, E_\chi \Pi_{\rm pn}\,.
\end{equation}
Henceforth, we can use the reference frame as a dressing field to produce the gauge-invariant operator:
\begin{equation}
    R_\chi^{-1}E_\chi \Pi_{\rm pn}
    \,,
\end{equation}
which we can view as an element of $\cB(\cH_{\rm pn})$, that is to say a logical operation. Active error correction can then be performed in two steps.
\begin{enumerate}
    \item \emph{Error detection}: the projective measurement $\{ P_\chi\}_{\chi \in \hat{G}}$ is performed, where the projectors $\{ P_\chi \}$ are explicitly given in equation \eqref{eq:projectors_on_isotypes}. From a gauge-theoretic point of view, this amounts to measuring the \emph{charge} $\chi$ of the state. From a QECC point of view, this is in turn equivalent to a \emph{syndrome measurement}, which determines the eigenvalues of a generating subset of $\mathcal{G}$. 
    \item \emph{Error-correction}: if the outcome $\chi=1$ is obtained, the state is already gauge-invariant and nothing is done; if the charge $\chi \neq 1$ is obtained the state is \emph{dressed} by the unitary $R_\chi^{-1}$ to bring it back into the invariant sector. 
\end{enumerate}
In this scheme, an error $E_\chi \in C_\chi \cap \mathcal{E}$ with $\chi \neq 1$ (resp.\ $E_1 \in C_1 \cap \mathcal{E}$) is (appropriately) \emph{corrected} whenever $R_\chi^{-1}E_\chi \Pi_{\rm pn}$ (resp.\ $R_\chi^{-1}E_1 \Pi_{\rm pn}$) is proportional to $\Pi_{\rm pn}$. Any other error in $E\in C_\chi \cap \mathcal{E}$ with $\chi \neq 1$ is detected, but not appropriately corrected. Reciprocally, if we are given a maximal set of correctable errors $\mathcal{E}\subset \mathcal{P}_n$ (which do exist by Proposition~\ref{prop:decomp_Pauli}), we can construct a complete set of frame fields $R= \{R_\chi\}_{\chi \in \hat{G}}$ that appropriately corrects $\mathcal{E}$ via dressing: for any $\chi \neq 1$, choose an element $E_\chi \in \mathcal{E}\cap C_\chi$ and set
\begin{equation}\label{eq:frame-errors}
    R_\chi := \eta_\chi \restr{E_\chi}{\cH_1}\quad\mathrm{with}\quad \eta_\chi \in \mathbb{C}\,, \vert \eta_\chi \vert = 1\,.
\end{equation}
Such $R_\chi$ will correct any error from $\mathcal{E}\cap C_\chi$, since two errors from this set are related by a phase on $\cH_{\rm pn}$ (see Proposition \ref{prop:correctable_CharTh}). Moreover, it is clear that any frame $R=\{R_\chi\}_{\chi \in \hat{G}\setminus \{1\}}$ able to correct $\mathcal{E}$ has to be of the form~\eqref{eq:frame-errors}. 

The next proposition formalizes what we just described in the language of recovery channels. 
\begin{prop}\label{propo:frame-from-errors}
   Let $R=\{R_\chi\}_{\chi \in \hat{G}\setminus \{1\}}$ be a complete set of frame fields and $\mathcal{O}_R: \mathcal{B}(\cH_{\rm kin}) \to \mathcal{B}(\cH_{\rm pn})$ the quantum channel defined by:
\begin{equation}\label{eq:recovery_frame}
    \forall \rho \in \mathcal{B}(\cH_{\rm kin})\,, \qquad  \mathcal{O}_R (\rho) = \Pi_{\rm pn} \rho \Pi_{\rm pn} + \sum_{\chi \in \hat{G}\setminus \{1\}} R_\chi^{-1} P_\chi \rho P_\chi R_\chi\,. 
\end{equation}
$\Oh_R$ corrects any error set $\mathcal{E} \subset \mathcal{P}_n$ containing the identity such that: for any $E \in \mathcal{E}$,
\begin{itemize}
    \item if $E\in C_1$, then $E \Pi_{\rm pn}=\eta \, \Pi_{\rm pn}$ for some $\eta \in \mathbb{C}$, $\vert \eta \vert =1$;
    \item if $E\in C_\chi$ with $\chi \neq 1$, then $R_\chi^{-1} E \Pi_{\rm pn}= \eta \, \Pi_{\rm pn}$ for some $\eta \in \mathbb{C}$, $\vert \eta \vert =1$. 
\end{itemize}  
Reciprocally, given a maximal set of correctable errors $\mathcal{E}\subset \mathcal{P}_n$ containing the identity, there exists a complete set of frame fields $R$ such that $\mathcal{O}_R$ corrects $\mathcal{E}$. Moreover, $R=\{ R_\chi \}_{\chi\in \hat{G}\setminus \{1\}}$ is unique up to a choice of global phase for each $R_\chi$. 
\end{prop}

\noindent{\bf Remark.} The recovery channel $\mathcal{O}_R$ is \emph{gauge-invariant}, in the sense that
\begin{equation}
    \mathcal{O}_{U^g\triangleright R} = \mathcal{O}_R \,.
\end{equation}
Hence, the error-correction capabilities of a complete set of reference frames are preserved by gauge transformations.

\medskip

In the next theorem, we show that a complete set of frame fields $R$ (in the sense of Definition \ref{def:complete_frame}) provides a canonical splitting of the kinematical Hilbert space into logical degrees of freedom and frame degrees of freedom. Moreover, the frame degrees of freedom constitute an ideal quantum reference frame for the gauge group $\mathcal{G}$. Hence, this proves that the notion of reference frame introduced in Definition \ref{def:complete_frame} provides the exact same data as the QRF splitting induced by a maximal error set described in Theorem \ref{thm_errorQRF}.
\begin{thm}\label{thm:nonlocal_fact}
    Let $\mathcal{E}\subset \mathcal{P}_n$ be a maximal set of correctable errors containing the identity, and $R$ a complete set of frame fields adapted to $\mathcal{E}$ (in the sense of Proposition~\ref{propo:frame-from-errors}). There exists a (frame-dependent) tensor product $\otimes_R$ such that $\cH_{\rm kin} = \cH_{\rm pn}\otimes_R \cH_{\rm gauge}$ and: 
    \begin{enumerate}
        \item $\cH_{\rm gauge}$ is a Hilbert space of dimension $2^{n-k}$ carrying a unitary representation $U_R$ of $G=\mathbb{Z}_2^{\times (n-k)}$, such that: 
        \begin{equation}
            \forall g \in G\,, \qquad U^g = {\rm id}_{\cH_{\rm pn}} \otimes_R U_R^g\,.
        \end{equation}
        \item $\cH_{\rm gauge}$ admits an orthonormal basis $\{ \vert \chi \rangle \,, \chi \in \hat{G}\}$ such that:
        \begin{equation}
            \forall \chi \in \hat{G}\,, \qquad \cH_\chi = \cH_{\rm pn}\otimes_R \vert \chi \rangle
        \end{equation}
        \item For any code state $\vert \psi \rangle \in \cH_{\rm pn}$ and $\chi \in \hat{G}$:
        \begin{equation}
            R_\chi ( \vert \psi \rangle \otimes_R \vert 1 \rangle)  = \vert \psi \rangle \otimes_R \vert \chi \rangle\,.
        \end{equation}
        \item For any error $E_\chi \in \mathcal{E} \cap C_\chi$ ($\chi \in \hat{G}$), there exists $\eta \in \mathbb{C}$ with $\vert \eta \vert =1$ such that: 
        \begin{equation}
            \forall \vert \psi \rangle \in \cH_{\rm pn} \,, \qquad E_\chi ( \vert \psi \rangle \otimes_R \vert 1 \rangle)  = \eta \vert \psi \rangle \otimes_R \vert \chi \rangle\,.
        \end{equation}
        \item $U_R$ is isomorphic to the regular representation of $G$.
    \end{enumerate}
\end{thm}
\begin{proof}
We first describe $\Hil_{\rm gauge}$ and then show that $\Hil_{\rm kin}$ can be written in the stated tensor product form.   Let $\cH_{\rm gauge}$ be a Hilbert space of dimension $2^{n-k}$ with orthonormal basis $\{ \vert \chi \rangle \,, \chi \in \hat{G}\}$. The representation $U_R$ is defined by:
    \begin{equation}
        \forall g \in G\,, \forall \chi \in \hat{G}\,, \qquad U_R^g \vert \chi \rangle := \chi (g) \vert \chi \rangle\,.
    \end{equation}
    $U_R$ contains one copy of each irreducible representation of $\mathbb{Z}_2^{\times (n-k)}$, so it must be isomorphic to the regular representation \cite{serre1977linear}. This observation can be made more explicit by introducing a new orthonormal basis $\{ \vert g \rangle \,, g \in \mathbb{Z}_2^{\times (n-k)}\}$, defined as:
    \begin{equation}\label{def:group_basis}
        \forall g \in G \,, \qquad \vert g \rangle := \frac{1}{\sqrt{2^{n-k}}} \sum_{\chi \in \hat{G}} \chi (g) \vert \chi \rangle \,. 
    \end{equation}
In other words, the basis $\{\vert g \rangle\}_{g \in G}$ and $\{\vert \chi \rangle\}_{\chi \in \hat{G}}$ are related by a Fourier transform, as described in Appendices \ref{app:Pontryagin} and \ref{app:reps}. For any $g$ and $h \in G$ we then have $U_R^g \vert h \rangle = \vert g h \rangle$, and 
\begin{equation}\label{eq:orthogonality_charac}
\langle g \vert h \rangle = \frac{1}{2^{n-k}} \sum_{\chi \in \hat{G}} \bar\chi (g) \chi (h)  = \delta_{g , h}
\end{equation}
follows from the orthogonality relations \eqref{eq:orthogonality_rel_dual} (which can be understood as orthogonality relations between columns of the character table of $G$, or equivalently, between lines of the character table of $\hat{G}$). This establishes claim number 5.
    
Next, we introduce a tensor product on $\cH_{\rm kin}$ as the bilinear map $\otimes_R:\cH_{\rm pn}\times \cH_{\rm gauge} \rightarrow \cH_{\rm kin}$ which maps $(\vert \varphi \rangle , \sum_{\chi\in\hat{G}}w_\chi \vert \chi \rangle)$ to
    \begin{equation}
        \vert \varphi \rangle \otimes_R \left( \sum_{\chi\in\hat{G}}w_\chi \vert \chi \rangle \right) := w_1 \vert \varphi \rangle + \sum_{\chi\in \hat{G},\, \chi \neq 1} w_\chi R_\chi \vert \varphi \rangle\,.
    \end{equation}
One can verify that $\otimes_R$ does obey the universal property defining a tensor product.\footnote{See e.g.\ \cite{geroch1985mathematical} for a physicist-friendly introduction to tensor products in terms of universal properties.} In other words, we are making the following identifications: for any $\vert \varphi \rangle \in \cH_{\rm pn}$,
\begin{align}
    \vert \varphi \rangle \otimes_R \vert 1 \rangle &\equiv \vert \varphi \rangle \,, \label{eq:def_tensor1}\\
    \forall \chi\in \hat{G}, \qquad \vert \varphi \rangle \otimes_R \vert \chi \rangle &\equiv R_\chi \vert \varphi \rangle \,. \label{eq:def_tensor2}
\end{align} 
In particular, claims number 2 and 3 hold. 

For any $g \in G$, any $\vert\varphi\rangle \in \cH_{\rm pn}$, and any $\chi\in \hat{G}$,  we have $\vert \varphi \rangle \otimes_R \vert \chi \rangle \in \cH_\chi$, so that:
\begin{equation}
    U^g \left( \vert \varphi \rangle \otimes_R \vert \chi \rangle \right) = \chi(g) \left( \vert \varphi \rangle \otimes_R \vert \chi \rangle \right) =  \vert \varphi \rangle \otimes_R \chi(g) \vert \chi \rangle  =   \vert \varphi \rangle \otimes_R U_R^g\vert \chi \rangle\,, 
\end{equation}
establishing claim 1.

Finally, claim 4 follows from claim 3 and Proposition~\ref{propo:frame-from-errors}: for any $\chi \in \hat{G}$ and $E_\chi \in C_\chi \cap \mathcal{E}$, $E_\chi \Pi_{\rm pn}$ coincides with $R_\chi \Pi_{\rm pn}$ up to a phase. 
\end{proof}

\noindent{\bf Remark.} Instead of defining a tensor product $\otimes_R$ on $\cH_{\rm kin}$, which provides an intrinsic definition of our tensor product factorization, we could have adopted an extrinsic point of view by introducing an isomorphism $t_R: \cH_{\rm kin} \to \cH_{\rm pn}\otimes \cH_{\rm gauge}$. More explicitly, we can define $t_R$ as the unique linear map verifying, for any $\vert \varphi \rangle \in \cH_{\rm pn}$,
\begin{align}
    t_R \vert \varphi \rangle &:= \vert \varphi \rangle \otimes \vert 1 \rangle \,, \\
    \forall \chi\in \hat{G}, \qquad t_R (R_\chi \vert \varphi \rangle)&:= \vert \varphi \rangle \otimes \vert \chi \rangle \,.
\end{align} 
These definitions are the analogs, in the extrinsic picture, of Eqs.~\eqref{eq:def_tensor1} and \eqref{eq:def_tensor2}. They provide an explicit isomorphism realizing the factorization $\cH_{\rm kin}\simeq \cH_{S} \otimes \cH_{R}$ from Theorem \ref{thm_errorQRF}, with $\cH_{S}= \cH_{\rm pn}$ and $\cH_{R}= \cH_{\rm gauge}$. Furthermore, we note that $t_R$ then defines an isometry, which corresponds, in the intrinsic picture, to the fact that the tensor product factorization $\cH_{\rm kin }= \cH_{\rm pn }\otimes_R \cH_{\rm gauge}$ is compatible with the inner product structures on $\cH_{\rm kin }$, $\cH_{\rm pn }$ and $\cH_{\rm gauge}$.  

\medskip

Let us illustrate Theorem \ref{thm:nonlocal_fact} on the $3$-qubit code. In that case, $\cH_{\rm kin}= (\mathbb{C}^2)^{\otimes 3}$ and we can define the representation $U:\mathbb{Z}_2^{\times 2} \to {\rm Aut}(\cH_{\rm kin})$ as:
\begin{equation}\label{eq:3qStab}
    U^{++} = I\,, \quad U^{+-} = Z_2 Z_3 \,, \quad U^{-+} = Z_1 Z_3 \,, \quad U^{--} = Z_1 Z_2\,.  
\end{equation}
We can then introduce explicit labels for the irreducible characters $\hat{G}= \{ \chi_1, \chi_2, \chi_3 ,\chi_4\}$, where:
\begin{equation}
    \begin{array}{llll}
\chi_1 (++) = 1\,, & \chi_1 (+-) = 1  \,, & \chi_1 (-+) = 1 \,, & \chi_1 (--) = 1 \,, \\   
\chi_2 (++) = 1\,, & \chi_2 (+-) = 1 \,, & \chi_2 (-+) = -1 \,, & \chi_2 (--) = -1 \,, \\   
\chi_3 (++) = 1\,, & \chi_3 (+-) = -1  \,, & \chi_3 (-+) = 1 \,, & \chi_3 (--) = -1 \,, \\
\chi_4 (++) = 1\,, & \chi_4 (+-) = -1  \,, & \chi_4 (-+) = -1 \,, & \chi_4 (--) = 1 \,, 
    \end{array}
\end{equation}
which in turn define the two-dimensional isotypes:
\begin{equation}
\begin{aligned}
    \cH_1 &= \cH_{\rm pn} = \Span\{ \vert 000\rangle ,\vert 111 \rangle \}\,, \\
    \cH_2 &= \Span\{ \vert 100\rangle ,\vert 011 \rangle \}\,, \\
    \cH_3 &=  \Span\{ \vert 010\rangle ,\vert 101 \rangle \}\,, \\
    \cH_4 &=  \Span\{ \vert 001\rangle ,\vert 110 \rangle \}\,.
\end{aligned}
\end{equation}
We will discuss three examples of maximal error sets. 

\begin{example}\label{ex:3qubit_single_flips} Let $\mathcal{E}= \{ I, X_1 , X_2 , X_3 \}$. Then $R=(R_2 , R_3 , R_4) := (\restr{X_1}{\cH_1} , \restr{X_2}{\cH_1} , \restr{X_3}{\cH_1})$ is a complete set of frame fields adapted to $\mathcal E$, and the channel $\mathcal{O}_R$ appropriately corrects any error from $\mathcal E$. It induces the tensor product $\otimes_R$, which decomposes $\cH_{\rm kin}$ into $\cH_{\rm pn}\otimes_R \cH_{\rm gauge}$ according to the following identifications of basis vectors:
\begin{align}
    \vert 000 \rangle \otimes_R \vert 1 \rangle \equiv \vert 000 \rangle \,, \qquad \vert 111 \rangle \otimes_R \vert 1 \rangle \equiv \vert 111 \rangle \,, \\
    \vert 000 \rangle \otimes_R \vert 2 \rangle \equiv \vert 100 \rangle \,, \qquad \vert 111 \rangle \otimes_R \vert 2 \rangle \equiv \vert 011 \rangle \,, \nonumber \\
    \vert 000 \rangle \otimes_R \vert 3 \rangle \equiv \vert 010 \rangle \,, \qquad \vert 111 \rangle \otimes_R \vert 3 \rangle \equiv \vert 101 \rangle \,, \nonumber \\
    \vert 000 \rangle \otimes_R \vert 4 \rangle \equiv \vert 001 \rangle \,, \qquad \vert 111 \rangle \otimes_R \vert 4 \rangle \equiv \vert 110 \rangle \,. \nonumber
\end{align}
Hence, logical information is encoded into the factor $\cH_{\rm pn}$ according to the \emph{majority rule} ($\vert 000 \rangle$ for a majority of zeroes, and $\vert 111 \rangle$ for a majority of ones), while the frame state tells us where a single-bit flip error may have occurred ($\vert 1 \rangle$ for no error, and $\vert j \rangle$ with $j\geq2$ for a bit flip error on the $(j-1)^{\rm th}$ qubit). This tensor factorization is useful because it cleanly separates gauge symmetries (aka syndrome measurements) from gauge-invariant (aka logical) data. Indeed, one has $U(\cdot) = {\rm id}_{\cH_{\rm pn}} \otimes_R U_R (\cdot)$ where $U_R$ takes the following diagonal form in the representation basis $(\vert 1 \rangle, \vert 2 \rangle, \vert 3 \rangle, \vert 4 \rangle)$:  
\begin{equation}\label{eq:matrix_rep-frame}
U_R^{+-} = \begin{pmatrix}
    1 & 0 & 0& 0 \\
    0 & 1 & 0& 0 \\
    0 & 0 & -1& 0 \\
    0 & 0 & 0& -1 
\end{pmatrix}\,, \quad 
U_R^{-+} = \begin{pmatrix}
    1 & 0 & 0& 0 \\
    0 & -1 & 0& 0 \\
    0 & 0 & 1 & 0 \\
    0 & 0 & 0& -1 
\end{pmatrix}\,, \quad 
U_R^{--} = \begin{pmatrix}
    1 & 0 & 0& 0 \\
    0 & -1 & 0& 0 \\
    0 & 0 & -1& 0 \\
    0 & 0 & 0& 1 
\end{pmatrix}\,.
\end{equation}
$U_R$ is isomorphic to the regular representation of $\mathbb{Z}_2 \times \mathbb{Z}_2$, which we can make explicit by introducing the group basis:
\begin{align}
 \vert ++\rangle &:= \frac{1}{2} \left( \vert 1 \rangle + \vert 2 \rangle + \vert 3 \rangle + \vert 4 \rangle\right) \,, \label{eq:group_basis-3bit}\\
  \vert +- \rangle &:= \frac{1}{2} \left( \vert 1 \rangle + \vert 2 \rangle - \vert 3 \rangle - \vert 4 \rangle\right)\,, \nonumber \\
   \vert -+ \rangle &:= \frac{1}{2} \left( \vert 1 \rangle - \vert 2 \rangle + \vert 3 \rangle - \vert 4 \rangle\right)\,, \nonumber \\
    \vert -- \rangle &:= \frac{1}{2} \left( \vert 1 \rangle - \vert 2 \rangle - \vert 3 \rangle + \vert 4 \rangle\right)\,. \nonumber
\end{align}
It is orthonormal and verifies $U_R^g \vert h \rangle = \vert gh \rangle$ for any $g,h \in  \mathbb{Z}_2\times \mathbb{Z}_2$. Finally, by construction, we have for any $\vert \varphi \rangle \in \cH_{\rm pn}$ and $j \in \{ 1, 2 , 3\}$:
\begin{equation}\label{eq:errors_3bit}
    X_j (\vert \varphi \rangle \otimes_R \vert 1 \rangle) = \vert \varphi \rangle \otimes_R \vert j+1 \rangle\,.  
\end{equation}
Errors that are correctable by the frame $R$ (that is to say, that are appropriately taken care of by the recovery operation $\mathcal{O}_R$ defined in \eqref{eq:recovery_frame}) can thus be identified with those operators in $\mathcal{B}(\cH_{\rm kin})$  whose restriction to $\cH_1 \equiv \cH_{\rm pn}\otimes_R \vert 1 \rangle \subset \cH_{\rm pn}\otimes_R \cH_{\rm gauge}$ only act on the frame Hilbert space $\cH_{\rm gauge}$. The quotient of this linear subspace by the equivalence relation $\sim$ can be identified with ${\rm id}_{\cH_{\rm pn}} \otimes_R \cB(\cH_{\rm gauge})$. Similarly, the algebra of logical operators can be identified with $\cB(\cH_{\rm pn})\otimes_R {\rm id}_{\cH_{\rm gauge}}$.
\eox
\end{example}

\begin{example}
    Let $\mathcal{E}' = \{ I , X_1 X_2 , X_2 X_3 , X_1 X_3 \}$, which is a legitimate maximal set of correctable errors (even though it is not relevant to practical implementations, since single-bit flip errors are typically more probable than double-bit flip errors). Then $R'=(R_2' , R_3' , R_4'):= (\restr{X_2 X_3}{\cH_1} , \restr{X_1 X_3}{\cH_1} , \restr{X_1 X_2}{\cH_1})$ is a complete set of frame fields adapted to $\mathcal{E}'$, which leads to the alternative tensor factorization $\otimes_{R'}$:
\begin{align}
    \vert 000 \rangle \otimes_{R'} \vert 1 \rangle \equiv \vert 000 \rangle \,, \qquad \vert 111 \rangle \otimes_{R'} \vert 1 \rangle \equiv \vert 111 \rangle \,, \\
    \vert 000 \rangle \otimes_{R'} \vert 2 \rangle \equiv \vert 011 \rangle \,, \qquad \vert 111 \rangle \otimes_{R'} \vert 2 \rangle \equiv \vert 100 \rangle \,, \nonumber \\
    \vert 000 \rangle \otimes_{R'} \vert 3 \rangle \equiv \vert 101 \rangle \,, \qquad \vert 111 \rangle \otimes_{R'} \vert 3 \rangle \equiv \vert 010 \rangle \,, \nonumber \\
    \vert 000 \rangle \otimes_{R'} \vert 4 \rangle \equiv \vert 110 \rangle \,, \qquad \vert 111 \rangle \otimes_{R'} \vert 4 \rangle \equiv \vert 001 \rangle \,.\nonumber
\end{align}
This tensor product identifies logical information in the non-trivially charged sectors indexed by $j \geq 2$ according to the \emph{minority rule} ($\vert 000 \rangle$ for a minority of zeroes, and $\vert 111 \rangle$ for a minority of ones), while the frame degrees of freedom identify which qubit has \emph{not} been affected by an eventual double-flip error (the frame is in state $\vert j \rangle$ whenever a double-flip has occurred on the two qubits other than the $(j-1)^{\rm th}$ one). Hence $\otimes_{R'}$ induces a splitting of the degrees of freedom that is somehow dual to that induced by $\otimes_R$. Formulas \eqref{eq:matrix_rep-frame}, \eqref{eq:group_basis-3bit} and \eqref{eq:errors_3bit} remain valid if one replaces $R$ by $R'$.
\eox
\end{example}

\begin{example}\label{ex:3qubit_third_qubit_logical} Finally, let $\mathcal{E}'' = \{ I , X_1 , X_2 , X_1 X_2 \}$, which is also a maximal set of correctable errors, even though it is even less relevant  to practical applications than $\mathcal{E}'$.\footnote{Indeed, the third qubit is now error-free, so it can be understood as a logical qubit. Nothing of practical value is gained by adjoining to it the first two ancilla qubits.} The compatible set of frame fields $R''=(R_2'' , R_3'' , R_4''):= (\restr{X_1}{\cH_1} , \restr{X_2}{\cH_1} , \restr{X_1 X_2 }{\cH_1})$ obeys the conditions of Proposition~\ref{propo:frame-from-errors} and Theorem~\ref{thm:nonlocal_fact}. It leads to yet another inequivalent tensor product $\otimes_{R''}$, defined as:
\begin{align}\label{eq:tensor_decomp_system1}
    \vert 000 \rangle \otimes_{R''} \vert 1 \rangle \equiv \vert 000 \rangle \,, \qquad \vert 111 \rangle \otimes_{R''} \vert 1 \rangle \equiv \vert 111 \rangle \,, \\
    \vert 000 \rangle \otimes_{R''} \vert 2 \rangle \equiv \vert 100 \rangle \,, \qquad \vert 111 \rangle \otimes_{R''} \vert 2 \rangle \equiv \vert 011 \rangle \,, \nonumber\\
    \vert 000 \rangle \otimes_{R''} \vert 3 \rangle \equiv \vert 010 \rangle \,, \qquad \vert 111 \rangle \otimes_{R''} \vert 3 \rangle \equiv \vert 101 \rangle \,, \nonumber\\
    \vert 000 \rangle \otimes_{R''} \vert 4 \rangle \equiv \vert 110 \rangle \,, \qquad \vert 111 \rangle \otimes_{R''} \vert 4 \rangle \equiv \vert 001 \rangle \,. \nonumber
\end{align}
Logical information in a charged sector $\cH_j$ with $j \geq 2$ can now be inferred by looking at the state of the third qubit alone (the logical state is $\vert 000 \rangle$ if the third qubit is 0 and $\vert 111 \rangle$ if the third qubit is 1). This is similar to what was observed relative to the kinematical decomposition $\cH_{\rm kin} = \mathbb{C}^2 \otimes (\mathbb{C}^2)^{\otimes 2}$ (see Section \ref{Sec:3qubitcodeasQRF}). However, we clearly see from Eq.~\eqref{eq:tensor_decomp_system1} that the non-local decomposition $\cH_{\rm kin} = \cH_{\rm pn} \otimes_{R''} \cH_{\rm gauge}$ is not equivalent to the kinematical decomposition.\footnote{The two tensor products had to be inequivalent since gauge transformations factorize differently relative to them: relative to $\otimes_{R''}$, the non-frame degrees of freedom are gauge-invariant, while relative to the original kinematical tensor product they are not. What Eqs.~\eqref{eq:tensor_decomp_system1} provide is an explicit isometry mapping the two tensor product structures into one another.}
\eox
\end{example}

\subsection{More general error-set/QRF relations}\label{Sec:generalerrorQRFcorresp}

In the previous two subsections we demonstrated that every maximal correctable Pauli error set generates a QRF $R$ and an associated tensor factorization $\Hil_{\rm kin}\simeq\Hil_R\otimes\Hil_S$ such that $\mathcal{G}$ acts trivially on its complement (cf.~Thms~\ref{thm_errorQRF} and \ref{thm:nonlocal_fact}).~Conversely, as per Theorem~\ref{thm_errorQRF}, any QRF with that property is associated with a unique equivalence class of correctable errors acting exclusively on its degrees of freedom.~Equivalently, recalling the expression \eqref{triv} of the QRF-disentangler/trivialization $T_R$ associated with a QRF $R$, every maximal set of correctable Pauli errors generates a QRF such that $T_R=I_R \otimes I_S$ is trivial.~Vice versa, any ideal QRF with $T_R=I_R \otimes I_S$ is associated with a unique equivalence class of correctable errors acting exclusively on it.

However, a generic QRF does not have the property that $\mathcal{G}$ acts nontrivially only on the frame. For example, what about the somewhat more general class of QRFs such that $\mathcal{G}$ still acts in a tensor product representation on $R$ and $S$, but non-trivially (though not necessarily faithfully) on $S$? Can we in that case also associate a unique equivalence class of correctable errors with $R$ as a generalization of the converse direction of Theorem~\ref{thm_errorQRF}? (When the representation is not of product form, the discussion is significantly more complicated.) The following observation tells us that in that case there are \emph{no} errors whatsoever whose code space restriction acts exclusively on $R$.

\begin{lem}\label{lem_none}
    Let $G$ be represented in a tensor product representation $g\mapsto U_R^g\otimes U_S^g$ on $\mathcal{H}_{\rm kin}=\mathcal{H}_R\otimes\mathcal{H}_S$ with $g\mapsto U^g_R$ faithful. An operator $E\in\cB(\mathcal{H}_{\rm kin})$ such that 
\begin{equation}
        E\,\Pi_{\rm pn}=e_R\otimes I_S
\end{equation}
    for some nonvanishing operator $e_R\in\cB(\mathcal{H}_{R})$ exists if only if $U_S^g=I_S$ for all $g\in G$.
\end{lem}
\begin{proof}
The proof can be found in App.~\ref{App:proofs}.
\end{proof}

We thus have to change the question we are asking. Instead, we find the following.
\begin{prop}\label{prop_nontriv}
    Consider a QRF-system partition $\mathcal{H}_{\rm kin}=\mathcal{H}_R\otimes\mathcal{H}_S$ with $\Hil_R\simeq(\mathbb{C}^{2})^{\otimes(n-k)}$ such that ${G}=\mathbb{Z}_2^{\times(n-k)}$ is represented with a unitary tensor product representation $g\mapsto U^g_R\otimes U_S^g$, where $g\mapsto U_R^g$ is faithful and $g\mapsto U_S^g$ non-trivial (though not necessarily faithful), and where the latter two representations on $R$ and $S$ are permitted to be projective. Then, there exists at least one orthonormal frame orientation basis  $\{\ket{g}_R\,|\,g\in{G}\}$, and for each such basis there is a set $\hat{\mathcal{E}}$ of $2^{n-k}$  errors $\hat{E}_\chi$ which
    \begin{itemize}
        \item[(i)] act trivially on $S$ in their \emph{unrestricted} form, 
        \item[(ii)] form a dual representation of the Pontryagin dual group $\hat{{G}}$, and
        \item[(iii)] define the \emph{unique} equivalence class $[\hat{\mathcal{E}}]$ of maximal sets of correctable errors  distinguished by the property that their code space restriction $\hat{E}_\chi\,\Pi_{\rm pn}$ acts trivially on the new system tensor factor ${\mathcal{H}}_S$  \emph{upon a TPS-refactorization with the QRF-disentangler}~\eqref{triv}. Namely,
\begin{equation}\label{Tprop:errors}
        {T}_R\,\hat{E}_\chi\,\Pi_{\rm pn}\,{T}^\dag_R = \ket{\chi}\!\bra{1}_R\otimes {I}_S,
        \end{equation}
        where $\ket{1}_R:=\frac{1}{\sqrt{|\mathcal{G}|}}\sum_g\ket{g}_R$ and $\ket{\chi}_R = \hat{U}^\chi \ket{1}_R$, $\chi \in \hat G$, with $\hat{E}_\chi = \hat{U}^\chi \otimes I_S$.
    \end{itemize}
\end{prop}

\noindent {\bf Remark:}  Per Eq.~\eqref{eq:trivialization_projective_sharp},  if the QRF is ideal, then ${T}_R\,(U_R^g\otimes U_S^g){T}_R^\dag =V_R^g\otimes I_S$ where $V_R^g \ket{h}_R = \ket{gh}_R$ for any $g,h \in G$. Therefore, \emph{after} trivialization the stabilizer group acts exclusively on the new frame factor $R$, as for the frames in Theorem~\ref{thm_errorQRF}.

\begin{proof}
Lemma~\ref{lemma:seedstate} implies that at least one orthonormal orientation basis exists for the frame Hilbert space $\mathcal{H}_R$ when $\mathcal{G}$ acts faithfully on it. Let $\{\ket{g}_R\,|\,g\in{G}\}$ be any such basis. Proceeding as in the converse direction of the proof of Theorem~\ref{thm_errorQRF}, it is clear that one can build the errors fulfilling (i) and (ii) via Eqs.~\eqref{Uhat} and~\eqref{Uhat2}. Noting that $[{T}_R,\hat{E}_\chi]=0$ by construction and that $T_R \Pi_{\rm pn} T_R^\dagger = \frac{1}{|G|}\sum_{g \in G} V_R^g \otimes I_S$ allows us to establish (iii). A proof of this latter relation is: 
\begin{align}
    T_R \Pi_{\rm pn} T_R^\dagger &= \frac{1}{|G|}\sum_{g,h,k \in G} \left( \ket{g}\!\bra{g}_R U_R^h \ket{k}\!\bra{k}_R \right) \otimes \left( U_S^{(g^{-1})}  U_S^h U_S^k \right) \\
    &= \frac{1}{|G|} \sum_{g,h,k \in G} c(h,k) c^*(g^{-1},hk) c^*(h,k) \left( \ket{g}\!\bra{g}_R  \ket{hk}\!\bra{k}_R \right) \otimes \left( U_S^{g^{-1}hk}\right) \\
    &= \frac{1}{|G|} \sum_{h,k \in G}  \underbrace{c^*((hk)^{-1},hk)}_{=1}  \left( \ket{hk}\!\bra{k}_R \right) \otimes \left( U_S^{e}\right) \\
    &= \frac{1}{|G|} \sum_{h \in G} \underbrace{\left( \sum_{k \in G} \ket{hk}\!\bra{k}_R \right)}_{= V_R^h} \otimes I_S = \frac{1}{|G|} \sum_{g \in G} V_R^g \otimes I_S\,. 
\end{align}
Claim (iii) then follows:
\begin{align}
    T_R \hat{E}_\chi \Pi_{\rm pn} T_R^\dagger &= \hat{E}_\chi T_R \Pi_{\rm pn} T_R^\dagger = \hat{E}_\chi \frac{1}{|G|}\sum_{h,k \in G} \ket{hk}\!\bra{k}_R \otimes I_S = \hat{E}_\chi \frac{1}{|G|}\sum_{h,k \in G} \ket{h}\!\bra{k}_R \otimes I_S \\
    &= \hat{E}_\chi \ket{1}\!\bra{1}_R \otimes I_S =  \left( \hat{U}^\chi \otimes I_S \right) \left( \ket{1}\!\bra{1}_R \otimes I_S \right) = \ket{\chi}\!\bra{1}_R \otimes I_S
\end{align}
\end{proof}

Many codes will admit several orthonormal orientation bases per choice of QRF subsystem, such that the latter can be associated with  multiple \emph{inequivalent} trivialization maps and equivalence classes  $[\hat{\mathcal{E}}]$ of maximal sets of correctable errors. 

\begin{example}\label{3qubit:nontrivialUSprop}
Consider the 3-qubit bit-flip code and choose $R=12$.~It is clear that we can equally well use the $X$- or $Y$-basis of $R$ to define its covariant orthonormal orientation basis.~When we choose the $X$-basis, the set of errors $\hat{E}_{\chi}$ residing in the Pontryagin dual group of Proposition~\ref{prop_nontriv} is given by $\hat{\mathcal{E}}=\{I_{12},X_1,X_2,X_1X_2\}$, and when we choose the $Y$-basis it is $\hat{\mathcal{E}}'=\{I_{12},Y_1,Y_2,Y_1Y_2\}$.~Both act trivially on $S=3$.~When we choose the $X$-basis, the disentangler reads
\begin{equation}
\begin{split}
    {T}_{R_X} &= (\ketbra{++}{++}_R + \ketbra{--}{--}_R) \otimes I_S + (\ketbra{+-}{+-}_R + \ketbra{-+}{-+}_R) \otimes Z_S \\[2mm]
    &= I_{12} \otimes \tfrac{1}{2}(I_3 + Z_3) + X_1 X_2 \otimes \tfrac{1}{2}(I_3 - Z_3) \label{eq:T_decoder}
\end{split}
\end{equation}
and clearly commutes with $\hat{\mathcal{E}}$ but not $\hat{\mathcal{E}}'$. This means that not only $\hat{\mathcal{E}}$ will continue to look $R$-local upon applying the disentangler, but also its code space restriction $\hat{\mathcal{E}}\,\Pi_{\rm pn}$ becomes $R$-local, which due to Lemma~\ref{lem_none} it was not before the refactorization. Indeed, suppressing tensor products, e.g.\ $X_1 \Pi_\mrm{pn} = \tfrac{1}{4}(X_1(I + Z_2 Z_3) + i Y_1(Z_2 - Z_3))$ is not $12$-local. By contrast, we have, for example, 
\begin{equation}
 {T}_{R_X}\,Y_1=-Y_1\,{T}_{R_X}+Y_1\left(I_3+Z_3\right)\,,
\end{equation}
so that ${T}_{R_X}Y_1\,\Pi_{\rm pn}\,{T}_{R_X}^\dag$ will contain operators that act non-trivially on $S=3$. 

Of course, we could have chosen the trivialization ${T}_{R_Y}$ with respect to the $Y$-basis instead, with the same conclusions but for $X$ and $Y$ interchanged.~Hence, for the \emph{same} QRF $R$, there exist inequivalent trivializations and thereby also inequivalent new tensor product structures in which the stabilizer group acts exclusively on the frame factor.~In contrast to the case when the stabilizer group acts trivially on $S$, the code space restriction of different maximal error sets that before restriction act exclusively on $R$ can be inequivalent, as here $\hat{\mathcal{E}}\,\Pi_{\rm pn}\not\sim\hat{\mathcal{E}}'\,\Pi_{\rm pn}$.
\eox
\end{example}

In conclusion, when $\mathcal{G}$ acts non-trivially on $S$, one cannot associate a unique equivalence class of correctable errors with the frame, but for each choice of orthonormal orientation basis one can.

\subsection{Encoding, decoding and recovery from disentanglers}

Our starting point for extracting QRFs out of Pauli stabilizer QECCs was to observe that the stabilizer group $\mathcal{G}$ is a representation of a QRF symmetry group $G = \mathbb{Z}_2^{\times(n-k)}$.
We began by constructing frames and systems that were local with respect to the innate TPS of the physical qubits, using the elements of $\mathcal{G}$ itself as a representation $g \mapsto U^g = U_R^g \otimes U_S^g$ with $U^g_R$ and $U^g_S$ nontrivial.
We then found that it is instead always possible to find a different tensor product factorization of $\Hil_\mrm{kin}\simeq \Hil_R \otimes \Hil_S$ such that both errors and $U^g$ act trivially on $S$, i.e., a frame $R$ for which $U^g = U_R^g \otimes I_S$ and $E \Pi_{\rm pn} = (e_R \otimes I_S)$ (Thm.~\ref{thm_errorQRF}).
Generally, if handed a frame alone for which both $U_R^g$ and $U_S^g$ are nontrivial, as typically for QRFs that are local in the innate TPS, then it is possible to find frame orientation bases and corresponding error sets that act trivially on $S$ upon a disentangling operation $T_R$ (Prop.~\ref{prop_nontriv}).
In this last subsection, we return to the local TPS and study cases in which errors and $U^g$ \emph{both} act nontrivially on $R$ and $S$.
Closing the loop on the QECC/QRF dictionary, we will study the error correcting properties of QRF disentanglers, with the punchline being that it is still possible to decode and recover from errors spread across $R$ and $S$ with the disentangler being an integral part of the recovery operation.

\begin{example}[Three-qubit code]
\label{Sec:3qubitQRFdisentangler}

Consider once again the 3-qubit code and the local QRF consisting of the first two qubits constructed in Sec.~\ref{Sec:3qubitcodeasQRF} and Ex.~\ref{3qubit:nontrivialUSprop}, with $R = 12$ and $S = 3$.
Choosing the seed state to be the $X$ eigenstate $\ket{++}_R$, the disentangler $T_R$ is given in Eq.~\eqref{eq:T_decoder} (here we omit the subscript on $R$ for brevity).
Consider two maximal correctable error sets, $\cE = \{I, X_1, X_2, X_3\}$ and $\hat{\cE}' = \{I, X_1, X_2, X_1 X_2\}$, which are the conventionally considered error set for the 3-qubit code and an error set that is localized to $R$, respectively.
For any $\ket{\bar\psi} \in \Hil_\mrm{code}$, we have that
\begin{equation} \label{eq:3qubitTerror}
    \begin{aligned}
        T_R \ket{\bar\psi} &= \ket{00}_{12} \otimes \ket{\psi}_3 \\
        T_R X_1 \ket{\bar\psi} = X_1 T_R \ket{\bar\psi} &= \ket{10}_{12} \otimes \ket{\psi}_3 \\
        T_R X_2 \ket{\bar\psi} = X_2 T_R \ket{\bar\psi} &= \ket{01}_{12} \otimes \ket{\psi}_3 \\
        T_R X_3 \ket{\bar\psi} = X_1 X_2 X_3 T_R \ket{\bar\psi} &= \ket{11}_{12} \otimes X_3 \ket{\psi}_3 .
    \end{aligned}
\end{equation}
The state of $R = 12$ therefore carries information about whether an $X$ error occurred, and the operators $X_1 X_2$ and $X_3$ both present the same syndrome.
Consequently, whether we correct for the error set $\mathcal{E}$ or $\hat{\cE}'$ depends on what we choose to do if we find that the state of $R$ is $\ket{11}_{12}$.
Taking no further action amounts to correcting for $\hat{\cE}'$, while applying $X_3$ amounts to correcting for $\mathcal{E}$.
\eox
\end{example}

It turns out that it is possible in general to define encoding, decoding, and recovery protocols using $T_R$ to correct errors spread across $R$ and $S$, as for the standard error set $\mathcal{E}$ in the example above.
Given a local, ideal QRF with a basis for $R$ of orthonormal frame orientation states $\ket{g}_R$ (see Sec.~\ref{Sec:genstabilizerQRF}), denote by $\ket{1}_R$ the uniform superposition over all frame states defined in Eq.~\eqref{eq:R_ready} and $T_R$ the disentangler defined in Eq.~\eqref{triv}.
Extend $T_R$ to a unitary on all of $\Hil_\mrm{kin}$and let it function as an encoding map in the following way: for any $\ket{\psi} \in \Hil_\mrm{logical} \simeq \Hil_S$, let $\ket{\bar\psi} \in \Hil_{\rm pn} \equiv \Hil_\mrm{code}$ be
\begin{equation} \label{eq:gen_encoding}
    \ket{\bar\psi} = T_R^\dag(\ket{1}_R \otimes \ket{\psi}_S),
\end{equation}
where we abuse notation slightly and use the fact that $\Hil_\mrm{logical} \simeq \Hil_S$ to directly write the logical state into $\ket{\psi}_S$ above.
Since $T_R$ is unitary, it of course follows that $T_R \ket{\bar{\psi}} = \ket{1}_R \otimes \ket{\psi}_S$ for any $\ket{\bar{\psi}} \in \Hil_\mrm{pn}$.
With these definitions in place, we have the following result:
\begin{prop}[\textbf{Action of local disentangler on error states}] \label{prop:general_recovery}
    Let $\Hil_\mrm{kin} = \Hil_R \otimes \Hil_S$ define a local QRF-system partition for a $[[n,k]]$ Pauli stabilizer code with $\Hil_R \simeq (\mathbb{C}^2)^{\otimes(n-k)}$ consisting of $n-k$ physical qubits and such that $G = \mathbb{Z}_2^{\times(n-k)}$ is represented with the code's stabilizers $g \mapsto U_R^g \otimes U_S^g$, where $g \mapsto U_R^g$ is faithful (but possibly projective), $g \mapsto U_S^g$ is non-trivial, and $\braket{g}{h}_R = \delta_{g,h}$ (i.e. $R$ is ideal). Let $\mathcal{E} = \{E_1, \dots, E_m\}$ be a set of correctable Pauli errors.
Then, for any any $\ket{\bar{\psi}} = T_R^\dagger(\ket{1}_R\otimes\ket{\psi}_S) \in \Hil_\mrm{pn}$, it follows that
    \begin{equation}
        T_R E_i \ket{\bar{\psi}} = \ket{w(E_i)}_R \otimes L_S(E_i)\ket{\psi}_S,
    \end{equation}
    for some collection of orthonormal states $\ket{w(E_i)}_R$ and unitary operators $L_S(E_i)$.
\end{prop}

\begin{proof}
The proof can be found in \App{App:proofs}.
\end{proof}

\medskip

\noindent
Prop.~\ref{prop:general_recovery} therefore gives us an explicit recovery protocol.
After decoding with $T_R$, perform an orthogonal measurement of $R$ onto $\{\ket{w(E)}_R\}_{E \in \mathcal{E}}$.
Upon obtaining the outcome $w(E)$, apply the unitary operator $L^{\dag}_S(E)$ to the state on $S$ to recover the logical state $\ket{\psi}_S$.

\begin{example}[Five-qubit code]

Consider again the 5-qubit code and its associated local QRF with $R=1234$ and $S=5$ that was introduced in Ex.~\ref{sec:5qubit}.
With the choice of seed state \eqref{eq:5qubitseed}, the disentangler \eqref{triv} again takes a compact form:
\begin{equation}
\begin{aligned}
    T_R &= IIII \otimes \tfrac{1}{4}(I + X + Y + Z) + IXXI \otimes \tfrac{1}{4}(I - X - Y + Z) \\
    &\qquad + XIIX \otimes \tfrac{1}{4} (I - X + Y - Z) + XXXX \otimes \tfrac{1}{4}(I + X - Y - Z)
\end{aligned}
\end{equation}
An encoded state $\ket{\bar{\psi}} \in \Hil_\mrm{code}$ decodes to
\begin{equation}
    T_R \ket{\bar{\psi}} = \ket{w_0}_{1234} \otimes X_5 \ket{\psi}_5,
\end{equation}
where the fixed fiducial state on $R = 1234$ is
\begin{equation}
\begin{aligned}
    \ket{w_0}_{1234} &= \frac{1}{4} \left( i \ket{0000} -\ket{0001}-\ket{0010}-\ket{0011}-\ket{0100}+\ket{0101}-i \ket{0110}-\ket{0111} \right. \\
    &\qquad \left. -\ket{1000}+i \ket{1001}+\ket{1010}+\ket{1011}-\ket{1100}+\ket{1101}-\ket{1110}-i \ket{1111} \right).
\end{aligned}
\end{equation}
The extra $X_5$ after decoding with $T_R$ is just a result of having started with the conventional choices \eqref{eq:5qubitencodedlogical01} for the logical $0$ and $1$ states instead of beginning with \eqref{eq:gen_encoding}; exchanging $\bar 0 \leftrightarrow \bar 1$ would of course get rid of the $X_5$.
It is straightforward but tedious to check that a single Pauli error $E_i = X_i, Y_i, Z_i$, $i \in \{1, \dots, 5\}$ results in a distinct orthogonal state on $R$, i.e.,
\begin{equation}
    T_R E_i \ket{\bar{\psi}} = \ket{w(E_i)}_{1234} \otimes L_5(E_i) \ket{\psi}_5.
\end{equation}
where $\braket{w(E_i)}{w(E_j)} = \delta_{ij}$ (and also $\braket{w(E_i)}{w_0} = 0$), and $L_5(E_i)$ is unitary.
A single qubit error can therefore be diagnosed by measuring $R$ and then conditionally rotating the state of $S$.
\eox
\end{example}

\section{Error duality}\label{sec_errorduality}

Let us now turn our focus to two classes of correctable error sets that turn out to be \emph{dual} to one another. In a gauge theory context, this has the flavor of electromagnetic duality and the errors can be interpreted as electric and magnetic charge excitations. This duality will have meaning more generally for QEC and QRFs, however, and be visible in the structural relations of these error sets as well as in their associated recovery operations. The abstract underpinning is provided by Pontryagin duality, which we review in Apps.~\ref{app:Pontryagin} and~\ref{app:reps} and which in the continuum indeed underlies electromagnetic duality \cite{Ben-Zvi}. There is one key difference to electromagnetic duality of the continuum: in our case, the representations of the stabilizer gauge group and its dual will \emph{not} commute, such that their charges turn out to be complementary and  cannot be determined simultaneously. Specifically, our magnetic charges will \emph{not} be stabilizer gauge-invariant, yet lead to meaningful correctable errors. This will also be illustrated in the surface codes of Sec.~\ref{sec:surfacecodes}, where both forms of eletro-magnetic duality arise. Owing to charge complementarity, our new duality is somewhat stronger.

This duality has already surfaced in the discussion of the error-set/QRF correspondence, especially in Theorem~\ref{thm_errorQRF}, and we shall extract it here more explicitly.~We begin by assuming that we are given unitary representations $U$ and $\hat U$ of the stabilizer group $G$ and its Pontryagin dual $\hat G$ on $\Hil_{\rm kin}$, respectively, that are \emph{dual} according to Definition~\ref{def:dual_reps_faithful}, i.e.\ that satisfy
\begin{equation}\label{dualrep}
    U^g\hat{U}^\chi=\chi(g)\hat{U}^\chi U^g\,,\qquad\forall\,g\in G\,,\chi\in\hat{G}\,,
\end{equation}
where $\chi\in\hat{G}$ labels the characters of the irreps of $G$.~Thus, $U(G),\hat{U}(\hat{G})$ do not commute.~As explained in App.~\ref{app:duality}, Eq.~\eqref{dualrep} implies that $U$ is \emph{faithful}.~This assumption is justified, as we have seen in Sec.~\ref{sec_genstab} that each $[[n,k]]$ stabilizer code with a faithful representation of $G=\mathbb{Z}_2^{\times(n-k)}$ gives rise to ideal QRFs in many different ways.~Furthermore, Proposition \ref{prop_QRFPcorresp} says that any orthonormal orientation basis of an ideal QRF associated with $\mathcal{G}=U(G)$ gives rise to a dual representation $\hat{\mathcal{G}}=\hat{U}(\hat{G})$.~Hence, each $[[n,k]]$ stabilizer code with faithful $U$ will admit a multitude of dual $\hat{U}$s.\footnote{Proposition \ref{prop_bijection} in App.~\ref{app:duality} provides a more general abstract characterization of all possible dual $(U,\hat U)$ pairs in terms of $*$-representations of a certain abstract $*$-algebra of functions on $\hat{G}\times G$, equipped with a twisted convolution product.} 

Before proceeding, we note that Eq.~\eqref{dualrep} is just a discrete version of the standard Weyl relations,
\begin{equation}\label{eq:Weyl}
U(s)\,V(t)=e^{-ist}\,V(t)\,U(s)\,,\qquad\forall\,t,s\in\mathbb{R}\,,
\end{equation}
in quantum mechanics, relating momentum and position space translations $U(s)=\exp(is\hat{q})$ and $V(t)=\exp(it\hat p)$, respectively, both of which are isomorphic to the translation group $(\mathbb{R},+)$ and Pontryagin duals of one another. Indeed, $\chi_t(s)=\exp(-ist)$ are the characters of position-space translations. Our notion of dual representations of two Pontryagin dual groups is thus nothing else than a discrete version of position-momentum duality and the standard Fourier transform in quantum mechanics. Our ensuing error duality may, therefore, also be read in these terms. To streamline our terminology, we nevertheless refer to the errors as ``electric'' and ``magnetic'' below, as this interpretation is somewhat more natural in a gauge theory context. Indeed, in a $\mathbb{Z}_2$ gauge theory, Pauli operators can quite literally correspond to electric charge excitations (e.g., see the surface codes in the next section).

\subsection{Pauli errors as  electric charge excitations}\label{ssec_pauliduality} 

We begin with a non-degenerate set of standard Pauli errors $E_\chi\in\mathcal{P}_n$ for a given $[[n,k]]$ stabilizer code, collecting their basic properties for contrast with the dual errors in Sec.~\ref{ssec_dualerrors}. In fact, the below equally holds for error sets comprised of any dual representation $\{\hat{U}^\chi\}_{\chi\in\hat{G}}$ of $\hat{G}$. In standard QEC settings, Pauli errors typically account for spin flip, phase flip, etc.\ errors. We will return to their interpretation for both sides of the QECC/QRF dictionary shortly. 

We recall from App.~\ref{app:stabilizer_reps} that the physical space can be decomposed orthogonally into the irreps of $G$
\begin{equation}
    \Hil_{\rm kin}=\bigoplus_{\chi\in\hat{G}}\Hil_\chi\,,
\end{equation}
where $\Hil_\chi$ is the isotype corresponding to  $\chi$ and the image of the orthogonal projector
\begin{equation}\label{pchi}
    P_\chi=\frac{1}{2^{n-k}}\sum_{g\in G}\chi(g)\,U^g\,.
\end{equation}
The code/perspective-neutral space is
\begin{equation}
    \Hil_{\rm pn}=\Hil_1\,
\end{equation}
and the errors, if they are correctable, map the code space unitarily into the other isotypes, which comprise the error spaces (cf.~Proposition \ref{prop:decomp_Pauli}):
\begin{equation}\label{E5}
    E_\chi(\Hil_{\rm pn})=\Hil_\chi.
\end{equation}
From an Abelian gauge theory perspective, where different irreps of the gauge group correspond to different electric charge sectors,\footnote{E.g., in electrodynamics, the spatial integral of the charge density $\rho$ in the Gauss law $\vec{\nabla}\cdot\vec{E}=\rho$ labels different charge sectors.} we may view the Pauli errors $E_\chi$ as \emph{electric charge excitations} and the characters $\chi$ as labeling  these electric charges.

The elements of the Pontryagin dual group translate between these error spaces
\begin{equation}\label{errortranslate}
    \hat{U}^\eta\left(\Hil_{\chi}\right)=\hat{U}^\eta\,P_\chi\left(\Hil_{\rm kin}\right)\underset{\eqref{Pchicov}}{=}\Hil_{\eta\chi}\,.
\end{equation}
So far we have not been specific about the choice of dual representation. A particularly natural restriction is to adapt the dual representation to the error structure of the code by demanding that
\begin{equation}
    \hat{U}^\eta\,E_\chi\,\Pi_{\rm pn}=E_{\eta\chi}\,\Pi_{\rm pn}\,.
\end{equation}
Setting $\chi=1$ tells us that the $\hat{U}^\eta$ act as errors on the code, thus comprising an error set with group structure as in the converse direction of Theorem~\ref{thm_errorQRF}. We will find this restriction to be realized in the examples below.  Note, however, that none of the following technical results in this section depend on this additional requirement.
Eq.~\eqref{E5} informs us that correctable Pauli errors map into the orthogonal complement of $\Hil_{\rm pn}$ and obey the KL condition in the simple form
\begin{equation}\label{KLsimple}
    \Pi_{\rm pn}\,E_\chi\,E_\eta\,\Pi_{\rm pn}=\delta_{\chi,\eta}\,\Pi_{\rm pn}\,.
\end{equation}
Now
\begin{equation} \label{ugpchi}U^g\,P_\chi=\chi(g)\,P_\chi\,,
\end{equation}
so that the error $E_\chi$ can be detected by measuring any generating set $\{U^g\}_{g\in G_{\rm gen}}$ of the stabilizer group $\mathcal{G}$.~The error syndrome is then given by the characters $\{\chi(g)\}_{g\in G_{\rm gen}}$.~When syndrome $\chi$ is found, one simply applies $E_\chi$ again to correct the error; this is a de-excitation of the electric charge $\chi$.~Alternatively, one can also measure the projectors $\{P_\chi\}_{\chi\in\hat{G}}$, as expressed by the following identity for arbitrary $\rho\in\mathcal{S}(\Hil_{\rm kin})$ (cf.~Lemma \ref{lem_det}):
\begin{equation}\label{projeqGtwirl}
\sum_{\chi\in\hat{G}}\,P_\chi\,\rho\,P_\chi = \frac{1}{2^{n-k}}\sum_{g\in G} U^g\,\rho\,U^g\,.
    \end{equation}
 Thus, the incoherent group averaging ($G$-twirl) over $G$ is equal to the projective measurement  $\{P_\chi\}_{\chi\in\hat{G}}$  of the $\chi$-charge sectors. Note that both sets of measurements $\{U^g\}_{g\in G_{\rm gen}}$ and $\{P_\chi\}_{\chi\in\hat{G}}$ are gauge-invariant and leave all isotypes invariant. In particular, they leave code states invariant.
 
 For a given error set $\mathcal{E}=\{E_\chi\}$ an error channel reads
 \begin{equation}
     \tilde{\mathcal{E}}:\mathcal{B}(\Hil_{\rm pn})\rightarrow\mathcal{B}(\Hil_{\rm kin})\,,\qquad\tilde{\mathcal{E}}(\rho_{\rm pn})=\sum_{k}\tilde{E}_k\rho_{\rm pn}\tilde{E}_k^\dag\,,
 \end{equation}
where the $\{\tilde{E}_k\}$ are  linear combinations of the $E_\chi$ such that $\sum_k\tilde{E}^\dag_k\tilde{E}_k = I$. The recovery channel returning the original perspective-neutral state, $\mathcal{O}(\tilde{\mathcal{E}}(\rho_{\rm pn}))=\rho_{\rm pn}$, may be checked to be given by
\begin{equation}\label{electricrecovery}
    \mathcal{O}:\mathcal{B}(\Hil_{\rm kin})\rightarrow\mathcal{B}(\Hil_{\rm pn})\,,\qquad\qquad\mathcal{O}(\rho)=\sum_{\chi\in\hat{G}}E_\chi^\dag\,P_\chi\,\rho\,P_\chi\,E_\chi\,.
\end{equation}
When we use a maximal set of Pauli errors, we can invoke Theorem \ref{thm_errorQRF} and the QRF $R$ generated by $\{E_\chi\}_{\chi\in\hat{G}}$ to write in the corresponding factorization $\Hil_{\rm kin}\simeq\Hil_R\otimes\Hil_S$ somewhat more explicitly
\begin{equation}\label{pchisimple}
    P_\chi=\ket{\chi}\!\bra{\chi}_R\otimes I_S\,,
\end{equation}
for an orthonormal basis $\{\ket{\chi}_R\}_{\chi\in\hat{G}
}$ on $\Hil_R$, given that the errors and stabilizer elements $U^g=U^g_R\otimes I_S$ act trivially on the system $S$. The isotypes then become
\begin{equation}\label{chiHS}
    \Hil_\chi\simeq \ket{\chi}_R\otimes \Hil_S\,
\end{equation}
and we can write the code space restriction of the errors as 
\begin{equation}\label{chione}
    E_\chi\,\Pi_{\rm pn}=\ket{\chi}\!\bra{1}_R\otimes I_S\,.
\end{equation}

\begin{example}[Three-qubit code]
\label{3qubit:errorduality1}
     Let us illustrate the above discussion for the simple example of the 3-qubit bit-flip code.~The stabilizer group $\mathcal G$ in Eq.~\eqref{eq:3qubitstab} furnishes a faithful unitary representation of $G= \mathbb{Z}_2 \times \mathbb{Z}_2$ on $\Hil_{\rm kin}\simeq(\mathbb C^2)^{\otimes3}$.~Explicitly (cfr.~Sec.~\ref{Sec:3qubitcodeasQRF})
\begin{equation}\label{eq:3qubitGrep}
    U^{++} = I\, , \qquad U^{+-} = Z_2 Z_3\, , \qquad
    U^{-+} = Z_1Z_3\, , \qquad U^{--} =Z_1 Z_2\,.
\end{equation}
The dual group $\hat{G}$ is the group of irreducible characters $\{\chi_i\}_{i=1,\dots,4}$ of $G$.~We take these to be
\begin{align}
\chi_1 (++) &= 1\,, \quad \chi_1 (+-) = 1  \,, \quad \chi_1 (-+) = 1 \,, \quad \chi_1 (--) = 1 \,, \label{3qubit:characters1}\\   
\chi_2 (++) &= 1\,, \quad \chi_2 (+-) = 1  \,, \quad \chi_2 (-+) = -1 \,, \quad \chi_2 (--) = -1 \,, \label{3qubit:characters2}\\   
\chi_3 (++) &= 1\,, \quad \chi_3 (+-) = -1  \,, \quad \chi_3 (-+) = 1 \,, \quad \chi_3 (--) = -1 \,, \label{3qubit:characters3} \\
\chi_4 (++) &= 1\,, \quad \chi_4 (+-) = -1  \,, \quad \chi_4 (-+) = -1 \,, \quad \chi_4 (--) = 1 \,,\label{3qubit:characters4} 
\end{align}
so that, as per Proposition~\ref{prop_QRFPcorresp}, choosing $R=12$ with orientation states given by the 2-qubit $X$-eigenstates $\{\ket{g}_R\}_{g\in G}=\{\ket{++},\ket{+-},\ket{-+},\ket{--}\}$, the QRF orientation basis associated with the representation \eqref{eq:3qubitGrep} of $G$ gives rise to a (faithful) unitary representation of the dual group $\hat{G}$ on $(\mathbb C^2)^{\otimes3}$ defined by the set of correctable Pauli errors $\hat{\mathcal E}=\{I, X_1,X_2,X_1X_2\}$ as
\begin{equation}\label{3qubit:dualGrep}
    \hat{U}^{\chi_1}=I \,, \qquad \hat{U}^{\chi_2} = X_1 \,, \qquad \hat{U}^{\chi_3} = X_2\,, \qquad \hat{U}^{\chi_4} = X_1 X_2\,.  
\end{equation}
It is in fact straightforward to check that the representations \eqref{eq:3qubitGrep}, \eqref{3qubit:dualGrep} are dual in the sense of Eq.~\eqref{dualrep} for the characters given in Eqs.~\eqref{3qubit:characters1}-\eqref{3qubit:characters4}.

$\Hil_{\rm kin}$ can be decomposed into the direct sum of the two-dimensional isotypes $\Hil_i$ corresponding to $\chi_i$
\begin{equation}\label{3qubit:Hchi}
\begin{aligned}
    \Hil_1 &= \cH_{\rm pn} = \Span\{ \vert 000\rangle ,\vert 111 \rangle \}\,, \\
    \cH_2 &= X_1(\cH_{\rm pn})=\Span\{ \vert 100\rangle ,\vert 011 \rangle \}\,, \\
    \cH_3 &= X_2(\cH_{\rm pn} )=\Span\{ \vert 010\rangle ,\vert 101 \rangle \}\,, \\
    \cH_4 &=  X_1X_2(\cH_{\rm pn} )=\Span\{ \vert 110\rangle ,\vert 001 \rangle \}\,.
\end{aligned}
\end{equation}
Clearly, the errors $\hat{U}^{\chi}\neq I$ translate between the different orthogonal isotypes (cf.~Eq.~\eqref{errortranslate}) and, as anticipated in Sec.~\ref{sssec_3qubit}, they can be detected by measuring any generating set of the stabilizer group \eqref{eq:3qubitGrep}.~The error syndrome consists of the strings of the corresponding characters in Eq.~\eqref{3qubit:characters1}-\eqref{3qubit:characters4}.

Lastly, let us emphasize that the representation of the dual group is not unique.~For example, choosing the orthonormal orientation basis $\{\ket{g}\}_{g\in G}$ to be given by the two-qubit $Y$-eigenstates rather than the $X$-eigenstates will give rise to a distinct faithful unitary representation $\hat{U}:\hat G=\mathbb Z_2\times\mathbb Z_2\to\text{Aut}((\mathbb C^2)^{\otimes 3})$ defined by the set of correctable Pauli errors $\hat{\mathcal E}'=\{I,Y_1,Y_2,Y_1Y_2\}.$
\eox
\end{example}

\subsection{Gauge fixing errors as magnetic charge excitations}\label{ssec_dualerrors}

There is another type of errors that do not irrevocably corrupt the logical data and that are natural in a gauge theory and QRF context. These are gauge fixing operations which map out of the gauge-invariant Hilbert space without an irreversible loss of gauge-invariant information. Any such gauge fixing freezes \emph{some} QRF into a particular orientation and is thus related to the Page-Wootters reductions into a QRF perspective that we encountered before. Pontryagin duality again helps to make this precise.

As discussed in App.~\ref{app:reps}, the physical space can also be decomposed orthogonally into dual irreps of the Pontryagin dual $\hat{G}$
\begin{equation}
    \Hil_{\rm kin}=\bigoplus_{g\in G}\hat{\mathcal{H}}_g\,,
\end{equation}
where $\hat\Hil_g$ is now the isotype corresponding to $g\in G$ (by Pontryagin duality the characters of $\hat{G}$ are now labeled by $g\in G$, cf.~App.~\ref{app:Pontryagin}) and the image of the orthogonal projector
\begin{equation}\label{PG}
    \hat{P}_g:=\frac{1}{2^{n-k}}\sum_{\chi\in\hat{G}}\chi(g)\, \hat{U}^\chi\,.
\end{equation}
In contrast to the case of Pauli errors, now each isotype overlaps with the code space
\begin{equation}\label{errorspacetranslate}
    \hat{P}_g\left(\Hil_{\rm pn}\right)=U^g\bigl(\hat{P}_e\left(\Hil_{\rm pn}\right)\bigr)\,,
\end{equation}
where we made use of the covariance in Eq.~\eqref{Pgcov} for dual representations applied to $G=\mathbb{Z}_2^{\times(n-k)}$,
\begin{equation}\label{Pgcov2}
    U^g\hat{P}_h U^g=\hat{P}_{gh}\,,
\end{equation}
and the gauge-invariance of $\Hil_{\rm pn}$. Thus, all overlaps with the code space are unitarily equivalent and related by stabilizer gauge transformations in $\Hil_{\rm kin}$; this is the analog of Eq.~\eqref{errortranslate}.~In fact, all overlaps are (up to a rescaling by $1/\sqrt{2^{n-k}}$) unitarily equivalent to the code space itself, $\sqrt{2^{n-k}}\hat{P}_g\left(\Hil_{\rm pn}\right)\simeq\Hil_{\rm pn}$, as appropriate for gauge fixings.~This is implied by each of the following two lemmas.~The first says that the projectors satisfy the KL condition, while the second provides a large class of unitaries whose code space restrictions coincide with these projectors. 

We are thus entitled to view $\{\hat{P}_g\}_{g\in G}$ both as a maximal set of correctable errors, as well as a set of valid gauge fixing projectors. As such, we shall refer to them as ``gauge fixing errors'' and emphasize that they are not errors in a gauge fixing operation, but the gauge fixing itself is the error mapping out of the invariant Hilbert space (similarly to how a Pauli error is not an error in the implementation of a Pauli operator).~Unlike Pauli errors, the gauge fixing ones do \emph{not} map into the orthogonal complement of the code/perspective-neutral space.  

\begin{lem}[\textbf{Gauge fixings are correctable}]
The projectors $\hat{P}_g$ onto the dual isotypes obey the KL condition \eqref{KL} \emph{diagonally}
\begin{equation}\label{KLgauge}
    \Pi_{\rm pn}\,\hat{P}_g\,\hat{P}_{g'}\,\Pi_{\rm pn}=\frac{1}{2^{n-k}}\delta_{g,g'}\,\Pi_{\rm pn}\,.
\end{equation}
\end{lem}

\begin{proof}
Invoking $\hat{P}_g\hat{P}_{g'}=\delta_{g,g'}\hat{P}_g$ and Eq.~\eqref{Pgcov2}, we have
\begin{equation}
\begin{split}
   \Pi_{\rm pn}\,\hat{P}_g\,\hat{P}_{g'}\,\Pi_{\rm pn}&=\delta_{g,g'}\frac{1}{2^{n-k}}\sum_{h\in G} U^h\hat{P}_g\,\Pi_{\rm pn}=\delta_{g,g'}\frac{1}{2^{n-k}}\sum_{h\in G} U^h\hat{P}_g\, U^h\,\Pi_{\rm pn}\\&=\delta_{g,g'}\frac{1}{2^{n-k}}\sum_{h\in G} \hat{P}_{hg}\,\Pi_{\rm pn}=\frac{1}{2^{n-k}}\delta_{g,g'}\Pi_{\rm pn}\,.
   \end{split}
\end{equation}
\end{proof}

These errors can be implemented unitarily.
\begin{lem}[\textbf{Unitary gauge fixing errors}]
Let $\{E_\chi\}_{\chi\in\hat{G}}$ be a maximal and non-degenerate set of correctable unitary Pauli errors with $E_1=I$ such that $E_\chi(\Hil_{\rm pn})\subset\Hil_{\rm pn}^\perp$ for $\chi\neq1$. Let also $h: \hat{G}\times G \to G$ be such that, for any $g\in G$, $h(\cdot, g): \hat{G}\to G$ is a bijection with $h(1,g)=g$. Then, for any $g \in G$, the operator
\begin{equation}\label{utilde}
    \hat{E}_g:=\sqrt{2^{n-k}}\sum_{\chi\in\hat{G}}\hat{P}_{h(\chi,g)}\,\Pi_{\rm pn}\,E_\chi\,,
\end{equation}
is a unitary on $\Hil_{\rm kin}$; furthermore, it obeys
\begin{equation}\label{eq:unitary gauge fixing}
    \hat{E}_g\,\Pi_{\rm pn}=\sqrt{2^{n-k}}\hat{P}_g\,\Pi_{\rm pn}\,,
\end{equation}
so that, up to scaling, the code space restriction of $\hat{E}_g$ agrees with $\hat{P}_g$.  
\end{lem}
There are therefore many correctable error sets $\{E_\chi\}_{\chi\in\hat{G}}$ defining a set of unitaries $\{\hat{E}_g\}_{g\in G}$ with the property that they implement the gauge fixing \eqref{eq:unitary gauge fixing}. For example, \emph{any}
maximal and non-degenerate set of correctable Pauli errors with $E_1=I$ will serve, as will \emph{any} dual representation of the Pontryagin dual, i.e.\ $\{E_\chi=\hat{U}^\chi\}_{\chi\in\hat{G}}$. Note also that we have the error set equivalences, according to Definition \ref{def_equiv}:
\begin{equation}\label{equivalences}
    \{\hat{E}_g\}_{g\in G}\sim\{\hat{P}_g\}_{g\in G}\sim\{\hat{U}^\chi\}_{\chi\in\hat{G}}\,.
\end{equation}
The first equivalence is provided by the lemma, the second follows from the fact that the elements of the Pontryagin dual representation $\hat{U}^\chi$ are linear combinations of the projectors $\{\hat{P}_g\}_{g\in G}$ as in Eq.~\eqref{Uhat}.

\begin{proof}
 Let $g\in G$. Using Eq.~\eqref{KLgauge}  and the fact that $h(\cdot, g)$ is a bijection, we have 
\begin{equation}
    \hat{E}_g^\dag\,\hat{E}_g=2^{n-k}\sum_{\chi, \eta\in \hat{G}}E_\chi\,\underbrace{\Pi_{\rm pn}\,\hat{P}_{h(\chi,g)}\,\hat{P}_{h(\eta,g)}\,\Pi_{\rm pn}}_{= \delta_{\chi , \eta }\, \Pi_{\rm pn}}\,E_\eta=\sum_{\chi\in \hat{G}}E_\chi\,\Pi_{\rm pn}\,E_\chi\,.
\end{equation}
Noting that, for any $\chi \in \hat{G}$,
\begin{equation}
    E_\chi\,\Pi_{\rm pn}\,E_\chi\underset{\eqref{Pj}}{=}\frac{1}{2^{n-k}}\sum_{g\in G}\,\chi(g)\,U^g\underset{\eqref{pchi}}{=}P_\chi\,,
\end{equation}
we conclude that $\hat{E}_g^\dag\,\hat{E}_g=\displaystyle \sum_{\chi \in \hat{G}}P_\chi = I$. Hence $\hat{E}_g$ is unitary.\footnote{We need not show that $\hat{E}_g\hat{E}_g^\dagger = I$ since the Hilbert spaces are finite dimensional.}

Finally, Eq.~\eqref{eq:unitary gauge fixing} follows from the KL conditions for the Pauli errors $\{E_{\chi}\}_{\chi \in \hat{G}}$ together with the condition $h(1,g)=1$:
\begin{equation}
    \hat{E}_g \Pi_{\rm pn} = \sqrt{2^{n-k}} \sum_{\chi \in \hat{G}} \hat{P}_{h(\chi, g)} \Pi_{\rm pn}\, E_\chi \Pi_{\rm pn} = \sqrt{2^{n-k}} \sum_{\chi \in \hat{G}} \hat{P}_{h(\chi, g)} \underbrace{\Pi_{\rm pn}\, E_\chi \, E_1\,  \Pi_{\rm pn}}_{\underset{{(\rm KL)}}{=} \delta_{\chi, 1}\, \Pi_{\rm pn}} 
    = \sqrt{2^{n-k}} \hat{P}_{g} \Pi_{\rm pn}\,.
\end{equation}
\end{proof}

Now, since for dual representations $U,\hat U$ (cf.~Eq.~\eqref{equaldim})
\begin{equation}
\dim\hat\Hil_g=\dim\Hil_\chi=\dim\Hil_{\rm pn}\,,\qquad\forall\,g\in G,\chi\in\hat{G}\,,
\end{equation}
we have by unitarity that
\begin{equation}\label{E5analog}
    \hat{E}_g\left(\Hil_{\rm pn}\right)=\hat\Hil_g\,,
\end{equation}
in analogy to Eq.~\eqref{E5}. Note that we now have that $E_e\neq I$ in contrast to the Pauli case where $E_\chi=I$. 

Perhaps somewhat more appropriately, we can therefore view $\{\hat{E}_g\}_{g\in G}$ as the set of unitary gauge fixing errors. Clearly, the two preceding lemmas imply the KL condition in the simple form analogous to Eq.~\eqref{KLsimple}
\begin{equation}\label{KLnice}
    \Pi_{\rm pn}\,\hat{E}^\dag_g\,\hat{E}_{g'}\,\Pi_{\rm pn}=\delta_{g,g'}\,\Pi_{\rm pn}\,.
\end{equation}
Given that the KL condition is obeyed, there exists a detection and recovery scheme. Since we also have
\begin{equation}\label{KLarbitrary}
    \Pi_{\rm pn}\,\hat{E}_g\,\Pi_{\rm pn}=\frac{1}{\sqrt{2^{n-k}}}\Pi_{\rm pn}\,,\qquad\forall\,g\in G\,,
\end{equation}
it might at first appear as though the incoherent group averaging operator $\sqrt{2^{n-k}}\,\Pi_{\rm pn}$ constitutes a blanket (i.e.\ error-independent) recovery for \emph{all} possible gauge fixing errors because it maps all gauge-fixed states back into the original perspective-neutral ones. However, for $n>k$, $\sqrt{2^{n-k}}\,\Pi_{\rm pn}$ is neither a projector nor a unitary and cannot be implemented in a quantum operation. There  do exist quantum operations invoking $\Pi_{\rm pn}$ as operation elements that correct gauge fixing errors but they are either trace-decreasing or non-deterministic such that the blanket recovery $\sqrt{2^{n-k}}\,\Pi_{\rm pn}$ or $\Pi_{\rm pn}$ is operationally not usefully implementable. We explain this in further detail in App.~\ref{app_nooperation}.

How can the gauge fixing errors be detected? In contrast to standard Pauli errors, we cannot measure the stabilizer measurements $\{U^g\}$ to determine the syndrome because these do not leave the error spaces invariant. In fact, Eqs.~\eqref{errorspacetranslate} and~\eqref{E5analog} tell us that they translate between the different gauge fixing error spaces. Similarly, we cannot determine the $g$-error by measuring the $\{P_\chi\}_{\chi\in\hat{G}}$ either because these too do not commute with the $\{\hat{P}_g\}_{g\in G}$ (cf.~Eq.~\eqref{Pscommutation}).

Instead, everything works in a dual way to the standard Pauli error case. Since
\begin{equation}\label{PgEgP1}
    \hat{P}_g\,\hat{E}_g\,\Pi_{\rm pn}=\hat{E}_g\,\Pi_{\rm pn}
\end{equation}
and, in analogy to Eq.~\eqref{ugpchi},
\begin{equation}\label{uchipg}
\hat{U}^\chi\,\hat{P}_g=\chi(g)\,\hat{P}_g\,,
\end{equation}
we can interpret the gauge fixing errors $\{\hat{E}_g\}$ as \emph{magnetic charge excitations} and the group elements (characters of the dual) $g\in G$ as labeling these dual magnetic charges. This is analogous to Abelian gauge theory, where the charges of the Pontryagin dual are magnetic.

Eq.~\eqref{uchipg} clarifies that we can now detect the error/charge $g$ by measuring any generating set $\{\hat{U}^\chi\}_{\chi\in\hat{G}_{\rm gen}}$ of the Pontryagin dual of the stabilizer group. The error syndrome is given by the set of characters for fixed $g$, $\{\chi(g)\}_{\chi\in\hat{G}_{\rm gen}}$, thus determining the magnetic charge $g$. When syndrome $g$ is found, the error is corrected by simply applying $\hat{E}_g^\dag$; this is a de-excitation of the magnetic charge $g$. One can equivalently detect the error by measuring the projectors $\{\hat{P}_g\}_{g\in G}$ instead. In analogy to Eq.~\eqref{projeqGtwirl}, this is manifested in the following identity for arbitrary $\rho\in\mathcal{S}(\Hil_{\rm kin})$ (whose proof is analogous to that of Lemma~\ref{lem_det} and thus omitted):
\begin{equation}
    \sum_{g\in G}\hat{P}_g\,\rho\,\hat{P}_g=\frac{1}{2^{n-k}}\sum_{\chi\in\hat{G}}\hat{U}^\chi\,\rho\,\hat{U}^\chi\,.
\end{equation}
The incoherent group average ($\hat{G}$-twirl) over the Pontryagin dual $\hat{G}$ therefore coincides with the projective measurement $\{\hat{P}_g\}_{g\in G}$ of the \emph{dual} $g$-charge sectors.

In contrast to the Pauli error case, the detection measurement is now performed with operator sets $\{\hat{U}^\chi\}_{\chi\in\hat{G}_{\rm gen}}$ or $\{\hat{P}_g\}_{g\in G}$, neither of which commute with the stabilizer group $\mathcal{G}$, thus neither of which leaves the code space invariant. This is because the gauge fixing errors do not map into the orthogonal complement of $\Hil_{\rm pn}$. Nevertheless, there is a functional recovery scheme. Let a gauge fixing error channel be given by
\begin{equation}
    \tilde{\mathcal{E}}:\mathcal{B}(\Hil_{\rm pn})\rightarrow\mathcal{B}(\Hil_{\rm kin})\,,\qquad \tilde{\mathcal{E}}(\rho_{\rm pn})=\sum_k\tilde{E}_k\,\rho_{\rm pn}\,\tilde{E}^\dag_k\,,
\end{equation}
where the $\{\tilde{E}_k\}$ are normalized linear combinations of the $\{\hat{E}_g\}_{g\in G}$ such that $\sum_k\tilde{E}^\dag_k\tilde{E}_k=I$. This accommodates the possibility that no gauge fixing error occurs. The appropriate recovery scheme is dual to the one for Pauli errors in Eq.~\eqref{electricrecovery}.

\begin{lem}[\textbf{Recovery schemes for gauge fixing errors}]\label{lem_gaugedetect}
    Consider a gauge fixing error set $\mathcal{E}=\{\hat{E}_g\}$ for some $[[n,k]]$ stabilizer code subject to a faithful representation of ${G}=\mathbb{Z}_2^{n-k}$. The error-correction operation $\mathcal{O}:\mathcal{B}(\Hil_{\rm kin})\rightarrow\mathcal{B}(\Hil_{\rm pn})$ is given by 
    \begin{equation}
        \mathcal{O}(\rho)=\sum_g\,\hat{E}_g^\dag\,\hat{P}_g\,\rho\,\hat{P}_g\,\hat{E}_g=\sqrt{2^{n-k}}\,\Pi_{\rm pn}\left(\sum_g\,\hat{P}_g\,\rho\,\hat{P}_g\right)\,\Pi_{\rm pn}\,.
    \end{equation}
Further, $\mathcal{O}(\rho)=\rho$ for $\rho\in\mathcal{S}(\Hil_{\rm pn})$, so code words are left intact by the correction if no error occurred.
\end{lem}

\begin{proof}
The proof is given in App.~\ref{App:proofs}.
\end{proof}

So far, we have been somewhat general and only assumed a dual pair of representations $U,\hat{U}$ of the stabilizer group and its Pontryagin dual to exist, but not made connection with QRFs. However, 
Proposition \ref{prop_QRFPcorresp} tells us that there is a one-to-one correspondence between covariant orientation bases of an ideal QRF and dual representations of $\hat{G}$ on the QRF Hilbert space $\Hil_R$. Furthermore, any stabilizer code with faithful representation of $G$ gives rise to a multitude of ideal QRFs and therefore dual representations.
In that case, we can always write
\begin{equation}\label{eq:Egdef}
\hat{P}_g:=\ket{g}\!\bra{g}_R\otimes I_S\,,
\end{equation}
where the $\{\ket{g}_R\}_{g\in G}$ comprise a covariant and orthonormal orientation basis of $\Hil_R$. 

This connects the discussion of gauge fixing errors with the Page-Wootters reduction maps into $R$'s perspective, as introduced in Sec.~\ref{sec:extvsintframe}, which are clearly equivalent
\begin{equation}
    \hat{P}_g=\ket{g}\!\bra{g}_R\otimes I_S\,\qquad\Leftrightarrow \qquad\mathcal{R}_R^g=\sqrt{2^{n-k}}(\bra{g}_R\otimes I_S)\,.
\end{equation}
In particular, we have the QRF-aligned states \cite{Hoehn:2021flk,Krumm:2020fws}
\begin{equation}\label{eq:PGpsipn}
    \hat{P}_g\,\ket{\psi}_{\rm pn}=\frac{1}{\sqrt{2^{n-k}}}\ket{g}_R\otimes\ket{\psi(g)}_{|R}\,,
\end{equation}
where $\ket{\psi(g)}_{|R}$ is the covariant Page-Wootters state in the perspective of $R$, so that
\begin{equation}
    U^{g'}\hat{P}_g\ket{\psi}_{\rm pn}=\frac{1}{\sqrt{2^{n-k}}}\ket{g'g}_R\otimes\ket{\psi(g'g)}_S\,
\end{equation}
and
\begin{equation}
    \hat{\Hil}_{g}=\ket{g}_R\otimes\Hil_{|R}\,
\end{equation}
with $\Hil_{|R}$ the reduced Hilbert space in $R$-perspective.

This holds for \emph{any} ideal QRF, including when $\mathcal{G}$ acts non-trivially also on $S$. For the error duality, we obtain a slightly stronger statement from a corollary of Theorem \ref{thm_errorQRF}, which says that any correctable set of Pauli errors from Sec.~\ref{ssec_pauliduality} is equivalent to errors gauge fixing the orientations of ideal QRFs on whose complement $\mathcal{G}$ acts trivially. By Proposition \ref{prop_QRFPcorresp}, each ideal orientation basis of such a QRF defines a dual representation $\hat{U}$ of $\hat{G}$.
\begin{cor}[\textbf{Equivalence of Pauli and gauge fixing error sets}] \label{cor:gaugefix-Pauli}
    Every maximal set of correctable Pauli errors in a stabilizer code with faithful representation of ${G}$ is equivalent, according to Definition~\ref{def_equiv}, to a complete set of gauge fixing errors $\mathcal{E}=\{\hat{P}_g=\ket{g}\!\bra{g}_R\otimes I_S\,|\,g\in\mathcal{G}\}$, where the $\ket{g}_R$ are the orientation states of some ideal QRF $R$ such that $U_S^g=I_S$ for all $g\in G$.
\end{cor}
 In other words, the information about \emph{every} correctable Pauli error also lies in the span of the Page-Wootters reduction maps $\{\mathcal{R}^g_R\}_{g\in G}$ of \emph{some} ideal QRF.
The converse of the corollary statement may not be true, i.e.\ not every maximal set of gauge fixing errors for a given QRF may be equivalent to a set of Pauli errors.

\begin{proof}
    This follows from the proof of Theorem~\ref{thm_errorQRF}, which establishes (i) that every maximal set of correctable Pauli errors generates a QRF $R$, such that the code space restriction of these errors acts exclusively on $R$ and generates the stabilizer group. Conversely, it establishes (ii) that for each QRF $R$ with the property that the stabilizer group acts trivially on its complement, there exists a dual representation $\hat{\mathcal{G}}$ of the Pontryagin dual, which defines the unique equivalence class of maximal sets of correctable errors whose code space restriction only acts on $R$. By Eq.~\eqref{equivalences} and Proposition \ref{prop_QRFPcorresp}, the Pontryagin dual is equivalent to gauge fixing errors with the stated properties.
\end{proof}

\begin{example}[Three-qubit code]
\label{3qubit:errorduality2}
To illustrate the \emph{magnetic} gauge fixing errors, let us continue with the three-qubit code of Example~\ref{3qubit:errorduality1} where $R=12$, $\{\ket{g}_R\}_{g\in G}=\{\ket{++},\ket{+-},\ket{-+},\ket{--}\}$ with $X\ket{\pm}=\pm\ket{\pm}$, and $S=3$.~Using the faithful representation \eqref{3qubit:dualGrep} of $\hat G$ provided by the maximal set of correctable errors $\hat{\mathcal E}=\{I,X_1,X_2,X_1X_2\}$ together with the characters \eqref{3qubit:characters1}-\eqref{3qubit:characters4}, the gauge fixing projectors \eqref{PG} read as
\begin{equation}\label{3qubitPG}
\begin{aligned}
    \hat{P}_{++}&=\frac{1}{4}(I+X_1+X_2+X_1X_2)=\ket{++}\!\bra{++}_R\otimes I_S\;,\\
     \hat{P}_{+-}&=\frac{1}{4}(I+X_1-X_2-X_1X_2)=\ket{+-}\!\bra{+-}_R\otimes I_S\;,\\
      \hat{P}_{-+}&=\frac{1}{4}(I-X_1+X_2-X_1X_2)=\ket{-+}\!\bra{-+}_R\otimes I_S\;,\\
       \hat{P}_{--}&=\frac{1}{4}(I-X_1-X_2+X_1X_2)=\ket{--}\!\bra{--}_R\otimes I_S\;.
\end{aligned}
\end{equation}
Their action on a generic code state $\ket{\psi_{\rm pn}}=a\ket{000}+b\ket{111}$, $|a|^2+|b|^2=1$, yields (cf.~Eq.~\eqref{eq:PGpsipn})
\begin{equation}
    \begin{aligned}
       \hat{P}_{++}\ket{\psi_{\rm pn}}&=\frac{1}{2}\ket{++}_R\otimes\ket{\psi(++)}_{|R}\\ 
       \hat{P}_{+-}\ket{\psi_{\rm pn}}&=\frac{1}{2}\ket{+-}_R\otimes\ket{\psi(+-)}_{|R}\\ 
       \hat{P}_{-+}\ket{\psi_{\rm pn}}&=\frac{1}{2}\ket{-+}_R\otimes\ket{\psi(-+)}_{|R}\\ 
       \hat{P}_{--}\ket{\psi_{\rm pn}}&=\frac{1}{2}\ket{--}_R\otimes\ket{\psi(--)}_{|R}
    \end{aligned}\qquad\text{with}\qquad \begin{aligned} \label{eq:givenRexample}
        \ket{\psi(++)}_{|R}&=\ket{\psi(--)}_{|R}=a\ket{0}+b\ket{1}\\
        \ket{\psi(+-)}_{|R}&=\ket{\psi(-+)}_{|R}=a\ket{0}-b\ket{1}\;.
    \end{aligned}
\end{equation}
$\Hil_{\rm kin}\simeq(\mathbb C^2)^{\otimes3}$ can be then decomposed into the two-dimensional orthogonal isotypes of $\hat G$
\begin{equation}\label{3qubit:hatGisotype}
    \begin{aligned}
\hat{\Hil}_{++}&=\ket{++}_R\otimes\Hil_{|R}\\ 
\hat{\Hil}_{+-}&=\ket{+-}_R\otimes\Hil_{|R}\\ 
\hat{\Hil}_{-+}&=\ket{-+}_R\otimes\Hil_{|R}\\ 
\hat{\Hil}_{--}&=\ket{--}_R\otimes\Hil_{|R}
\end{aligned}\qquad\text{with}\qquad\Hil_{|R}\simeq\Hil_S=\text{span}\{\ket0,\ket1\}\;.
\end{equation}
Unlike the isotypes $\Hil_{i\neq1}$ corresponding to $\chi_{i\neq 1}$, which were all orthogonal to the code space $\Hil_{\rm pn}=\Hil_1=\text{span}\{\ket{000},\ket{111}\}$ (cf.~Eq.~\eqref{3qubit:Hchi}), the dual isotypes \eqref{3qubit:hatGisotype} all overlap with $\Hil_{\rm pn}$.~Moreover, similarly to the dual transformations $\hat{U}^{\chi}$ in Eq.~\eqref{3qubit:dualGrep} which translate between the isotypes $\Hil_{i}$, the stabilizer gauge transformations $U^g$ in Eq.~\eqref{eq:3qubitGrep} translate between the isotypes \eqref{3qubit:hatGisotype}.

Using $\hat{\mathcal E}=\{I,X_1,X_2,X_1X_2\}$ as the maximal, non-degenerate set of correctable errors in Eq.~\eqref{utilde}, a unitary implementation of the gauge fixing errors \eqref{3qubitPG} on $\Hil_{\rm pn}$ can be checked to be

\begin{equation}\label{3qubit:unitaryGF}
    \begin{array}{lll}
        \displaystyle E_{++} =\frac{1}{2}(I-Y_1Y_2+Z_1X_2Z_3+X_1Z_2Z_3)\,, & \qquad & E_{++} \Pi_{\rm pn}= 2\hat{P}_{++}\Pi_{\rm pn}\;, \\[2mm]
        \displaystyle E_{+-}=\frac{1}{2}(I-iZ_1Y_2+X_1Z_2Z_3-iY_1X_2Z_3)\,, & \qquad & E_{+-} \Pi_{\rm pn}=2\hat{P}_{+-}\Pi_{\rm pn}\;, \\[2mm]
        \displaystyle E_{-+}=\frac{1}{2}(I-iY_1Z_2+Z_1X_2Z_3-iX_1Y_2Z_3)\,, & \qquad & E_{-+}\Pi_{\rm pn}=2\hat{P}_{-+}\Pi_{\rm pn}\;, \\[2mm]
        \displaystyle E_{--}=\frac{1}{2}(I-Y_1Y_2-iY_1Z_3-iY_2Z_3)\,, & \qquad & E_{--}\Pi_{\rm pn}=2\hat{P}_{--}\Pi_{\rm pn}\;,
    \end{array}
\end{equation}
with $\Pi_{\rm pn}=P_1$ given in Sec.~\ref{Sec:3qubitcodeasQRF}, Eq.~\eqref{eq:3qubit:PicodePipn}.~It is straightforward to check that Eqs.~\eqref{PgEgP1} and \eqref{uchipg} are satisfied for the projectors \eqref{3qubitPG} and their unitary implementation \eqref{3qubit:unitaryGF}.~The latter can then be interpreted as the \emph{magnetic charge excitations} labeled by the elements of $G$ (the characters of $\hat G$).~Lastly, the dual basis states $\{\ket{g}_R\}_{g\in G}$ and $\{\ket{\chi_i}_R\}_{i=1,\dots,4}$ are related via Fourier transform as 
\begin{equation}
\begin{aligned}
\ket{\chi_1}_R=\frac{1}{2}(\ket{++}_R+\ket{+-}_R+\ket{-+}_R+\ket{--}_R)=&\ket{00}_R\;,\\
\ket{\chi_2}_R=\frac{1}{2}(\ket{++}_R+\ket{+-}_R-\ket{-+}_R-\ket{--}_R)=&\ket{10}_R\;,\\
\ket{\chi_3}_R=\frac{1}{2}(\ket{++}_R-\ket{+-}_R+\ket{-+}_R-\ket{--}_R)=&\ket{01}_R\;,\\
\ket{\chi_4}_R=\frac{1}{2}(\ket{++}_R-\ket{+-}_R-\ket{-+}_R+\ket{--}_R)=&\ket{11}_R\;.
\end{aligned}
\end{equation}
\eox
\end{example}

Let us thus now invoke the error-induced TPS $\Hil_{\rm kin}\simeq\Hil_R\otimes\Hil_S$ from Theorem \ref{thm_errorQRF} such that $U^g=U^g_R\otimes I_S$ acts trivially on the system $S$. Then, in analogy to Eq.~\eqref{chione}, we can write more explicitly
\begin{equation}
\hat{E}_g\,\Pi_{\rm pn}=\ket{g}\!\bra{1}_R\otimes I_S\,,
\end{equation}
and, in analogy to Eq.~\eqref{chiHS},
\begin{eqnarray}
    \hat{\Hil}_{\rm g}\simeq\ket{g}_R\otimes\Hil_S\,.
\end{eqnarray}
Finally, using the duality of the representations $U,\hat U$ \eqref{dualrep}, we have for the basis states $\{\ket{\chi}_R\}_{\chi\in\hat{G}}$ of the error set from Eq.~\eqref{pchisimple}
\begin{equation}\label{MUB}
    \braket{\chi}{g}_R=\bra{1}\hat{U}^\chi_R \,U_R^g\ket{e}_R=\chi(g)\bra{1}U_R^g\,\hat{U}_R^\chi\ket{e}_R=\chi(g)\braket{1}{e}_R=\chi(g)\bra{1}E_e\ket{1}\underset{\eqref{KLarbitrary}}{=}\frac{\chi(g)}{\sqrt{2^{n-k}}}
\end{equation}
since $U^g_R\ket{1}_R=\ket{1}_R$ and $\hat{U}_R^\chi\ket{e}_R=\ket{e}_R$, given that $\chi=1$ and $g=e$ label the trivial representations of $G$ and $\hat{G}$, respectively. Thus, we have the Fourier transform
\begin{equation}\label{Fourier}
    \ket{g}_R=\frac{1}{\sqrt{2^{n-k}}}\sum_{\chi\in\hat{G}}\chi(g)\ket{\chi}_R\,,
\end{equation}
manifesting the duality between the two sets of errors and corresponding charges and leading to
\begin{equation}
    P_\chi\,\hat{P}_g = \frac{\chi(g)}{\sqrt{2^{n-k}}}\ket{\chi}\!\bra{g}_R\otimes I_S\,,\qquad\forall\,g\in G\,,\chi\in\hat{G}\,.
\end{equation}
The projectors onto the respective charge sectors thus do not commute. When we know the value of the charge $g$ we know nothing about the dual charge $\chi$, and vice versa. This becomes especially apparent in such a TPS but holds in general for charges associated to a dual pair of representations $U,\hat{U}$, as explained in App.~\ref{app:duality}.

\medskip

Let us finally look at nonlocal tensor product factorizations, but in the language of Sec.~\ref{sec:dressing_field}. To any  complete set of frame fields $R= \{R_\chi\}_{\chi \in \hat{G}\setminus\{1\}}$, one can naturally associate a representation of the dual group $\hat{U}_R$. Furthermore, such $\hat{U}_R$ turns out to be always dual to $U$, and thus defines projection operators $\{\hat{P}^R_g\}_{g\in G}$ whose action on the code subspace can be interpreted as gauge-fixing errors (modulo a rescaling by $\sqrt{|G|}$). Let us start by defining $\hat{U}_R$.
\begin{defn}\label{def:dual_rep_frome_frames}
    Let $R= \{R_\chi\}_{\chi \in \hat{G} \setminus\{1\}}$ be a complete set of frame fields, and let $\otimes_R$ be the tensor product defined in Theorem \ref{thm:nonlocal_fact}. One thus has $\cH_{\rm kin} = \cH_{\rm pn} \otimes_R \cH_{\rm gauge}$, with $\cH_{\rm gauge}$ spanned by the orthonormal basis $\{ \vert \chi \rangle \,, \, \chi \in \hat{G}\}$. For any $\vert \psi \rangle \in \cH_{\rm pn}$ and any $\chi, \eta \in \hat{G}$, we define:
    \begin{equation}
        \hat{U}^\chi_R (\vert \psi \rangle \otimes_R \vert \eta \rangle) := \vert \psi \rangle \otimes_R \vert \chi \eta \rangle \,.
    \end{equation}
    The map $\hat{U}_R : \chi \mapsto \hat{U}^\chi_R$ defines a unitary representation of $\hat{G}$, which is explicitly isomorphic to $2^k$ copies of the regular representation of $\hat{G}$.
\end{defn}
For notational convenience, let us set
\begin{equation}
    R_1 := I\,,
\end{equation}
so that a complete set of frame fields now defines a collection of isometries $\{R_\chi \}_{\chi \in \hat{G}}$ indexed by the full dual group. 
Note that, by construction, we have:
\begin{equation}
    \forall \chi \in \hat{G}\,,\quad \hat{U}^\chi_R \Pi_{\rm pn} = R_\chi \Pi_{\rm pn}\,. 
\end{equation}
It is straightforward to express $\hat{U}_R$ directly in terms of the frame fields and the representation $U$.
\begin{lem}
    For any $\chi \in \hat{G}$, we have
    \begin{equation}\label{eq:dual_rep_associated_to_frame}
        \hat{U}^\chi_R = \sum_{\eta \in \hat{G}} R_{\chi \eta} R_{\eta}^{-1} P_\eta = \frac{1}{|G|} \sum_{g \in G} \sum_{\eta \in \hat{G}} \eta(g) R_{\chi \eta} R_{\eta}^{-1} U^g\,.  
    \end{equation}
\end{lem}
\begin{proof}
    Let $\chi \in \hat{G}$. For any $\vert \psi \rangle \in \cH_{\rm pn}$ and $\eta' \in \hat{G}$, we have $\vert \psi \rangle \otimes_R \vert \eta' \rangle:= R_{\eta'} \vert \psi \rangle$, so that:
    \begin{align*}
        \sum_{\eta \in \hat{G}} R_{\chi \eta} R_{\eta}^{-1} P_\eta (\vert \psi \rangle \otimes_R \vert \eta' \rangle) &= \sum_{\eta \in \hat{G}} R_{\chi \eta} R_{\eta}^{-1} P_\eta R_{\eta'} \vert \psi \rangle = R_{\chi \eta'} R_{\eta'}^{-1} R_{\eta'} \vert \psi \rangle \\ 
        &= R_{\chi \eta'} \vert \psi \rangle  = \vert \psi \rangle \otimes_R \vert \chi \eta'\rangle = \hat{U}_R^\chi (\vert \psi \rangle \otimes_R \vert \eta'\rangle)\,. 
    \end{align*}
    This establishes the first equality, and the second one follows by plugging-in the expression of the projector $P_\chi$ \eqref{eq:projector_isotype}.
\end{proof}
We can then check that $\hat{U}_R$ is indeed dual to the representation $U$. 
\begin{prop}
    The unitary representations $U$ and $\hat{U}_R$ are dual to each other (in the sense of Definition \ref{def:dual_reps_faithful}).
\end{prop}
\begin{proof}
    For any $\vert \psi\rangle \in \cH_{\rm pn}$, any $g \in G$, and any $\chi, \eta \in \hat{G}$, we have
    \begin{equation}
        U^g \hat{U}_R^\chi (\vert \psi \rangle \otimes_R \vert \eta \rangle ) = U^g (\vert \psi \rangle \otimes_R \vert \chi \eta \rangle ) = \vert \psi \rangle \otimes_R U^g_R \vert \chi \eta \rangle  = \chi(g) \eta(g) \vert \psi \rangle \otimes_R  \vert \chi \eta \rangle 
    \end{equation}
    and
    \begin{equation}
         \hat{U}_R^\chi U^g (\vert \psi \rangle \otimes_R \vert \eta \rangle ) = \hat{U}_R^\chi ( \vert \psi \rangle \otimes_R U^g_R \vert \eta \rangle ) = \eta(g) \hat{U}_R^\chi (\vert \psi \rangle \otimes_R \vert \eta \rangle )  = \eta(g) \vert \psi \rangle \otimes_R  \vert \chi \eta \rangle\,, 
    \end{equation}
    so that $U^g \hat{U}_R^\chi = \chi(g) \hat{U}_R^\chi U^g$.
\end{proof}
This result should be physically intuitive: it is an analogue for finite Abelian groups of the more standard observation that, in quantum mechanics, finite translations in momentum space are related to finite translations in configuration space by Weyl commutation relations. In the case of quantum mechanics on the line (resp.~circle), the relevant groups are $G=\mathbb{R}$ (resp.~$G = {\rm U}(1)$) and $\hat{G} = \mathbb{R}$ (resp.~$\hat{G} = \mathbb{Z}$).
Finally, according to Appendix \ref{app:reps}, rescaled orthogonal projectors $\{\sqrt{|G|} \hat{P}^R_g\}_{g \in G}$ which project on the various isotypes of the representation $\hat{U}_R$ can be obtained by means of a Fourier transform; namely:
\begin{equation}\label{eq:dual_proj_exact}
    \forall g \in G\,, \qquad \sqrt{|G|} \hat{P}^R_g := \frac{1}{\sqrt{\vert G\vert}} \sum_{\chi \in \hat{G}} \chi(g) \hat{U}^\chi_R\,. 
\end{equation}
As was explained above, the action of such operators on the code subspace $\cH_{\rm pn}$ can be interpreted as gauge-fixing errors. We thus have the following proposition.
\begin{prop}\label{prop:gauge-fixing_errors_frame}
    To a complete set of frame fields $R=\{ R_\chi \}_{\chi \in \hat{G}}$, we can associate a set of gauge-fixing errors $\{ \hat{E}_g\}_{g \in G}$, whose restrictions to $\cH_{\rm pn}$ obey
    \begin{equation}
        \hat{E}_g \Pi_{\rm pn} = \sqrt{\vert G\vert} \hat{P}^R_g \Pi_{\rm pn} = \frac{1}{\sqrt{|G|}} \sum_{\chi \in \hat{G}} \chi(g) R_\chi \Pi_{\rm pn}\,. 
    \end{equation}
\end{prop}
\begin{proof}
    Since $\hat{U}_R^\chi$ and $R_\chi$ agree on the code subspace, we have 
    \begin{equation}\label{eq:dual_proj_code}
        \hat{P}^R_g \Pi_{\rm pn}= \frac{1}{\vert G\vert} \sum_{\chi \in \hat{G}} \chi(g) \hat{U}^\chi_R \Pi_{\rm pn} = \frac{1}{\vert G\vert} \sum_{\chi \in \hat{G}} \chi(g) R_\chi \Pi_{\rm pn} \,,
    \end{equation}
    from which the claim follows.
\end{proof}
We thus conclude that the gauge-fixing errors can be interpreted as Fourier transforms of the frame fields. Let us give an illustration of those findings in the $3$-qubit code.
\begin{example}
Let us consider the $3$-qubit code together with the complete set of frame fields $$R= (R_{1} , R_{2} , R_{3} , R_{4}) = (\restr{I}{\cH_1}, \restr{X_1}{\cH_1} , \restr{X_2}{\cH_1} , \restr{X_3}{\cH_1})\,,$$
as already discussed in example \ref{ex:3qubit_single_flips}. A dual representation and a dual error set has already been constructed for the same code in examples \ref{3qubit:errorduality1} and \ref{3qubit:errorduality2}. Recall that the representation $U$ is defined by \eqref{eq:3qubitGrep}, and the dual group is expressed in terms of the characters $(\chi_1 , \chi_2, \chi_3, \chi_4)$ from equations \eqref{3qubit:characters1}-\eqref{3qubit:characters4}. Definition \ref{def:dual_rep_frome_frames} yields, by direct evaluation of expression \eqref{eq:dual_rep_associated_to_frame} (which is straightforward but somewhat tedious), the following dual representation:
\begin{align}
    \hat{U}^{\chi_1}_R &= I\,, \\
    \hat{U}^{\chi_2}_R &= \frac{1}{2} \left( X_1 + X_2 X_3 + X_1 Z_2 Z_3 - X_2 X_3 Z_2 Z_3  \right)\,, \\
    \hat{U}^{\chi_3}_R &= \frac{1}{2} \left( X_2 + X_1 X_3 + X_2 Z_1 Z_3 - X_1 X_3 Z_1 Z_3  \right)\,, \\
    \hat{U}^{\chi_4}_R &= \frac{1}{2} \left( X_3 + X_1 X_2 + X_3 Z_1 Z_2 - X_1 X_2 Z_1 Z_2 \right)\,. 
\end{align}
Note in particular that $\hat{U}^{\chi_i}_R \Pi_{\rm pn} = R_i \Pi_{\rm pn}$ for every $1 \leq i \leq 4$. The dual representation so constructed is therefore adapted to the set of frame fields $R$, or equivalently, to the error set $\cE = \{I, X_1 , X_2, X_3\}$. By contrast, the construction of example \ref{3qubit:errorduality1} produced a dual representation adapted to $R''= (\restr{I}{\cH_1}, \restr{X_1}{\cH_1} , \restr{X_2}{\cH_1} , \restr{X_1 X_2}{\cH_1})$, or equivalently, to the error set $\cE'' = \{ I , X_1 , X_2 , X_1 X_2 \}$ (see example \ref{ex:3qubit_third_qubit_logical}). Equation \eqref{eq:dual_proj_exact} allows to determine the form of the dual projectors $\{\hat{P}_g^R\}_{g \in G}$. For instance, one finds:
\begin{align}
    \hat{P}^R_{++} &= \frac{1}{4} I + \frac{1}{8} \left( X_1 + X_2 + X_3 \right)+ \frac{1}{8} \left( X_2 X_3+ X_1 X_3 + X_1 X_2 \right) + \frac{1}{8} \left( X_1 Z_2 Z_3+ X_2 Z_1 Z_3 + X_3 Z_1 Z_2 \right) \\
    &\qquad - \frac{1}{8} \left( X_2 X_3 Z_2 Z_3 + X_1 X_3 Z_1 Z_3 + X_1 X_2 Z_1 Z_2 \right) \,.\nonumber 
\end{align}
When acting on the code subspace, a rescaled version of those projectors determine the action of the dual error set $\{ \hat{E}_g \}_{g \in G}$; according to \eqref{eq:dual_proj_code}, we have:
\begin{align}
    \hat{E}_{++} \Pi_{\rm pn } &= 2 \hat{P}^R_{++} \Pi_{\rm pn} = \frac{1}{2} \left( I + X_{1} + X_{2} + X_{3}  \right) \Pi_{\rm pn} \,, \\
    \hat{E}_{+-} \Pi_{\rm pn } &= 2 \hat{P}^R_{+-} \Pi_{\rm pn} = \frac{1}{2} \left( I + X_{1} - X_{2} - X_{3}  \right) \Pi_{\rm pn} \,, \\
    \hat{E}_{-+} \Pi_{\rm pn } &= 2 \hat{P}^R_{-+} \Pi_{\rm pn} = \frac{1}{2} \left( I - X_{1} + X_{2} - X_{3}  \right) \Pi_{\rm pn} \,, \\
    \hat{E}_{--} \Pi_{\rm pn } &= 2 \hat{P}^R_{-+} \Pi_{\rm pn} = \frac{1}{2} \left( I - X_{1} - X_{2} + X_{3}  \right) \Pi_{\rm pn} \,.
\end{align}
\end{example}

\subsection{Summary and interpretation of error duality}

In summary, we  have the duality: 
\begin{eqnarray}
    \text{stabilizer gauge group}\quad {G}&\longleftrightarrow& \hat{{G}}\quad\text{Pontryagin dual}\nonumber\\
    \text{stabilizer gauge transformation}\quad U^g&\longleftrightarrow&\hat{U}^\chi\quad\text{dual group transformation}\nonumber\\
    \text{electric charge label}\quad \chi\in\hat{{G}}&\longleftrightarrow& g\in{G}\quad\text{magnetic charge label}\nonumber\\
    \text{electric charge excitation}\quad E_\chi&\longleftrightarrow& \hat{E}_g\quad\text{magnetic charge excitation}\nonumber\\
    \text{electric charge sector projector}\quad P_\chi&\longleftrightarrow& \hat{P}_g\quad\text{magnetic charge sector projector}\nonumber\\
    \text{elec.\ charge measurem.}\quad\{U^g\}_{g\in G_{\rm gen}}\text{ or }\{P_\chi\}_{\chi\in\hat{G}}&\longleftrightarrow&\{\hat{U}^\chi\}_{\chi\in\hat{G}_{\rm gen}}\text{ or }
\{\hat{P}_g\}_{g\in G}\quad\text{magn.\ charge measurem.}\nonumber\\
\text{electric syndrome}\quad\{\chi(g)\}_{g\in G_{\rm gen}}&\longleftrightarrow&\{\chi(g)\}_{\chi\in\hat{G}_{\rm gen}}\quad\text{magnetic syndrome}\nonumber
\end{eqnarray}
This constitutes an \emph{error duality} between standard correctable Pauli errors (as well as dual representations $\hat{\mathcal{G}}$ of $\hat{G}$) and gauge fixing errors. The counterparts for gauge fixing errors to all relevant properties, syndrome measurements and recovery schemes of Pauli errors are obtained, according to this dictionary, by simply swapping the dual $\chi\leftrightarrow g$ labels in the relevant expressions.

In analogy to Abelian gauge theory, we interpreted the charges of the stabilizer gauge group and its dual as electric and magnetic, respectively. Accordingly, we interpreted the standard correctable Pauli errors as electric and the new gauge fixing errors as magnetic charge excitations. 
Owing to their duality according to \eqref{dualrep}, we have that the representations $\mathcal{G}=U(G)$ of the stabilizer gauge group and $\hat{\mathcal{G}}=\hat{U}(\hat{G})$ of its Pontryagin dual do \emph{not} commute. This stands in contrast to the electromagnetic duality in the continuum, where the $\rm{U}(1)$ gauge transformations commute with their Pontryagin dual such that both their charges a gauge-invariant. Consequently, we have seen that electric and magnetic charges are complementary in our case: maximal information about one is at the expense of any information about the other and they lead to bases for QRFs that are mutually unbiased, Eq.~\eqref{MUB}. Specifically, the magnetic charges are \emph{not} gauge-invariant. However, $\{U^g\}$ and $\{P_\chi\}$ commute, as do $\{\hat{U}^\chi\}$ and $\{\hat{P}_g\}$ which enables the electric and magnetic charge or, equivalently, syndrome measurements. 

This difference from the standard form of electromagnetic duality will also become visible in the surface codes of Sec.~\ref{sec:surfacecodes}, where we will encounter both. What we present here as error duality thus appears to be a new and somewhat stronger form of Abelian charge duality.

The error duality is not unique to a fixed pair of Pauli and gauge fixing error sets. In stabilizer QEC, we are usually given a fixed representation $U(G)$. Now fix also a dual representation $\hat{U}(\hat{G})$ leading to a fixed gauge fixing error set defined by $\{\hat{P}_{g}\}_{g\in G}$ in Eq.~\eqref{PG}. This gauge fixing error set is dual, according to the above dictionary, to \emph{any} maximal set of correctable Pauli errors, as well as the elements $\{\hat{U}^\chi\}_{\chi\in\hat{G}}$ of \emph{any} representation $\hat{U}$ of $\hat{G}$ that is dual to $U$. Conversely, fix any maximal set of correctable Pauli errors. This is dual to the gauge fixing error set of \emph{any} dual representation of the Pontryagin dual $\hat{G}$. Consequently, for every electric charge set there exist many dual magnetic charge sets, and vice versa. 

The arguably cleanest form of the duality arises for the ideal QRFs generated by any maximal set of correctable Pauli errors according to Theorem~\ref{thm_errorQRF}. In that case, all of the elements in the above dictionary act only on the QRF Hilbert space $\Hil_R$ and lead to the Fourier transform Eq.~\eqref{Fourier} between the dual error bases. Nevertheless, even then the duality is not unique, as a given ideal QRF will generally admit many orthonormal orientation bases and thus dual representations of $\hat{G}$.

We have here taken a gauge theory/QRF perspective on these dual sets of errors. However, as per our QECC/QRF dictionary, they also have interpretations in standard QEC or quantum communication scenarios. The Pauli errors account for the standard errors, such as bit and phase flips, in the former, and in the latter they constitute transformations mapping external-frame-independent information into superselection sectors that no longer are invariant under external frame transformations, as discussed in Sec.~\ref{sec_dictionary}. The dual gauge fixing errors have not previously been discussed in such scenarios and essentially amount to just removing redundancy.

\section{Surface codes}\label{sec:surfacecodes}

As a step forward in connecting QEC with actual gauge theories, we shall now illustrate our findings in the context of \emph{surface codes} \cite{Kitaev:1997wr, Bravyi:1998sy}, which can be understood as two-dimensional lattice gauge theories with finite structure group.~We will focus on the simplest possible choice of structure group, namely $\mathbb{Z}_2$, as any such surface code can indeed be interpreted as a stabilizer code.~Such a code is also an example of homological code, exploiting the $\mathbb{Z}_2$-homology of surfaces.~For illustrative purposes, we will start by discussing the case of a planar square lattice with boundary \cite{Bravyi:1998sy}.~We will then turn to codes defined on closed surfaces of genus $g \geq 1$, as originally introduced by Kitaev \cite{Kitaev:1997wr}; this includes the well-known toric code.~Depending on how exactly we approach the latter codes, they may not be directly covered by the general results of the previous sections, because they may involve a representation of the gauge group that is not faithful.~However, even in such a situation, they can be analyzed with similar methods.~Moreover, we will see that the non-faithful character of the representation is directly related to a key physical property of such systems, namely: global charge conservation. 

\subsection{Surface code with boundary: square lattice with $\mathbb{Z}_2$ structure group}\label{subsec:square_code}

\subsubsection{Combinatorial structure of the lattice}

Surface codes with boundaries were introduced in \cite{Bravyi:1998sy}. Let us focus on the simplest model described in that reference, which is based on a regular square lattice with rectangular shape, as illustrated in Fig.~\ref{fig:surface_code_disk}. We denote by $H$ (resp.\ $L$) the height (resp.\ width) of this architecture, which we define as the number of horizontal lines (resp.\ the number of vertical lines minus one) of the grid. Importantly, the lattice has two types of boundaries: the so-called \emph{rough} boundaries at the top and bottom, and the \emph{smooth} boundaries on the left and right sides. Let us denote by $\mathcal{E}$ and $\mathcal{V}$ the set of edges and vertices of the lattice. There are $(L+1)(H+1)$ vertical edges and $LH$ horizontal ones, therefore $|\cE|= 2LH + L +H +1$. As for the vertices, there are $\vert \cV \vert= LH +H$ of them. Finally we need to specify a set of \emph{faces} or \emph{plaquettes} for this architecture, which we denote by $\mathcal{F}$. A face is defined to be: 1) a collection of four edges lying on the boundary of a unit square in the lattice; or 2) a collection of three edges bordering an incomplete square which touches one of the two rough boundaries (see e.g.\ the structure of the plaquette labelled by $Z_{f_2}$ in Fig.~\ref{fig:surface_code_disk}). There are $ \vert \mathcal{F} \vert = LH+L$ faces in our architecture. Mathematically speaking, the graph structure together with the set of faces $\mathcal{F}$ defines a discrete orientable surface, also known as a \emph{combinatorial map}, with the topology of a disk. Obviously, experimental implementations of surface codes are also realized by laying quantum degrees of freedom on two-dimensional materials, hence the name. 

We will denote by $\cP(\cV)$,  $\cP(\cE)$ and $\cP(\cF)$ the power sets of $\cV$, $\cE$ and $\cF$ respectively. The discrete surface interpretation of our network also provides us with a \emph{boundary map} $\partial: \cP(\cV) \to \cP(\cE)$, which maps any $V \subset \cV$ to the subsets of its edges that lie on its boundary $\partial V \subset \cE$. More precisely, the map $\partial$ is defined as follows: if $v \in \cV$ has four incident edges $e_1, e_2, e_3, e_4$, then $\partial \{ v \} := \{ e_1 , e_2 , e_3, e_4\}$; if $v \in \cV$ has three incident edges $e_1, e_2, e_3$, then $\partial \{ v \} := \{ e_1 , e_2 , e_3 \}$; for any $V_1, V_2 \in \cV$, we furthermore impose $\partial (V_1 \Delta V_2 ) = (\partial V_1) \Delta (\partial V_2)$, where $\Delta$ denotes the symmetric difference of sets. In other words, the map $\partial$ is a group homomorphism from $(\cP(\cV), \Delta)$ to $(\cP(\cE), \Delta)$, which we have defined by its action on the set of $\cP(\cV)$-generators $\{ \{v\} \, \vert \, v \in \cV\}$. Likewise, we have a \emph{coboundary map} $\delta: \cP(\cF) \to \cP(\cE)$, which can be understood as the boundary map of the dual lattice.\footnote{In the dual lattice, faces take the role of dual vertices while vertices take the role of dual faces. By contrast, the notion of edge is self-dual, in the sense that an $e \in \cF$ can be considered both as an edge and a dual edge. We will repeatedly (and often implicitly) make use of this fact throughout the text.} Namely, if $f \in \cF$ is made up of four edges $e_1 , e_2 , e_3, e_4$ (resp. three edges $e_1 , e_2 , e_3$), then $\delta \{f\}:= \{e_1 , e_2 , e_3 , e_4\}$ (resp. $\delta \{f\}:= \{e_1 , e_2 , e_3 \}$), and for any $F_1, F_2 \subset \cF$ we have the group homomorphism relation $\delta (F_1 \Delta F_2) = (\delta F_1 ) \Delta (\delta F_2)$. 

\begin{figure}[htb]
    \centering
    \includegraphics[scale=.7]{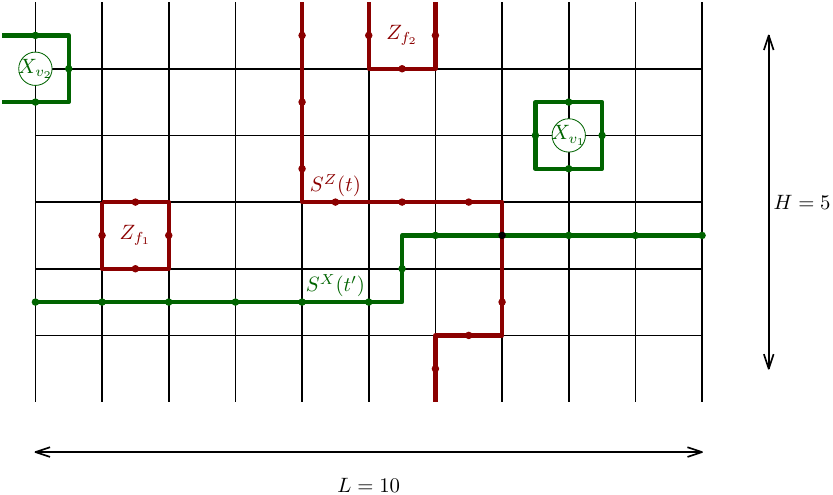}
    \caption{Surface code on a rectangular lattice of size $L \times H$; it has $(L+1)\times(H+1)$ vertical edges and $L\times H$ horizontal ones, supporting $2LH+L+H+1$ physical qubits in total. Illustrated in red and green are plaquettes and vertex operators, as well as non-trivial string operators which capture invariant data about the one logical qubit stored in the code subspace.}
    \label{fig:surface_code_disk}
\end{figure}

\subsubsection{Gauge group and its dual}

The abstract gauge group underlying this model can be taken to be the product of the power sets $\cP(\cV)$ and $\cP(\cF)$, with group law given by the symmetric difference of subsets. More precisely, we set
\begin{equation}\label{eq:group_surface}
G := \{ (V, F) \, \vert \, V \subset \cV \,, F \subset \cF \}     
\end{equation}
and, for any $(V_1 , F_1), (V_2 , F_2) \in G$,
\begin{equation}\label{eq:surface_group_law}
    (V_1 , F_1)(V_2 , F_2):= (V_1 \Delta V_2, F_1 \Delta F_2 )\,.
\end{equation}
Equation \eqref{eq:surface_group_law} defines an Abelian group structure with unit element $(\emptyset , \emptyset )$. Furthermore, for any $(V,F) \in G$ we have $(V,F)^{2}=(\emptyset , \emptyset )$, which makes $G$ isomorphic to $\mathbb{Z}_2^{\times (|\cV|+|\cF|)}$.\footnote{In a slightly different language, this definition of the abstract gauge group exploits the fact that the group law of $\mathbb{Z}_2$ is the same as exclusive OR in a Boolean algebra.}

We can build characters of $G$ in the following way. Given $(\hat{V},\hat{F}) \subset (\cV, \cF)$, we define
\begin{equation}\label{eq:charac_surface}
    \chi_{(\hat{V},\hat{F})}: G \to \{1,-1\}\,, \quad (V,F) \mapsto  (-1)^{| \hat{V} \cap  V|+ | \hat{F} \cap  F|}\,.
\end{equation}
This does define a character because, for any $(V_1 , F_1), (V_2, F_2) \in G$: 
\begin{align}
    \chi_{(\hat{V},\hat{F})} \left((V_1 , F_1) (V_2, F_2)\right) &=  \chi_{(\hat{V},\hat{F})} \left((V_1 \Delta V_2, F_1 \Delta F_2)\right) = (-1)^{|\hat{V} \cap (V_1 \Delta V_2)|} (-1)^{|\hat{F} \cap (F_1 \Delta F_2)|} \\
    &= (-1)^{|(\hat{V} \cap V_1) \Delta (\hat{V} \cap V_2)|}  (-1)^{|(\hat{F} \cap F_1) \Delta (\hat{F} \cap F_2)|} \\ 
    &   = (-1)^{|\hat{V} \cap  V_1| +|\hat{V} \cap V_2 | - 2 |\hat{V} \cap V_1 \cap V_2|} (-1)^{|\hat{F} \cap  F_1| +|\hat{F} \cap F_2 | - 2 |\hat{F} \cap F_1 \cap F_2|} \\
    & = (-1)^{|\hat{V} \cap  V_1| +|\hat{F} \cap  F_1| } (-1)^{|\hat{V} \cap V_2 | +|\hat{F} \cap F_2 |}  =   \chi_{(\hat{V},\hat{F})} \left((V_1 , F_1)\right) \, \chi_{(\hat{V},\hat{F})}\left( (V_2, F_2)\right) \,.
\end{align}

By symmetry of the roles of $G$ and $\hat{G}$ in the previous calculation, we also have: for any $(\hat{V}_1, \hat{F}_1), (\hat{V}_2, \hat{F}_2) \in \cP(\cV)\times \cP(\cF)$, 
\begin{equation}
    \chi_{(\hat{V}_1,\hat{F}_1)} \, \chi_{(\hat{V}_1,\hat{F}_2)} = \chi_{(\hat{V}_1 \Delta \hat{V}_2,\hat{F}_1 \Delta \hat{V}_2)}\,.
\end{equation}
This makes it clear that the characters $\{ \chi_{(\hat{V},\hat{F})} \, \vert \, \hat{V} \subset \cV\,, \hat{F}\subset \cF\}$ form a group isomorphic to $G$, hence they make up all of $\hat{G}$. Moreover, this shows that, similarly to $G$ itself, we can think of the group law in $\hat{G}$ as being represented by the symmetric difference.

\subsubsection{Representation of the gauge group, gauge constraints, and error spaces}

To the combinatorial map underlying the model, we associate a kinematical Hilbert space $\cH_{\rm kin}$ and a commuting set of constraint operators $\mathcal{G}$ which, together, define the invariant subspace $\mathcal{H}_{\rm pn}$. The kinematical degrees of freedom consist of one qubit per edge of the lattice, meaning that
\begin{equation}
    \cH_{\rm kin} = \bigotimes_{e \in \mathcal{E}} \mathbb{C}^2 \simeq (\mathbb{C}^2)^{\otimes \vert \cE \vert} \,. 
\end{equation}
For each edge $e$, we fix a computational basis $\Span\{ \vert 0 \rangle_e \,, \vert 1 \rangle_e\}$ in the associated $\mathbb{C}^2$, as well as Pauli operators $\{ I_e , X_e , Y_e, Z_e\}$ adapted to this basis. 

Next, we define \emph{vertex operators} and \emph{plaquette operators}. For any vertex $v \in \mathcal{V}$, the vertex operator $X_v$ is defined to be
\begin{equation}\label{eq:vertex_cst}
    X_v := \prod_{e \; \mathrm{incident}\; \mathrm{to} \; v } X_e\,.
\end{equation}
If $v$ lies on one of the smooth boundaries, then $X_v$ is a product of three bit-flip operators, and it is a product of four such operators otherwise (see Fig.~\ref{fig:surface_code_disk} for two examples). Similarly, for any $f \in \mathcal{F}$, the plaquette operator $Z_f$ is defined to be:
\begin{equation}\label{eq:face_cst}
    Z_f := \prod_{e \in f} Z_e\,.
\end{equation}
$Z_f$ is a product of three $Z$ operators if it lies on one of the rough boundaries, and is again a product of four such operators otherwise. Vertex operators obviously commute among themselves, and similarly for plaquette operators. Moreover, any pair of plaquette and vertex operators share zero or two edges, hence they also commute. Hence
\begin{equation}
    \mathcal{G} := \big\langle X_v , Z_f \, \vert \,  v \in  \mathcal{V} \,, f \in \mathcal{F} \big\rangle
\end{equation}
defines an Abelian subgroup of the Pauli group $\mathcal{P}_{\vert \cE \vert}$. Furthermore, we can view $\mathcal{G}$ as the image of the \emph{unitary  representation} $U:G \to {\rm Aut}(\cH_{\rm kin})\,, g \mapsto U^g$, defined by 
\begin{equation}\label{eq:U_surface_code}
    \forall (V, F) \in G\,, \qquad U^{(V,F)} := \prod_{v \in V} X_v \prod_{f \in F} Z_f = \prod_{e \in \partial V} X_e \prod_{e \in \delta F} Z_e \,.
\end{equation}
The representation $U$ verifies the technical assumption of Appendix \ref{app:stabilizer_reps}: it is \emph{faithful}\footnote{This claim can be justified by studying the $\mathbb{Z}_2$ homology of the discrete surface.} and $-I \notin {\rm Im}(U)= \mathcal{G}$. Hence $\uU: \mathbb{C}[G] \to \cB(\cH_{\rm kin})$ defines a faithful $*$-algebra representation (see Appendix \ref{app:reps}). In particular, we have 
\begin{equation}
    \mathcal{G} \simeq \mathbb{Z}_2^{\times (2LH+L+H)}\,.
\end{equation}
The pair $(\mathcal{H}_{\rm kin}, \mathcal{G})$ thus defines a stabilizer code with $n = 2LH + L + H + 1$ and $n-k = 2LH+L+H$. It can therefore encode
\begin{equation}
    k = 1
\end{equation}
logical qubit in the perspective-neutral Hilbert space
\begin{equation}
    \cH_{\rm pn} := \big\{ \vert \Psi \rangle \in \cH_{\rm kin} \, \vert \,  \forall v \in \mathcal{V}, \forall f \in \mathcal{F}\,, X_v \vert \Psi \rangle = \vert \Psi \rangle = Z_f \vert \Psi \rangle \big\} \simeq \mathbb{C}^2\,.
\end{equation}

At this point, we would like to emphasize that the gauge-theoretic interpretation of the model we are relying on here is \emph{different from the standard one}. Indeed, in the standard $\mathbb{Z}_2$-lattice gauge theory interpretation, the gauge group is generated by the vertex operators alone, while plaquette operators are interpreted as \emph{flatness constraints} for lattice $\mathbb{Z}_2$-holonomies. Note that those two sets of constraints are mapped into one-another by an \emph{electromagnetic duality} (analogous to the one of Maxwell theory) which exchanges the original lattice for its dual, as well as the roles of $X$ and $Z$. Which of those two sets we interpret as gauge conditions (resp. flatness conditions) is therefore a matter of choice. But for both choices, we can interpret $\cH_{\rm pn}$ as the \emph{flat sector} (or vacuum sector) of a $\mathbb{Z}_2$ lattice gauge theory. In this interpretation, the gauge group of the model is isomorphic to $\mathbb{Z}_2^{\times \vert \cV \vert}$ (or $\mathbb{Z}_2^{\times \vert \cF \vert}$ in the dual picture), while in our alternative interpretation it is isomorphic to $\mathbb{Z}_2^{\times (\vert \cV \vert+ \vert \cF \vert)}$. 

\medskip

Following the analysis of App.~\ref{app:stabilizer_reps}, $\mathcal{\cH_{\rm kin}}$ decomposes as a direct sum of $\vert \hat{G} \vert =  2^{\vert \cV \vert +\vert \cF \vert} = 2^{n-k}$ orthogonal isotypes, each of dimension $2^k = 2$. Each such isotype is labeled by an element of $\chi_{(\hat{V}, \hat{F})} \in \hat{G}$. Physically speaking, $\hat{V}$ (resp.\ $\hat{F}$) is the subset of vertices (resp.\ faces) that support \emph{gauge defects} or \emph{non-trivially charged quasiparticles}, where: we say that a state has a \emph{defect} or \emph{non-zero charge} at some vertex $v$ (resp.\ face $f$) whenever it is an eigenstate of $X_v$ (resp.\ $Z_f$) with eigenvalue $-1$. We then have the orthogonal decomposition:
\begin{equation}\label{eq:isotypes_surface}
    \cH_{\rm kin} = \bigoplus_{(\hat{V}, \hat{F})\,, \hat{V} \subset \mathcal{V}\,, \hat{F} \subset \mathcal{F}} \cH_{\hat{V},\hat{F}}\,, 
\end{equation}
where, for any $(\hat V, \hat F)$,\footnote{Here, $\mathds{1}_A$ denotes the indicator function of a set $A$.} 
\begin{equation}
\cH_{\hat V, \hat F} = \big\{ \vert \Psi \rangle \in \cH_{\rm kin} \, \vert \, \forall v \in \cV\,, X_v \vert \Psi \rangle = (-1)^{\mathds{1}_{\hat V}(v)} \vert \Psi \rangle   \; \mathrm{and}  \; \forall f \in \cF\,, Z_f \vert \Psi \rangle = (-1)^{\mathds{1}_{\hat F}(f) } \vert \Psi \rangle \big\} \simeq \mathbb{C}^2\,.
\end{equation}
In particular, the code subspace is the \emph{no-defect (or vacuum) sector}  
\begin{equation}
    \cH_{\rm pn} = \cH_{\emptyset , \emptyset}\,.
\end{equation}

Following Kitaev \cite{Kitaev:1997wr}, we can represent any error $E\in \mathcal{P}_{\vert \cE \vert}$ as a product of string operators, where string operators are defined as follows. For any open path $t$ in the direct lattice, we define
\begin{equation}
S^Z (t) = \prod_{e \in t} Z_e\,.
\end{equation}
Such an operator acts trivially on the vacuum $\cH_{\emptyset , \emptyset }$ whenever $t$ forms a loop, or when its two endpoints $v_1$ and $v_2$ lie on the same rough boundary; in that case $S^Z(t)$ can be written as a product of plaquette operators $Z_f$. When $v_1$ and $v_2$ lie on distinct rough boundaries (as illustrated in Fig.~\ref{fig:surface_code_disk}), $S^Z(t)$ still commutes with all the constraints so it leaves $\cH_{\rm  pn}$ invariant, but acts non-trivially on it. Such an operator therefore implements a non-trivial logical operation, which we can declare to be the logical $Z$ operator of the single logical qubit stored in $\cH_{\rm pn}$. In all the other cases, $S^Z (t)$ isometrically maps $\cH_{\emptyset , \emptyset}$ to a non-trivially charged sector $\cH_{V, \emptyset}$, where $V$ is the subset of vertices in $\{ v_1 , v_2\}$ that do not lie on a rough boundary. See Fig.~\ref{fig:dressing} for an example of a path $\gamma_v$ generating a single vertex defect $v$ (or ``electric'' charge). Finally, the restriction $\restr{S^Z (t)}{\cH_{\rm pn}}$ only depends on the (relative) homotopy class $[t]$ of the path $t$: indeed, the face constraints ensure that we can deform the path across a face at will.\footnote{In particular, note that an endpoint vertex $v$ that lies on a rough boundary can be translated at will along that same boundary (thanks to the incomplete face constraints). By contrast, endpoint vertices that do not lie on a rough boundary --- which support defects --- are kept fixed by such deformations.} A second set of string operators is introduced to generate face defects (or ``magnetic'' charges in the usual sense of electromagnetic duality, not to be confused with the error duality from Sec.~\ref{sec_errorduality}). For any path $t'$ in the \emph{dual} graph, we define
\begin{equation}
    S^X (t') = \prod_{e \in t'} X_e\,.
\end{equation}
Such an operator enjoys the same properties as $S^Z(t)$, up to a duality transformation that exchanges the roles of vertices and faces, as well as the roles of rough boundaries and smooth boundaries. In particular, if $t'$ connects the two smooth boundaries as illustrated in Fig.~\ref{fig:surface_code_disk}, $S^X (t')$ implements a non-trivial logical operation on $\cH_{\rm pn}$. Since it also anticommutes with $S^Z (t)$, we can take it to represent the logical $X$ operator of the qubit stored in $\cH_{\rm pn}$. Any endpoint of $t'$ which does not lie on a smooth boundary will create a face defect, as illustrated in Fig.~\ref{fig:dressing} (with a single face defect created by a dual string operator supported on $\gamma_f'$). Finally, the restriction $\restr{S^X (t')}{\cH_{\rm pn}}$ only depends on the homotopy class $[t']$ of $t'$, where allowed deformations are determined by the vertex constraints. In particular, an endpoint of $t'$ that lies on a smooth boundary is allowed to slide along this boundary, while endpoints that create face defects are kept fixed.     

Generalizing, any Pauli operator in $\mathcal{P}_{\vert \cE \vert}$ may be written (in general, non-uniquely) in the form 
\begin{equation}\label{eq:errors_string_op}
    E_{T , T'} := \eta \left( \prod_{i=1}^{n_1} S^Z(t_i) \right) \left( \prod_{i=1}^{n_2} S^X(t_i') \right)\,,
\end{equation}
where $\eta \in \{ \pm 1 , \pm i\}$, $T=\{ t_1, \ldots , t_{n_1} \}$ is a set of direct paths, and $T'=\{ t_1', \ldots , t_{n_2}' \}$ is a set of dual paths. For any defect sector $(\hat{V} , \hat{F})\neq (\emptyset , \emptyset) $, we can find a charge-creating error $E_{T,T'} \in C_{\chi_{(\hat{V}, \hat{F})}}$ (see Appendix \ref{app:stabilizer_reps}) which isometrically maps the code subspace $\cH_{\emptyset , \emptyset}$ to the defect subspace $\cH_{\hat{V} , \hat{F}}$. This allows one to create an arbitrary number of vertex and face defects over the vacuum.\footnote{String operators also allow to ``move those excitations around'', and thereby define a notion of spin-statistics. Famously, one finds that electric and magnetic excitations obey (Abelian) anyonic statistics relative to one another (even though they behave like bosons among themselves).} Moreover, the restriction $\restr{E_{T,T'}}{\cH_{\emptyset , \emptyset}}$ only depends on the respective homotopy classes $[T]$ and $[T']$ of the systems of paths $T$ and $T'$ (where, as in the case of a single path, allowed homotopy deformations are induced by the action of plaquette or vertex operators). 

Coming back to the qualitative discussion of Sec.~\ref{ssec_interlude}, the $\mathbb{Z}_2^{|\cV|+|\cF|}$ symmetry of this model is interpreted as a gauge symmetry at the level of the code subspace or vacuum sector $\cH_{\rm pn}= \cH_{\emptyset, \emptyset}$, while errors are interpreted as non-trivially charged excitations over this vacuum. The group $\mathbb{Z}_2^{|\cV|+|\cF|}$ acting non-trivially on such excitations, the gauge symmetry of the vacuum sector $\cH_{\emptyset, \emptyset}$ can be thought of as \emph{emergent}. This interpretation is particularly relevant when $\cH_{\emptyset , \emptyset}$ is realized as the lowest energy eigenspace of some Hamiltonian, in which case the restriction to $\cH_{\emptyset , \emptyset}$ can effectively --- but \emph{imperfectly} --- be implemented by cooling a thermal state to near-zero temperature. This is the scenario that was originally envisioned in \cite{Kitaev:1997wr}. 

\subsubsection{Homological description of the code subspace and two notions of duality}\label{sec:surface_code_homological_descr}

Even though we did not rely on that interpretation so far, let us recall the standard interpretation of this code as a $\mathbb{Z}_2$-homological code. This will allow us to further stress the differences between the two types of dualities that are at play in this model: standard electromagnetic duality on the one hand, which relates the homology of the surface to its cohomology; and the notion of error duality we spelled out in Sec.~\ref{sec_errorduality}, which relates the group $G$ to its dual. For the reader unfamiliar with (co)chain complexes, we note that the material of this subsection  will not be relied upon in our subsequent explicit illustration of error-generated QRFs and dual errors.

To start with, the boundary and coboundary homomorphisms described previously can be extended into full chain and cochain complexes
\begin{equation}
    \partial : \; \cP(\cV) \overset{\partial_1}{\to} \cP(\cE) \overset{\partial_2}{\to} \cP(\cF)\,, \qquad \mathrm{and} \quad  \qquad \delta : \; \cP(\cF) \overset{\delta_1}{\to} \cP(\cE) \overset{\delta_2}{\to} \cP(\cV)\,,
\end{equation}
such that\footnote{As usual in this context, we allow ourselves to omit the index $i$ in $\partial_i$ (resp. $\delta_i$) whenever it is clear from context which map we are referring to. With this convention, the maps $\partial$ and $\delta$ already introduced at the beginning of the present section are in fact $\partial_1$ and $\delta_1$.}
\begin{equation}\label{eq:homology_condition}
    \partial_2 \circ \partial_1 = 0\;(\mathrm{or}\;\partial^2 = 0) \quad \mathrm{and} \quad \delta_2 \circ \delta_1 = 0\;(\mathrm{or}\;\delta^2 = 0)\,.
\end{equation}
This is achieved as follows. For any edge $e\in \cE$, we define: $\partial_2 \{e\}$ as the set of (at most two) faces $e$ belongs to; and $\delta_2 \{e\}$ as the set of (at most two) vertices $e$ is incident to. The maps $\partial_2: \cP(\cE) \to \cP(\cF)$ and $\delta_2: \cP(\cE) \to \cP(\cV)$ are then the (unique) group homomorphisms respecting those conditions. 

Next, we can introduce the group of \emph{cycles} as the subgroup $\cZ:= \ker \partial_2 \leq \cP(\cE)$ and the group of \emph{boundaries} as $\cB:= {\rm Im}(\partial_1)$. Owing to \eqref{eq:homology_condition}, $\cB$ is a (normal) subgroup of $\cZ$ so we can define the \emph{homology group} by the quotient
\begin{equation}
    \faktor{\cZ}{\cB} \simeq \mathbb{Z}_2\,.
\end{equation}
Likewise, we can introduce the groups of \emph{cocycles} $\cZ^* := \ker{\delta_2}$ and \emph{coboundaries} $\cB^* := {\rm Im}(\delta_1)$, and construct the \emph{cohomology group}
\begin{equation}
    \faktor{\cZ^*}{\cB^*} \simeq \mathbb{Z}_2\,.
\end{equation}

To interpret the code in homological terms, it is convenient to introduce the following notation: for any subset of edges $E \subset \cE$, we denote by $\vert E\rangle \in \cH_{\rm kin}$ the computational basis state for which any qubit with support in $E$ is in state $\vert 1 \rangle$ while any other qubit is in state $\vert 0 \rangle$; in other words:
\begin{equation}
    \vert E \rangle := \left( \bigotimes_{e \notin E} \vert 0 \rangle \right) \otimes \left( \bigotimes_{e \in E} \vert 1 \rangle \right)\,.
\end{equation}
To construct our homological basis of $\cH_{\rm pn}$, we start out from the state $\vert \emptyset \rangle = \vert 0 \rangle^{\otimes \vert \cE\vert}$, which is left invariant by any plaquette operator $Z_f$, and average it over the subgroup of the stabilizer group generated by vertex operators. This defines a \emph{logical zero} state
\begin{equation}
    \vert \bar 0\rangle := \frac{1}{\sqrt{2^{\vert \cV\vert }}} \sum_{V \subset \cV} \left( \prod_{v \in V} X_v\right) \vert \emptyset \rangle = \frac{1}{\sqrt{2^{\vert \cV\vert }}} \sum_{V \subset \cV} \vert \partial V \rangle = \frac{1}{\sqrt{2^{\vert \cV\vert }}} \sum_{E \in \cB} \vert E \rangle\,.
\end{equation}
Acting with the logical $X$ operator represented by e.g. $S^X (t')$, yields the \emph{logical one} state: 
\begin{equation}
    \vert \bar 1\rangle := S^X(t') \vert \bar 0 \rangle = \frac{1}{\sqrt{2^{\vert \cV\vert }}} \sum_{E \in \cB} \vert E \Delta t' \rangle = \frac{1}{\sqrt{2^{\vert \cV\vert }}} \sum_{E \in \cZ , \, E \notin \cB} \vert E \rangle \,.
\end{equation}
The last equality holds since, for any $E \in \cB$, $E\Delta t'$ is in the same homology class as $t'$, and: $\partial_2 t' = \emptyset$ (hence $E\Delta t'$ is a cycle); $t'$, as a path that connects two distinct boundaries, is not a boundary (hence $E \Delta t'$ is not a boundary). We can therefore interpret $\vert \bar 0 \rangle$ (resp. $\vert \bar 1\rangle$) as an average over the trivial (resp. nontrivial) homology class of $\cZ$. In particular, the two-dimensional nature of $\cH_{\rm pn}$ can be traced back to the fact that the homology group $\faktor{\cZ}{\cB}$ of the lattice has order $2$.

In an analogous way, one could construct a basis of $\cH_{\rm pn}$ that is well adapted to the cohomology of the surface. This is most simply obtained by exchanging, in the previous paragraph: the boundary operator $\partial$ for the coboundary operator $\delta$, the plaquette constraints $\{ Z_f\}_{f \in \cF}$ for the vertex constraints $\{ X_v\}_{v \in \cV}$, the qubit basis $(\vert 0 \rangle , \vert 1 \rangle)$ (which diagonalizes $Z$ operators) for the basis $(\vert +\rangle , \vert - \rangle)$ (which diagonalizes $X$ operators), a nontrivial dual string operator such as $S^X(t')$ for a non-trivial string operator such as $S^Z(t)$ (see Figure \ref{fig:surface_code_disk}). This operation is most naturally interpreted as the action of the standard electromagnetic duality that exchanges the roles of the lattice and dual lattice in this system. From an algebraic point of view, it can be understood as exchanging the roles of $\mathbb{Z}_2$-homology and $\mathbb{Z}_2$-cohomology: 
\begin{equation}
    \faktor{\cZ}{\cB} \quad \underset{\mathrm{standard}\;\mathrm{EM}\;\mathrm{duality}}{\longleftrightarrow} \quad  \faktor{\cZ^*}{\cB^*} \,.
\end{equation}
Furthermore, $\faktor{\cZ^*}{\cB^*}$ can naturally be understood as the Pontryagin dual of $\faktor{\cZ}{\cB}$ with the following canonical pairing: for any $[t'] \in \faktor{\cZ}{\cB}$ (where $t'$ denotes a path in the dual lattice and $[t']$ its equivalent class) and $[t] \in \faktor{\cZ^*}{\cB^*}$ (where $t$ denotes a path in the direct lattice and $[t]$ its equivalent class), one sets $\langle [t],[t']\rangle:= (-1)^{|t\cap t'|}$.\footnote{The quantity $|t\cap t'| \,{\rm mod}\, 2$ is independent of the choice of representatives $t,t'$, and defines the \emph{intersection number} of the classes $[t]$, $[t']$.} Electromagnetic duality can in this sense be understood as Pontryagin duality for the pair of dual groups $\left( \faktor{\cZ}{\cB}, \faktor{\cZ^*}{\cB^*}\right)$

By contrast, the error duality spelled out in Sec.~\ref{sec_errorduality} relates the full gauge group $G= \cP(\cV) \times \cP(\cF)$ and its dual $\hat{G}$. We can use the maps $\partial_1$ and $\delta_1$, which are both injective in the present model, to identify $G$ with $\cB \times \cB^*$. Hence, in the language of (co)homology, we have:
\begin{equation}
    G=\cB \times \cB^* \quad \underset{\mathrm{error}\;\mathrm{duality}}{\longleftrightarrow} \quad  \hat{G}=\widehat{\cB \times \cB^*}\,. 
\end{equation}
This makes it clear that, even though the standard electromagnetic duality of this model and the general notion of error duality introduced in Sec.~\ref{sec_errorduality} can both be understood as instances of Pontryagin duality, they are very much distinct.

\subsubsection{Refactorization associated to a complete set of frame fields}

Let $(\hat{V}, \hat{F}) \in \cV \times \cF \setminus \{(\emptyset , \emptyset )\} $. A frame for the defect sector $\cH_{\hat{V} , \hat{F}}$ is an isometry $R_{\hat{V} , \hat{F}}: \cH_{\emptyset , \emptyset}\to \cH_{\hat{V} , \hat{F}}$ induced by an error $E_{T_{\hat V} , T_{\hat F}'}$ (of the form \eqref{eq:errors_string_op} with, for definiteness, $\eta=1$), where $T_{\hat V}$ (resp.\ $T_{\hat F}'$) is a system of paths with endpoints $\hat V$ (resp.\ $\hat F$). Since $R_{\hat V , \hat F}$ only depends on the homotopy classes $[T_{\hat V}]$ and $[T_{\hat F}']$, we reach the interesting conclusion that \emph{a frame for $\cH_{\hat V, \hat F}$ (in the sense of Definition \ref{def:complete_frame}) is determined by a choice of homotopy class for a system of paths and dual paths connecting the defects $(\hat V , \hat F)$} (among themselves or to one of the boundaries). Our observations together with Proposition~\ref{propo:frame-from-errors} tell us that a complete frame $R= \{ R_{\hat V , \hat F} = \restr{E_{T_{\hat V} , T_{\hat F}'}}{\cH_{\rm pn}} \,, (\hat V , \hat F)\neq (\emptyset , \emptyset) \}$ will correct any non-trivial error of the form \eqref{eq:errors_string_op} that is homotopically equivalent to one of the individual frames, i.e.\ such that $([T], [T'])= ([T_{\hat V}], [T_{\hat F}'])$ for some $(\hat V , \hat F)\neq (\emptyset , \emptyset)$.

Given a complete set of frame fields $R= \{ R_{ \hat V , \hat F}  \,, (\hat V , \hat F)\neq (\emptyset , \emptyset) \}$, we can also define the tensor product $\otimes_R$ provided by Theorem~\ref{thm:nonlocal_fact} (or equivalently Theorem~\ref{thm_errorQRF}), which makes the identification: 
\begin{equation}
    \cH_{\rm kin} = \cH_{\emptyset , \emptyset } \otimes_R \cH_{\rm gauge}\,,
\end{equation}
where $\cH_{\rm gauge}$ is spanned by a canonical orthonormal basis labeled by characters 
\begin{equation}\label{eq:surface_code_square_factorization}
    \cH_{\rm gauge} = \Span \big\{ \vert \chi_{(\hat V , \hat F)} \rangle \, \vert \, \hat V \subset \mathcal{V} \,, \hat F \subset \mathcal{F} \big\}\,.
\end{equation}
In terms of the homological encoding introduced in the previous subsection, we have, for any $(\hat{V}, \hat{F})\in \hat{G}$:
\begin{align}
    \vert \bar 0\rangle \otimes_R \vert \chi_{(\hat{V}, \hat{F})}\rangle &:= \left( \prod_{t \in T_{\hat{V}}} S^Z(t)\right) \left( \prod_{t' \in T'_{\hat{F}}} S^X(t') \right) \vert \bar 0 \rangle = \frac{1}{\sqrt{2^{|\cV|}}} \sum_{E\in \cB} \left( \prod_{t \in T_{\hat{V}}} S^Z(t)\right) \vert E \Delta \left( \underset{t' \in T'_{\hat{F}}}{\Delta}  t' \right)\rangle \\
    &= \frac{1}{\sqrt{2^{|\cV|}}} \sum_{E\in \cB} \prod_{t \in T_{\hat{V}}}(-1)^{\vert t \cap E \vert} \prod_{t' \in T'_{\hat{F}}} (-1)^{|t \cap t'|} \vert E \Delta \left( \underset{t' \in T_{\hat{F}}}{\Delta}  t' \right)\rangle
\end{align}
and  
\begin{align}
    \vert \bar 1\rangle \otimes_R \vert \chi_{(\hat{V}, \hat{F})}\rangle &:= \left( \prod_{t \in T_{\hat{V}}} S^Z(t)\right) \left( \prod_{t' \in T'_{\hat{F}}} S^X(t') \right) \vert \bar 1 \rangle  \\
    &= \frac{1}{\sqrt{2^{|\cV|}}} \sum_{E\in \cZ ,\, E \notin \cB} \prod_{t \in T_{\hat{V}}}(-1)^{\vert t \cap E \vert} \prod_{t' \in T'_{\hat{F}}} (-1)^{|t \cap t'|} \vert E \Delta \left( \underset{t' \in T_{\hat{F}}}{\Delta}  t' \right)\rangle\,.
\end{align}
The coefficients in those expressions are function of the intersection numbers of paths and dual open paths anchored at $\hat{V}$ and $\hat{F}$, which can be understood in terms of the homology of a punctured surface (i.e. the original surface with any vertex from $\hat{V}$ and any face from $\hat{F}$ removed).

Relative to the decomposition \eqref{eq:surface_code_square_factorization}, the action of the gauge group $\mathcal{G}$ can be encoded in a new representation 
\begin{equation}
    \mathcal{G}_R = \langle X_v^R , Z_f^R \,\vert\, v \in \mathcal{V}\,,  f \in \mathcal{F} \rangle 
\end{equation}
of $\mathbb{Z}_2^{\times (\vert \cV \vert+\vert \cF \vert)}$ on $\mathcal{H}_{\rm gauge}$, defined as 
\begin{equation}
    X_v^R \vert \chi_{(\hat V , \hat F)} \rangle = (-1)^{\mathds{1}_{\hat V}(v)} \vert \chi_{(\hat V , \hat F)} \rangle\,, \qquad Z_f^R \vert \chi_{(\hat V , \hat F)} \rangle = (-1)^{\mathds{1}_{\hat F}(f)} \vert \chi_{(\hat V , \hat F)} \rangle\,.
\end{equation}
In particular, it is clear that those new plaquette and vertex operators commute and square to the identity, as they should. For any $v \in \mathcal{V}$ and $f \in \mathcal{F}$, we then have:
\begin{equation}
    X_v = {\rm id}_{\cH_{\emptyset , \emptyset}} \otimes_R X_v^R\,, \qquad Z_f = {\rm id}_{\cH_{\emptyset , \emptyset}} \otimes_R Z_f^R\,. 
\end{equation} 
Any operator in $\mathcal{G}_R$ can be written uniquely as a product of the form 
\begin{equation}\label{eq:binary_encoding}
U_R^{(V,F)} = \prod_{v \in V} X_v^R \prod_{f \in F} Z_{f}^R \,, 
\end{equation}
where $(V,F) \in G$. The map $U_R: G \to {\rm Aut}(\cH_{\rm kin}), g \mapsto U_R^g$ then defines a unitary representation isomorphic to the regular representation. Its associated group basis (see \eqref{def:group_basis}) can be taken to be
\begin{equation}
    \forall (V,F) \in G\,, \qquad \vert (V,F) \rangle :=  
    \frac{1}{\sqrt{2^{\vert \cV \vert +\vert \cF \vert}}} \sum_{V \subset \mathcal{V}, F \subset \mathcal{F}} (-1)^{\vert \hat{V}\cap V \vert + \vert \hat{F}\cap F \vert }  \, \vert   \chi_{(\hat{V} , \hat{F})} \rangle  \,.
\end{equation}
With this definition at hand, we then have $U_R^g \vert h \rangle = \vert gh \rangle$ and $\langle g \vert h \rangle = \delta_{g, h}$ for any $g,h \in G$, hence $U_R$ is indeed isomorphic to the regular representation of the gauge group $G$ and $R$ furnishes an ideal QRF.

Finally, to conclude our illustration of Theorem~\ref{thm:nonlocal_fact}, we observe that an error $E_{T, T'} \in \mathcal{P}_{\vert \cE \vert}$ is going to be correctable by the frame $R$ if, for any code state $\vert \Psi_{\rm code} \rangle \in \cH_{\rm pn}$ (cf.~Eq.~\eqref{eq:def_tensor1}):
\begin{equation}
    E_{T,T'} \vert \Psi_{\rm code} \rangle \equiv E_{T,T'} \left( \vert \Psi_{\rm code} \rangle \otimes_R \vert \emptyset , \emptyset \rangle \right) = \eta  \vert \Psi_{\rm code} \rangle \otimes_R \vert \chi_{(\hat V , \hat F)} \rangle\,,
\end{equation}
for some phase $\eta$ and defect labels $(\hat V , \hat F )$. In other words, correctable errors are those that only affect the state of the defect Hilbert space $\mathcal{H}_{\rm gauge}$ defined by the frame. The frame degrees of freedom are in this sense redundant. However, as already observed repeatedly in previous sections, a different choice of frame will in general lead to an inequivalent splitting between logical and redundant degrees of freedom. 

From a practical QECC point of view, the code subspace $\mathcal{H}_{\emptyset , \emptyset}$ may be realized as the degenerate ground state of a many-body 2d material. In such scenario, one is interested in correcting the most probable errors. For example, in an independent and identically distributed (i.i.d.) error model with low error probability, it is more likely to have small collections of defects rather than large swathes. So, constructing a complete reference frame may be an overkill for practical applications, as only a small number of defect types are expected to occur with significantly high probability. For instance, if one only cares about correcting errors that bring a code state into the single-defect Hilbert space (which is generated by defects of the form $(V, F)= (\{ v\}, \emptyset)$ or $(V, F) = (\emptyset , \{ f\})$), the notion of frame can be seen to coincide with a notion of \emph{dressing field} which should be familiar from standard gauge theory. Indeed, in these sectors of the Hilbert space, choosing a frame amounts to choosing a system of paths (resp.\ dual paths) connecting any vertex (resp.\ face) to one of the rough boundaries (resp.\ smooth boundaries). This is a discrete realization of a  frame field for the pertinent gauge group defined by Wilson lines anchored at a boundary where gauge transformations act trivially \cite{Carrozza:2021gju,Araujo-Regado:2024dpr, Araujo-Regado:2025ejs}. This choice can be specified by two maps: $\cV \ni v \mapsto \gamma_v$ and $\cF \ni f \mapsto \gamma_f'$. We then define $R_{\{ v\}, \emptyset} \equiv S^Z(\gamma_v)$ and $R_{ \emptyset , \{ f\}} \equiv S^X(\gamma_f')$. The collection of such operators provides us with a kind of dressing field for a single-particle defect which may be created anywhere on the lattice or its dual. For instance, suppose that a single defect is created at vertex $v$ by some error $E= S^Z(t_v)$, where $t_v$ is a simple path connecting $v$ to one of the rough boundaries. $E$ is not gauge-invariant but $R_{\{v\}, \emptyset} E = S^Z(t_v \cup \gamma_v)$ is, since $t_v \cup \gamma_v$ is a path hooked at both of its ends to one of the rough boundaries. If $t_v \cup \gamma_v$ is connected to the same boundary at both ends, then $\restr{S^Z(t_v \cup \gamma_v)}{\cH_{\emptyset , \emptyset}} = {\rm id}_{\cH_{\emptyset, \emptyset}}$ which means that $E$ is appropriately corrected by $R_{\{v\}, \emptyset}$. By contrast, if $t_v \cup \gamma_v$ connects the top and the bottom boundaries, $R_{\{v\}, \emptyset} E$ acts like the logical $Z$ operator on $\cH_{\emptyset , \emptyset}$, which means that $E$ is not appropriately corrected by $R_{\{v\}, \emptyset}$. 

This construction is illustrated in Fig.~\ref{fig:dressing}. It only partially fixes the tensor factorization $\cH_{\rm kin}= \cH_{\rm {pn}}\otimes_R \cH_{\rm gauge}$, but again, that might be all one needs in a practical situation.
Nonetheless, having a complete set of frame fields in hand opens the door to other partial fixings of the logical-gauge tensor factorization that could be adapted to a given noise model with biased or correlated noise.

\begin{figure}
    \centering
    \includegraphics[scale=0.7]{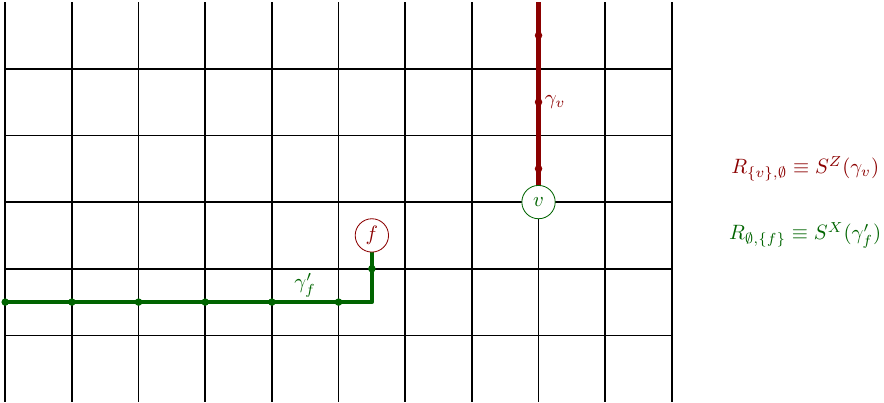}
    \caption{In the single-defect sector, the choice of reference frame is specified by a choice of path $\gamma_v$ for every vertex $v$, and a choice of dual path $\gamma_f'$ for every face $f$. An error creating a defect at location $v$ (resp.\ $f$) is dressed by $R_{\{v\}, \emptyset} = S^Z (\gamma_v)$ (resp.\ $R_{\emptyset , \{f\}} = S^X (\gamma_f')$) to form a gauge-invariant (or, equivalently, a logical) operator.}
    \label{fig:dressing}
\end{figure}

\subsubsection{Dual representation associated to a complete set of frame fields, and dual errors}\label{subsec:duality_square}

Next, let us illustrate the error duality of Sec.~\ref{sec_errorduality} in surface codes with boundary.
Given the faithful representation $U$ introduced in \eqref{eq:U_surface_code}, we can look for a dual representation $\hat{U}$, as defined in Definition \ref{def:dual_reps_faithful}. The dual representation $\hat{U}$ is far from unique, but let us provide two explicit construction methods. We first construct a dual representation adapted to the frame $R = \{ R_{\hat{V}, \hat{F}} \,, \, (\hat{V}, \hat{F}) \in \hat{G}\}$ considered in the previous section (with the additional convention $R_{\emptyset , \emptyset} := I$), and discuss how dual errors act in this case. The resulting representation will not take value in the Pauli group; for completeness, we will construct another dual representation that does take value in the Pauli group in the next section. 

Following \eqref{eq:dual_rep_associated_to_frame}, we can define the dual representation associated to $R$ as follows: for any $(\hat{V}, \hat{F}) \in \cP(\cV)\times \cP (\cF)$, we set
\begin{equation}\label{eq:surface_code_nonlocal_dual_rep}
\hat{U}_R^{\chi_{(\hat{V}, \hat{F})}} := \sum_{\hat{V}'\subset \cV ,\, \hat{F}'\subset \cF} E_{T_{\hat{V} \Delta \hat{V}'}, T'_{\hat{F} \Delta \hat{F}'}} E_{T_{ \hat{V}'}, T'_{\hat{F}'}} P_{\chi_{(\hat{V}', \hat{F}')}}\,.
\end{equation}
For each $(\hat{V}', \hat{F}') \in \cP(\cV) \times \cP(\cF)$, the operator $E_{T_{\hat{V} \Delta \hat{V}'}, T'_{\hat{F} \Delta \hat{F}'}} E_{T_{ \hat{V}'}, T'_{\hat{F}'}}$ appearing in this expression is a product of string and dual string operators anchored at $\hat{V}$ and $\hat{F}$, supplemented by a (possibly empty) product of loop and dual loop operators. It is in particular an element of the Pauli group. However, this operator is in general dependent on the defect sector $(\hat{V}', \hat{F}')$ one is considering, so that $\hat{U}_R^{\chi_{(\hat{V}, \hat{F})}}$ is in general \emph{not} itself an element of the Pauli group.

Continuing with our illustration of error duality, for every group element $(V, F) \in \cP(V) \times \cP(F)$, the isotype projector defined in \eqref{eq:dual_proj_exact} takes the form:
\begin{equation}
    \hat{P}^R_{(V,F)} := \frac{1}{2^{|\cV|+|\cF|}} \sum_{\hat{V} \subset \cV,\, \hat{F} \subset \cF} \chi_{(\hat{V},\hat{F})}((V,F)) \, \hat{U}_R^{\chi_{(\hat{V}, \hat{F})}} = \frac{1}{2^{|\cV|+|\cF|}} \sum_{\hat{V} \subset \cV,\, \hat{F} \subset \cF} (-1)^{\vert \hat{V}\cap V |+ |\hat{F}\cap F|} \, \hat{U}_R^{\chi_{(\hat{V}, \hat{F})}}\,.
\end{equation}
Following Proposition \ref{prop:gauge-fixing_errors_frame}, the dual representation $\hat{U}^R$ is associated to a maximal set of gauge-fixing errors $\{ \hat{E}_{(V,F)}\}_{(V,F)\in G}$ whose restrictions to the code subspace obey
\begin{align}\label{eq:surface_code_nonlocal_dual_errors}
    \forall V \subset \cV\,, \forall F \subset \cF\,,\qquad  \hat{E}_{(V,F)} \Pi_{\rm pn} =& \sqrt{2^{|\cV|+|\cF|}} \hat{P}^R_{(V,F)} \Pi_{\rm pn} \\
    &= \frac{1}{\sqrt{2^{|\cV|+|\cF|}}} \sum_{\hat{V}\subset \cV,\, \hat{F}\subset \cF} (-1)^{|\hat{V}\cap V|+|\hat{F}\cap F|} E_{T_{\hat{V}}, T'_{\hat{F}}} \Pi_{\rm pn}\,. \nonumber
\end{align}
Such errors are entirely determined by the homotopy classes of the systems of paths and dual paths $( T_{\hat{V}}, T'_{\hat{F}})_{(\hat{V}, \hat{F})\in \hat{G}}$. The non-trivial sum entering their definition (which can be interpreted as a kind of Fourier transform) implies that they cannot in general be represented by operators from the Pauli group.

\subsubsection{Example of dual representation taking value in the Pauli group}\label{sec:surface_code_duality_Pauli}

For the $3$-qubit code, it was show in example \ref{3qubit:errorduality1} that it is possible to construct dual representations that take value in the Pauli group, and are in particular ``local'' relative to the kinematical tensor product structure. The resulting gauge-fixing errors are also local in the sense that they effectively only act on the first two qubits (see example \ref{3qubit:errorduality2}). By contrast, the dual representation \eqref{eq:surface_code_nonlocal_dual_rep} does not take value in the Pauli group and the dual errors \eqref{eq:surface_code_nonlocal_dual_errors} act in a nonlocal fashion on the lattice. It is thus natural to wonder whether dual representations that are local in this sense also exist for the surface code. The answer is yes, and the purpose of this subsection is to provide one such example.

To this effect, let us first introduce a \emph{spanning forest}\footnote{A \emph{forest} is a subgraph without loop, and its connected components are \emph{trees}. We will say that it is \emph{spanning} if it visits every $v \in \cV$.} $T \subset \mathcal{E}$ in the direct lattice and a \emph{dual spanning forest $T' \subset \mathcal{E}$}\footnote{$T'$ is therefore a subgraph of the dual network without loop, and its connected components are dual trees. We will say that it is \emph{spanning} if it visits every $f \in \cF$. Remember also that the set of dual edges can be canonically identified with the set of edges $\cE$, which allows us to view the dual subgraph $T'$ as a subset of $\cE$.} in the dual lattice such that: a) $T \cap T' = \emptyset$; b) for any vertex $v \in \cV$, there is a unique simple path $\gamma_v$ in $T$ which connects it to one of the rough boundaries; c) for any face $f \in \cF$, there is a unique simple dual path $\gamma_f'$ in $T'$ which connects it to one of the smooth boundaries. In particular, each tree in $T$ (resp.\ dual tree in $T'$) has a unique edge lying on the rough (resp.\ smooth) boundary, which we call its \emph{root-edge}. It is possible to show that the condition imposed on $(T,T')$ can always be satisfied; we will postpone a complete justification of this fact to footnote \ref{foot:forests} in Sec.~\ref{subsec:example_dual_rep_toric_code}, where we will have to deal with a similar but slightly simpler combinatorial problem, but see Figure~\ref{fig:surface_code-tree} for an example of such a pair $(T,T')$. Given any such pair $(T,T')$, we can then define:
\begin{equation}\label{eq:dual_rep_square}
    \forall \hat{V} \subset \cV\,, \forall \hat{F} \subset \cF \,, \qquad \hat{U}^{\chi_{(\hat V, \hat F)}} :=  \prod_{v \in \hat V} S^Z (\gamma_v)  \prod_{f \in \hat F} S^X (\gamma_f') = \left( \prod_{v \in \hat V} \prod_{e \in \gamma_v} Z_e \right) \left( \prod_{f \in \hat F} \prod_{e \in \gamma_f'} X_e \right) \,,
\end{equation}
where $\gamma_v$ and $\gamma'_f$ are the unique simple paths defined above. For any $(\hat V , \hat{F} ) \subset \cV \times \cF \setminus \{ (\emptyset , \emptyset) \}$, $\hat{U}^{\chi_{(\hat V, \hat F)}}$ is a product of string operators which create defects (resp.\ dual defects) over the vacuum at any $v \in \hat V$ (resp.\ $f \in \hat F$). Furthermore, we have
\begin{align}\label{eq:proof_duality_square}
    U^{(V,F)} \hat{U}^{\chi_{(\hat V, \hat F)}} &= \prod_{v \in V} X_v \prod_{f \in F} Z_f \prod_{\hat{v}\in \hat{V}} S^Z (\gamma_{\hat{v}}) \prod_{\hat{f}\in \hat{F}} S^X (\gamma_{\hat{f}}') \\ 
    &= \left( \prod_{v \in V} X_v  \prod_{\hat{v}\in \hat{V}} S^Z (\gamma_{\hat{v}}) \right) \left( \prod_{f \in F} Z_f \prod_{\hat{f}\in \hat{F}} S^X (\gamma_{\hat{f}}') \right)  \nonumber \\ 
    &= \left( (-1)^{\vert \hat{V} \cap V \vert}  \prod_{\hat{v}\in \hat{V}} S^Z (\gamma_{\hat{v}}) \prod_{v \in V} X_v \right) \left( (-1)^{\vert \hat{F} \cap F \vert }  \prod_{\hat{f}\in \hat{F}} S^X (\gamma_{\hat{f}}') \prod_{f \in F} Z_f \right)  \nonumber \\
    &= (-1)^{\vert \hat{V} \cap V \vert + \vert \hat{V} \cap V \vert} \prod_{\hat{v}\in \hat{V}} S^Z (\gamma_{\hat{v}})  \prod_{\hat{f}\in \hat{F}} S^X (\gamma_{\hat{f}}') \prod_{v \in V} X_v \prod_{f \in F} Z_f  \nonumber \\
    &= \chi_{(\hat V, \hat F)}((V,F)) \,  \hat{U}^{\chi_{(\hat V, \hat F)}} U^{(V,F)} \nonumber 
\end{align}
for any $(V, F) \in G$ and $\chi_{(\hat V , \hat F)} \in \hat{G}$. Therefore, the defining relation of duality is verified (see Definition \ref{def:dual_reps_faithful}). Moreover, the properties of $(T,T')$ ensure that $\hat{U}$ does define a (faithful) representation. In particular, the fact that $T \cap T' = \emptyset$ is crucial to guarantee that we get an ordinary representation, and not merely a projective representation. Indeed, for any $(\hat{V_1}, \hat{F_1}), (\hat{V_2}, \hat{F_2})  \in \cV \times \cF$, we find:
\begin{align}\label{eq:proof_rep_square} 
 \hat{U}^{\chi_{(\hat{V_1}, \hat{F_1})}} \hat{U}^{\chi_{(\hat{V_2}, \hat{F_2})}} &= \prod_{v_1 \in \hat{V_1}} S^Z (\gamma_{v_1})  \prod_{f_1 \in \hat{F_1}} S^X (\gamma_{f_1}') \prod_{v_2 \in \hat{V_2}} S^Z (\gamma_{v_2})  \prod_{f_2 \in \hat{F_2}} S^X (\gamma_{f_2}')  \\ 
 & = \prod_{v_1 \in \hat{V_1}} S^Z (\gamma_{v_1})  \prod_{v_2 \in \hat{V_2}} S^Z (\gamma_{v_2})  \prod_{f_1 \in \hat{F_1}} S^X (\gamma_{f_1}')  \prod_{f_2 \in \hat{F_2}} S^X (\gamma_{f_2}') \nonumber  \\
 & = \prod_{v \in \hat{V_1} \Delta \hat{V_2}}  S^Z (\gamma_{v})  \prod_{f \in \hat{F_1}\Delta \hat{F_2}} S^X (\gamma_{f}')  = \hat{U}^{\chi_{(\hat{V_1}\Delta \hat{V_2}, \hat{F_1}\Delta \hat{F_2})}} = \hat{U}^{\chi_{(\hat{V_1}, \hat{F_1})}\chi_{(\hat{V_2}, \hat{F_2})}}\,,    \nonumber 
\end{align}
where $T \cap T' = \emptyset$ was used to commute direct and dual string operators in the second line, and the fact that any string operator squares to the identity was used in the third. As a result, $\hat{U}$ is a representation of the Pontryagin dual $\hat{G}$, and it is dual to $U$ in the sense of Definition \ref{def:dual_reps_faithful}. 

\begin{figure}
    \centering
    \includegraphics[scale=.7]{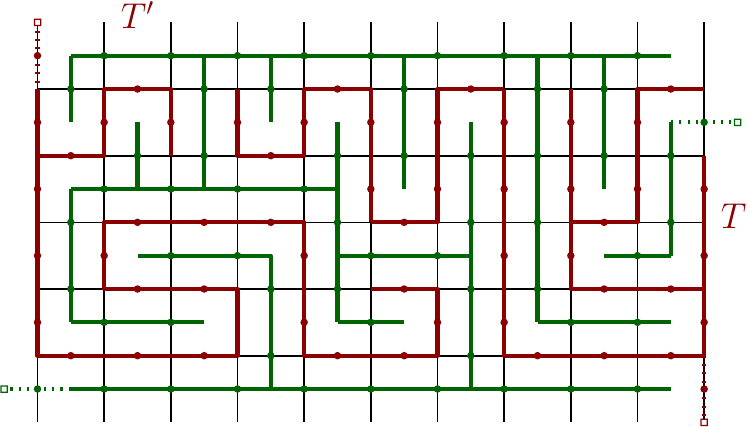}
    \caption{Example of a spanning forest $T$ and a dual spanning forest $T'$ (we have defined $T'$ as a subset of $\cE$, but it is here represented as a subset of dual edges to facilitate visualization). Here $T$ and $T'$ have two connected components each. For any $v \in \cV$ (resp.\ $f\in \cF$), there is a unique simple path $\gamma_v \in T$ (resp.\ dual path $\gamma_f' \in T'$) which connects $v$ (resp.\ $f$) to one of the rough (resp.\ smooth) boundaries. For convenience, the root-edges of $T$ and $T'$ are represented as dotted lines. Note that the edge in the top-right corner does not belong to either $T$ or $T'$.}
    \label{fig:surface_code-tree}
\end{figure}

To conclude, we note that the dual representation we just constructed has support on all the qubits of the lattice save one (i.e. the one located at the top-right corner edge of Fig.~\ref{fig:surface_code-tree}). In the spirit of what was done in example \ref{3qubit:errorduality1} for the $3$-qubit code, we can therefore interpret that particular qubit as a logical subsystem $S$, while the qubits with support on the forests $(T,T')$ constitute a local frame $R$ for that subsystem. In other words, we have obtained a local tensor factorization of the kinematical Hilbert space of the type considered in Sec.~\ref{Sec:genstabilizerQRF}. The dual errors $\{\hat{E}_{(V,F)}\}_{(V,F) \in G}$ associated to the dual representation $\hat{U}$ (which we will not write down explicitly, even though it could be done) act on the code subspace in the same way as the rescaled projectors $\{\sqrt{2^{|\cV|+|\cF|}}\hat{P}_{(V,F)} \}_{(V,F)\in G}$, which can be expressed via a Fourier transform in terms of the representation $\hat{U}$. In particular, those errors act locally on the forests $(T,T')$; each can be interpreted as projecting down on a gauge-fixed description of the code, in which the state of the frame qubits from $(T,T')$ has been fixed to a particular orientation state.

\subsection{$\mathbb{Z}_2$ surface code on a closed orientable surface of arbitrary genus}

\subsubsection{Hilbert space and charge sectors}

Consider now a surface code on an orientable surface of genus $g\geq 1$, defined as the flat sector of a $\mathbb{Z}_2$ gauge theory \cite{Kitaev:1997wr}; the case $g=1$ constitutes the well-known toric code.  That is, keeping the notations of the previous section, we now have a closed combinatorial map with $|\cF|$ faces (or plaquettes), $|\cE|$ edges and $|\cV|$ vertices with
\begin{equation}
2 - 2g = |\cF| - |\cE| + |\cV|\,.    
\end{equation}
We impose no restriction on the degree of the vertices or faces; in particular, the underlying graph need not be regular. The kinematical Hilbert space is again
\begin{equation}
\cH_{\rm kin} = \bigotimes_{e \in \mathcal{E}} \mathbb{C}^2 = (\mathbb{C}^2)^{\otimes |\cE |} \,,
\end{equation}
and we select a canonical basis in each of the factors to define a Pauli group $\mathcal{P}_{|\cE|}$. Vertex and plaquette operators are defined as in Eqs.~\eqref{eq:vertex_cst} and \eqref{eq:face_cst}.\footnote{No incomplete plaquette or vertex operator is necessary since the combinatorial map has no boundary.} The constraints all commute (because any pair of plaquette and vertex operators share an even number of edges), generating some Abelian gauge group $\mathcal{G}$. The code subspace $\mathcal{H}_{\rm pn}$ is again defined as the gauge-invariant subspace of $\cH_{\rm kin}$. In contrast to the previous code, the plaquette and vertex operators are not independent: there are two non-trivial relations (which result from the fact that any edge is the boundary of exactly two faces, and similarly for the dual lattice)
\begin{equation}\label{eq:relations}
    \prod_{v\in \mathcal{V}} X_v = I\,, \qquad \prod_{f\in \mathcal{F}} Z_f = I\,. 
\end{equation}
As a result, $\mathcal{G}$ defines a unitary representation of $\mathbb{Z}_2^{\times (\vert \cV \vert +\vert \cF \vert)}$, but it is \emph{not faithful}. We cannot directly apply Lemma~\ref{lemma:dim_Hj} to decompose $\cH_{\rm kin}$ into isotypes, but all is not lost as following the same method will lead to the desired result. 

To make this precise, let us rely on the same definition of abstract gauge group $G$ as in \eqref{eq:group_surface} and \eqref{eq:surface_group_law}. We then introduce a unitary representation $U: G \to {\rm Aut}(\cH_{\rm kin})$ as defined in \eqref{eq:U_surface_code}.  By construction, we have ${\rm Im}(U)= \mathcal{G}$, but $U$ is not faithful since:
\begin{equation}
U^{-1}(\{I\}) = \{ (\emptyset, \emptyset) , (\cV, \emptyset), (\emptyset, \cF), (\cV, \cF) \} \simeq \mathbb{Z}_2 \times \mathbb{Z}_2\,. 
\end{equation}
Decomposing $U$ into isotypes leads to the direct sum
\begin{equation}
    \cH_{\rm kin} = \bigoplus_{(\hat{V}, \hat{F})\,, \hat{V} \subset \mathcal{V}\,, \hat{F} \subset \mathcal{F}} \cH_{\hat{V},\hat{F}}\,, 
\end{equation}
where, as in the previous subsection, we have labeled isotypes of $\cH_{\rm kin}$ by the locus of vertex and face defects. Let us determine the respective dimensions of these isotypes. Let $(\hat{V} , \hat{F}) \in \cV \times \cF$ be a defect configuration. Applying the same character theory argument as in Lemma~\ref{lemma:dim_Hj}, we find that 
\begin{align}
    \dim (\cH_{\hat V , \hat F}) &= \frac{1}{2^{|\cV|+|\cF|}} \sum_{h \in \{ (\emptyset, \emptyset) , (\cV, \emptyset), (\cV, \emptyset), (\cV, \cF) \}} \underbrace{\overline{\chi_U (h)}}_{2^{\vert \cE\vert }} \chi_{(\hat V , \hat F)} (h) \\
    &= 2^{2 g -2} \left( (-1)^{|\hat{V}\cap \emptyset|+|\hat{F}\cap \emptyset|} + (-1)^{|\hat{V}\cap \cV|+|\hat{F}\cap \emptyset|} + (-1)^{|\hat{V}\cap \emptyset|+|\hat{F}\cap \cF|} + (-1)^{|\hat{V}\cap \cV|+|\hat{F}\cap \cF|} \right) \\
    &= 2^{2 g -2} \left( 1 + (-1)^{\vert \hat V \vert } + (-1)^{\vert \hat F \vert } + (-1)^{\vert \hat V \vert + \vert \hat F \vert} \right) \\
    &= 2^{2 g -2} \left( 1 + (-1)^{\vert \hat V \vert } \right) \left( 1 + (-1)^{\vert \hat F \vert } \right) \,,
\end{align}
which allows us to conclude that
\begin{equation}
    \dim (\cH_{\hat V , \hat F}) = \begin{cases}
       2^{2g} \quad \mathrm{if} \quad \vert \hat V \vert \in 2\mathbb{N} \quad \mathrm{and}\quad \vert \hat F \vert \in 2\mathbb{N}\,,\\
      0 \qquad \mathrm{otherwise}\,.
    \end{cases}
\end{equation}
In other words, the fact that the representation is not faithful translates into global charge conservation, which imposes that vertex (resp.\ face) defects can only be created \emph{by pairs}.\footnote{This can also directly be gleaned from the identities in Eq.~\eqref{eq:relations}.} This is the key (and obviously well-known) difference between the present model and the one discussed in the previous subsection. It is however nice to recover this fact from representation theory alone.

To summarize, this code decomposes into $2^{|\cV|+|\cF|-2}$ isotypes of the group $G \simeq \mathbb{Z}_2^{\times(|\cV|+|\cF|)}$, each of dimension $2^{2g}$. This means that a fourth of all possible isotypes appear in the decomposition of $\mathcal{H}_{\rm kin}$. In particular, the code subspace can support $2g$ logical qubits:
\begin{equation}
    \mathcal{H}_{\rm pn} \equiv \cH_{\emptyset , \emptyset } \simeq (\mathbb{C}^2)^{\otimes 2g}\,.
\end{equation}
From there, the analysis of the previous section generalizes straightforwardly. In particular, selecting a complete reference frame $R$ is equivalent to fixing the homotopy classes of a system of defect-connecting paths in each sector. Once this is done, we can factorize the kinematical Hilbert space as:
\begin{equation}
    \cH_{\rm kin} = \cH_{\emptyset , \emptyset} \otimes_R \cH_{\rm gauge}\,, \quad \mathrm{with} \quad  \cH_{\rm gauge}\simeq (\mathbb{C}^2)^{\otimes(|\mathcal{V}|+|\mathcal{F}|-2)}\,.
 \end{equation}
The main observation made in the previous examples, that gauge transformations act solely on $\cH_{\rm gauge}$ via some unitary representation $U_R$, remains true. The only thing that changes is that $U_R$ is not isomorphic to the regular representation of $G$ anymore: indeed, it only contains a copy of one in four of its irreducible representations. Thus, a complete frame $R$ now furnishes a QRF that no longer is ideal.

\subsubsection{Alternative gauge group and its faithful representation}

In the way we have formulated it in Sec.~\ref{sec_errorduality} and definition  \ref{def:dual_reps_faithful}, our notion of error duality cannot be directly applied to the representation $U:G \to {\rm Aut}(\cH_{\rm kin})$, since it is not faithful. However, taking the quotient by $U^{-1}(\{I\})= \{ (\emptyset, \emptyset) , (\cV, \emptyset), (\emptyset, \cF), (\cV, \cF) \}$, we obtain a faithful representation $\mathtt{U}: \mathtt{G} \to {\rm Aut}(\cH_{\rm kin})$ of the quotient group 
\begin{equation}
    \mathtt{G}:= G / U^{-1}(\{I\}) \simeq \mathbb{Z}_2^{\times (|\cV|+ |\cF|)} /(\mathbb{Z}_2 \times \mathbb{Z}_2)\simeq \mathbb{Z}_2^{\times (|\cV|+ |\cF| - 2)}\,.
\end{equation}
$\mathtt{G}$ can just as well serve as abstract gauge group. From a discrete $\mathbb{Z}_2$-homology point of view, it is nothing but the product of the group of boundaries with the group of coboundaries of the discrete surface i.e. $\mathtt{G}= \cB \times \cB^*$ (see Sec.~\ref{sec:surface_code_homological_descr}). Indeed, the group $(\cP(\cV), \Delta)$ considered modulo the equivalence relation: \begin{equation}
\forall V, W\subset \cV\,, \qquad V \sim W \quad \Leftrightarrow \quad V \sqcup W = \cV
\end{equation}
can be identified with the group of \emph{boundaries} of the surface, whose elements are of the form $\partial V$ ($V \in \cV$). Likewise, the group $(\cP(\cF), \Delta)$ modulo the identification of $\hat{F} \subset \cF$ with its complement $\cF \setminus \hat{F}$  can be identified with the group of \emph{coboundaries} of the surface, whose elements are in the image of the co-boundary map $\delta$ (see Sec.~\ref{sec:surface_code_homological_descr}). Thus, the elements of $\mathtt{G}$ can be labeled by equivalence classes $(\mathtt{V},\mathtt{F})$, where $V,W\in\mathtt{V}$ if and only if $W=\mathcal{V}\setminus V$ and $F,F'\in\mathtt{F}$ if and only if $F'=\mathcal{F}\setminus F$.

As compared to $G$, which is generated by \emph{local} structures (vertices and dual vertices), $\mathtt{G}$ has a more \emph{non-local} character, being generated by boundaries and coboundaries. Hence, $\mathtt{G}$ would probably not qualify as a local gauge group associated to a local gauge theory. But it can nonetheless serve as a gauge group in the broad sense we are relying on in the present paper. In the remainder of this subsection and in the next, we will consider $\mathtt{G}$ as the defining gauge group of closed surface codes. 

$\mathtt{G}$ being the quotient of $G$ by a normal subgroup of order $4$, the group of characters of $\mathtt{G}$ is isomorphic to a subgroup of index $4$ of $\hat{G}$: namely, $\hat{\mathtt{G}}$ can be identified with the subgroup of characters $\chi_{(\hat{V}, \hat{F})} \in \hat{G}$ labeled by $(\hat{V} , \hat{F})\in \cV \times \cF$ with both $|\hat{V}|$ and $|\hat{F}|$ \emph{even}. Indeed, in that case, definition \eqref{eq:charac_surface} descends well to the quotient labeled by the equivalence classes $(\mathtt{V},\mathtt{F})$ since, for any $(V,F) \subset \cV \times \cF$,
\begin{equation}\label{nonfaithrep}
    \begin{split}
    \hat{V}\cap (\cV \setminus V)= \hat{V} \setminus (\hat{V}\cap V) &\quad \Rightarrow \quad (-1)^{|\hat{V}\cap (\cV \setminus V)|} = (-1)^{|\hat{V}|- |\hat{V}\cap V|} \underset{|\hat{V}|\in 2 \mathbb{N}} = (-1)^{|\hat{V}\cap V|}\,,  \\
        \hat{F}\cap (\cF \setminus F)= \hat{F} \setminus (\hat{F}\cap F) &\quad \Rightarrow \quad (-1)^{|\hat{F}\cap (\cF \setminus F)|} = (-1)^{|\hat{F}|- |\hat{F}\cap F|} \underset{|\hat{F}|\in 2 \mathbb{N}} = (-1)^{|\hat{F}\cap F|}\,.
\end{split}
\end{equation}
This provides us with $2^{|\cV|+ |\cF|-2}= |\mathtt{G}|$ independent characters, which therefore make up for all of $\hat{\mathtt{G}}$.

\subsubsection{Dual representations and error duality}\label{subsec:example_dual_rep_toric_code}

As in the case of the surface code with boundary from Sec.~\ref{subsec:square_code}, we can construct two types of representations that are dual to the defining representation $U$. On the one hand, given some complete set of frame fields $R=(R_{{\chi}_{(\hat{V}, \hat{F})}})_{(\hat{V}, \hat{F})\in \hat{\mathtt{G}}}$, which depends on a choice of system of paths and dual paths for each sector $(\hat{V}, \hat{F})\in \hat{\mathtt{G}}$, we can define a dual representation $\hat{U}_R$ that is naturally associated to $R$. Such dual representation does not take value in the Pauli group and its associated gauge-fixing errors act nonlocally relative to the kinematical tensor product structure (by contrast, they act locally relative to the tensor product structure $\otimes_R$ induced by $R$). Such construction proceeds exactly as in Sec.~\ref{subsec:duality_square}, so we will not comment further on it. On the other hand, it is possible to construct dual representations that take value in the Pauli group, and whose associated gauge-fixing errors only act on a subset of the kinematical qubits, which collectively constitute a local frame. To get there, we can proceed similarly as in Sec.~\ref{sec:surface_code_duality_Pauli}, with some modifications that we now describe.

To this effect, let $(T,T')$ be two subsets of $\cE$ such that: a) $T \cap T' = \emptyset$; b) $T$ is a spanning tree of the direct lattice; c) $T'$ is a spanning tree of the dual lattice. It is clear that such conditions can be met, for instance in the following way: select first a spanning tree $T$ of the direct combinatorial graph; one can show that the dual combinatorial map  minus $T$ is connected, so one can select a spanning tree $T'$ in this map; clearly, $T'$ is also spanning for the original dual map. Furthermore, there are exactly $2g$ edges left from the original map which are not in $T \cup T'$.\footnote{\label{foot:forests}The same argument can be upgraded to justify the existence of the forests $(T,T')$ from Sec.~\ref{sec:surface_code_duality_Pauli}. Starting from a planar lattice such as the one represented in Fig.~\ref{fig:surface_code-tree}, we add one new vertex $v_{\rm sink}$ and one new face $f_{\rm sink}$ such that: any edge from the rough boundaries is incident to $v_{\rm sink}$ and any edge from the smooth boundaries belongs to $f_{\rm sink}$. We obtain in this way a combinatorial map with the topology of a closed non-orientable surface of \emph{non-orientable genus} equal to one. In our notation, this means that $2g =1$. Hence, by the same argument as in the main text, we can construct a spanning tree $\tilde{T}$ and a dual spanning tree $\tilde{T}'$ which are disjoint and cover all the edges of the map but $2g=1$ of them. Removing the vertex $v_{\rm sink}$ and the face $f_{\rm sink}$ to return to the original lattice, the edges from $\tilde{T}$ produce a spanning forest $T$ while the edges from $\tilde{T}'$ produce a dual spanning forest (indeed, since we are removing a vertex and a face, connectedness may be lost). The pair $(T,T')$ obeys all the conditions postulated in Sec.~\ref{sec:surface_code_duality_Pauli}.
} Furthermore, let $v_0 \in \cV$ (resp.\ $f_0 \in \cF$) denote a marked vertex (resp.\ a marked face) adjacent to an edge of $T$ (resp.\ to an edge of $T'$) which we will call \emph{root-vertex} (resp.\ \emph{root face}). See Figure \ref{fig:duality_torus} for an example in genus $1$. Associated to the pair $(T,T')$ we have a map $\cV \times \cV \ni (v_1 , v_2) \mapsto \gamma_{v_1 , v_2}$ and a map $\cF \times \cF \ni (f_1 , f_2) \mapsto \gamma_{f_1 , f_2}'$ where: $\gamma_{v_1 , v_2}$ denotes the unique simple path in $T$ with $\{ v_1 , v_2 \}$ as endpoints; and $\gamma_{f_1 , f_2}'$ denotes the unique simple dual path in $T'$ with $\{ f_1 , f_2 \}$ as endpoints. Note that, for any $(v_1 , v_2) \in \cV \times \cV$ and $(f_1 , f_2) \in \cF \times \cF$, we have $\gamma_{v_1, v_2} \cap \gamma_{f_1, f_2}$ so that
\begin{equation}
    [S^Z (\gamma_{v_1 , v_2}), S^X (\gamma_{f_1 , f_2}')] = 0\,.
\end{equation}
We also have
\begin{equation}
    \gamma_{v_1, v_2} = \gamma_{v_0 , v_1} \Delta \gamma_{v_0 , v_2}\,, \qquad \gamma_{f_1, f_2}' = \gamma_{f_0 , f_1}' \Delta \gamma_{f_0 , f_2}'\,,
\end{equation}
and therefore
\begin{equation}\label{eq:pairs_surface}
    S^Z( \gamma_{v_1, v_2}) = S^Z( \gamma_{v_0 , v_1}) S^Z(  \gamma_{v_0 , v_2}) \,, \qquad S^X( \gamma_{f_1, f_2}' ) = S^X( \gamma_{f_0 , f_1}' ) S^X( \Delta \gamma_{f_0 , f_2}')\,.
\end{equation}
We can thus define a representation of $\hat{\mathtt{G}}$ similarly to \eqref{eq:dual_rep_square}, namely: for any $\hat{V} \subset \cV$ and $\hat{F} \subset \cF$ with $|\hat{V}|, |\hat{F}|\in 2 \mathbb{N}$, 
\begin{equation}\label{eq:dual_rep_closed}
    \hat{\mathtt{U}}^{\chi_{(\hat V, \hat F)}} :=  \prod_{v \in \hat V} S^Z (\gamma_{v_0, v})  \prod_{f \in \hat F} S^X (\gamma_{f_0 , f}') \,.  
\end{equation}
Owing to \eqref{eq:pairs_surface}, this product can be expressed as a product of operators $S^Z( \gamma_{v_1, v_2})$ and $S^Z( \gamma_{f_1, f_2}')$ where $\{ v_1 , v_2\}$ (resp.\ $\{ f_1 , f_2\}$) run over disjoint pairs of vertices (resp.\ faces) in $\hat{V}$ (resp.\ $\hat{F}$). Such operators thus provide a particular family of charge-creating errors, as considered in the original reference \cite{Kitaev:1997wr}, where charges are created by pairs and located at the endpoints of a string. By requiring the strings and dual strings to lie in $T$ and $T'$, we are making sure that the family of errors so obtained form a group. Moreover, this group is dual to $\mathcal{G}={\rm Im}(U)= {\rm Im}(\mathtt{U})$. Indeed, the same algebraic manipulations as those performed in \eqref{eq:proof_duality_square} and \eqref{eq:proof_rep_square} (while keeping \eqref{nonfaithrep} in mind) allow to establish that: $\hat{\mathtt{U}}$ does define a representation of the Pontryagin dual $\hat{\mathtt{G}}$, and it is dual to the representation $\mathtt{U}$ of $\mathtt{G}$ (in the sense of Definition \ref{def:dual_reps_faithful}).

\begin{figure}
    \centering
    \includegraphics[scale=0.7]{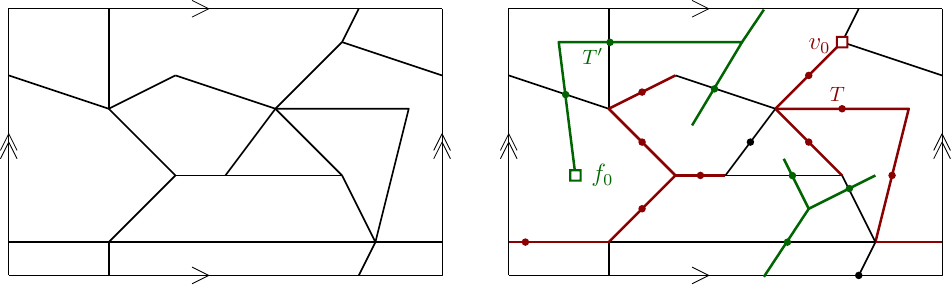}
    \caption{Left: a (non-regular) combinatorial map of genus $g=1$ (horizontal and vertical sides of the rectangle are identified according to the arrows). Right: a pair of non-overlapping spanning trees $(T,T')$ on the same map. Edges (and qubits) from $T$ are represented in red, while dual edges (and qubits) from $T'$ are represented in green. The root-vertex $v_0$ and the root-face $f_0$ are also emphasized, and the $2g =2$ qubits represented in black belong to neither of the two trees: they constitute a local subsystem $S$ of the type investigated in Sec.~\ref{Sec:genstabilizerQRF}, while qubits with support in $T\cup T'$ serve as frame degrees of freedom.}
    \label{fig:duality_torus}
\end{figure}

Analogously to our discussion from Sec.~\ref{sec:surface_code_duality_Pauli}, we see that the dual representation $\hat{\mathtt{U}}$ takes value in the Pauli group and has support on a strict subset of the kinematical qubits. In the spirit of Sec.~\ref{Sec:genstabilizerQRF}, we can interpret those as degrees of freedom supporting a local QRF $R$, while the leftover qubits can be understood as logical qubits pertaining to a subsystem $S$. In the present case, there are $2g$ qubits in $S$ and $|\cE| - 2g$ qubits in $R$ (see Fig.~\ref{fig:duality_torus}). The dual errors $\{\hat{E}_{(V,F)}\}_{(V,F) \in \mathtt{G}}$ associated to $\hat{\mathtt{U}}$ are gauge-fixing errors that force the state of the frame qubits to point in a particular definite orientation of the gauge group.

\section{Discussion}
\label{sec:discussion}

In this section, we give a technical discussion of how our results relate to other characterizations of redundancy and correctability in QEC, as well as how our results relate to operator algebra quantum error correction (OAQEC).
Non-technical remarks about our results and future directions follow in the Conclusion.

\subsection{QRF perspective on redundancy and correctability}\label{ssec_redundancy}

Quantum reference frames, in particular via Cor.~\ref{cor:KL-via-QRFs} and Prop.~\ref{prop_nontriv}, offer a clean perspective on the relationship between redundancy and the correctability of errors in stabilizer quantum error correcting codes.
Following the corollary, given a set of correctable errors $\mathcal{E}$, it is always possible to construct a QRF $R$ and a tensor product structure $\Hil_\mrm{physical} = \Hil_R \otimes \Hil_S$ such that both the action of the errors and the stabilizers are entirely localized to $R$.
For a $[[n,k]]$ Pauli stabilizer code, it follows that $\log_2(\dim \Hil_R) = n-k$.
Therefore, as a joint $+1$ eigenstate of all the stabilizers, any encoded state $\ket{\bar\psi} \in \Hil_\mrm{code}$ necessarily takes the form $\ket{\bar\psi} = \ket{0}_R \otimes \ket{\psi}_S$, where $\ket{0}_R$ is \emph{the} joint $+1$ eigenstate of the stabilizers \emph{as operators on R}.
This makes concrete, via QRFs, the picture of a $[[n,k]]$ stabilizer code as consisting of $n-k$ virtual stabilizer qubits and $k$ logical qubits \cite{PreskillNotes,Poulin:2005wry,Kao:2023ehh}.

The QRF picture nicely complements other information theoretic characterizations of correctability.
Recalling the formulation of error correction with quantum channels (Eq.~\eqref{qeccrel}), the question of whether a given error channel $\tilde{\mathcal{E}}$ (constructed out of errors from $\mathcal{E}$) is correctable becomes the question of whether the recovery channel $\Oh$ exists.
A necessary and sufficient condition for $\Oh$ to exist is that, for any $\bar{\rho}, \bar{\sigma} \in \cB(\Hil_\mrm{code})$, their relative entropy should satisfy \cite{Petz:1986tvy,Petz:1988usv}
\begin{equation}
    D(\tilde{\mathcal{E}}(\bar{\rho}) \, \Vert \, \tilde{\mathcal{E}}(\bar\sigma)) = D(\bar\rho \, \Vert \, \bar\sigma).
\end{equation}
In other words, recalling that relative entropy is a measure of the distinguishability of states, the error channel should not decrease the distinguishability of encoded states.\footnote{In general, the quantum data processing inequality says that acting with a quantum channel can only decrease or preserve relative entropy \cite{UhlmannRelativeEntropy}.}
Now if $\tilde{\mathcal{E}} = \tilde{\mathcal{E}}_R \otimes \mrm{id}_S$, $\bar \rho = \ketbra{0}{0}_R \otimes \rho_S$, and $\bar{\sigma} = \ketbra{0}{0}_R \otimes \sigma_S$, a straightforward computation shows that
\begin{equation}
    D(\tilde{\mathcal{E}}(\bar{\rho}) \, \Vert \, \tilde{\mathcal{E}}(\bar\sigma)) = D(\rho_S \, \Vert \, \sigma_S) = D(\bar\rho \, \Vert \bar\sigma).
\end{equation}

Another equivalent criterion comes from dilating the error channel $\tilde{\mathcal{E}}$ to an isometry $V : \Hil_\mrm{physical} \to \Hil_\mrm{physical} \otimes \Hil_E$ with an environment $\Hil_E$ such that $\tilde{\mathcal{E}}(\rho) = \Tr_E[V \rho V^\dagger]$.
Further introduce a reference system $S'$ with $\dim S' = \dim \Hil_\mrm{code}$, and let
\begin{equation}
    \ket{\Gamma}_{RSS'} = \sum_{\mu=1}^{\dim\Hil_\mrm{code}} \ket{\bar{\mu}}_{RS}\otimes\ket{\mu}_{S'}
\end{equation}
denote a maximally-entangled state across the code subspace and $S'$.
Then, a necessary and sufficient condition for $\tilde{\mathcal{E}}$ to be correctable is that
\begin{equation}
    \rho_{S'E} = \Tr_{RS}\left[ V \otimes I_{S'} \ketbra{\Gamma}{\Gamma}_{RSS'} (V \otimes I_{S'})^\dagger  \right]
\end{equation}
factorizes across $S'$ and $E$, i.e. $\rho_{S'E} = \rho_{S'} \otimes \rho_E$.
In other words, if we view errors as being due to uncontrolled interactions with an unmonitored environment, then they are correctable only if the environment develops no correlations with the logical data that is stored in the code subspace.
This condition follows immediately if we know that we can write $V = V_R \otimes I_S$ and $\ket{\bar\mu}_{RS} = \ket{0}_R\otimes\ket{\mu}_S$, as guaranteed by the QRF perspective.

\subsection{Connection with operator algebra QEC}

Proposition \ref{prop_nontriv} also reveals a connection between QRFs and operator algebra quantum error correction (OAQEC) \cite{BenyOAQEC1,BenyOAQEC2}.
The point of contact with OAQEC is twofold.
First, the QRF-disentangler $T_R$ associated with the frame orientation basis $\{\ket{g}_R\}_{g\in G}$ of Proposition \ref{prop_nontriv} allows us to explicitly factorize the kinematical (or equivalently, physical) space into the standard form for subspace QEC.
Second, the algebra of relational observables of the logical system $S$ relative to the frame $R$ is correctable for the errors $\{\hat{E}_{\chi}\}_{\chi\in\hat G}$ in a sense to be clarified shortly.  

In OAQEC, the physical space decomposes as $\Hil_{\rm physical}=\oplus_k (A_k\otimes B_k)\oplus\mathcal K$, where $A_k$ and $B_k$ are the ancillary and logical sectors, respectively, and $\mathcal K$ is the total error space.
The usual setting of subspace QECCs is recovered as the special case in which $\Hil_{\rm physical}=(A\otimes B)\oplus\mathcal K$ with $\dim A=1$.
Denoting by $\Pi$ the projector onto $A\otimes B$, a set of operators $\Lambda$ on $\Hil_{\rm physical}$ is said to be correctable for an error channel $\tilde{\mathcal E}:\cB(\Pi(\Hil_{\rm physical}))\to\cB(\Hil_{\rm physical})$ on states in $\Pi(\Hil_{\rm physical})$ if there exists a recovery channel $\mathcal O$ such that $\Pi(\mathcal O\circ\tilde{\mathcal E})^\dag(L)\Pi=\Pi L\Pi$, for any $ L \in\Lambda$ (cf.~\cite[Definition 8]{BenyOAQEC2}).

In our present QRF setting, the refactorized physical space decomposes as 
\begin{equation}
T_R(\Hil_{\rm kin})=(\ket{1}_R\otimes\Hil_S)\oplus\mathcal K\;,
\end{equation}
with the total error space $\mathcal K$ given by
\begin{equation}
\mathcal K=\bigoplus_{\chi\neq1} T_R(\Hil_{\chi})\;,
\end{equation}
with $\Hil_{\chi}=P_{\chi}(\Hil_{\rm kin})$, $\chi\neq1$, the isotypes of the non-trivial representations of the stabilizer group $G$ (cf.~App.~\ref{app:reps}). The frame $R$ thus serves as the ancillary system encoding all the gauge-redundant information which, after trivialization, is in the fixed ready state $\ket{1}_R$ on which the errors \eqref{Tprop:errors} as well as the gauge transformations solely act.~In contrast, the non-redundant logical information is entirely encoded in the system $S$.~Correspondingly, as a direct consequence of claim \emph{(iii)} in Proposition \ref{prop_nontriv}, we have
\begin{equation}\label{eq:commrefactor}
[{T}_R\Pi_{\rm pn}\,\hat{E}_{\chi}^\dagger \hat{E}_{\eta}\,\Pi_{\rm pn}\,{T}^\dag_R, \ket{1}\!\bra{1}_R\otimes f_S]=0\qquad\forall\, \hat{E}_\chi, \hat{E}_\eta\in\hat{\mathcal E}\,,\; f_S\in\cB(\Hil_S)\,.
\end{equation}
Equivalently, owing to the unitarity of $T_R$ on $\Hil_{\rm kin}$ for ideal frames, we have
\begin{equation}\label{eq:commrefactor2}
[\Pi_{\rm pn}\,\hat{E}_{\chi}^\dagger \hat{E}_{\eta}\,\Pi_{\rm pn}, \Pi_{\rm pn}(\ket{e}\!\bra{e}_R\otimes f_S)\Pi_{\rm pn}]=0\qquad\forall\, \hat{E}_{\chi}, \hat{E}_{\eta}\in\hat{\mathcal E}\,,\; f_S\in\cB(\Hil_S)\,.
\end{equation}
In \cite[Thm.~2]{BenyOAQEC1} and \cite[Thm.~9]{BenyOAQEC2}, the above commutation condition was identified as the generalization of the Knill-Laflamme condition to operator algebras in that it provides the necessary and sufficient condition for the algebra of operators under consideration to be correctable for a set of errors.~Therefore, we see that the algebra of relational observables describing $S$ relative to $R$ in orientation $e\in G$,\footnote{Strictly speaking, for the operators in Eq.~\eqref{eq:relationalSobs} to be observables, $f_S$ should be Hermitian.~The discussion of this section however applies also to generic frame-dressed/relational operators with $f_S\in\cB(\Hil_S)$ so that we shall not restrict them to be Hermitian.} i.e.,
\begin{equation}\label{eq:relationalSobs}
\mathcal A_{S|R}^e=\bigl\{|G|\cdot\Pi_{\rm pn}(\ket{e}\!\bra{e}_R\otimes f_S)\Pi_{\rm pn}\in\cB(\Hil_{\rm pn})\;,\;f_S\in\cB(\Hil_S)\bigr\}
\end{equation}
is correctable for the errors $\{\hat{E}_{\chi}\}_{\chi\in\hat{G}}$ on states in $\Hil_{\rm pn}$.~That is, given the trace-preserving error channel
\begin{equation}
    \tilde{\mathcal E}(\bullet)=\frac{1}{|G|}\sum_{\chi\in\hat G}\hat{E}_{\chi}\bullet\hat{E}_{\chi}^{\dagger}\;,
\end{equation}
there exists a channel $\hat{\mathcal O}$ such that
\begin{equation}
\Pi_{\rm pn}(\hat{\mathcal O}\circ\tilde{\mathcal E})^\dag(\mathcal A^e_{S|R})\Pi_{\rm pn}=\mathcal A^e_{S|R}\;.
\end{equation}
By inspection, the recovery channel is given by the electric error recovery \eqref{electricrecovery} specialized to the errors $\hat{E}_{\chi}$, namely,
\begin{equation}
\hat{\mathcal O}(\bullet)=\sum_{\chi\in\hat{G}}\hat{E}_{\chi}^\dag P_{\chi}\bullet P_{\chi}\hat{E}_{\chi}\;.
\end{equation}
Equivalently, the channel $\mathcal T_R\circ\hat{\mathcal O}\circ\mathcal T_R^\dag$, with $\mathcal T_R(\bullet)=T_R\bullet T_R^\dag$, is the recovery channel for correcting the operators $\ket{1}\!\bra{1}_R\otimes f_S$ for the errors $\hat{E}_{\chi}$ on the refactorized code space $T_R(\Hil_{\rm pn})=\ket{1}_R\otimes\Hil_S$ (recall that $[{T}_R,\hat{E}_\chi]=0$).~Lastly, we note that $\mathcal T_R(\mathcal A_{S|R}^e)=\ket{1}\!\bra{1}_R\otimes\cB(\Hil_S)=\cB(T_R(\Hil_{\rm pn}))$, so that it contains all subalgebras of $\cB(T_R(\Hil_{\rm pn}))$ correctable for the errors $\hat{E}_{\chi}$ on states in $T_{R}(\Hil_{\rm pn})$ (cf.~\cite[Corollary 10]{BenyOAQEC2}).

\section{Conclusion}\label{sec:conclusion}

In this work, we proposed a correspondence between quantum gauge theories and quantum error correcting codes and we took the initial steps towards formalizing this correspondence through the apparatus of quantum reference frames.
The result was a QECC/QRF dictionary, summarized in Table~\ref{tab_dictionary} and elaborated for Pauli stabilizer codes in the text's main body.

A $[[n,k]]$ Pauli stabilizer QECC generically gives rise to QRFs with symmetry group $\mathbb{Z}_2^{\times(n-k)}$.
A multitude of QRFs can be extracted from any given QECC, corresponding to a multitude of ways of separating logical information from redundant degrees of freedom.
Moreover, this separation of the code into frame and system degrees of freedom can be local with respect to the innate tensor product structure of the physical qubits (such as with the local frames in Sec.~\ref{sec_genstab}) or it can endow the physical Hilbert space with an entirely different tensor product structure (such as with the nonlocal frames in Sec.~\ref{sec_genstab}).

A choice of QRF is closely connected with the choice of correctable error set given a code subspace.
One of our core results (Thm~\ref{thm_errorQRF} and Cor.~\ref{cor:KL-via-QRFs}) was to show that there is a one-to-one correspondence between (equivalence classes of) correctable error sets and QRFs such that errors (and stabilizers) act only on frame degrees of freedom, leaving the (logical) system degrees of freedom untouched.
The correspondence between gauge-invariant information in gauge theories and logical information in QECCs is thus further substantiated.

A significant part of our investigations focused on errors.
Initially it might seem a bit strange to think of errors in a gauge theory context, given the correspondence between the (gauge-invariant) perspective-neutral Hilbert space and the code subspace.
Errors can take encoded states outside of the code subspace, yet in gauge theories, we do not usually think of taking a state outside of the perspective-neutral Hilbert space into a configuration that violates the gauge constraint.
Nevertheless, we can view the process of violating gauge-invariance by fixing a gauge as an error, and such gauge-fixing errors end up being totally general:
given a Pauli stabilizer QECC that corrects a specified set of Pauli errors, there exists a QRF with respect to which these errors act as gauge-fixing errors, and they act only on the frame degrees of freedom (Cor.~\ref{cor:gaugefix-Pauli}).
Furthermore, we can think of a Pauli error as a sort of electric excitation (Sec.~\ref{ssec_pauliduality}) and the corresponding gauge-fixing error as a sort of magnetic excitation (Sec.~\ref{ssec_dualerrors}), with the two being related by Pontryagin duality.
We then illustrated all of these findings in surface codes (Sec.~\ref{sec:surfacecodes}), which are veritable stabilizer code/gauge theory hybrids.

We finish by mentioning a handful of directions for future work.
Two immediate goals are of course to connect more general gauge theories to quantum error correction, and to further elaborate the QRF-to-QECC direction of the dictionary.
For the former, we aim to formalize the notion of errors as charge excitations and electromagnetic duality in more general contexts \cite{Masazumi1,Masazumi2}.
It will also be interesting to study more involved lattice gauge theories (e.g., non-Abelian lattice gauge theories) or lattice gauge theories of direct physical interest to high energy physics (e.g., lattice quantum electrodynamics or lattice quantum chromodynamics).
For the latter, the PN-framework for general QRFs \cite{delaHamette:2021oex} has the structure of group-based QECCs.
We thus aim to systematically study the error correcting properties of QRFs as group-based QECCs and to investigate the types of codes that one obtains from QRFs beyond the ideal QRFs with symmetry group $\mathbb{Z}_2^{\times(n-k)}$ that we considered here.

To begin with, our results should generalize straightforwardly to QECCs based on representations of $\mathbb{Z}_p^{\times (n-k)}$, taking value into a generalized Pauli group associated to a collection of $n$ $p$-level systems (with $p \geq 3$). Beyond that, we will aim at an extension of our analysis to arbitrary finite Abelian groups, and more generally still, to arbitrary locally-compact Abelian groups. Some of the mathematical structures that we rely on, such as Fourier analysis, will no longer be available for more general groups, at least not in the guise of Pontryagin duality. We should therefore expect new technical and conceptual challenges as we move towards non-Abelian gauge theory and non-Abelian generalizations of stabilizer codes. Nonetheless, it will be especially interesting to construct QECCs out of QRF setups that have non-Abelian and non-compact symmetry groups, as well as non-ideal QRFs that have fuzzy orientation states.
Viewed as a recipe for building codes, QRFs with non-Abelian symmetry groups immediately provide tools to systematically write down and characterize the error correcting properties of non-Abelian stabilizer codes. 
Existing examples of such codes include $XS$ \cite{Ni:2014clx} and $XP$ \cite{Webster:2022kdn} codes, but the QRF formalism opens up the possibility of constructing non-Abelian stabilizer codes in principle from any non-Abelian symmetry group.
In a similar vein, QRFs with compact and non-compact symmetry groups should provide systematic ways to construct QECCs that are suitable for continuous variable and bosonic systems.

Non-ideal QRFs with unsharp frame states---i.e., frame states that still span the frame Hilbert space but are not orthogonal---are particularly intriguing, as these will give rise to approximate quantum error correcting codes (AQECCs).
AQECCs that have a lower resource draw than exact QECCs are of increased interest for near-term quantum devices, but much less is known about them in comparison to their exact counterparts.
In terms of QRFs, one can deliberately construct fuzzy orientation states for theories that would otherwise admit ideal frames.
One can also force all frames to be fuzzy by, e.g., choosing a symmetry group that is larger than what can be faithfully represented by the frame. 
Concretely, this latter strategy amounts to attempting to ``squeeze'' larger codes into smaller representatives at the expense of having fuzzy codewords.

In the QECC-to-QRF direction, code distance (the minimum weight of a logical operator) did not feature in our QECC/QRF dictionary.
As a quantity that mostly has to do with specifics of the group representation, it is less clear what its corresponding entry might be on the QRF side of the dictionary; nevertheless, it has been speculated that code distance might be related to energy gaps in gauge theories \cite{Bao:2023lrb}.
We leave this question to future work.

A key \emph{practical} application can be anticipated to arise in quantum simulations of gauge theories and the Standard Model (SM), which is a field that is rapidly expanding after the successful simulation of simple gauge systems on simple quantum computers \cite{Banuls:2019bmf,Martinez:2016yna,Kokail:2018eiw,Yang:2020yer,Bauer:2022hpo,DiMeglio:2023nsa,Chakraborty:2020uhf,Honda:2021aum}.
Furthermore, initial steps toward quantum simulations of features of quantum gravity have been taken \cite{Brown:2019hmk,Gharibyan:2020bab,Shapoval:2022xeo,Li:2017gvt,Li:2019holentropysim}, as well as toward simulating cosmological particle creation \cite{Maceda:2024rrd} and false vacuum decay \cite{Darbha:2024cwk,Darbha:2024srr}.
Classical simulations of gauge theories have so far almost exclusively relied on Monte-Carlo methods.
However, these are prone to the notorious sign problem in state sum simulations and are incapable of extracting more from a lattice gauge theory than static information (such as particle masses). By contrast, simulating dynamical information and especially emergent phenomena (such as bound particle states) of the SM necessitates a wholly different approach because the parameter space becomes far too large for a classical computer to handle. This is why the particle physics community invests great hope into simulations on a quantum computer \cite{Bauer:2022hpo,DiMeglio:2023nsa}. Crucially, quantum simulations sidestep the sign problem on account of relying on a Hamiltonian formulation.

The current quantum simulations of lattice gauge theories and quantum gravity aspects are not yet sophisticated enough to deal with QEC and rather invoke error mitigation techniques. However, a current roadmap for quantum simulations of particle physics \cite{Bauer:2022hpo,DiMeglio:2023nsa} foresees the implementation of QECCs in these simulations in the coming years.
Indeed, significant progress has been made recently in the experimental implementation of QEC protocols \cite{Ryan-Anderson:2021wzx,Postler:2021ddz,Zhao:2021chc,Acharya:2024btg,Reichardt:2024saa,Putterman:2024kpg,Brock:2024vkc,GoogleQuantumAIandCollaborators:2024efv}, fuelling the anticipation that full-fledged quantum simulations will involve QEC in the long run. 

There are, in fact, first works exploring the idea to align the gauge group of a lattice gauge theory with the stabilizer group, and this can be shown to lead to more efficient QECCs for simulating lower-dimensional Abelian lattice models \cite{Rajput:2021trn,Spagnoli:2024mib}. While these works take a purely practical perspective on the alignment of the code space with the gauge-invariant Hilbert space of a lattice gauge theory, they are compatible with our correspondence and provide a promising example for the practical usefulness of our more fundamental perspective.

\section*{Acknowledgments}
\addcontentsline{toc}{section}{Acknowledgments}

We thank Ning Bao, ChunJun (Charles) Cao, Tobias Haas, Masazumi Honda, Aleksander Kubica, Ognyan Oreshkov, and John Preskill for helpful discussions during the preparation of this manuscript, as well as Stefan Eccles and Joshua Kirklin for initial discussions about this line of research. PAH thanks Mischa Woods in particular for discussion and feedback on an early note set of this QECC/QRF correspondence in 2020. SC acknowledges support from IMB, which receives support from the EIPHI Graduate School (contract ANR-17-EURE-0002). SC also acknowledges support of the Institut Henri Poincar\'{e} (UAR 839 CNRS-Sorbonne Universit\'{e}), and LabEx CARMIN (ANR-10-LABX-59-01), for a research stay during which part of the present work was done. The work of ACD and PAH was supported in part by funding from the Okinawa Institute of Science and Technology Graduate University. ACD is also grateful for support from the Centre for Quantum Information and Communication (QuIC) at the Universit\'e de Bruxelles, which hosted ACD for a cross-consortium visit during which work on this manuscript was completed. PAH was further supported by the Foundational Questions Institute under grant number FQXi-RFP-1801A. FMM's research at Western University is supported by Francesca Vidotto's Canada Research Chair in the Foundation of Physics, and NSERC Discovery Grant ``Loop Quantum Gravity:~from Computation to Phenomenology.'' Western University is located in the traditional territories of Anishinaabek, Haudenosaunee,  L\=unaap\'eewak and Chonnonton Nations. This work is also supported by the NSF grant PHY-2409402, PHY-2110273. FMM also thanks the Okinawa Institute of Science and Technology and the Perimeter Institute for hospitality through the last stages of this work. Research at Perimeter Institute is supported in part by the Government of Canada through the Department of Innovation, Science and Economic Development and by the Province of Ontario through the Ministry of Colleges and Universities. This project/publication was also made possible through the support of the ID\# 62312 grant from the John Templeton Foundation, as part of the \href{https://www.templeton.org/grant/the-quantum-information-structure-of-spacetime-qiss-second-phase}{\color{black}\textit{``The Quantum Information Structure of Spacetime''} Project (QISS)}.~The opinions expressed in this project/publication are those of the author(s) and do not necessarily reflect the views of the John Templeton Foundation.

\appendix

\section{QRFs and tensor product structure changes}\label{app_QRFchange}

In Sec.~\ref{ssec_PN}, we briefly described QRF transformations between two choices of QRF, $R$ and $R'$, which in general split the kinematical Hilbert space differently into redundant frame and complementary system information, $\Hil_{\rm kin}\simeq\Hil_R\otimes\Hil_S$ and $\Hil_{\rm kin}\simeq\Hil_{R'}\otimes\Hil_{S'}$. That is, each defines a TPS on $\Hil_{\rm kin}$, which is nothing but an equivalence class of unitaries
\begin{equation}
    t_R:\Hil_{\rm kin}\rightarrow\Hil_R\otimes\Hil_S\,,
\end{equation}
and similarly for $R'$, where $t_R\sim\tilde{t}_R$ if $t_R\circ \tilde{t}_R^{\,\dag}$ is a combination of tensor product unitaries $W_R\otimes W_S$ and permutations of tensor factors of equal dimension (e.g., see \cite{Cotler:2017abq}). 
The two partitions of $\Hil_{\rm kin}$ are thus related by $v:\Hil_R\otimes\Hil_S\rightarrow\Hil_{R'}\otimes\Hil_{S'}$, where 
\begin{equation}
    v=t_{R'}\circ t_R^\dag\,,
\end{equation}
which is a change of TPS when $v$ is non-local and which, for notational simplicity, we had left implicit in the main text. It is, however, crucial for the generalization of the QRF changes from the standard case that starts from a rigid split $\Hil_{\rm kin}=\Hil_R\otimes\Hil_{R'}\otimes\Hil_S$ to the general product partitions\footnote{One can easily generalize to the case when $\Hil_{\rm kin}$ is isomorphic to a direct sum of tensor factors, but this will not be relevant here.} of $\Hil_{\rm kin}$ that become important in this work. A more explicit discussion of QRF transformations among a continuum of distinct choices will appear in \cite{TBH}, including examples, but we summarize those changes diagrammatically in Fig.~\ref{fig:QRFchange}. 
\begin{figure}[h!]
 \[
    \xymatrix{
        && \ar[dll]
        _-{t_{R}}\underset{(\text{with }U^g)}{\mathcal{H}_{\rm kin}} \ar[dd]|(0.525)\hole|(0.58)\hole|(0.645)\hole
        ^(.335){\Pi_{\rm pn}} \ar[drr]^-{t_{R'}} &&\\
\ar[dd]|(.475){t_R\Pi_{\rm pn}t_R^\dag=\frac{1}{|\mathcal G|}\sum_g U_R^g\otimes U_S^g\qquad}\underset{(\text{with }U_R^g\otimes U_S^g)}{\mathcal{H}_{R}\otimes\mathcal{H}_S} \ar[rrrr]_{v\,=\,t_{R'} \circ\, t_R^\dag} &&&& \underset{(\text{with }U_{R'}^g\otimes U_{S'}^g)}{\mathcal{H}_{R'}\otimes\mathcal{H}_{S'}}\ar[dd]|(.475){\qquad t_{R'}\Pi_{\rm pn}t_{R'}^\dag=\frac{1}{|\mathcal G|}\sum_g U_{R'}^g\otimes U_{S'}^g}\\
      && \ar[dll]
        _-{t_{R}}\Hil_{\rm pn}\ar[drr]^-{t_{R'}} &&\\
     \ar[dd]_-{\mathcal R_R^g}t_R(\Hil_{\rm pn})\ar[rrrr]_{v\,=\,t_{R'} \circ\, t_R^\dag} &&&& t_{R'}(\Hil_{\rm pn})\ar[dd]^{\mathcal R_{R'}^{g'}}\\
     &&&&\\
     \Hil_{|R}\ar[rrrr]^(.428){V_{R\to R'}\,=\,\mathcal R_{R'}^{g'}\,\circ\,v\,\circ\, (\mathcal R_{R}^{g})^\dag}_(.56){=\,(\mathcal R_{R'}^{g'}\,\circ\,t_{R'})\,\circ\,(\mathcal R_{R}^{g}\,\circ\,t_{R})^\dag} &&&& \Hil_{|R'}
} \]
\caption{QRF changes as quantum coordinate transformations in the general context of refactorizing $\Hil_{\rm kin}$.~Note that the intermediate TPS change $v$ on $\Hil_{\rm kin}$ does not affect the compositional ``coordinate change'' form of the transformation familiar from the standard case \cite{delaHamette:2021oex,Hoehn:2019fsy,Hoehn:2021flk,Vanrietvelde:2018dit,Vanrietvelde:2018pgb}.~Note also that $V_{R\to R'}$ is a change of TPS on $\Hil_{\rm pn}$ \cite{Hoehn:2023ehz}.~This generalizes previous constructions in the literature and will be discussed in greater detail in \cite{TBH}.}\label{fig:QRFchange}
\end{figure}

In particular, the usual compositional form of QRF transformations as ``quantum coordinate transformations'' survives in the more general case thanks to the compositional form of the kinematical TPS changes
\begin{eqnarray}
    V_{R\to R'}&=&\left(\mathcal{R}_{R'}^{g'}\circ t_{R'}\right)\circ\left(\mathcal{R}_R^g\circ t_R\right)^\dag\\
    &\underset{\eqref{PWred}}{=}&|\mathcal{G}|\left(\bra{g'}_{R'}\otimes I_{S'}\right)(t_{R'}\Pi_{\rm pn}t_{R'}^\dag)\,v\,t_R\,\Pi_{\rm pn}t_R^\dag\left(\ket{g}_R\otimes I_S\right)\nonumber\\
    &=&|\mathcal{G}|\left(\bra{g'}_{R'}\otimes I_{S'}\right)t_{R'}\Pi_{\rm pn}\,t_R^\dag\left(\ket{g}_R\otimes I_S\right)\,.
\end{eqnarray}
In the special case that $S=R'\otimes\tilde{S}$, $v$ becomes trivial (a change of local basis) and we recover the explicit standard form of QRF transformations, e.g.\ see \cite{delaHamette:2021oex,Hoehn:2021flk,Hoehn:2019fsy,delaHamette:2020dyi},
\begin{equation}
    V_{R\to R'}=\sum_{g''}\ket{g''g}_R\otimes\bra{g''^{-1}g'}_{R'}\otimes U_{\tilde{S}}^{g''}\,.
\end{equation}
This constitutes a change of TPS on $\Hil_{\rm pn}$ \cite{Hoehn:2023ehz}.

The same discussion as above and in Fig.~\ref{fig:QRFchange} applies of course also to the change of trivialization/disentangler, which in the general case reads
\begin{equation}
    T_{R'}\circ v\circ T_R^\dag=\left(T_{R'}\circ t_{R'}\right)\circ\left(T_R\circ t_R\right)^\dag\,.
\end{equation}

\section{Projective unitary representations and cocycles}\label{app_projrep}

In Sec.~\ref{sec:extvsintframe}, we considered a possibly projective unitary representation of some group $G$ on the QRF Hilbert space $\mathcal{H}_R$:\begin{equation}\label{eq:projective_rep_R}
    U^g_R\,U^{h}_R=c(g,h)\,U^{gh}_R\,,\quad\forall\,g,h\in G\,,
\end{equation}
where $|c(g,h)|=1$ is a phase term that differs from $1$ whenever $U_R$ is projective. As $U_R^e=I_R$, we have the condition that $c(g,e)=c(e,g)=1$, for all $g\in G$. We now exhibit additional constraints on the phases and some freedom in choosing them. These results will be invoked in proving several technical results across the main body. Note that in this section, $G$ need not be finite Abelian.

As a result of associativity, the phases $\{c(g,h)\}_{(g,h) \in G\times G}$ must obey the following \emph{cocycle conditions}:
\begin{equation}\label{eq:cocycle_cond}
\forall g,h , k \in G\,, \qquad c(g,hk) c(h,k) = c(g,h) c(gh , k )\,,
\end{equation}
and we will refer to the function $c$ as a \emph{cocycle}. For any $g \in G$, we have the freedom to redefine $U_R^g$ as $\tilde{U}_R^g := \lambda_g U_R^g$, where $\lambda_g$ is a phase (with the condition that $\lambda_e = 1$ given our assumption that the neutral element $e$ must always be represented by the identity $I_R$). $\tilde{U}_R$ then defines a unitary projective representation with cocycle $\tilde{c}$ obeying:
\begin{equation}\label{eq:change_of_cocycle}
\forall g,h \in G\,, \qquad     \tilde{c}(g,h) \lambda_{gh} = c(g,h) \lambda_g \lambda_h\,.
\end{equation}
$U_R$ and $\tilde{U}_R$ are considered equivalent as projective representations. Similarly, we say that two cocycles $c$ and $\tilde{c}$ are equivalent whenever there exist phases $\{\lambda_g\}_{g \in G}$ (with $\lambda_e = 1$) such that \eqref{eq:change_of_cocycle} holds. Finally, a cocycle $c$ is said to be trivial if it is equivalent to the unit cocycle, namely: there exist phases $\{\lambda_g\}_{g \in G}$ (with $\lambda_e = 1$) such that
\begin{equation}\label{eq:trivial_cocycle}
    \forall g,h \in G\,, \qquad c(g,h) = \frac{\lambda_{gh}}{\lambda_g \lambda_h}\,.
\end{equation}
In such a situation, the projective representation $U_R$ from \eqref{eq:projective_rep_R} is in fact equivalent to a non-projective representation.

Before proceeding, let us use the freedom to redefine phases to fix some of the coefficients of the cocycle $c$. First, we note that, for any $g \in G$
\begin{equation}
    c(g,g^{-1}) = c(g^{-1},g)\,.
\end{equation}
Indeed, since $U_R(g)U_R(g^{-1}) = c(g,g^{-1})I_R$, we have $U_R(g)^{-1}=c^*(g,g^{-1}) U_R(g^{-1})$; it follows that $I_R = c^*(
g,g^{-1}) U_R(g^{-1}) U_R(g) = c^*(g,g^{-1}) c(g^{-1},g) I_R$, and therefore $c^*(g,g^{-1}) c(g^{-1},g)=1$. Now, if for any $g \in G$, we define $\lambda_g = \lambda_{g^{-1}}$ to be a square-root of $c^*(g,g^{-1})= c^*(g^{-1},g)$ (with $\lambda_e = 1$) and use \eqref{eq:change_of_cocycle} to introduce an equivalent cocycle $\tilde{c}$, we find that
\begin{equation}
    \tilde{c}(g,g^{-1}) = c(g,g^{-1}) \frac{\lambda_g \lambda_{g^{-1}}}{\lambda_e} = c(g,g^{-1}) c^*(g, g^{-1}) = 1\,.
\end{equation}
This shows that we are free to redefine the cocycle $c$ in \eqref{eq:projective_rep_R} so that:
\begin{equation}\label{eq:cocycle_gauge_restriction2}
    \forall g \in G\,, \qquad c(g, g^{-1})=1\,,
\end{equation}
which we will assume from now on.\footnote{In Sec.~\ref{Sec:genstabilizerQRF}, the $U_R^g$ will be strings of Pauli operators, and so by construction they already satisfy $c(g,g^{-1}) = 1$, since strings of Pauli operators square to the identity.} As a consequence of \eqref{eq:cocycle_gauge_restriction2}, the following useful properties hold: for any $g \in G$,
\begin{equation}\label{eq:inverse_of_UR2}
    (U_R^g)^\dagger = (U_R^g)^{-1} = U_R^{(g^{-1})}\,; 
\end{equation}
and for any $g,h \in G$,
\begin{equation}\label{eq:conjugate_of_cocycle2}
    c^{*}(g,h) = c(h^{-1} , g^{-1} ) \,.
\end{equation}
Indeed, a proof of the preceding property goes as follows:
\begin{align*}
I_R &= U_R(gh) U_R(gh)^{-1} = U_R(gh) U_R(h^{-1} g^{-1}) = c(g,h) c(h^{-1}, g^{-1}) U_R(g)U_R(h)U_R(h^{-1})U_R(g^{-1}) \\
&= c(g,h) c(h^{-1} ,g^{-1}) U_R(g)U_R(h)U_R(h)^{-1} U_R(g)^{-1} = c(g,h) c(h^{-1} ,g^{-1}) I_R\,,
\end{align*}
which can only hold if $c(g,h) c(h^{-1} ,g^{-1}) =1$, that is, if $c^{*}(g,h) = c(h^{-1} , g^{-1})$.

\section{Harmonic analysis on a finite Abelian group}\label{app:Pontryagin}

Let $G$ be a finite Abelian group (the relevant case for the main body of the paper is $G= \mathbb{Z}_2^{\times (n-k)}$). The \emph{Pontryagin dual} of $G$ is defined to be the group of homomorphisms from $G$ to the unit circle $\hat{G}:= {\rm Hom}(G, {\rm U}(1))$, equipped with the pointwise product law: that is, for any $\chi , \eta \in \hat{G}$, we define $\chi\eta: g \mapsto \chi(g)\eta(g)$. Given that $G$ is Abelian, $\hat{G}$ is nothing but the group of irreducible characters of $G$, and we have $|\hat{G}|= |G|$. Since $G$ has finite order, there is also some $k \geq 2$ such that, for any $\chi \in \hat{G}$, ${\rm Im}(\chi)$ is a subset of the subgroup of $k^{\rm th}$-roots of unity (in the main body of the paper, we have ${\rm Im}(\chi) \subset \{1, -1\}$, hence we can take $k=2$). We also have a canonical pairing between $G$ and $\hat{G}$, defined by
\begin{align}
    \langle \cdot , \cdot \rangle: \hat{G}\times G &\to {\rm U}(1) \\
    (\chi , g)& \mapsto \chi(g) \nonumber
\end{align}

The \emph{Pontryagin duality theorem} states that $\hat{\hat{G}}$ is canonically isomorphic to $G$. Indeed, it turns out that the irreducible characters of $\hat{G}$ are provided by the evaluation maps ${\rm ev}_g: \hat{G}\ni \chi \mapsto \chi(g)$ ($g \in G$), which verify ${\rm ev}_{g_1}{\rm ev}_{g_2}= {\rm ev}_{g_1 g_2}$ for any $g_1 , g_2 \in G$, and therefore provide a canonical isomorphism with $G$.

\medskip

The \emph{group algebra} of $G$ is the algebra of formal linear combinations of elements of $G$ with complex coefficients or, equivalently, the algebra of complex functions on $G$ with convolution as product. In more detail, one first defines $\ell^2(G)$ as the complex vector space of formal linear combinations of elements in $G$. Given some $f \in \ell^2(G)$, we will use the notation
\begin{equation}
f = \sum_{g \in G} f(g)\, g\,,    
\end{equation}
where, for any $g \in G$, $f(g) \in \mathbb{C}$ denotes the value of the function $f$ at $g$. In particular, $\ell^2(G)$ can equivalently be thought as the vector space of complex functions on $G$ (which are always square-integrable since $G$ is finite), hence the notation. One can then equip $\ell^2(G)$ with a product $\star: \ell^2(G) \times \ell^2(G) \to \ell^2(G)$, defined to be the unique bilinear map such that:
\begin{equation}
    \forall g , h \in G \subset \ell^2(G)\,, \qquad g \star h := gh\,.
\end{equation}
We recognize $\star$ as a convolution product since, for any $f_1,f_2 \in \ell^2(G)$, we we have:
\begin{equation}
f_1 \star f_2 := \sum_{g\in G} \sum_{h \in G} f_1(g) f_2(h) \, gh =  \sum_{g \in G} \Biggl( \underbrace{\sum_{h \in G} f_1(gh^{-1}) f_2 (h)}_{f_1 \star f_2 (g)} \Biggr) \, g\,.
\end{equation}
This makes $\left( \ell^2(G), \star \right)$ an algebra. 

Finally, one can introduce an involution $^*: \ell^2(G) \to \ell^2(G)$, defined as the unique \emph{antilinear} map such that:
\begin{equation}
\forall g \in G\subset \ell^2(G)\,, \qquad g^{*} := g^{-1}\,.    
\end{equation}
In other words, for any $f \in \ell^2(G)$ and $g \in G$, we have
\begin{equation}
    f^*(g):= \overline{f(g^{-1})}\,.
\end{equation}
One can check that, for any $f_1 , f_2 \in \ell^2(G)$
\begin{equation}
    \left( f_1 \star f_2 \right)^{*} = f_2^{*} \star f_1^{*}\,,
\end{equation}
which makes $\left( \ell^2(G), \star , ^* \right)$ a $*$-algebra. 

\medskip

We have an isomorphism between the vector spaces $\ell^2(G)$ and $\ell^2(\hat{G})$, given by the unique linear map $\cF: \ell^2(G) \to \ell^2(\hat{G})$ such that:
\begin{align}
\forall g \in G\,, \qquad \cF[g] := \frac{1}{\sqrt{|G|}}\sum_{\chi \in \hat{G}} \bar\chi(g)\, \chi\,.
\end{align}
In other words, for any $f\in \ell^2(G)$ and $\chi \in \hat{G}$, we have:
\begin{equation}
    \mathcal{F}[f](\chi)= \frac{1}{\sqrt{|G|}} \sum_{g \in G} f(g) \bar\chi (g)\,.
\end{equation}
$\cF$ defines a \emph{Fourier transform}; its inverse is the unique linear map $\cF^{-1}: \ell^2(\hat{G}) \to \ell^2(G)$ such that:
\begin{equation}
\forall \chi \in \hat{G}\,, \qquad \cF^{-1}[\chi] := \frac{1}{\sqrt{|G|}}\sum_{g \in G} \chi(g)\, g\,.
\end{equation}
Equivalently, for any $\hat{f} \in \ell^2(\hat{G})$ and $g \in G$, we have 
\begin{equation}
    \mathcal{F}^{-1}[\hat{f}](g)= \frac{1}{\sqrt{|G|}} \sum_{\chi \in \hat{G}} \hat{f}(\chi) \chi (g)\,.
\end{equation}

As usual, the Fourier transform maps the convolution product to the pointwise product, which with our notations can be expressed as:
\begin{align}
\forall f_1 , f_2 \in \ell^2(G)\,, \chi \in \hat{G}\,, \qquad \cF [f_1 \star f_2] (\chi) &= \cF[f_1](\chi) \, \cF[f_2](\chi)  \,,  \\
\forall \hat{f}_1 , \hat{f}_2 \in \ell^2(\hat{G})\,, g \in G\,, \qquad \cF^{-1} [\hat{f}_1 \star \hat{f}_2] (g) &= \cF^{-1}(\hat{f}_1)(g) \, \cF^{-1}(\hat{f}_2)(g) \,.
\end{align}
In particular, note that $\mathcal{F}$ does \emph{not} define a $*$-algebra isomorphism from $(\ell^2(G), \star , *)$ to $(\ell^2(\hat{G}), \star , *)$. Rather, it is a $*$-algebra isomorphism from $(\ell^2(G), \star , *)$ to $(\ell^2(\hat{G}), \cdot  , \overline{\phantom{x}})$ where $\cdot$ (resp. $\overline{\phantom{x}}$) denotes the pointwise product (resp. standard complex conjugation) on the vector space of complex functions on $\hat{G}$. 

\section{Representations and operator-valued Fourier transform}\label{app:reps}

Let $G$ be a finite Abelian group and $U: G \to {\rm Aut}(\cH_{\rm kin}), g \mapsto U^g$ a unitary representation of $G$ on a finite-dimensional Hilbert space $\cH_{\rm kin}$.\footnote{In the main body of the paper, the image of $g$ by $U$ is denoted $U^g$. This notation is convenient as it allows to add extra labels referring to subsystems and frames as subscripts. Since we do not need such extra labels in the present Appendix, we find more convenient to use the notation $U(g)$.} We can decompose $\cH_{\rm kin}$ into irreducible representations labeled by characters \cite{serre1977linear, simon1996representations}, or in other words, by elements of the Pontryagin dual $\hat{G}$. Given $\chi \in \hat{G}$, the \emph{isotype} $\cH_\chi$ is defined as the span of all the irreducible subspaces in $\cH_{\rm kin}$ which are isomorphic to the representation $\chi$. We then have a canonical decomposition of the full Hilbert space into isotypes:
\begin{equation}\label{eq:decomp_isotypes}
\cH_{\rm kin} = \bigoplus_{\chi \in \hat{G}} \cH_\chi \,,    
\end{equation}
where for any $\chi \in \hat{G}$, $\cH_\chi$ turns out to be the image of the orthogonal projector \cite{serre1977linear, simon1996representations}
\begin{equation}\label{eq:projector_isotype}
   P_\chi := \frac{1}{|G|} \sum_{g \in G} \bar\chi(g) \, U(g)\,. 
\end{equation}
Owing to the unitarity of $U$, the decomposition \eqref{eq:decomp_isotypes} is orthogonal, which means that $P_\chi P_\eta = \delta_{\chi, \eta} P_\chi$ for any $\chi , \eta \in \hat{G}$.

Another useful way of understanding the unitary representation $U$ is as a $*$-algebra representation. For this purpose, let $\mathcal{B}(\cH_{\rm kin})$ denote the $*$-algebra of linear operators on $\cH_{\rm kin}$ (with the operator product as product, and the adjoint $^\dagger$ as involution). There is a unique $*$-algebra representation $\uU: \ell^2(G) \to \mathcal{B}(\cH_{\rm kin})$ associated to $U$: $\uU$ is defined to be the unique linear map verifying $\uU(g):= U(g)$ for any $g \in G$. One can think of $U$ and $\uU$ interchangeably as they encode the exact same data. For instance, the fact that $U$ is unitary is equivalent to the fact that $\uU$ preserves the $*$-structure. However, $U$ being faithful (as a group representation) is \emph{not} equivalent to $\uU$ being faithful (as a $*$-algebra representation). If $\uU$ is faithful (meaning that $\ker \uU = \{ 0\}$), then $U$ is faithful (meaning that $U^{-1}({\rm id}_{\rm \cH_{\rm kin}}) = \{ e \}$),\footnote{Indeed, if $U$ is not faithful, then we have $g \neq e$ such that $U(g)={\rm id}_{\rm \cH_{\rm kin}}$. Hence $\uU(g - e)= U(g)- {\rm id}_{\rm \cH_{\rm kin}} = 0$, and since $g - e \neq 0$ we conclude that $\uU$ is not faithful.} but the converse is not true because $\{U(g)\}_{g\in G}$ need not be linearly independent.

\medskip

Let us now suppose that the dual group $\hat{G}$ is also represented on $\cH_{\rm kin}$, by a unitary representation $\hat{U}: \hat{G} \to {\rm Aut}(\cH_{\rm kin}), \chi \mapsto \hat{U}(\chi) = \hat{U}^\chi$.\footnote{In the main body of the paper, the image of $\chi$ by $\hat{U}$ is denoted $\hat{U}^\chi$. In the present Appendix, we find it more convenient to use the notation $\hat{U}(\chi)$.} Employing the same notations as above, we denote by $\hat{\uU}: \ell^2(\hat{G}) \to \mathcal{B}(\cH_{\rm kin})$ the unique $*$-algebra representation that linearly extends $\hat{U}$. $\hat{U}$ induces a second decomposition of $\cH_{\rm kin}$:
\begin{equation}\label{eq:decomp_isotypes-dual}
\cH_{\rm kin} = \bigoplus_{g \in G}  \hat{\cH}_{g}\,, 
\end{equation}
where $\hat{\cH}_g$ is the isotype associated to ${\rm ev}_g \in \hat{\hat{G}} \simeq G$. Note how we are relying on the Pontryagin duality theorem to view an irreducible character of $\hat{G}$ as an element of the group $G$. For any $g \in G$, $\hat{\cH}_g$ is the image of the orthogonal projector 
\begin{equation}\label{eq:projector_isotype-dual}
    \hat{P}_g  := \frac{1}{|G|} \sum_{\chi \in \hat{G}} \bar\chi(g) \, \hat{U}(\chi)\,,
\end{equation}
and we have $\hat{P}_g \hat{P}_h = \delta_{g,h} \hat{P}_g$ for any $h\in G$.

\medskip

Let us assume, in addition, that $\uU$ and $\hat{\uU}$ are both \emph{faithful} (that is, as group algebra representations). We can then induce an isomorphism $\tilde{\cF}: \uU\left( \ell^2(G) \right) \to \hat{\uU}( \ell^2(\hat{G}) )$ by requiring the following diagram to commute:
$$
\xymatrix{
\ell^2(G)\ar[rr]^-{\uU}\ar@<-2.5pt>[dd]_{\mathcal F} && \uU(\ell^2(G))\ar@<-25pt>[dd]^-{\tilde{\mathcal F}}\subset\mathcal B(\Hil_{\rm kin})\\
&&\\
\ell^2(\hat{G})\ar@<-2.5pt>[uu]_-{\mathcal F^{-1}}\ar[rr]^-{\hat \uU} && \hat \uU(\ell^2(\hat{G}))\ar@<30pt>[uu]^-{\tilde{\mathcal F}^{-1}}\subset\mathcal B(\Hil_{\rm kin})
}
$$
$\tilde{\cF}$ then defines an \emph{operator-valued Fourier transform}, which proves useful to relate $P_\chi$ to $\hat{U}(\chi)$ (resp. $\hat{P}_g$ to $U(g)$). To see this, let us first observe that
\begin{align}
    \forall \chi \in \hat{G}\,, \qquad \uU\left( \cF^{-1}[\bar\chi]\right) &= \sqrt{|G|} P_\chi  \,,\\
    \forall g \in G\,,  \qquad \hat{\uU}\left( \cF[g]\right) &= \sqrt{|G|} \hat{P}_g  \,.
\end{align}
We can compose those two equations by $\tilde{\cF}$ and $\tilde{\cF}^{-1}$ respectively, to obtain
\begin{align}
    \forall \chi \in \hat{G}\,, \qquad \hat{U}(\bar\chi) &=  \tilde{\cF}\left( \sqrt{|G|} P_\chi \right)  \,,\label{dualfourier}\\
    \forall g \in G\,,  \qquad U(g) &=  \tilde{\cF}^{-1}\left(\sqrt{|G|} \hat{P}_g  \right)\,.
\end{align}
Hence, each representation can be related to its projectors by means of the operator-valued Fourier transform. 

\section{Dual representations}\label{app:duality}

Let $G$ be a finite Abelian group. We can introduce a notion of duality for representations as follows.
\begin{defn}[Duality of representations]\label{def:dual_reps_faithful}
Let $U:G \to {\rm Aut}(\cH_{\rm kin})$ and $\hat{U}:\hat{G} \to {\rm Aut}(\cH_{\rm kin})$ be two unitary representations on a finite-dimensional Hilbert space $\cH_{\rm kin}$. We will say that $U$ and $\hat{U}$ are \emph{dual to each other} if: for any $g \in G$ and any $\chi \in \hat{G}$, 
\begin{equation}\label{eq:dual_reps_faithful}
U(g) \hat{U}(\chi) = \chi(g) \, \hat{U}(\chi) U(g)\,.
\end{equation}
\end{defn}
As explained around Eq.~\eqref{eq:Weyl}, this notion of dual representations is simply a discrete version of position-momentum duality and the Weyl relations in standard quantum mechanics.

The commutation relations among elements of two dual representations are thus entirely determined by the canonical pairing between the group and its dual (see Appendix \ref{app:Pontryagin}). If $U$ is given, this does not fix $\hat{U}$ uniquely, but as we will see, this provides interesting structure. 

Note that Definition \ref{def:dual_reps_faithful} can be overly constraining, in the sense that if $U$ is fixed, a dual $\hat{U}$ does not necessarily exist. This is for instance the case whenever $U$ fails to be faithful, since equation \eqref{eq:dual_reps_faithful} trivially fails for $g \in U^{-1}(\{ {\rm id}_{\cH_{\rm kin}}\}) \setminus \{ e\}$. It is actually possible to refine the definition of duality in order to accommodate non-faithful representations. However, since in the present paper we will only apply this notion to a faithful $U$ (even $\uU$ will always be faithful in our examples), we will refrain from doing so here. We also note that Definition \ref{def:dual_reps_faithful} is compatible with the homomorphism property of the representations. For instance, if we have $g_1, g_2 \in G$ and $\chi \in \hat{G}$ such that $U(g_1)\hat{U}(\chi) = \chi(g_1)\hat{U}(\chi)U(g_1)$ and $U(g_2)\hat{U}(\chi) = \chi(g_2)\hat{U}(\chi)U(g_2)$, then it follows from the homomorphism property of $U$ that $U(g_1 g_2)\hat{U}(\chi) = \chi(g_1g_2)\hat{U}(\chi)U(g_1 g_2)$. A similar statement holds for the dual representation $\hat{U}$, so in practice, we only need to establish \eqref{eq:dual_reps_faithful} for $g$ and $\chi$ drawn from generator subsets in order to prove that $U$ and $\hat{U}$ are dual. 

Let us assume from now on that $U:G \to \mathrm{Aut}(\cH_{\rm kin})$ and $\hat{U}:\hat{G} \to \mathrm{Aut}(\cH_{\rm kin})$ are dual to each other. The duality condition from Definition \ref{def:dual_reps_faithful} immediately leads to: 
\begin{align}
    \forall g,h \in G \,, \qquad U(h)^\dagger \hat{P}_g U(h) &= \hat{P}_{gh} \,, \label{Pgcov}\\
    \forall \chi,\eta \in \hat{G} \,, \qquad \hat{U}(\eta) P_\chi \hat{U}(\eta)^\dagger &= P_{\chi \eta}\,.\label{Pchicov}
\end{align}
The first equation defines a right-action of $G$ on $\{ \hat{P}_g \,|\, g \in G\}$, while the second defines a left-action of $\hat{G}$ on $\{ P_\chi \,|\, \chi \in \hat{G}\}$ (when the characters of $G$ take value in $\{1, -1\}$, as is the case in stabilizer codes, the distinction between right- and left-action becomes irrelevant). In particular, the two previous equations make clear that all the projectors in the family $\{ \hat{P}_g\}_{g \in G}$ (resp. $\{ P_\chi \}_{\chi \in \hat{G}}$) necessarily have the same trace, hence: $\forall g \in G, \chi \in \hat{G}$,
\begin{equation}\label{equaldim}
    \dim (\cH_\chi) =  \frac{\dim(\cH_{\rm kin})}{|G|} = \dim (\hat{\cH}_g)\,.
\end{equation}
We will see that this equation holds in particular for stabilizer codes, and can be derived in that context without relying on the existence of a dual representation (see Appendix \ref{app:stabilizer_reps}). 

The commutation relations also imply that, when we perform a coherent group-average on a dual projector (resp.\ a projector) by the representation $U$ (resp.\ by the dual representation $\hat{U}$), the result is proportional to the identity. This has a nice interpretation in terms of quantum states: the state $\hat{\Psi}_g := \frac{|G|}{\dim(\cH_{\rm kin})} \hat{P}_g$ (which describes a maximally-mixed state on $\hat{\cH}_g$, that is to say a state with definite dual charge $g$ about which we have no other information) is mapped under coherent group-averaging by $U$ to the maximally-mixed state on $\cH_{\rm kin}$
\begin{equation}
   \frac{1}{|G|} \sum_{h \in G} U(h) \hat{\Psi}_g U(h)^\dagger = \frac{1}{\dim(\cH_{\rm kin})} \sum_{h\in G} \hat{P}_{gh^{-1}} = \frac{1}{\dim(\cH_{\rm kin})} \sum_{h\in G} \hat{P}_{h} = \frac{{\rm id}_{\cH_{\rm kin}}}{\dim(\cH_{\rm kin})}\,.
\end{equation}
Likewise, a coherent group-average by the dual representation $\hat{U}$ of the definite-charge maximally-mixed state $\Psi_\chi := \frac{|G|}{\dim(\cH_{\rm kin})} P_\chi$ produces the maximally-mixed state on $\cH_{\rm kin}$:
\begin{equation}
   \frac{1}{|G|} \sum_{\eta \in \hat{G}} \hat{U}(\eta) \Psi_\chi \hat{U}(\eta)^\dagger =  \frac{{\rm id}_{\cH_{\rm kin}}}{\dim(\cH_{\rm kin})}\,.
\end{equation}
Those relations show that the charge and dual charges are complementary: if one knows the value of one of them with certainty, one has essentially no information about the other.

The previous commutation relations also imply the following  identities: for any $g \in G$ and $\chi \in \hat{G}$, we have 
\begin{equation}
\begin{split}\label{Pscommutation}
    \hat{P}_g P_\chi &= \frac{1}{|G|} \sum_{h \in G} \bar{\chi}(h) U(h) \hat{P}_{hg} = \frac{1}{|G|} \sum_{\eta \in \hat{G}} \bar{\eta}(g) P_{\eta \chi} \hat{U}(\eta)\,,\\
    {P}_\chi \hat{P}_g& = \frac{1}{|G|} \sum_{h \in G} \chi(h) \hat{P}_{hg} U(h)  = \frac{1}{|G|} \sum_{\eta \in \hat{G}} \eta(g) \hat{U}(\eta) P_{\eta \chi} \,.
    \end{split}
\end{equation}

\medskip

To conclude this appendix, let us provide a characterization of the duality structure of Definition \ref{def:dual_reps_faithful} in terms of $*$-algebra representations. To this effect, we introduce a $*$-algebra $\mathcal{A}(G)$ which can be understood as a \emph{twisted} version of the standard group algebra $\ell^2(\hat{G} \times G)$. As a vector space, $\mathcal{A}(G)$ is the space of complex functions on $\hat{G} \times G$. As before, we identify an $f\in \mathcal{A}(G)$ to the formal sum of elements of the group
\begin{equation}
    f= \sum_{(\chi, g) \in \hat{G} \times G} f(\chi, g)\; (\chi, g)\,.
\end{equation}
We can then define a product as the unique bilinear map $\star:\mathcal{A}(G) \times \mathcal{A}(G) \to \mathcal{A}(G)$ such that: for any $(\chi_1 , g_1), (\chi_2 , g_2) \in \hat{G} \times G$,
\begin{equation}
    (\chi_1 , g_1)\star(\chi_2 , g_2) = \chi_2 (g_1) \;(\chi_1 \chi_2 , g_1 g_2) \,.
\end{equation}
Note that this product is \emph{noncommutative}. This makes $(\mathcal{A}(G), \star)$ an algebra, thanks to the following lemma.
\begin{lem}
    The product $\star$ is associative. Furthermore, for any $f_1 , f_2 \in \mathcal{A}(G)$, and any $(\chi, g) \in \hat{G} \times G$, we have
    \begin{equation}\label{eq:twisted_conv}
        f_1 \star f_2 (\chi , g) = \sum_{(\eta , h) \in \hat{G} \times G} f_1 (\chi \bar\eta , g h^{-1}) \, \eta(g h^{-1}) \, f_2 (\eta , h) \,.
    \end{equation}
\end{lem}
\begin{proof}
   Let $(\chi_1 , g_1), (\chi_2 , g_2), (\chi_3 , g_3) \in \hat{G} \times G$. We straightforwardly have
   \begin{align}
       \left( (\chi_1 , g_1)\star (\chi_2 , g_2)\right)\star (\chi_3 , g_3) &= \chi_2 (g_1) \, (\chi_1 \chi_2 , g_1 g_2) \star (\chi_3 , g_3) = \chi_2 (g_1) \chi_3 (g_1 g_2) \, (\chi_1 \chi_2 \chi_3 , g_1 g_2 g_3) \nonumber\\
       &= \chi_2 (g_1) \chi_3 (g_1) \chi_3 (g_2) \, (\chi_1 \chi_2 \chi_3 , g_1 g_2 g_3) = \chi_3 (g_2) \chi_2\chi_3 (g_1)  \, (\chi_1 \chi_2 \chi_3 , g_1 g_2 g_3) \nonumber \\
       &= \chi_3(g_2)\, (\chi_1 , g_1) \star (\chi_2 \chi_3 , g_2 g_3 )  = (\chi_1 , g_1) \star \left( \chi_3 (g_2) \,  (\chi_2 \chi_3 , g_2 g_3 ) \right) \nonumber \\
       &= (\chi_1 , g_1) \star \left( (\chi_2 , g_2 ) \star ( \chi_3 ,  g_3 ) \right)\,.
   \end{align}
   $\star$ being a bilinear map, it follows that for any $f_1 , f_2 , f_3 \in \mathcal{A}(G)$, $f_1 \star (f_2 \star f_3) = (f_1 \star f_2 ) \star f_3$.

   Moreover, for any $f_1 , f_2 \in \mathcal{A}(G)$ we have
   \begin{align}
       f_1 \star f_2 &= \sum_{(\chi , g) \in \hat{G} \times G} \sum_{(\eta , h) \in \hat{G} \times G} f_1(\chi , g) \, f_2(\eta , h)\; (\chi, g) \star (\eta , h) \\ 
       &= \sum_{(\chi , g) \in \hat{G} \times G} \sum_{(\eta , h) \in \hat{G} \times G} f_1(\chi , g) \, f_2(\eta , h)\eta(g)\; (\chi \eta, g h)  \\
       &= \sum_{(\chi , g) \in \hat{G} \times G} \sum_{(\eta , h) \in \hat{G} \times G} f_1(\chi \bar\eta , gh^{-1}) \, f_2(\eta , h)\eta(g h^{-1})\; (\chi , g ) \\
       &= \sum_{(\chi , g) \in \hat{G} \times G} \left( \sum_{(\eta , h) \in \hat{G} \times G} f_1(\chi \bar\eta , gh^{-1}) \, \eta(g h^{-1}) \, f_2(\eta , h) \right) \; (\chi , g )\,,
   \end{align}
   hence $f_1 \star f_2 (\chi , g) = \underset{{(\eta , h) \in \hat{G} \times G}}{\sum} f_1 (\chi \bar\eta , g h^{-1}) \, \eta(g h^{-1}) \, f_2 (\eta , h)$ for any $(\chi , g ) \in \hat{G} \times G$.
   \end{proof}
Formula \eqref{eq:twisted_conv} shows that $\star$ is a \emph{twisted convolution product}: were it not for the term $\eta(gh^{-1})$, we would indeed obtain the standard convolution product for the group $\hat{G} \times G$. 
We can also introduce an antilinear map $^*: \mathcal{A}(G) \to \mathcal{A}(G)$, which to any $f \in \mathcal{A}(G)$ associates the function: 
\begin{equation}
    f^* = \sum_{(\chi, g) \in \hat{G} \times G} \chi(g) \overline{f(\chi, g)}\; (\bar\chi, g^{-1})\,.
\end{equation}
In other words, we have: $f^*(\chi ,g) = \chi(g) \overline{f(\bar\chi, g^{-1})}$ for any $(\chi, g )\in \hat{G} \times G$. This makes $(\mathcal{A}(G) ,\star , ^*)$ a $*$-algebra thanks to the next lemma.
\begin{lem}
    The antilinear map $^*$ is an involution, and it is compatible with the product $\star$ in the sense that: for any $f_1 , f_2 \in \mathcal{A}(G)$,
    \begin{equation}
        \left( f_1 \star f_2 \right)^* = f_2^* \star f_1^* 
    \end{equation}
As a result, $(\mathcal{A}(G) ,\star , ^*)$ defines a $*$-algebra.
\end{lem}
\begin{proof}
Let us check that $^*$ is an involution. For any $f \in \mathcal{A}(G)$, $\chi \in \hat{G}$ and $g \in G$, we have
\begin{equation}
    \left(f^{*}\right)^{*} (\chi, g)= \chi(g) \overline{f^{*} (\bar\chi , g^{-1})} = \chi(g) \overline{\bar\chi(g^{-1}) \overline{f (\chi , g)} } = \chi(g) \bar\chi(g) f (\chi , g)  = f(\chi , g)\,.
\end{equation}

Let us then verify that the involution is compatible with the twisted convolution. Let $f_1 , f_2 \in \mathcal{A}(G)$ and $(\chi , g) \in \hat{G}\times G$. We have
\begin{align}
    \left( f_1 \star f_2 \right)^* (\chi, g) &= \chi(g) \overline{f_1 \star f_2 (\bar\chi , g^{-1})} = \chi(g)\sum_{(\eta, h) \in  \hat{G}\times G} \overline{f_1(\bar\chi \bar\eta , g^{-1}h^{-1})} \, \overline{\eta(g^{-1} h^{-1})} \, \overline{f_2(\eta , h)} \\
    &= \chi(g)\sum_{(\eta, h) \in  \hat{G}\times G}  \overline{\eta \chi}(hg) f_1^* ( \eta \chi , hg) \, \eta(h g) \, \eta(h^{-1})f_2^* (\bar\eta , h^{-1}) \\
    &= \sum_{(\eta, h) \in  \hat{G}\times G}   f_1^* ( \eta \chi , hg) \, \overline{\eta\chi}(h) \, f_2^* (\bar\eta , h^{-1}) = \sum_{(\eta, h) \in  \hat{G}\times G} f_1^* ( \eta , h ) \, \bar\eta(hg^{-1}) \, f_2^* (\chi \bar\eta , gh^{-1}) \\
    &= \sum_{(\eta, h) \in  \hat{G}\times G} f_2^* (\chi \bar\eta , gh^{-1}) \, \eta(gh^{-1}) \,  f_1^* ( \eta , h ) = \left( f_2^* \star f_1^* \right) (\chi , g)\,.
\end{align}
\end{proof}

We then have the following equivalence.

\begin{prop}\label{prop_bijection}
    Let $G$ be a finite Abelian group and $\cH_{kin}$ a finite-dimensional Hilbert space. There is a bijection between the following two sets.
    \begin{enumerate}
        \item The set of pairs $(\hat{U}, U)$ where: $U$ is a unitary representation of $G$ on $\cH_{\rm kin}$; $\hat{U}$ is a unitary representation of $\hat{G}$ on $\cH_{\rm kin}$; $\hat{U}$ and $U$ are dual to each other.
        \item The set of $*$-representations of $\mathcal{A}(G)$ on $\mathcal{B}(\cH_{\rm kin})$.
    \end{enumerate}
\end{prop}
\begin{proof}
    Let $(\hat{U},U)$ be a pair of dual representations, as defined in \emph{1}. Let $\mathcal{U}: \mathcal{A}(G) \to \cB(\cH_{\rm kin})$ be the linear map defined by 
    \begin{equation}
        \forall ( \chi , g ) \in \hat{G}\times G\,, \qquad \mathcal{U}((\chi,g)):= \hat{U}(\chi)U(g)\,. 
    \end{equation} 
    For any $(\chi , g),\, (\eta , h) \in \hat{G}\times G$, we have:
    \begin{align}
        \mathcal{U}\left((\chi,g)\right)\mathcal{U}((\eta,h)) &= \hat{U}(\chi)U(g)\hat{U}(\eta)U(h) = \eta(g)\, \hat{U}(\chi)\hat{U}(\eta) U(g)U(h) \\
        &= \eta(g)\, \hat{U}(\chi\eta) U(gh) = \eta(g)\, \mathcal{U}((\chi\eta , gh)) = \mathcal{U}( \eta(g) (\chi\eta , gh)) \\
        &= \mathcal{U}( (\chi,g) \star (\eta , h)) \,, 
    \end{align}
    where the duality between $\hat{U}$ and $U$ was used in the first line. Secondly, for any $(\chi , g)\in \hat{G}\times G$ we have:
    \begin{align}
        \mathcal{U}\left((\chi,g)\right)^\dagger &= \left( \hat{U}(\chi)U(g) \right)^\dagger = U(g)^\dagger \hat{U}(\chi)^\dagger = U(g^{-1}) \hat{U}(\bar\chi) \\   
        &= \bar\chi (g^{-1}) \, \hat{U}(\bar\chi) U(g^{-1}) = \chi(g) \, \mathcal{U}((\bar \chi , g^{-1})) = \mathcal{U}\left(\chi(g) \, (\bar \chi , g^{-1})\right) \\
        &= \mathcal{U}\left( (\chi, g)^* \right)\,,
        \end{align}
    where the duality was used in the second line. As a result, $\mathcal{U}$ is a $*$-representation of $\mathcal{A}(G)$. 
    
Reciprocally, given $\mathcal{U}$ a $*$-representation of $\mathcal{A}(G)$, we define $U: G \to {\rm Aut}(\cH_{\rm kin})$ and $\hat{U}: \hat{G} \to {\rm Aut}(\cH_{\rm kin})$ by
\begin{equation}
        \forall (\chi , g) \in \hat{G}\times G\,, \qquad U(g):=\mathcal{U}((1,g))\,, \qquad \hat{U}(\chi) := \mathcal{U}((\chi, e))\,. 
    \end{equation} 
    Then $U$ is a unitary representation of $G$ since, for any $g, h \in G$:
    \begin{equation}
        U(gh) = \mathcal{U}\left( (1,gh) \right) = \mathcal{U}\left( (1,g) \star (1, h) \right) = \mathcal{U}\left( (1,g) \right) \mathcal{U}\left( (1,{h}) \right) = U(g)U(h)\,,
    \end{equation}
    and
    \begin{equation}
        U(g)^\dagger = \mathcal{U}\left( (1,g) \right)^\dagger = \mathcal{U}\left( (1,g)^* \right) = \mathcal{U}\left( (1,g^{-1}) \right) = U(g^{-1})\,. 
    \end{equation}
    Likewise, $\hat{U}$ is a unitary representation of $\hat{G}$ since, for any $\chi, \eta \in \hat{G}$:
    \begin{equation}
        \hat{U}(\chi \eta) = \mathcal{U}\left( (\chi \eta,e) \right) = \mathcal{U}\left( (\chi,e) \star (\eta, e) \right) = \mathcal{U}\left( (\chi,e) \right) \mathcal{U}\left( (\eta,e) \right) = \hat{U}(\chi){\hat{U}}(\eta)\,,
    \end{equation}
    and
    \begin{equation}
        \hat{U}(\chi)^\dagger = \mathcal{U}\left( (\chi,e) \right)^\dagger = \mathcal{U}\left( (\chi,e)^* \right) = \mathcal{U}\left( (\bar\chi,e) \right) = \hat{U}(\bar\chi)\,. 
    \end{equation}
Finally, $U$ and $\hat{U}$ are dual to each other since, for any $g\in G$ and $\chi \in \hat{G}$:
\begin{equation}
\begin{split}
    U(g)\hat{U}(\chi)&=\mathcal{U}((1,g))\,\mathcal{U}((\chi,e))=\mathcal{U}((1,g)\star(\chi,e))=\mathcal{U}(\chi(g)(\chi,g))\\
    &=\chi(g)\,\mathcal{U}((\chi,g))=\chi(g)\,\mathcal{U}((\chi,e)\star(1,g))=\chi(g)\,\hat{U}(\chi)U(g)\,.
    \end{split}
\end{equation}
This establishes the claimed bijection. 
\end{proof}

\section{Representation-theoretic analysis of stabilizer codes}\label{app:stabilizer_reps}

Let us apply the representation-theoretic tools recalled in Appendices \ref{app:Pontryagin} and \ref{app:reps} to stabilizer codes. We fix some integer $n \in \mathbb{N}^*$ and a kinematical Hilbert space $\cH_{\rm kin}= (\mathbb{C}^2)^{\otimes n}$ equipped with its standard inner product $\langle \cdot \vert \cdot \rangle$. Let also $\mathcal{P}_n$ denote the Pauli group canonically acting on $\cH_{\rm kin}$. A \emph{stabilizer code} encoding $k \in \mathbb{N}^{*}$ ($k < n$) logical qubits can be defined and analyzed in terms of the representation theory of the group $G = \mathbb{Z}_2^{\times (n-k)}$. More precisely, a code is fully specified by a choice of \emph{faithful} \emph{unitary} representation $U:G \to {\rm Aut}(\cH_{\rm kin})$, such that $\mathcal{G}:= {\rm Im}(U) \subset \mathcal{P}_n \setminus \{ - I \}$. We first observe that the two technical assumptions we make about $U$ have a simple $*$-algebra interpretation. 
\begin{lem}\label{lem:faithful_alg_rep}
    Let $U:G \to {\rm Aut}(\cH_{\rm kin})$ a unitary representation such that ${\rm Im}(U)\subset \cP_n$, and $\uU: \ell^2(G) \to \mathcal{B}(\cH_{\rm kin})$ its associated $*$-algebra representation. Then the following two propositions are equivalent:
    \begin{enumerate}
        \item $U$ is faithful (as a group representation) and $-I \notin {\rm Im}(U)$;
        \item $\uU$ is faithful (as a $*$-algebra representation).
    \end{enumerate}
\end{lem}
\begin{proof}
    Suppose that $\uU$ is faithful. Then $U$ is also faithful (see Appendix \ref{app:reps}). Furthermore, we necessarily have $-I \notin {\rm Im}(U)$ otherwise we could find some $g \in G\setminus \{e\}$ such that $U(g)= - I$, which would imply that $\uU(g +e)= U(g)+ I = 0$, and therefore that $\uU$ is not faithful. 

    Reciprocally, suppose that $U$ is faithful and $- I \notin {\rm Im}(U)$. Then, for any $g \in G \setminus \{ e\}$, we have $U(g) \notin \{I , -I , i I , - i I\}$; indeed: $I$ is excluded since $U$ is faithful; $-I$ is excluded by assumption; $i I$ and $-i I$ are excluded otherwise their square $- I$ would also be in ${\rm Im}(U)$. Given that any Pauli matrix has vanishing trace, this implies that $\Tr(U(g))=0$ for any $g \neq e$. To prove that $\uU$ is faithful, consider $f \in \ell^2(G)$ such that $\uU(f) = 0$. This translates into the condition 
    \begin{equation}
        \sum_{h \in G} f(h) \, U(h) = 0\,.
    \end{equation}
    For any $g \in G$, we obtain after multiplying the previous equation by $U(g)^\dagger$ and taking its trace:
    \begin{equation}
        \sum_{h \in G} f(h) \Tr(U(hg^{-1})) = 0 \quad \Rightarrow \quad f(g) \Tr(I) =0 \quad \Rightarrow \quad f(g)=0\,.
    \end{equation}
 It results that $f =0$, from which we conclude that $\uU$ is faithful.
\end{proof}

From now on, let $U$ be as in the previous lemma, so that $\uU$ is a faithful $*$-algebra representation of $\ell^2(G)$. We then have the orthogonal decomposition of $\cH_{\rm kin}$ into isotypes of $U$ (see equation \eqref{eq:decomp_isotypes}):
\begin{equation}\label{eq:decomp_isotypes-stab}
    \cH_{\rm kin} = \bigoplus_{\chi \in \hat{G}} \cH_\chi \,,
\end{equation}
where $\chi$ labels the $2^{n-k}$ inequivalent irreducible representations of $G= \mathbb{Z}_2^{\times (n-k)}$ \cite{serre1977linear, simon1996representations}. The isotype associated to the trivial character $\chi=1$ is the invariant subspace of $\cH_{\rm kin}$. In other words, it is the \emph{logical} or \emph{perspective-neutral} Hilbert space:
\begin{equation}\label{eq:H_trivial_code_pn}
    \cH_1 = \cH_{\rm code} = \cH_{\rm pn}\,.
\end{equation} 
Given $\chi \in \hat{G} \simeq \mathbb{Z}_2^{\times (n-k)}$, the orthogonal projector onto $\cH_\chi$ (see equation \eqref{eq:projector_isotype}) is given by
\begin{equation}\label{eq:projectors_on_isotypes}
    P_\chi:= \frac{1}{2^{n-k}} \sum_{g \in \mathbb{Z}_2^{\times (n-k)}} \chi(g) \, U(g) \,,
\end{equation}
where we used the fact that $\chi(g) \in \{ +1 , -1\}$ is real for any $g \in \mathbb{Z}_2^{\times (n-k)}$. Note that this provides a generalization of the coherent group-averaging projector to arbitrary charge sectors; indeed, we have $P_1=\Pi_{\rm pn}$. 

Let us also recall some useful facts from character theory \cite{serre1977linear, simon1996representations}. First, we have a canonical inner-product $(\cdot , \cdot)$ on the vector space of class functions on $G=\mathbb{Z}_2^{\times (n-k)}$, defined as follows: for any two complex class functions $f_1 , f_2$ on $G$, 
\begin{equation}\label{eq:inner_class}
    (f_1 , f_2) := \frac{1}{|G|} \sum_{g \in G} \overline{f_1 (g)} f_2 (g)\,.
\end{equation}
The irreducible characters of $G$, that is to say elements of $\hat{G}$, form an orthonormal basis of the space of class functions. We have in particular the orthogonality relations
\begin{equation}\label{eq:orthogonality_rel_charac}
    \forall \chi, \eta \in \hat{G}\,, \qquad (\chi , \eta) = \frac{1}{2^{n-k}} \sum_{g \in G} \chi(g) \eta(g)= \delta_{\chi , \eta}\,. 
\end{equation}
The character $\chi_U$ of the representation $U$ is the class function defined by $\chi_U(g):= \Tr(U(g))$ for any $g \in G$. We can decompose it in the orthonormal basis $\{ \chi \}_{\chi \in \hat{G}}$, from which one can infer that:
\begin{equation}\label{eq:dim_charac}
  \forall \chi \in \hat{G}\,,\qquad \dim(\cH_\chi) = (\chi_U , \chi)\,.   
\end{equation}
Finally, the orthogonality relations \eqref{eq:orthogonality_rel_charac} applied to the dual group $\hat{G}$ lead to the dual relations:
\begin{equation}\label{eq:orthogonality_rel_dual}
    \forall g, h \in G\,, \qquad  \frac{1}{2^{n-k}} \sum_{\chi \in \hat{G}} \chi(g) \chi(h)= \delta_{g , h}\,. 
\end{equation}

We first observe that, given our assumptions, all the isotypes of $U$ are non-trivial and of equal dimension. 
\begin{lem}\label{lemma:dim_Hj}
For any $\chi \in \hat{G}$, $\dim (\cH_\chi) = 2^k$.
\end{lem}
\begin{proof}
    Let $\chi \in \hat{G}$. From \eqref{eq:dim_charac}, we know that 
    \begin{equation}
        \dim(\cH_\chi) = (\chi_U , \chi) = \frac{1}{2^{n-k}} \sum_{g \in G} \overline{\chi_U (g)} \chi (g)\,,
    \end{equation}
where $(\cdot,\cdot)$ is the standard inner-product defined in \eqref{eq:inner_class}. As already explained in the proof of Lemma \ref{lem:faithful_alg_rep}, the fact that $U$ is faithful and $- I \notin {\rm Im}(U)$ implies that $\chi_U(g) = {\rm tr} (U(g)) = 0$ unless $g=e$. As a result, we have:
\begin{equation}
    \dim(\cH_\chi) = \frac{1}{2^{n-k}} \underbrace{\overline{\chi_U(e)}}_{2^n} \underbrace{\chi(e)}_{1} = 2^k\,.
\end{equation}
\end{proof}
Next, let us determine how Pauli operators, which model errors, interact with the decomposition of the Hilbert space $\cH_{\rm kin}$. To this effect, we first observe that the $*$-algebra $\mathcal{B}(\cH_{\rm kin})$ can be decomposed into isotypes of the adjoint representation associated to $U$.
\begin{lem}\label{lem:decomp_algebra}
    Let us endow $\cB(\cH_{\rm kin})$ with its standard Hilbert-Schmidt inner product,\footnote{That is, for any $A,B \in \cB(\cH_{\rm kin})$, we set $\langle A \vert B \rangle := \Tr\left( A^\dagger B\right)$.} which makes it a Hilbert space. We have the orthogonal decomposition
    \begin{equation}
        \mathcal{B}(\cH_{\rm kin}) = \bigoplus_{\chi \in \hat{G}} \mathcal{B}_\chi\,,
    \end{equation}
    where:
\begin{equation}
    \forall \chi \in \hat{G}\,, \qquad \cB_\chi := \{ \varphi \in \cB(\cH_{\rm kin}) \, \vert \, \forall g \in G\,, \; U(g) \varphi U(g)^\dagger = \chi(g) \varphi \}\,.
\end{equation}
Furthermore, for any $\chi \in \hat{G}$, $\cB_\chi$ is a $*$-subalgebra of $\cB(\cH_{\rm kin})$.
\end{lem}
\begin{proof}
     The adjoint action $\cB(\cH_{\rm kin}) \ni \varphi \mapsto U(g) \triangleright \varphi := U(g) \varphi U(g)^\dagger$ ($g \in G$) defines a \emph{unitary representation} of $G$ on $\cB(\cH_{\rm kin})$. We can thus decompose $\cB(\cH_{\rm kin})$ into orthogonal isotypes $\cB_\chi$ labeled by $\chi \in \hat{G}$ (see equation \eqref{eq:decomp_isotypes}). For any $\chi \in \hat{G}$, $\cB_\chi$ is a union of one-dimensional representations isomorphic to $\chi$; in other words, $\cB_\chi$ is the subspace of operators $\varphi$ verifying $U(g)\triangleright\varphi = \chi(g) \varphi$ for any $g\in G$. Furthermore, it is clear that this relation is stable under product, so: $\varphi_1 \varphi_2 \in \cB_\chi$ whenever $\varphi_1 \in \cB_\chi$ and $\varphi_2 \in \cB_\chi$. Owing to the fact that $\chi(g)$ is real for any $g \in G$, $\cB_\chi$ is also stable under the adjoint $^\dagger$.\footnote{Indeed, for any $\varphi \in \cB_\chi$ and $g\in G$, we have $U(g) \triangleright \varphi^\dagger = \left( U(g) \triangleright \varphi\right)^\dagger = \left( \chi(g) \varphi\right)^\dagger = \bar\chi(g) \varphi^\dagger = \chi(g) \varphi^\dagger$.} Hence, $\cB_\chi$ is a $*$-subalgebra of $\cB(\cH_{\rm kin})$.
\end{proof}

As for the decomposition of $\cH_{\rm kin}$, character theory allows us to determine the dimension of each isotype of $\cB(\cH_{\rm kin})$.
\begin{lem}\label{lemma:dim_B}
For any $\chi \in \hat{G}$, $\dim (\cB_\chi) = 2^{n+k}$.
\end{lem}
\begin{proof}
Given that $\cB(\cH_{\rm kin}) \simeq \cH_{\rm kin}^* \otimes \cH_{\rm kin}$, the character $\eta_U$ of the adjoint representation obeys:
    \begin{equation}
        \forall g \in G\,, \qquad \eta_U (g)= \bar\chi_U (g)\chi_U(g) = \vert \chi_U(g) \vert^2\,.
    \end{equation}
  Reasoning as in Lemma \ref{lemma:dim_Hj} above, we have for any $\chi \in \hat{G}$
    \begin{equation}
        \dim(\cB_\chi) = (\eta_U , \chi) = \frac{1}{2^{n-k}} \underbrace{\vert \chi_U (e) \vert^2}_{(2^n)^2} \underbrace{\chi(e)}_{1} = 2^{n+k}\,. 
    \end{equation}
    Hence, all the isotypes are non-trivial and of the same dimension.     
\end{proof}
A crucial observation is that any element from the Pauli group $\mathcal{P}_n$ lies in one of the isotypes of $\mathcal{B}(\cH_{\rm kin})$, and this fully determines how errors act on the code subspace. 
\begin{prop}\label{prop:decomp_Pauli}
    The Pauli set $\mathcal{P}_n$ decomposes as the disjoint union
      \begin{equation}
    \mathcal{P}_n = \bigsqcup_{\chi \in \hat{G}} C_\chi\,, \qquad {\rm where:}\qquad \forall \chi \in \hat{G}\,,\quad  C_\chi := \cP_n \cap \cB_\chi \,.
\end{equation}
Moreover, we have the following properties: for any $\chi , \eta \in \hat{G}$ and $E_\chi \in C_\chi$
\begin{enumerate}
    \item $\Span(C_\chi ) = \cB_\chi$ and $\dim( \Span(C_\chi ) ) = 2^{n+k}$ ;
    \item $E_\chi (\cH_\eta ) = \cH_{\chi\eta }$ ;
    \item $E_\chi P_\eta E_\chi = P_{\chi\eta}$ .
\end{enumerate}
\end{prop}
\begin{proof}
Let $E\in \mathcal{P}_n$. Notice that any two Pauli operators either commute or anti-commute. Hence, for any $g \in G$, one has either $U(g)E = E U(g)$ or $U(g)E=-E U(g)$, so $U(g) \triangleright E = \pm E$. It follows that $E$ generates a one-dimensional stable subspace of the adjoint representation, which is thus irreducible and necessarily lies in one of the isotypes of $\mathcal{B}(\cH_{\rm kin})$. As a result, $\mathcal{P}_n = \underset{\chi \in \hat{G}}{\bigsqcup} C_\chi$.

Next, using that $\Span (\cP_n) = \cB(\cH_{\rm kin})$ together with Lemma \ref{lem:decomp_algebra}, we obtain:
\begin{equation}
    \bigoplus_{\chi \in \hat{G}} \cB_\chi = \bigoplus_{\chi \in \hat{G}} \Span(C_\chi)\,.
\end{equation}
In addition, we have $\Span(C_\chi) \subset \cB_\chi$ for any $\chi \in \hat{G}$, so the previous identity can only hold if $\Span(C_\chi) = \cB_\chi$ for any $\chi \in \hat{G}$. Lemma \ref{lemma:dim_B} implies in particular that $\dim( \Span(C_\chi ) ) = 2^{n+k}$ for any $\chi \in \hat{G}$, which establishes \emph{1}.

Let $\chi, \eta \in \hat{G}$, and $E_\chi \in C_\chi$. For any $v \in \cH_\eta$ we have:
\begin{equation}
\forall g \in G\,, \qquad    U(g) E_\chi v = \underbrace{U(g)E_\chi U(g)^\dagger}_{=\chi(g)E} \underbrace{U(g)v}_{= \eta(g)v} = \chi\eta (g) \, E_\chi v\,,
\end{equation}
hence $E_\chi v \in \cH_{\chi \eta}$. We thus have $E_\chi(\cH_\eta)\subset \cH_{\chi \eta}$, and given that $E_\chi$ is unitary, $E_\chi(\cH_\eta)= \cH_{\chi \eta}$. This establishes \emph{2}. Finally, we can employ the explicit expression \eqref{eq:projectors_on_isotypes} to write: 
\begin{align}
    E_\chi P_\eta E_\chi &= \frac{1}{2^{n-k}}\sum_{g \in G} \eta(g) E_\chi U(g) E_\chi = \frac{1}{2^{n-k}}\sum_{g \in G} \eta(g) E_\chi \underbrace{U(g) E_\chi U(g)^\dagger}_{= \chi(g) E_\chi} U(g) \nonumber  \\ &= \frac{1}{2^{n-k}}\sum_{g \in G} \chi\eta(g)  \underbrace{(E_\chi)^2}_{= I} U(g) = P_{\chi \eta}\,. 
\end{align}
This establishes \emph{3}.
\end{proof}

In view of those results, we can call an element $E_\chi \in C_\chi$ ($\chi \in \hat{G}$) an \emph{error of type $\chi$}. Indeed, any such $E_\chi$ isometrically maps the code subspace $\cH_{\rm code}= \cH_1$ to the \emph{error subspace} or \emph{charged sector} $\cH_\chi$.  

Let us focus on $\chi=1$ as a special case. By definition, $C_1$ is nothing but the centralizer of $\mathcal{G}$ in $\mathcal{P}_n$: $C_1 = C_{\mathcal{P}_n}(\mathcal{G})$. Alternatively, by Proposition \ref{prop:decomp_Pauli}, we can characterize $C_1$ as the subset of Pauli operators that preserve the code subspace. Such operators can therefore be used to implement logical operations. More precisely, we have
\begin{equation}
    \Span (C_1) = \cB_1 = \{ \varphi \in \mathcal{B}(\cH_{\rm kin}) \, \vert \, \forall g \in G\,, [U(g) , \varphi] = 0\}
\end{equation}
which is equivalent to 
\begin{equation}
\Span (C_{\mathcal{P}_n}(\mathcal{G})) = C_{\mathcal{B}(\cH_{\rm kin})} (\mathcal{G}) \; \left(= C_{\Span(\cP_n)} (\mathcal{G}) \right)\,.
\end{equation}
In other words, \emph{logical operations are represented by gauge-invariant operators on $\cH_{\rm kin}$}. Furthermore, two operators $\varphi , \tilde{\varphi} \in \Span(C_1) = \cB_1$ are logically equivalent ($\varphi \sim \tilde{\varphi}$) if and only if they coincide on the code subspace, namely: 
\begin{equation}
    \varphi \sim \tilde{\varphi} \qquad \Leftrightarrow \qquad 
\restr{\varphi}{\cH_{1}} = \restr{\tilde{\varphi}}{\cH_{1}}
    \,.
\end{equation}
Since any element of $\Span(C_1)=\cB_1$ preserves any isotype $\cH_\chi$ of $\cH_{\rm kin}$ (by Proposition \ref{prop:decomp_Pauli}), $\Span(C_1)=\cB_1$ can be identified with a subspace of
\begin{equation}
     \bigoplus_{\chi \in \hat{G}} \mathcal{B}(\cH_\chi)\,.
\end{equation}
Owing to Lemma~\ref{lemma:dim_Hj} we have $\dim (\mathcal{B}(\cH_\chi)) = 2^k \times 2^k$ for any $\chi \in \hat{G}$, and $\dim( \Span{C_1}) = 2^{n+k}$ by Proposition~\ref{prop:decomp_Pauli}. Hence, by dimensionality, we necessarily have 
\begin{equation}
\Span(C_1) = \cB_1 \simeq \bigoplus_{\chi \in \hat{G}} \mathcal{B}(\cH_\chi)\,.
\end{equation}
Taking the quotient by the equivalence relation $\sim$, we conclude that
\begin{equation}\label{eq:logical_operations}
\faktor{\Span(C_1)}{\sim} = \faktor{\cB_1}{\sim} \simeq \cB(\cH_1)\,,
\end{equation}
which we identify with the vector space of logical operators. To summarize, \emph{logical operators are nothing but gauge-invariant operators on $\cH_{\rm kin}$, considered modulo equivalence on $\cH_1$}. 

This observation can be generalized to an arbitrary $\eta \in \hat{G}$ as follows. Thanks to Proposition \ref{prop:decomp_Pauli}, we can establish the isomorphism 
\begin{equation}
\Span(C_\eta)= \cB_\eta \simeq     \bigoplus_{\chi \in \hat{G}} \cB(\cH_\chi, \cH_{\chi \eta})\,,
\end{equation}
where, for any $\chi \in \hat{G}$, $\mathcal{B}(\cH_\chi, \cH_{\chi\eta})$ denotes the vector space of linear operators from $\cH_\chi$ to $\cH_{\chi\eta}$. Taking the quotient by $\sim$, we then obtain
\begin{equation}\label{eq:errors_modulo_eq}
\faktor{\Span(C_\eta)}{\sim} = \faktor{\cB_\eta}{\sim} \simeq \mathcal{B}(\cH_1, \cH_\eta)\,,
\end{equation} 
We can interpret this quotient as the vector space of \emph{errors of type $\eta$}, which are also, in gauge-theoretic language, \emph{creation-operators carrying a charge $\eta$}. 

\section{No useful blanket recovery operation for gauge fixing errors}\label{app_nooperation}

Here, we briefly elucidate why neither $\sqrt{2^{n-k}}\,\Pi_{\rm pn}$ nor $\Pi_{\rm pn}$ constitute an operationally useful blanket recovery for arbitrary gauge fixing errors, despite Eq.~\eqref{KLarbitrary}. 

Let us begin with $\sqrt{2^{n-k}}\,\Pi_{\rm pn}$ which is neither a projector nor a unitary and thus cannot be implemented in a quantum operation. To see this, recall that a quantum operation can be written as
\begin{equation}
    \mathcal{O}(\rho)=\sum_i K_i\rho K_i^\dag\,,\qquad\qquad\qquad
    \sum_i K_i^\dag K_i\leq I.
\end{equation}
Now, suppose that $K_1=\sqrt{2^{n-k}}\,\Pi_{\rm pn}$. Then we must have
\begin{equation}
    0\leq\sum_{i\geq2}K_i^\dag K_i= I-2^{n-k}\,\Pi_{\rm pn}\,,
\end{equation}
which, however, is unattainable because $I-2^{n-k}\,\Pi_{\rm pn}$ is not a non-negative operator; indeed, any code state is an eigenstate with eigenvalue $1-2^{n-k}<0$ for $n>k$. This also means that there can be no valid linear combination $K'_1=K_1+\sum_{j\neq2}c_j K_j$ without violating the normalization.

There is, however, a \emph{trace-decreasing} quantum operation that does correct the errors, though not deterministically:
\begin{equation}
    \mathcal{O}(\rho)=\Pi_{\rm pn}\rho\Pi_{\rm pn}.
\end{equation}
If $\rho\in\mathcal{S}(\mathcal{H}_{\rm code})$, then $\mathcal{O}(\rho)=\rho$, leaving the trace invariant. If, on the other hand, $\rho'=\hat{E}_g\,\rho\,\hat{E}_g^\dag$, then by virtue of Eq.~\eqref{KLarbitrary}
\begin{equation}
    \mathcal{O}(\rho')=\frac{1}{2^{n-k}}\,\rho\,,
\end{equation}
so that the trace has decreased. (The same conclusion holds for the full error channel $\rho'=1/2^{n-k}\sum_g \hat{E}_g\,\rho\,\hat{E}_g^\dag$.)

Of course, there is also a valid projective measurement channel, which is trace-preserving and contains the previous operation:
\begin{equation}
    \rho\mapsto \Pi_{\rm pn}\rho\Pi_{\rm pn}+\Pi_{\rm pn}^\perp\,\rho\,\Pi_{\rm pn}^\perp
\end{equation}
with $\Pi_{\rm pn}^\perp=I-\Pi_{\rm pn}$. 
But, assuming $\braket{\psi}{\psi}_{\rm pn}=1$, we have thanks to Eq.~\eqref{KLarbitrary}
\begin{equation}
   \bra{\psi}\hat{E}_g\,\Pi_{\rm pn}\,\hat{E}_g^\dag\,\ket{\psi}_{\rm pn}=\frac{1}{2^{n-k}}\,.
\end{equation}
This projective measurement channel is thus a non-deterministic one that corrects a gauge fixing error with probability $1/2^{n-k}$ and this can be arbitrarily close to zero for large $n-k$.

We conclude that the blanket recovery $\sqrt{2^{n-k}}\,\Pi_{\rm pn}$ or $\Pi_{\rm pn}$ is operationally not usefully implementable.

\section{Proofs}\label{App:proofs}

In this appendix we collect the proofs of some statements from the main body. \\

\noindent\textbf{Lemma~\ref{lem:PicodeasGA} (Equality of the code and perspective-neutral projectors).}
\emph{The orthogonal projector $\Pi_{\rm code}$ onto $\mathcal \Hil_{\rm code}\subset\Hil_{\rm physical}$ can be written as a coherent group averaging projector over the stabilizer group.
That is,
\begin{equation}\label{eq:app:stabilizerGA}
    \Pi_{\rm code}=\frac{1}{|G|}\sum_{g\in G}U^g \equiv \Pi_\mrm{pn}\;,
\end{equation}
with $U^g$, $g\in G=\mathbb Z_2^{\times(n-k)}$, a unitary representation of the stabilizer group on $\Hil_{\rm physical}$.
}

\begin{proof}
Let $\Pi_{\rm pn}=\frac{1}{|G|}\sum_{g\in G}U^g$.~Clearly, $\Pi_{\rm pn}^{\dagger}=\Pi_{\rm pn}$ and $\Pi_{\rm pn}\circ\Pi_{\rm pn}=\Pi_{\rm pn}$, so $\Pi_{\rm pn}$ is an orthogonal projector.~Moreover, $\Pi_{\rm pn}\ket{\psi_{\rm code}}=\ket{\psi_{\rm code}}$ for any $\ket{\psi_{\rm code}}\in\Hil_{\rm code}$, that is $\Hil_{\rm code}\subset \text{Im}(\Pi_{\rm pn})$.~Similarly, for any $\ket{\psi}\in\Hil_{\rm physical}$ and $U^g$, $g\in\mathcal G$, we have $U^g\Pi_{\rm pn}\ket{\psi}=\Pi_{\rm pn}\ket{\psi}$, that is $\text{Im}(\Pi_{\rm pn})\subset \Hil_{\rm code}$.~Therefore, $\text{Im}(\Pi_{\rm pn})=\Hil_{\rm code}$.
\end{proof}

\noindent\textbf{Lemma~\ref{lemma:URexistence} (Existence of a faithful representation of $G$ on a subset of physical qubits).}
\emph{Let $0 \leq k \leq n$ be fixed.
    Let $\mathcal{G} = \{U^g ~ | ~ g \in G = \mathbb{Z}_2^{\times(n-k)}\}$, where $U^e = I^{\otimes n}$, be a Pauli stabilizer group on $\Hil \simeq (\mathbb{C}^2)^{\otimes n}$, i.e., $-I \notin \mathcal{G}$ and $\forall~g,h\in G$: $U^g \in \mathcal{P}_n$, $[U^g,U^h] = 0$, $U^g U^h = U^{gh} \in \mathcal{G}$, and $U^g = U^h$ if and only if $g = h$.
    Then, there exists a choice of $A \subseteq \{1, \dots, n\}$ consisting of $k$ distinct elements such that $\pi_A(U^g) = I^{\otimes n}$  if and only if $g=e$, where $\pi_A = \prod_{a \in A} \pi_a$ and with
    $$
    \begin{aligned}
        \pi_a ~ : ~ &\mathcal{P}_n \to \mathcal{P}_n \\
        & i^\lambda \, \Oh_1 \cdots \Oh_{a-1} \Oh_a \Oh_{a+1} \cdots \Oh_n ~ \mapsto ~ \Oh_1 \cdots \Oh_{a-1} I \Oh_{a+1} \cdots \Oh_n
    \end{aligned}
    $$
    denoting the operator that sets the $a^\text{th}$ letter of an $n$-qubit Pauli string to the identity and removes any prefactor that is not 1.}

\begin{proof}
Let us start by recalling that the Pauli group up to phases can be viewed as a $GF(2)$-vector space as follows (see, e.g., \cite{gottesman2006quantum}).
An $n$-qubit Pauli string is mapped onto an element of $GF(2)^{2n}$ by concatenating two length-$n$ strings of zeros and ones.
The entries of the first string are $0$ when the corresponding letter of the Pauli string is $I$ or $Z$ and $1$ when the corresponding letter is $X$ or $Y$, and the entries of the second string are essentially inverted: $0$ for $I$ or $X$ and $1$ for $Z$ or $Y$.
For example, a Pauli string $XIIZY$ has a corresponding vector representation $(10001 \, | \, 00011)$ (it is customary to separate the two substrings with a vertical bar).
Multiplication of two Pauli strings \emph{modulo the phase} therefore corresponds to addition of vectors in $GF(2)^{2n}$ and scalar multiplication is defined over the field $GF(2)$. That $-I \notin \mathcal{G}$ is important for the vector representation of the stabilizer group to be faithful.
In particular, this prevents $U^g$ and $\pm i U^g$, which would have the same representations as vectors in $GF(2)^{2n}$, from simultaneously being in $\mathcal{G}$ (on account of appropriate powers of them multiplying to $-I$).

Lifted to the representation on $GF(2)^{2n}$, the projection $\pi_A$ for $A \subset \{1, \dots, n\}$ sets the $i^\text{th}$ and $(n+i)^\text{th}$ entries of the vector $\mbf{u}(g)$ representing $g$ to $0$ if $i \in A$. Here we abuse notation and use $\mathcal{G}$ to also refer to the representation of $G = \mathbb{Z}_2^{\times(n-k)}$ on $GF(2)^{2n}$, and we let $\pi_A$ directly act on this representation as just described, too.~In this language, Lemma~\ref{lemma:URexistence} can be restated as follows.
\begin{lem} \label{lem:Rchoice}
    Let $\mathcal{G} = \{\mbf{u}(g)\}$ be a faithful representation of $G = \mathbb{Z}_2^{\times (n-k)}$ on the vector space $GF(2)^{2n}$ as described above.~There exists a choice of $A \subset \{1, \dots, n\}$ with $|A| = k$ such that $\{\pi_A \mbf{u}(g)\}$ is still a faithful representation of $G$ on $GF(2)^{2n}$.
\end{lem}
Note that it immediately follows that, if $\{\mbf{u}(g)\}$ is a representation of $G$, then $\{\pi_A \mbf{u}(g)\}$ is also a representation of $G$ (but possibly unfaithful) for any $A$ since $\pi_A$ is a linear operator on $GF(2)^{2n}$ and since the group operation is represented by addition of vectors.~Furthermore, we note that in this language, the vectors in a set $\{\mbf{v}^1, \dots, \mbf{v}^m\} \subset GF(2)^{2n}$ are linearly independent if the unique solution to $\alpha_1 \mbf{v}^1 + \cdots + \alpha_m \mbf{v}^m = 0$ with $\alpha \in GF(2)^m$ is $\alpha = 0$.

We shall now establish a technical result that forms the crux of the Lemma~\ref{lem:Rchoice}.
\begin{lem}\label{lem:technicalZ2}
    Let $n \geq 2$ and $1 \leq k \leq n-1$. Let $N = \binom{n}{k}$, and let $\mathcal{A} = \{A_\mu \subset \{1, \dots, n\}, |A_\mu| = k\}_{\mu=1}^N$ be the list of the distinct $k$-index subsets of $\{1, \dots, n\}$, ordered according to some order that we choose.~Let $M$ be a $2n \times N$ matrix with entries in $GF(2)$ and with the following properties:
    \begin{enumerate}
        \item $[M]_{(2i-1)\mu} = [M]_{(2i)\mu} = 0$ whenever $i \notin A_\mu$ for $i = 1, \dots, n$.
        \item In each column $\mu$, there is at least one index $i$ such that either or both $[M]_{(2i-1)\mu} \neq 0$, $[M]_{(2i)\mu} \neq 0$.
    \end{enumerate}
    Then $M$ has at least $n-k+1$ linearly independent columns in $GF(2)^{2n}$.
\end{lem}

\begin{proof} For future use, we begin by making the following observations about $M$.
The assumptions imply that in any given column, there are at least $2(n-k)$ zeros and at most $2n-1$ zeros.
For any given $i \in \{1, \dots, n\}$, are at least $\binom{n-1}{n-k-1} = \binom{n-1}{k}$ columns in which $[M]_{(2i-1)\mu} = [M]_{(2i)\mu} = 0$ is guaranteed.

We prove the claim by induction.
As a base case, let $n = 2$ and $k = 1$.
Up to exchanging its columns, we may uniquely write $M$ as
\begin{equation}
    M = \left( \begin{array}{cc}
    x_1 & 0 \\
    y_1 &  0 \\
    0 & x_2 \\
    0 & y_2
    \end{array} \right)
\end{equation}
where $x_i$ and $y_i$ cannot both be zero for each $i \in \{1, 2\}$.
The two columns are clearly linearly independent, and so $\mathrm{rank}\, M = 2$.

Now, suppose that the claim holds for $n \geq 2$ and $1 \leq k \leq n-1$.
Let $N' = \binom{n+1}{k'}$ for $1 \leq k' \leq n$ and let $M'$ be a $2(n+1)\times N'$ matrix with the assumed properties.
We must show that $\mathrm{rank}\, M' \geq n - k' +2$.

Choose the order on $\mathcal{A}$ such that $[M']_{(2n+1)\mu} = [M']_{(2n+2)\mu} = 0$ is guaranteed for $\mu = 1, \dots, N$, where here $N = \binom{n}{k'}$.
Consequently, for $1 \leq k' \leq n-1$, the submatrix $M_L \equiv [M']_{1:2n,1:N}$ satisfies the inductive assumptions.\footnote{The notation $[B]_{i:j,\mu:\nu}$ denotes the submatrix consisting of rows $i$ through $j$ and columns $\mu$ through $\nu$ of $B$.}
(We will treat $k' = n$ as a special case that we will come back to later.)
\begin{figure}[h!]
    \centering
    \includegraphics[scale=0.6]{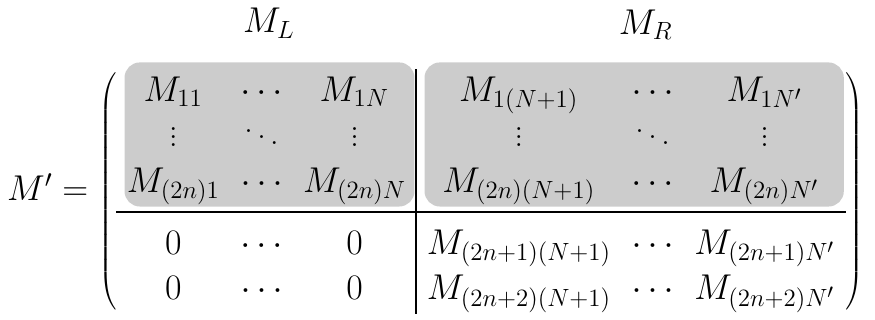}
\end{figure}

\noindent To see this, note that each column of $M_L$ contains at least $2(n-k')$ zeros and at most $2n-1$ zeros.
Furthermore, the columns can be indexed by the distinct subsets $A_\mu$ of $\{1, \dots, n\}$ that specify (or rather, whose complements specify) the entries of $M_L$ that necessarily must be zero. 
It follows that $M_L$, and hence also $M'$, have rank at least $n-k'+1$.

If, in the last $\binom{n}{k'-1}$ columns of $M'$, there is at least one column $\mu$ such that $[M']_{(2n+1)\mu} \neq 0$ or $[M']_{(2n+2)\mu} \neq 0$, then $\mathrm{rank}\, M' \geq n-k'+2$ as required.
Otherwise, if $[M']_{(2n+1)\mu} = [M']_{(2n+2)\mu} = 0$ for all $1 \leq \mu \leq N'$, then the submatrix $M_R \equiv [M']_{1:2n,(N+1):N'}$ satisfies the inductive assumptions for $1 \leq k'-1 \leq n-1$, or equivalently $2 \leq k' \leq n$.
(Similarly, we will come back to $k' = 1$ as a special case later.)
This is because, since none of the $[M']_{(2n+1)\mu}$ or $[M']_{(2n+2)\mu}$ were any of the necessarily-zero entries for $N+1 \leq \mu \leq N'$, there must still be at least $2(n-k'+1)$ zeros in each column of $M_R$ and at least one non-zero entry.
(In other words, $M_R$ is a ``$(n,k'-1)$'' instance of the base case.)
The rank of $M_R$, and hence also $M'$, is therefore at least $n-k'+2$.

For $n \geq 2$ and $1 \leq k \leq n-1$, the last step is to prove the special cases $k = 1$ and $k = n-1$.
When $k = 1$, up to swaps of its columns, $M$ is a pseudo-diagonal matrix of the form
\begin{equation}
    M = \left(
    \begin{array}{cccc}
        x_1 & 0 & \cdots & 0 \\
        y_1 & 0 & \cdots & 0 \\
        0 & x_2 & \cdots & 0 \\
        0 & y_2 & \ddots & 0 \\
        \vdots & \vdots & \ddots & \vdots \\
        0 & 0 & \cdots & x_n \\
        0 & 0 & \cdots & y_n
    \end{array}
    \right)
\end{equation}
in which $x_i$ and $y_i$ cannot both be zero for each $i = 1, \dots, n$, and so $\mathrm{rank} \, M = n$.
The case $k = n-1$ is a sort of inverted version of the $k=1$ case.
Here, the columns of $M$ can be arranged such that the pseudo-diagonal is necessarily zero, $[M]_{(2i-1)i} = [M]_{(2i)i} = 0$ for $1 \leq i \leq n$, i.e.,
\begin{equation}
    M = \left(
    \begin{array}{cccc}
        0 & M_{12} & \cdots & M_{1n} \\
        0 & M_{22} & \cdots & M_{2n} \\
        M_{31} & 0 & \cdots & M_{3n} \\
        M_{41} & 0 & \ddots & M_{4n} \\
        \vdots & \vdots & \ddots & \vdots \\
        M_{(2n-1)1} & M_{(2n-1)2} & \cdots & 0 \\
        M_{(2n)1} & M_{(2n)2} & \cdots & 0
    \end{array}
    \right)
\end{equation}
Given that any column must contain at least one non-zero entry, we see that $\mathrm{rank} \, M \geq 2$.
\end{proof}
We can now finally prove Lemma~\ref{lem:Rchoice}.~For any $n \geq 1$, the cases $k = 0$ and $k = n$ are trivial.~For $n \geq 2$, we prove the claim by contradiction.~Let us assume the opposite, namely, that for any $1 \leq k \leq n-1$ and any $A_\mu \in \mathcal{A}$, the set $\{\pi_{A_\mu} \mbf{u}(g)\}$ is an unfaithful representation of $\mathbb{Z}_2^{\times (n-k)}$ as vectors on $GF(2)^{2n}$.~Equivalently, for each $A_\mu$, there exists a vector $\mbf{u}(g_\mu) \neq 0$ such that $\pi_{A_\mu} \mbf{u}(g_\mu) = 0$.~Per Lemma~\ref{lem:technicalZ2}, the matrix constructed by concatenating column vectors $(\mbf{u}(g_{A_1}) ~ \mbf{u}(g_{A_2}) ~ \cdots ~ \mbf{u}(g_{A_N}))$ has rank greater than or equal to $n-k+1$.~Since each $\mbf{u}(g_\mu) \in \mathcal{G}$, this contradicts the fact that $\mathcal{G}$ has $n-k$ linearly independent vectors, which follows because $\mathbb{Z}_2^{\times(n-k)}$ has $n-k$ independent generators.
\end{proof}   

\noindent\textbf{Lemma~\ref{lemma:seedstate} (Existence of a seed state for local QRFs).}
    \emph{Let $\{P^0 = I^{\otimes m}, P^1, \dots, P^{2^m-1}\} \subset \mathcal{P}_m$, $m\geq1$, be a set of distinct Pauli strings with no non-unit phases that is closed under multiplication up to constants,~i.e., for all $i, j$, $P^i \neq P^j$ if $i \neq j$ and $P^i P^j \propto P^k$ for some $k$.~Then, there exists a state $\ket{\phi} \in (\mathbb{C}^2)^{\otimes m}$ such that $\bra{\phi} P^i \ket{\phi} = 0$ $\forall~i\neq 0$.}

\begin{proof}
    Let $\mathcal{W} =\Span\{P^i ~ | ~ i \neq 0\}$ denote the subspace of $\cB(\Hil)$, $\Hil\simeq(\mathbb{C}^2)^{\otimes m}$, spanned by the non-identity $P^i$'s.
Let us also equip the vector space of linear operator with the Frobenius inner product,
\begin{equation}
    \langle A, B \rangle = \Tr(A^\dagger B) \quad \forall \; A, B \in \cB(\Hil).
\end{equation}
Note that the elements of $\mathcal{P}_m$ therefore form an orthogonal basis for $\cB(\Hil)$
\begin{equation}
    \langle P^i, P^j \rangle = 2^m \delta_{ij}\,,
\end{equation}
where we implicitly let the index to run beyond $i = 2^m-1$ all the way to $i = 4^m-1$ to label all elements of the Pauli group.~The vanishing expectation value condition on $\ket{\phi}$ can then be rewritten as
\begin{equation}
    0 = \bra{\phi} P^i \ket{\phi} = \Tr(P^i \ketbra{\phi}{\phi}) = \langle P^i, \ketbra{\phi}{\phi} \rangle \quad i = 1, \dots, 2^m-1.
\end{equation}
If $\Phi \in \cB(\Hil)$ is a projector, then $\Phi^2 = \Phi$, and if it is rank-1, then $\Vert \Phi \Vert = 1$.~This latter point follows from the fact that the operator norm induced by the Frobenius inner product can also be written as the sum of the squares of the eigenvalues of the operator, i.e.
\begin{equation}
    \Vert A \Vert^2 = \langle A, A \rangle = \sum_{i} \lambda_i(A)^2.
\end{equation}
Therefore, we may equivalently state the objective of the lemma as follows.~We must show that there exists a unit-norm operator $\Phi$ in the orthogonal complement of $\mathcal{W}$ satisfying $\Phi^2 = \Phi$.

Let us look at the cases $m=1$ and $m=2$ first, and then consider the $m\geq 3$ case separately.~For $m=1$, $\mathcal W$ is spanned by a single Pauli operator $P^1\in\{X,Y,Z\}$.~The orthogonal complement $\mathcal W^{\perp}$ is thus spanned by the identity $I$ and the remaining two Pauli operators $P^2,P^3\neq P^1$.~$\Phi=\ket\phi\!\bra\phi$ with $\ket\phi$ any of the two eigenstates of $P^2$ or $P^3$ can then be taken as the desired projector.  

For $m=2$, $\mathcal W=\{A_1A_2,B_1B_2,(A_1A_2)(B_1B_2)\}$ with $A_i,B_i\in\{I,X,Y,Z\}$ and $A_i\neq B_i$ for at least one $i$, which amounts to the following four distinct possibilities:
\begin{enumerate}
\item $\mathcal W$ is generated by $IA$ and $IB$ with $A\neq B\neq I$;\footnote{For our purposes, there is no need to further distinguish between this case and $\mathcal W=\{AI, BI, (AB)I\}$ as it works analogously.}
\item $\mathcal W$ is generated by $IA$ and $BI$ with $A,B\neq I$;
\item $\mathcal W$ is generated by $IA$ and $B_1B_2$ with $A,B_1,B_2\neq I$ and $A\neq B_2$;
\item $\mathcal W$ is generated by $A_1A_2$ and $B_1B_2$ with $A_i\neq B_i\neq I$, $i=1,2$.
\end{enumerate}
In the first case, $\Phi=\ket\phi\!\bra\phi$ with $\ket\phi$ any of the four two-qubit Bell states will be such that $\Tr(P\Phi)=0$ for any $P\in\mathcal W$, thus providing us with the desired projector.~Cases 2.-4. are those for which it is possible to find an operator $\mathcal O=O_1O_2$ with $O_i\neq I$ such that, for each element of $\mathcal W$, there is at least one $i\in\{1,2\}$ for which $O_i$ anticommutes with the i-th qubit operator contanined in the element of $\mathcal W$.~Specifically, we have: $O_1\neq I, B$ (2 choices left), $O_2\neq I,A$ (2 choices left) for case 2;~$O_1\neq I, B_1$ (2 choices left), $O_2\neq I,A, B_2$ (1 choice left) for case 3;~and $\mathcal O=A_1B_2$ or $\mathcal O=A_2B_1$ for case 4.~In either of these cases, we can take $\Phi=\ket\phi\!\bra\phi$ with $\ket\phi$ given by the tensor product of eigenstates of $O_1$ and $O_2$. 

For $m\geq3$, we will use the following lemma.
\begin{lem}\label{lem:mgeq3}
    When $m \geq 3$, There exists a commuting subgroup of $\mathcal{P}_m$ generated by $\ell \geq 2$ non-identity generators that are contained in $\mathcal{W}^\perp$.
\end{lem}

\begin{proof}
We will show that there exist two distinct Pauli strings $Q^1, Q^2 \in \mathcal{W}^\perp$ such that $Q^1 Q^2 = Q^2 Q^1$ and $Q^1 Q^2 \in \mathcal{W}^\perp$.~Note that if $Q^1$ and $Q^2$ commute, then the product $Q^1 Q^2$ commutes with $Q^1$ and with $Q^2$ automatically.~Let now $Q^1 \in \mathcal{W}^\perp$, $Q^1 \neq I^{\otimes m}$ and let us count the number of operators $Q^2 \in \mathcal{W}^\perp$, $Q^2 \neq Q^1$ and $Q^2 \neq I^{\otimes m}$ such that $Q^1 Q^2 \in \mathcal{W}^\perp$.~We can count this number by process of elimination.~First, there are $4^m - (2^m - 1)$ operators in $\mathcal{W}^\perp$, and we must additionally subtract 1 for $Q^1$ and 1 for $I^{\otimes m}$.~Next, for every $P^i \in \mathcal{W}$, we must rule out the string $Q^1 P^i$, which we know must be in $\mathcal{W}^\perp$ since $\mathcal{W}$ is closed under multiplication by Paulis (i.e., if $Q^1 P^i \in \mathcal{W}$ and $P^i \in \mathcal{W}$, then $Q^1 = (Q^1 P^i) P^i \in \mathcal{W}$, in contradiction with our assumption).~There are $2^m-1$ choices of $P^i$, and so altogether, the number of possible choices for $Q^2$ such that $Q^1 Q^2 \in \mathcal{W}^\perp$ is
\begin{equation}
N_\mrm{closed} = (4^m - (2^m-1)) - (2^m - 1) - 1 - 1 = 4^m - 2^{m+1}.
\end{equation}
Next, the total number of Pauli strings that anticommute with $Q^1$ is
\begin{equation}
    N_\mrm{anticomm} = \tfrac{1}{2} 4^m.
\end{equation}
One way to see this is to remark that $N_\mrm{anticomm} = 4^m - N_\mrm{comm}$.
Any Pauli string with binary vector $(x \, | \, z)$ that commutes with $Q^1$ satisfies $(a_1 \, | \, b_1) \cdot (x \, | \, z) = 0$, which has a $(2m-1)$-dimensional solution space and hence $2^{(2m-1)} = \tfrac{1}{2}4^m$ distinct strings as solutions.
For $m \geq 3$, it follows that $N_\mrm{anticomm} < N_\mrm{closed}$, and every possible choice of $Q^2$ cannot anticommute with $Q^1$.
Therefore, there must exist a $Q^2 \in \mathcal{W}^\perp$ that commutes with $Q^1$ and such that $Q^1 Q^2 \in \mathcal{W}^\perp$.
\end{proof}

Per the above Lemma~\ref{lem:mgeq3}, let $\mathcal{V} = \Span\{Q^0, Q^1, \dots, Q^{2^\ell-1}\} \subseteq \mathcal{W}^\perp$ be the span of a set of commuting Pauli operators that are closed under multiplication, with $Q^0 = P^0 = I^{\otimes m}$ and $\ell \geq 3$.
As an ansatz for $\Phi$, write
\begin{equation} \label{eq:Phi_ansatz}
    \Phi = \sum_{i=0}^{2^\ell-1} c_i Q^i.
\end{equation}
We shall now show that there exists at least one non-trivial solution for the unknown coefficients $c_i$ such that $\Phi^2=\Phi$ and $\Vert \Phi \Vert = 1$.

The condition $\Vert \Phi \Vert = 1$ reads
\begin{equation}
    1 = \langle \Phi, \Phi \rangle = \sum_{i,j} \bar{c}_i c_j \langle Q^i, Q^j \rangle = \sum_{i,j} \bar{c}_i c_j 2^m \delta_{ij} = 2^m \sum_i |c_i|^2\,,
\end{equation}
so we must have
\begin{equation}
    \sum_{i=0}^{2^\ell-1} |c_i|^2 = \frac{1}{2^m}\,.
\end{equation}
Since $\Tr \Phi = 1$, we also conclude that
\begin{equation}
    1 = \sum_i c_i \Tr Q^i = 2^m c_0\;,
\end{equation}
and so
\begin{equation}
    c_0 = \frac{1}{2^m}\,.
\end{equation}
Next, let us write out $\Phi^2$
\begin{equation}
    \Phi^2 = \sum_{i,j=0}^{2^\ell-1} c_i c_j Q^i Q^j = \sum_{i=0}^{2^\ell-1} c_i^2 I^{\otimes m} + \sum_{i \neq j} c_i c_j Q^i Q^j\;.
\end{equation}
For simplicity, choose the $c_i$ to be real to write
\begin{align}
    \Phi^2 &= \frac{1}{2^m} I^{\otimes m} +  \sum_{i \neq j} c_i c_j Q^i Q^j \\
    &= \frac{1}{2^m} I^{\otimes m} + 2 \sum_{0<j} \frac{1}{2^m}c_j Q^j + \sum_{0\neq i \neq j} c_i c_j Q^i Q^j\;,
\end{align}
where, in the last line, we separated out the case $Q^i = I^{\otimes m}$ from the sum.~As a further simplication, let us also suppose that $c_1 = c_2 = \dots = c_{2^\ell-1} = c$ are all equal, so that we arrive at
\begin{equation}
    \Phi^2 = \frac{1}{2^m} I^{\otimes m} + \frac{2c}{2^m} \sum_{0<j} Q^j + c^2 \sum_{0\neq i\neq j} Q^i Q^j.
\end{equation}
This must be equal to $\Phi$, as given in line \eqref{eq:Phi_ansatz}.~To deduce the value of $c$, compare the inner products $\langle Q^k, \Phi\rangle$ and $\langle Q^k, \Phi^2\rangle$ for $k \neq 0$:
\begin{equation}
    \langle Q^k, \Phi\rangle = \langle Q^k, \Phi^2\rangle \quad \Rightarrow \quad 2^m c = 2 c + c^2 \sum_{0 \neq i \neq j} \langle Q^k, Q^i Q^j \rangle
\end{equation}
The sum counts the number of ordered pairs $Q^i, Q^j$ such that $Q^i Q^j = Q^k$, with $i, j, k \neq 0$.
To count these pairs, notice that we can write
\begin{equation}
    Q^k = Q^i (Q^i Q^k) \equiv Q^i Q^j.
\end{equation}
For fixed $Q^k$, there are thus $2^\ell-2$ choices of $Q^i$ that result in a distinct, non-identity operator $Q^j = Q^i Q^k$.~The number of ordered pairs is therefore $2^\ell - 2$.~Altogether, we have that
\begin{equation}
    2^m c = 2c + c^2 2^m (2^\ell - 2)\;,
\end{equation}
from which we readily obtain
\begin{equation}
    c = 0, ~ \frac{1 - 2^{1-m}}{2^\ell-2}.
\end{equation}
The nonzero solution for $c$ yields the desired seed state thus completing our proof.
\end{proof}

\noindent\textbf{Lemma~\ref{lem_none}.}
    \emph{Let $G$ be represented in a tensor product representation $g\mapsto U_R^g\otimes U_S^g$ on $\mathcal{H}_{\rm kin}=\mathcal{H}_R\otimes\mathcal{H}_S$ with $g\mapsto U^g_R$ faithful. An operator $E\in\cB(\mathcal{H}_{\rm kin})$ such that 
\begin{equation}\label{er}
        E\,\Pi_{\rm pn}=e_R\otimes I_S
\end{equation}
    for some nonvanishing operator $e_R\in\cB(\mathcal{H}_{R})$ exists if only if $U_S^g=I_S$ for all $g\in G$.}
\begin{proof}
Writing an arbitrary $E\in\cB(\mathcal{H}_{\rm kin})$  as
\begin{equation}
    E=\sum_{i,g,g'}c_{g,g',i}\ket{g}\!\bra{g'}\otimes e_i\,,
\end{equation}
for some complex coefficients $c_{g,g',i}$ and some operator basis $e_i\in\cB(\mathcal{H}_S)$, we find (upon a sum label redefinition)
\begin{equation}
    E\,\Pi_{\rm pn}=\sum_{i,g,g',g''}c_{g,g''g',i}\ket{g}\!\bra{g'}\otimes e_i\,U_S^{g''}\,.
\end{equation}
Since $\ket{g}\!\bra{g'}$ constitutes a basis for $\cB(\mathcal{H}_R)$, Eq.~\eqref{er} can only be satisfied, provided
\begin{equation}\label{er2}
    \sum_{i,g''}c_{g,g''g',i}\,e_i\,U_S^{g''}=k_{g,g'}I_S\,,\qquad\forall\,g,g'\in G\,,
\end{equation}
for some complex numbers $k_{g,g'}$. Now pick any $g,g'$ such that $k_{g,g'}\neq0$ and any $g*\in G$ such that $U_S^{g*}\neq I_S$ and multiply both sides of the previous equation with it from the right.~The l.h.s.\ yields (upon a sum label redefinition)
\begin{equation}
\sum_{i,g''}c_{g,g''g*^{-1}g',i}e_i\,U_S^{g''}\underset{\eqref{er2}}{=}k_{g,g*^{-1}g'}I_S\,.
\end{equation}
The r.h.s., on the other hand, gives 
\begin{equation}
k_{g,g'}U_S^{g^*}\not\propto I_S\,.
\end{equation}
We thus have a contradiction when $U_S^g\neq I_S$ for at least one $g\in G$.

By contrast, when $U_S^g=I_S$ for all $g\in G$, then a choice $c_{g,g',i}=\tilde c_{g,g'}\tilde c_i$ with $\sum_i \tilde{c}_i\,e_i=I_S$ produces an operator $E$ with the desired properties. 
\end{proof}

\begin{lem}[\textbf{Detection is equal to incoherent group averaging}]\label{lem_det}
    The syndrome detection operation for Pauli errors in $[[n,k]]$ stabilizer codes with faithful representation of $\mathcal{S}$ coincides with the \emph{incoherent group average} over the stabilizer group, i.e.\ its so-called $G$-twirl:
    \begin{equation}
        \mathcal{M}_{\rm syn}(\rho)=\sum_{\chi\in\hat{G}}\,P_\chi\,\rho\,P_\chi = \frac{1}{|G|}\sum_{g\in G} U^g\,\rho\,U^g=\mathcal{G}(\rho)\,,
    \end{equation}
    where $P_\chi$ is the orthogonal projector onto the error space $\Hil_\chi$. 
\end{lem}

\begin{proof}
    The projector is given by Eq.~\eqref{eq:projectors_on_isotypes}.
 The set $\{P_\chi\}_{\chi\in\hat{G}}$ clearly defines the elements of a trace-preserving operation because $\sum_{\chi\in\hat{G}}P_\chi=I$. Similarly, the $G$-twirl defines a trace-preserving operation because $\sum_g \tilde{E}_g^\dag\,\tilde{E}_g=I$, where $\tilde{E}_g=\frac{1}{\sqrt{2^{n-k}}} U^g$. Now, 
    \begin{equation}
     \sum_{\chi}P_\chi \rho P_\chi = \frac{1}{2^{2(n-k)}} \sum_{g_1 , g_2} \underbrace{\sum_\chi \bar{\chi}(g_1) \chi (g_2)}_{2^{n-k} \delta_{g_1, g_2}} U^{g_1}  \rho U^{g_2} = \frac{1}{2^{n-k}} \sum_g U^{g} \rho U^{g}\,.
    \end{equation}
\end{proof}

\noindent\textbf{Proposition~\ref{prop:general_recovery}} \textbf{(Action of local disentangler on error states).} \emph{Let $\Hil_\mrm{kin} = \Hil_R \otimes \Hil_S$ define a local QRF-system partition for a $[[n,k]]$ Pauli stabilizer code with $\Hil_R \simeq (\mathbb{C}^2)^{\otimes(n-k)}$ consisting of $n-k$ physical qubits and such that $G = \mathbb{Z}_2^{\times(n-k)}$ is represented with the code's stabilizers $g \mapsto U_R^g \otimes U_S^g$, where $g \mapsto U_R^g$ is faithful (but possibly projective), $g \mapsto U_S^g$ is non-trivial, and $\braket{g}{h}_R = \delta_{g,h}$ (i.e. $R$ is ideal). Let $\mathcal{E} = \{E_1, \dots, E_m\}$ be a set of correctable Pauli errors.
Then, for any any $\ket{\bar{\psi}} = T_R^\dagger(\ket{1}_R\otimes\ket{\psi}_S) \in \Hil_\mrm{pn}$, it follows that}
\begin{equation}
    T_R E_i \ket{\bar{\psi}} = \ket{w(E_i)}_R \otimes L_S(E_i)\ket{\psi}_S,
\end{equation}
\emph{for some collection of orthonormal states $\ket{w(E_i)}_R$ and unitary operators $L_S(E_i)$.}

\begin{proof}
First, from the definitions of $T_R$ in Eq.~\eqref{triv} and $\ket{1}_R$ in Eq.~\eqref{eq:R_ready}, and since the QRF is ideal, we observe that
\begin{equation}
    \ket{\bar{\psi}} = \frac{1}{\sqrt{|\mathcal{G}|}} \sum_{g \in \mathcal{G}} \ket{g}_R \otimes U_S^{g} \ket{\psi}_S .
\end{equation}
Let us also assume without loss of generality that $\mathcal{E}$ is a maximal set of correctable Pauli errors, since it can always be extended to one if needed.
Relabeling $i \rightarrow \chi$ for notational homogeneity, each Pauli string $E_\chi \in \mathcal{E}$ belongs to a distinct $C_\chi \subset \mathcal{B}_\chi$ (hence, a distinct isotype) of $\Hil_\mrm{kin}$, and there are $2^{n-k}$ such strings in $\mathcal{E}$ (see \App{app:stabilizer_reps}, in particular Lem.~\ref{lem:decomp_algebra} and Prop.~\ref{prop:decomp_Pauli}).
As an element of $\cB_\chi$, we have that
\begin{equation}\label{eq:error_in_isotype}
    U^g E_\chi U^g = \chi(g) E_\chi .
\end{equation}
(Recall that, as strings of Pauli operators with a prefactor of $\pm 1$, the stabilizers satisfy $(U^g)^\dagger = U^g$.)

Next, we manipulate $T_R E_\chi \ket{\bar\psi}$ in the following way:
\begin{align}
    T_R E_\chi \ket{\bar{\psi}} &= \left(\sum_h \ketbra{h}{h}_R \otimes (U_S^h)^\dagger \right) E_\chi \left( \frac{1}{\sqrt{|\mathcal{G}|}}  \sum_g \ket{g}_R \otimes U_S^g \ket{\psi}_S \right) \\
    &= \frac{1}{\sqrt{|\mathcal{G}|}} \sum_{h} \ket{h}_R \otimes \left( \sum_g (\bra{h}_R\otimes (U_S^h)^\dagger) E_\chi (\ket{g}_R \otimes U_S^g) \right) \ket{\psi}_S \label{eq:TEstep1} \\
    &= \frac{1}{\sqrt{|\mathcal{G}|}} \sum_{h} \ket{h}_R \otimes \left( \sum_g \bra{e}_R((U_R^h)^\dagger \otimes (U_S^h)^\dagger) E_\chi (U_R^g \otimes U_S^g) \ket{e}_R \right) \ket{\psi}_S \label{eq:TEstep2} \\
    &= \frac{1}{\sqrt{|\mathcal{G}|}} \sum_{h} \ket{h}_R \otimes \left[ \bra{e}_R \left( \sum_g (U^h)^\dagger E_\chi U^g \right) \ket{e}_R \right] \ket{\psi}_S
\end{align}
To go from line \eqref{eq:TEstep1} to \eqref{eq:TEstep2}, we used the fact that $\ket{g}_R = U_R^g \ket{e}_R$.~Next, we observe that $(U^h)^\dagger = U^h$, we insert $I = (U^h)^2$ between $E$ and $U^g$, we and invoke Eq.~\eqref{eq:error_in_isotype} to write
\begin{align}
    T_R E_\chi \ket{\bar{\psi}} &= \frac{1}{\sqrt{|\mathcal{G}|}} \sum_{h} \ket{h}_R \otimes \left[ \bra{e}_R \left( \sum_g (U^h E_\chi U^h)(U^h U^g) \right) \ket{e}_R \right] \ket{\psi}_S \\
    &= \frac{1}{\sqrt{|\mathcal{G}|}} \sum_{h} \ket{h}_R \otimes \left[ \bra{e}_R \chi(h) E_\chi \left( \sum_g U^{hg} \right) \ket{e}_R \right] \ket{\psi}_S \\
    &= \left( \frac{1}{\sqrt{|\mathcal{G}|}} \sum_{h} \chi(h) \ket{h}_R \right) \otimes (|\mathcal{G}| \bra{e}_R E_\chi \Pi_{\mrm{pn}} \ket{e}_R ) \ket{\psi}_S
\end{align}
The first bracketed term is the inverse group-valued Fourier transform (cf. Eq.~\eqref{Fourier}), which we denote here by
\begin{equation}
    \ket{w(E_\chi)}_R := \frac{1}{\sqrt{|\mathcal{G}|}} \sum_{h} \chi(h) \ket{h}_R,
\end{equation}
and let us denote the second bracketed term by
\begin{equation}
    L_S(E_\chi) := |\mathcal{G}| \bra{e}_R E_\chi \Pi_{\mrm{pn}} \ket{e}_R .
\end{equation}
so that
\begin{equation} \label{eq:TE}
T_R E_\chi \ket{\bar{\psi}} = \ket{w(E_\chi)}_R \otimes L_S(E_\chi)\ket{\psi}_S.
\end{equation}

To complete the proof, we are required to show that the $\ket{w(E_\chi)}_R$ are orthonormal and that the $L_S(E_\chi)$ are unitary.
Orthonormality follows from properties of characters:
\begin{align*}
    \braket{w(E_\chi)}{w(E_{\chi'})}_R &= \frac{1}{|\mathcal{G}|} \sum_{g,h} \chi^*(g) \chi'(h) \braket{g}{h}_R \\
    &= \frac{1}{|\mathcal{G}|} \sum_{g} \chi^*(g) \chi'(g) \\
    &= \delta_{\chi,\chi'}
\end{align*}
The last line follows because the characters of the irreps of a finite group are orthogonal. 
(Strictly speaking the conjugation is redundant here, because $\chi(g) \in \{+1, -1\}$.)

We can demonstrate unitary by brute force computation.
We first remark that, as strings of Pauli operators, we can write $E_\chi = (E_\chi)_R \otimes (E_\chi)_S$, where $(E_\chi)_R$ and $(E_\chi)_S$ are shorter substrings that square to the identity.
We then proceed with the following manipulations:
\begin{align}
    L_S(E_\chi)^\dagger L_S(E_\chi) &= |\mathcal{G}|^2 \bra{e}_R \Pi_\mrm{pn} E_\chi (\ketbra{e}{e}_R \otimes I_S) E_\chi \Pi_\mrm{pn} \ket{e}_R \\
    &= \sum_{h,g} \bra{e}_R (U_R^h \otimes U_S^h)^\dagger((E_\chi)_R \otimes (E_\chi)_S)  (\ketbra{e}{e}_R \otimes I_S) ((E_\chi)_R \otimes (E_\chi)_S) (U_R^g \otimes U_S^g) \ket{e}_R \\
    &= \sum_{h,g} \bra{h} (E_\chi)_R \ketbra{e}{e} (E_\chi)_R \ket{g} (U_S^h)^\dagger U_S^g
\end{align}
To go to the second line we also used the definition of $\Pi_\mrm{pn}$ and the fact that $\Pi_\mrm{pn}^\dagger = \Pi_\mrm{pn}$.
Next we swap the two matrix elements and combine $(U_S^h)^\dagger = U_S^h$ and $U_S^g$ according to Eq.~\eqref{eq:phasedefinitionS}:\footnote{$(U_R^h)^\dagger = U_R^h$ and $(U_S^h)^\dagger = U_S^h$, since these are strings of Pauli operators and, by definition, the $U_R^h$ have a trivial prefactor while the $U_S^h$ can have a prefactor of $+1$ or $-1$. See Eq.~\eqref{eq:URg_generalstab}.}
\begin{align}
    L_S(E_\chi)^\dagger L_S(E_\chi) &= \sum_{h,g} \bra{e} (E_\chi)_R \ketbra{g}{h} (E_\chi)_R \ket{e} c^*(h,g) U_S^{hg} \\
    &= \sum_{g, g'} \bra{e} (E_\chi)_R \ketbra{g}{g'g} (E_\chi)_R \ket{e} c^*(gg',g) U_S^{g'}
\end{align}
The last step follows from re-indexing the sum with $h = g'g~= gg'$.
From the definition of $c(g,g')$, it follows that
\begin{equation}
    c(gg',g) = c(g',g),
\end{equation}
and so in particular
\begin{align}
    \bra{g'g}_R c^*(gg',g) &= \left(c(gg',g) \ket{g'g}_R \right)^\dagger \\
    &= \left( c(g',g) U^{g'g}_R \ket{e}_R \right)^\dagger \\
    &= \left( U^{g'}_R U^g_R \ket{e}_R \right)^\dagger \\
    &= \bra{g}_R (U^{g'}_R)^\dagger \\
    &= \bra{g}_R U^{g'}_R,
\end{align}
whence
\begin{align}
    L_S(E_\chi)^\dagger L_S(E_\chi) &= \sum_{g'} \bra{e} (E_\chi)_R \left(\sum_g \ketbra{g}{g}_R \right) U_R^{g'} (E_\chi)_R \ket{e} U_S^{g'} \\
    &= \sum_{g'} \bra{e} (E_\chi)_R U_R^{g'} (E_\chi)_R \ket{e} U_S^{g'}.
\end{align}
Because $(E_\chi)_R$ and $U_R^{g'}$ are both strings of Pauli operators, they either commute or anticommute, and so let us write $U_R^{g'} (E_\chi)_R = (-1)^{f(g',\chi)} (E_\chi)_R U_R^{g'}$, where $f \in \{+1,-1\}$.
It then follows that
\begin{align}
    L_S(E_\chi)^\dagger L_S(E_\chi) &= \sum_{g'} (-1)^{f(g',\chi)} \bra{e} U_R^{g'} \ket{e} U_S^{g'} \\
    &= \sum_{g'} (-1)^{f(g',\chi)} \braket{e}{g'}_R U_S^{g'} \\
    &= (-1)^{f(e,\chi)} U_S^e \\
    &= I_S,
\end{align}
where the last line follows because $U_R^e = I^{\otimes(n-k)}$ and $U_S^e = I^{\otimes k}$ by assumption.
The computation of $L_S(E) L_S(E)^\dagger$ is analogous.
\end{proof}

\noindent\textbf{Lemma~\ref{lem_gaugedetect}} (\textbf{Recovery schemes for gauge fixing errors})\textbf{.}
    \emph{Consider a gauge fixing error set $\mathcal{E}=\{\hat{E}_g\}$ for some $[[n,k]]$ stabilizer code subject to a faithful representation of ${G}=\mathbb{Z}_2^{n-k}$. The error-correction operation $\mathcal{O}:\mathcal{B}(\Hil_{\rm kin})\rightarrow\mathcal{B}(\Hil_{\rm pn})$ is given by 
    \begin{equation}
        \mathcal{O}(\rho)=\sum_g\,\hat{E}_g^\dag\,\hat{P}_g\,\rho\,\hat{P}_g\,\hat{E}_g=\sqrt{2^{n-k}}\,\Pi_{\rm pn}\left(\sum_g\,\hat{P}_g\,\rho\,\hat{P}_g\right)\,\Pi_{\rm pn}\,.
    \end{equation}
Further, $\mathcal{O}(\rho)=\rho$ for $\rho\in\mathcal{S}(\Hil_{\rm pn})$, so code words are left intact by the correction if no error occurred.
}

\begin{proof}
The error correction operation elements $\hat{E}_g^\dag\,\hat{P}_g$ are  normalized: 
\begin{equation}
    \sum_{g\in G} \hat{P}_g\,\hat{E}_g\,\hat{E}_g^\dag\,\hat{P}_g=I\,.
\end{equation}
Moreover, setting $\tilde{E}_k=\sum_g c_{g,k}\hat{E}_g$ and thus $\hat{P}_g\tilde{E}_k\,\Pi_{\rm pn}=\sqrt{2^{n-k}}\,c_{g,k}\hat{P}_g\,\Pi_{\rm pn}$, for $\rho_{\rm pn}\in\mathcal{S}(\mathcal{H}_{\rm pn})$ we have
\begin{eqnarray}
\mathcal{O}\left(\tilde{\mathcal{E}}(\rho_{\rm pn})\right)&=&\sum_{g,k}E^\dag_{g}\,\hat{P}_{g}\,\tilde{E}_k\,\rho_{\rm pn}\,\tilde{E}_k^\dag\,\hat{P}_{g}\,\hat{E}_{g}=2^{n-k}\sum_{g,k} |c_{g,k}|^2\,\hat{E}_g^\dag\,\hat{P}_g\,\rho_{\rm pn}\,\hat{P}_g\,\hat{E}_g\nonumber\\
&\underset{\eqref{eq:unitary gauge fixing}}{=}&\sum_{g,k}|c_{g,k}|^2\,\rho_{\rm pn}=\rho_{\rm pn}\,.
\end{eqnarray}
The last equality follows from $\tilde{E}_k=\sum_g c_{g,k}\hat{E}_g$ and
\begin{equation}
  I=\sum_k\tilde{E}_k^\dag\tilde{E}_k=2^{n-k}\sum_{k,g,\chi}|c_{g,k}|^2\,E_\chi\Pi_{\rm pn}\hat{P}_{g_\chi}\Pi_{\rm pn}E_\chi=\sum_{k,g,\chi}|c_{g,k}|^2\,E_\chi\Pi_{\rm pn}E_\chi  =\sum_{k,g,\chi}|c_{g,k}|^2\,P_\chi=\sum_{k,g}|c_{g,k}|^2\,.\nonumber
\end{equation}
In particular, $\mathcal{O}(\rho_{\rm pn})=\rho_{\rm pn}$, so code words are unaffected by the operation if no error occurred.
\end{proof} 

\clearpage

\bibliographystyle{utphys-modified}
\bibliography{qrf-qecc.bib}

\end{document}